\documentclass[a4paper,11pt,oneside]{amsart}
\usepackage[utf8]{inputenc}
\usepackage{amssymb,amsmath,amsthm}
\usepackage[foot]{amsaddr}
\usepackage{graphicx}
\usepackage[dvipsnames]{xcolor}
\usepackage[all,cmtip]{xy}
\usepackage{booktabs}
\usepackage[numbers,sort&compress]{natbib}
\usepackage{mathrsfs}
\usepackage{microtype}
\usepackage{tikz}
\usetikzlibrary{decorations.markings}

\usepackage{url}
\usepackage[colorlinks=true,linkcolor=red,citecolor=blue]{hyperref}
\usepackage{cleveref}

\newtheorem{thm}{Theorem}[section]

\newtheorem{lem}[thm]{Lemma}
\newtheorem{prop}[thm]{Proposition}
\newtheorem{cor}[thm]{Corollary}
\theoremstyle{definition}
\newtheorem{rem}[thm]{Remark}
\newtheorem{defn}[thm]{Definition}
\newtheorem{eg}[thm]{Example}
\numberwithin{equation}{section}
\newcommand{\Li}[1]{\operatorname{Li}_{#1}}
\newcommand{\set}[1]{\{#1\}}
\newcommand{\abs}[1]{\lvert #1 \rvert}
\newcommand{\norm}[1]{\lVert #1 \rVert}
\newcommand{\iu}{\mathrm{i}}
\newcommand{\ipi}{\iu \pi}
\newcommand{\td}[1][]{\mathrm{d}^{#1}}
\DeclareMathOperator{\codim}{codim}
\DeclareMathOperator{\Res}{Res}
\DeclareMathOperator{\Var}{Var}
\DeclareMathOperator{\var}{var}
\newcommand{\tvar}{\operatorname{\widetilde{\var}}} % var with tilde
\DeclareMathOperator{\End}{End}
\DeclareMathOperator{\Aut}{Aut}
\DeclareMathOperator{\Hom}{Hom}
\DeclareMathOperator{\sgn}{sgn}
\newcommand{\is}[2]{\langle {#1} \,\vert\, {#2} \rangle}
\newcommand{\dR}{\mathsf{dR}}

\newcommand{\inv}[1]{ {#1}^{-1} }
\newcommand{\mb}[1]{\mathbb{#1}}
\newcommand{\ti}[1]{\widetilde{#1}}
\newcommand{\rk}[1]{\mathrm{rank}({#1}) }
\newcommand{\rest}[1]{|_{#1}}

\newcommand{\RR}{\mathbb{R}}
\newcommand{\CC}{\mathbb{C}}
\newcommand{\ZZ}{\mathbb{Z}}
\newcommand{\QQ}{\mathbb{Q}}
\newcommand{\PP}{\mathbb{P}}
\newcommand{\Upol}{\mathcal{U}}
\newcommand{\Fpol}{\mathcal{F}}
\newcommand{\II}{\mathcal{I}}
\newcommand{\Transpose}{\intercal}
\newcommand{\defas}{\mathrel{\mathop:}=}
\newcommand{\Sphere}{\mathbb{S}} % (real) unit sphere
\newcommand{\oBall}{\mathbb{B}} % open ball
\newcommand{\cBall}{\overline{\oBall}} % closed ball
\newcommand{\oBallD}{\mathbb{D}} % open ball
\newcommand{\cBallD}{\overline{\oBallD}} % closed ball
\newcommand{\setcompl}{\mathsf{c}} % set complement
\newcommand{\fibSphere}{\delta}
\newcommand{\PLm}{m} % codimension of critical stratum
\DeclareMathOperator{\id}{id}
\newcommand{\irrone}[1]{\mathfrak{#1}}
\newcommand{\asyO}[1]{\mathcal{O}(#1)}
\newcommand{\vcell}{\zeta}

\newcommand{\vcyc}{\nu}
\newcommand{\dvcyc}{\widetilde{\vcyc}}
\newcommand{\conc}{\star} % path concatenation
\newcommand{\CV}{\leqslant}
\newcommand{\nCV}{\nleqslant}
\newcommand{\nMil}{\mu} % Milnor number
\newcommand{\simple}[1]{#1_{\mathsf{s}}}

\newcommand{\orcidicon}[1]{\href{https://orcid.org/#1}{\includegraphics[height=2.5ex]{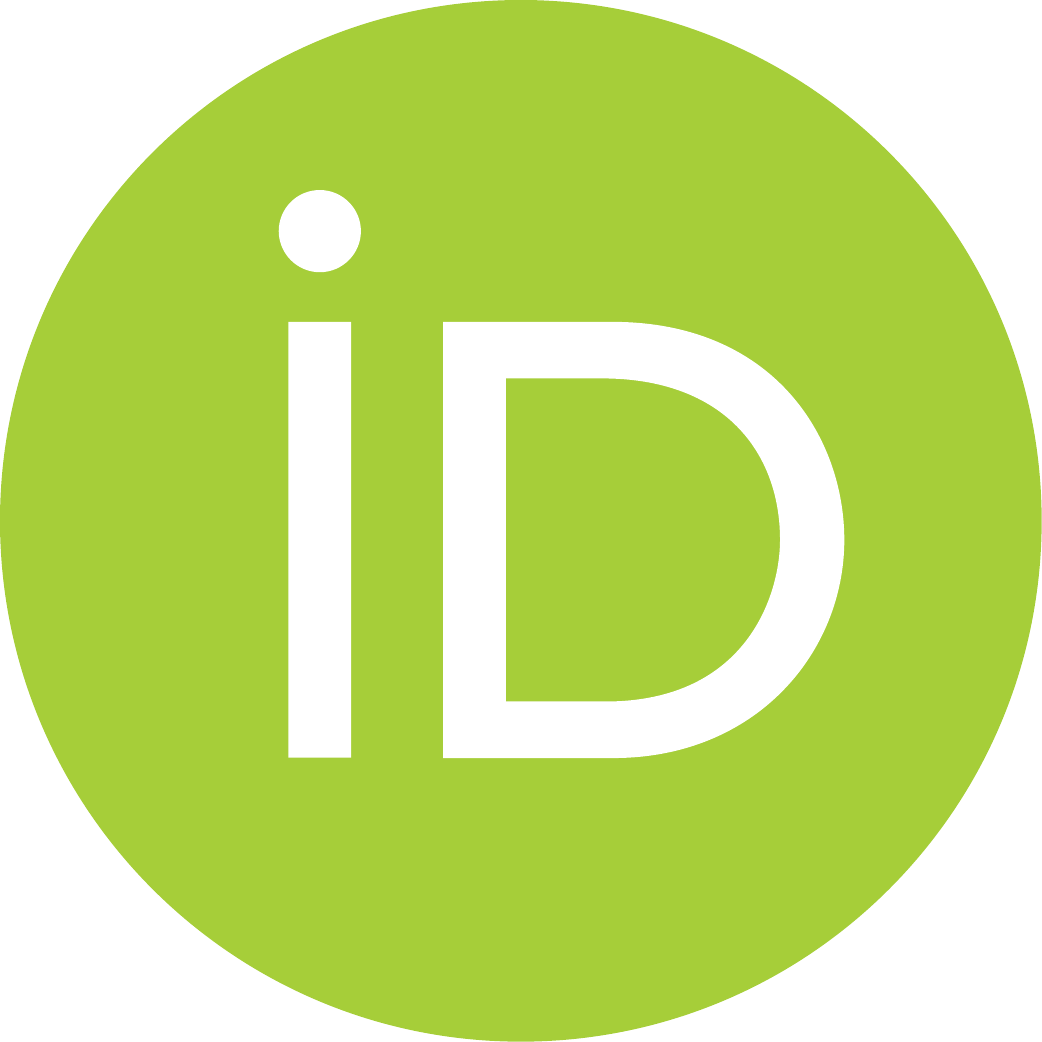}}}
\newcommand{\setm}{\setminus}

\title{Hierarchies in relative Picard-Lefschetz theory}
\author{Marko Berghoff$^{(1,2)}$ \protect\orcidicon{0000-0002-9108-3045}}
\address[1]{Mathematical Institute, University of Oxford, Oxford OX2 6GG, UK}
\address[2]{(now at) Institut f\"ur Mathematik, Humboldt Universit\"{a}t zu Berlin, Rudower Chaussee 25, 12489 Berlin, Germany}
\email{berghoffmj@gmail.com}
\author{Erik Panzer$^{(1)}$ \protect\orcidicon{0000-0002-9897-5812}}
\email{erik.panzer@maths.ox.ac.uk}

\begin{document}

\begin{abstract}
We prove a relative version of the Picard-Lefschetz theorem, describing the variation of relative homology groups $H_d(Y_t \setm A_t,B_t\setm A_t)$ in the fibers of a smooth fiber bundle $Y \to T$ of complex manifolds with $A\cup B \subset Y$ transverse. From this we derive the vanishing of certain iterated variations, a system of constraints dubbed ``hierarchy''.

As applications, we rederive the known analytic structure of Aomoto polylogarithms and massive one loop Feynman integrals. Moreover, we introduce the ``simple type'' to prove hierarchy constraints in degenerate cases where the Picard-Lefschetz formula does not apply, e.g.\ the massless triangle or the ice cream cone Feynman diagram.
We compare our findings with a ``classical" hierarchy of iterated variations (from 1960's $S$-matrix theory) and show how our setup not only explains, but also refines the latter. In order to do so, we need to further resolve the geometry of Feynman motives: We boldly blow up what no one has blown up before.   
\end{abstract}

\date{\today}
\maketitle

\tableofcontents

\section{Introduction}

\subsection{Monodromy}\label{sec:intro-mon}
A holomorphic surjection $\pi\colon Y \rightarrow T$ of connected complex manifolds defines a family $Y_t = \pi^{-1}(t)$ of analytic varieties that depend on a parameter $t\in T$. The set $L\subset T$ of critical values has measure zero, %\cite{Sard:MeasureCrit},
and the smooth fibres $Y_t$ over the complement of $L$ glue into a smooth fibre bundle if $\pi$ is proper. % \cite{Ehresmann:ConnFibre}.
Such local trivializations $\pi^{-1}(U)\cong Y_{t_0}\times U$ over $t_0\in U \subseteq T \setm L$ show that the homology groups $H_d(Y_t;\ZZ)$ of the fibres
% form a local system of Abelian groups
are locally constant over $T\setm L$.
The monodromy representation
\begin{equation*}
    \pi_1(T\setm L,t_0) \longrightarrow \Aut H_d (Y_{t_0};\ZZ)
\end{equation*}
is induced by parallel transport along paths. It is an important tool to study periods of algebraic varieties and variations of Hodge structures \cite{Griffiths:PeriodsSummary}.

More generally, two analytic subvarieties $A,B\subset Y$ give rise to a family of pairs $B_t \setm A_t \subset Y_t \setm A_t$ where $A_t = A\cap Y_t$ and $B_t=B\cap Y_t$.
% $A_t = A\cap \pi^{-1}(t)$ and $B_t = B\cap \pi^{-1}(t)$.
We abbreviate the associated relative homology groups as
\begin{equation*}
    H_d(Y \setm A, B)_t = H_d(Y_t \setm A_t,B_t \setm A_t; \ZZ).
\end{equation*}
To describe their monodromy, we seek:
\begin{enumerate}
    \item[\textbf{(L)}]
    a subset $L \subset T$ such that the pairs $(Y_t \setm A_t,B_t \setm A_t)$ vary as a pair of topological fibre bundles over $T \setm L$, and
    \item[\textbf{(V)}]
    the representation $\rho\colon \pi_1(T \setm L,t_0) \longrightarrow \Aut H_d(Y \setm A, B)_{t_0} $.
\end{enumerate}
The \emph{Landau variety} of the pair $(Y\setm A,B\setm A)$ with respect to $\pi$ is a closed analytic subvariety $L$ of $T$ with property \textbf{(L)}. We will describe the monodromy $\gamma_*=\rho(\gamma)$ along a loop $\gamma\in\pi_1(T\setm L,t_0)$ in terms of the \emph{variation}
\begin{equation*}
    \Var_{\gamma} = \gamma_* - \id \in \End H_d(Y\setm A,B)_{t_0}.
\end{equation*}

\subsection{Parameter integrals}\label{sec:intro-integrals}
The monodromy data encodes the analytic continuation \textbf{(V)}, and a bound \textbf{(L)} on the singularities, of parameter integrals
\begin{equation*}%\label{eq:function}
    \II(t) = \int_{\sigma_t} \omega_t %\omega|_{Y_t-A_t}
\end{equation*}
of integrands with poles on $A$, over domains with boundary in $B$. Let
\begin{itemize}
    \item $n$ the complex dimension of the fibres $Y_t$,
    \item $\omega\in\Omega^n(Y\setm A)$ a holomorphic differential form,
    \item $\omega_t=\omega|_{Y_t\setm A_t}$ its restriction to a fibre,
    \item $\sigma_{t_0}$ an $n$-chain in $Y_{t_0}\setm A_{t_0}$ with boundary in $B_{t_0}\setm A_{t_0}$.
\end{itemize}
The class $[\sigma_{t_0}]$ extends to a multivalued section 
$[\sigma_t]\in H_n(Y\setm A,B)_t$ over $T\setm L$, representable locally by a continuous family $\sigma_t$ of chains.
Since $\omega_t$ is smooth on the compact support of $\sigma_t$, the integrals $\II(t)$ converge and define a multivalued holomorphic function on $T\setm L$. As restrictions of a global form $\omega$, the integrands $\omega_t$ are single-valued. Analytic continuation $\gamma \cdot \II$ along a loop $\gamma\in\pi_1(T\setm L,t)$ is thus completely determined by the monodromy:
\begin{equation*}
    \gamma \cdot \II(t) = \int_{\gamma_* [\sigma_{t}]} \omega_{t}
    = \II(t) + \int_{\Var_\gamma [\sigma_{t}]} \omega_{t}.
\end{equation*}

By Riemann's extension theorem, %\cite[\S 7]{GrauertRemmert:CAS},
$\omega$ extends holomorphically over any irreducible components of $A$ with complex codimension $\geq 2$. We hence drop all such components and assume that $A$ is pure of codimension one. Locally,
\begin{equation*}
    \omega = \sum_{i_1<\cdots<i_n} \frac{P_{i_1\!\cdots i_n}(z)}{Q(z)}\ \td z_{i_1} \wedge \ldots \wedge \td z_{i_n}
\end{equation*}
can be written with holomorphic functions $P$ and $Q$ such that $A$ is the vanishing locus $\set{Q=0}$ of the denominator.

Similarly, the singularities of $\II$ are covered already by only those components $\ell_i \subseteq L$ with complex codimension one.
Items \textbf{(L)} and \textbf{(V)} above translate into finding an upper bound $L$ on the singularities $\bigcup_i \ell_i$ of $\II$, and studying the behaviour of $\II$ near each $\ell_i$. This purely homological approach is blind to the specific choice of the integrand $\omega$, and some integrands produce integrals with fewer singularities (consider e.g.\ $\omega=0$). We are not concerned with identifying precisely which components $\ell_i \subseteq L$ of the Landau variety are genuine singularities of the integral $\II$ for a given $\omega$---instead, the homological approach provides insights that hold for all integrands.

\subsection{Picard-Lefschetz theory}
Let
\begin{equation*}
    \irrone{A} = \{ A_1, A_2, \ldots \}
    \quad\text{and}\quad
    \irrone{B} = \{ B_1, B_2, \ldots \}
\end{equation*}
denote the finite sets of irreducible components of $A$ and $B$. We assume that these are smooth hypersurfaces with transverse intersections. 
Suppose that $\ell\subseteq L$ is a component such that $A\cup B$ has a unique, non-degenerate critical point $p(t)$ over $t\in \ell$. Let $A_1,\ldots,A_i,B_1,\ldots,B_j$ denote the hypersurfaces that contain $p(t)$. Then for a small loop $\gamma$ in $T\setm L$, that is based at $t_0$ and winds around $\ell$, the Picard-Lefschetz formula
\begin{equation}\label{eq:intro-PL}
    \Var_{\gamma} (\sigma) = (-1)^{(n+1)(n+2)/2} \cdot 
    \is{\dvcyc_p}{\sigma} \cdot \vcyc_p
\end{equation}
determines the variation of $\sigma\in H_n(Y\setm A,B)_{t_0}$ in terms of an intersection number $\is{\dvcyc_p}{\sigma}$ and vanishing cycles $\vcyc_p\in H_n(Y\setm A,B)_{t_0}$, $\dvcyc_p\in H_n(Y\setm B,A)_{t_0}$. The latter are constructed out of a sphere $\Sphere^{n-i-j}$ localized near $p(t)$ and embedded in the submanifold $A_1\cap\ldots\cap A_i\cap B_1\cap\ldots\cap B_j$.

For a single hypersurface $\abs{\irrone{A}\sqcup\irrone{B}}\leq 1$, the Picard-Lefschetz theorem is well-known \cite{Vassiliev:RamLac,Vassiliev:AppliedPL}. A proof, allowing for arbitrary isolated critical points, can be found in \cite[\S5]{Lamotke:HomIsoSing}. 
There are even generalizations for isolated critical points on complete intersections \cite{Hamm:LokTopKom,Looijenga:IsolatedCI}.
However, the case with \emph{multiple} hypersurfaces has not been studied as much. Arrangements with $\abs{\irrone{A}}>1$ are treated in \cite{FFLP}, but without boundary ($B=\varnothing$). The generalization to $B\neq\varnothing$ is briefly mentioned in \cite{Pham:Singularities}, but important details are left out.

The first result of this paper is a detailed, self-contained proof of \eqref{eq:intro-PL} in the general case where $A,B\neq\varnothing$. The main new insights are:
\begin{itemize}
    \item We correct an error in the literature:
All sources claim that all linear pinches ($i+j=n+1$) have zero variation.\footnote{See the last paragraph in \cite[\S I.9]{Vassiliev:RamLac} or \cite[\S I.8]{Vassiliev:AppliedPL}, and footnote~11 on \cite[p.~95]{Pham:Singularities}.} As \cref{ss:Li1} shows, this is not the case. Only the linear pinches of ``pure $\irrone{A}$ type'' ($j=0$) or ``pure $\irrone{B}$ type'' ($i=0$) have zero variation, whereas linear pinches on \emph{mixed} intersections of $\irrone{A}$'s \emph{and} $\irrone{B}$'s have non-zero variation.
    \item We prove a fairly general vanishing statement for iterated variations, see \cref{sec:intro-hierarchy}.
\end{itemize}

The main point of this paper is therefore that it pays off to keep track of which elements of $\irrone{A}$ and $\irrone{B}$ contain a given critical point.
The distinction between $\irrone{A}$ and $\irrone{B}$ is important, despite the fact that the Landau variety $L$ depends only on the combined arrangement $A\cup B$.

\subsection{Hierarchy principle}\label{sec:intro-hierarchy}

Consider an iterated variation around two components $\ell,\ell'$ of the Landau variety. The \emph{hierarchy principle} is a combinatorial criterion that ensures $\Var_{\gamma'}\circ\Var_{\gamma}=0$.
Suppose for example that the fibres over $\ell$ and $\ell'$ have a single non-degenerate critical point $p$ and $p'$, respectively.
Then according to \eqref{eq:intro-PL}, we are asking if $\is{\dvcyc_{p'}}{\vcyc_p}=0$.

The set of all hypersurfaces in $\irrone{A}\sqcup\irrone{B}$ that contain $p$ encodes a pair
\begin{equation*}
    (I,J)=(\set{i\colon p\in A_i},\set{j\colon p\in B_j}),
\end{equation*}
called the \emph{type} of $p$. Let $(I',J')$ denote the type of $p'$.
The variation around $\ell$ localizes in a small neighbourhood of $p$, which meets only those $A_i$ and $B_j$ with $i\in I, j \in J$. Hence the vanishing cycle $\vcyc_p$ is in the image of
\begin{equation*}
    H_n\Big(Y\setm A,\bigcup\nolimits_{j\in J} B_j\Big)_{t_0} \longrightarrow H_n(Y\setm A, B)_{t_0},
\end{equation*}
and therefore the boundary of $\vcyc_p$ remains forever confined to the $B_j$ with $j\in J$---even after continuation from $t_0$ to any $t \in T \setm L$.

Furthermore, $\vcyc_p$ is an iterated Leray coboundary (tube) around each $A_i$ with $i\in I$. It follows that for any $i\in I$, $\vcyc_p$ is in the kernel of
\begin{equation*}
    H_n(Y\setm A,B)_{t_0} \longrightarrow H_n\Big(Y\setm \bigcup\nolimits_{k\neq i} A_k, B\Big)_{t_0}.
\end{equation*}
So if some $i\in I$ is not in $I'$, then $\vcyc_p$ becomes zero in the localization of $\Var_{\gamma'}$ near $p'$. A dual argument for the vanishing cycles $\dvcyc_{p'}$ shows that $\is{\dvcyc_{p'}}{\vcyc_p}=0$ holds also whenever there exists $j\in J'$ that is not in $J$.

To summarize, the iterated variation $\Var_{\gamma'}\circ \Var_{\gamma}$ can only be nonzero if $p$ and $p'$ are \emph{compatible} (denoted $p'\CV p$), that is, if
\begin{equation}\label{eq:intro-compatible}
    I'\supseteq I \quad\text{and}\quad J'\subseteq J.
\end{equation}
Therefore, a non-zero iteration of variations $\Var_n\circ\cdots\circ\Var_1$ requires that the sets of participating $\irrone B$-elements get smaller ($I_1\supseteq I_2\supseteq\ldots$), whereas the sets of $\irrone A$-elements can only get bigger ($J_1\subseteq J_2\subseteq\ldots$).

In \cref{sec:hierarchy} we generalize this property in various ways:
\begin{itemize}
    \item For \emph{linear} simple pinches (\cref{defn:pinches}), compatibility requires strict containment: $I'\supsetneq I$ and $J'\subsetneq J$.
    \item When there are multiple critical points in a fiber, we call two components of the Landau variety \emph{compatible} ($\ell'\CV \ell$), if all critical points over $\ell'$ are compatible with all critical points over $\ell$.
    \item We can allow critical points that are not simple pinches, or critical \emph{sets} $p\subset A^I \cap B^J$ that are not even isolated. Then \eqref{eq:intro-compatible} becomes
    \begin{equation}\label{eq:intro-simple-hierarchy}
        I'\supseteq \simple{I}\quad\text{and}\quad \simple{J}'\subseteq J,
    \end{equation}
    in terms of the \emph{simple components} $\simple{I}\subseteq I$ and $\simple{J}'\subseteq J$ (\cref{def:simple-type}).
    Here a component $A_i$ or $B_j$ is called simple, if $p$ remains critical even if that component is removed from the arrangement.
\end{itemize}

The upshot is a relation $\CV$ on the components $\ell$ of $L$ (\cref{def:rel_landau_components}). The hierarchy principle (\cref{corr:povar}) states then that
\begin{equation}\label{eq:intro-hierarchy}
    \Var_{\gamma'} \circ \Var_{\gamma} = 0 \quad \text{if} \quad \ell'\nCV \ell. 
\end{equation}

As an application, we show in \cref{sec:Aomoto,ssec:oneloop} how to quickly derive the known Landau variety and hierarchy for Aomoto polylogarithms and massive one loop Feynman integrals.

\subsection{Feynman integrals}
The application of homological methods to Feynman integrals has a long history \cite{HwaTeplitz:HomFI,Golubeva:InvestigationFeynHom}.
Approaches using the \emph{momentum representation}, like \cite{Federbush:SomeHomology6,Westwater:SixthLadder}, are complicated by the need to compactify the integration domain, which poses unsolved challenges beyond one loop \cite{Muehlbauer:MomLanIso}.

This problem does not appear in the \emph{parametric representation}. However, in this representation, the integration domain has boundary. An analysis of (iterated) variations therefore requires a relative version of the (hierarchy) Picard-Lefschetz theorems.
Using this setup, Boyling derived constraints of the form \eqref{eq:intro-hierarchy} for massive one-loop graphs \cite[\S 2]{Boyling:HomParaFeyn}.

The Picard-Lefschetz formula \eqref{eq:intro-PL} is however not usually applicable to more general Feynman integrals, because the critical sets are often not simple and isolated.
Using our generalized hierarchy principle, we can nevertheless obtain constraints on iterated variations in such cases. We carry this out in detail for Feynman graphs with two loops (sunrise and ice cream cone) or one loop and zero masses (triangle), see \cref{sec:feynman}.

We find that the distinction in \eqref{eq:intro-simple-hierarchy} between the type of a singularity $\ell\subset L$ vs.\ its simple components is indeed crucial: applying \eqref{eq:intro-compatible} to singularities that are not simple pinches leads to wrong conclusions. This need for a refined analysis was noted for Feynman integrals in \cite{LandshoffOlivePolkinghorne:Hierarchical}, correcting a too optimistic expectation for the $\irrone{B}$-hierarchy from \cite{LandshoffPolkinghorneTaylor:MandelstamRepr}. In \cref{ssec:oneloop,sec:icecream} we demonstrate that also the $\irrone{A}$-hierarchy is interesting for Feynman integrals.

\subsection{Outlook}
The techniques developed in this paper provide a framework to explain variation constraints of the form \eqref{eq:intro-hierarchy}. We hope that these methods will be useful to prove constraints for infinite families of Feynman integrals, and thereby scattering amplitudes. For example, \emph{extended Steinmann relations} \cite[\S 3]{CDDHMP:CosmicSteinmann} are applied with great effect in high order perturbation theory. For now, these relations and generalizations thereof \cite{DrummondFosterGurdrogan:ClusterAdjN4} are conjectures.

\begin{rem}
    Originally, Steinmann relations \cite{Lassalle:AnaMany1,Steinmann:WFkommII} refer to amplitudes---not individual Feynman diagrams---and they only apply in a restricted (``physical'') region of parameter space. This translates into the statement $\Var_{\gamma'}\Var_{\gamma}\sigma=0$ for certain chains $\sigma$, but not all. In the physical region, first type singularities have simple pinch type and can thus be studied with Picard-Lefschetz theory \cite{Pham:DiffusionMultiple,Pham:SingMultiScatt,HMSV:SeqDiscOnShell}. The resulting hierarchy has a simple formulation, but it does not apply outside the physical region, ignores singularities of second type, and requires generic masses.\footnote{For example, external particles are not allowed to all have the same mass.}

    We do not make any such assumptions.
\end{rem}

\subsection{Outline of the paper} 
In \cref{sec:stratsandlandau} we review the solution of problem \textbf{(L)}, that is, how to find the Landau variety of a pair $(Y \setm A, B \setm A)$ with respect to a smooth map of complex manifolds $\pi\colon Y\to T$. We discuss (canonical) stratifications, Thom's first isotopy lemma (including a smooth version of it), how to compute $L$ and how this applies to the study of parameter integrals. 

\Cref{sec:PL} deals then with problem \textbf{(V)} via Picard-Lefschetz theory. We review the classical Picard-Lefschetz theorem and prove its generalization to relative homology. In the process, we give a concise account of all necessary technical tools: Definition of the monodromy and variation operators, localization arguments, characterization of \textit{simple pinch} critical points, Leray's residue and (co)boundary maps, and definition of the \textit{vanishing chains}. After the stage is set, we state and prove the (relative) Picard-Lefschetz theorem, \cref{thm:PL}, then discuss refinements for linear pinches (\cref{rem:lin-pinch-types}) and generalizations to arbitrary simple pinches (\cref{eq:general-simple-pinch}).  

In \cref{sec:hierarchy} we introduce our hierarchy principle. We first present it for simple pinches (\cref{eq:pinch-hierarchy,eq:itvar_B_linear,eq:itvar_A_linear,eq:itvar_n-m_odd}), then for certain components of the Landau variety (\cref{corr:simple-povar}), and finally state it in its most general form (corollaries \ref{corr:povar_upstairs} and \ref{corr:povar}) which allows for non-isolated critical sets.

The following two sections apply the previously developed theory to two families of examples. We rigorously derive their Landau varieties and discuss implications (and limitations) of the hierarchy principle. In \cref{sec:polylogs} we consider polylogarithms: While \cref{sec:dilog} deals with a specific (slightly degenerate) example, the dilogarithm, \cref{sec:Aomoto} discusses a tame family, the Aomoto polylogarithms, in full detail. Finally, in \cref{sec:feynman} we study Feynman integrals. After introducing some notation, we start with massive one loop graphs (\cref{ssec:oneloop}) where our setup works out of the box. Then we consider some specific examples where this is not the case, the massless triangle, the massive sunrise and the ice cream cone. We show how additional blow ups (beyond the Feynman motive) turn $A\cup B$ transverse, so that we can apply our methods. This allows us to explain and refine what is called \textit{breakdown of the hierarchical principle} in \cite{Boyling:HomParaFeyn,LandshoffOlivePolkinghorne:Hierarchical}.

The appendices provide several technical details, that are hard to find in the existing literature:
\begin{itemize}
\item
\Cref{sec:homology-groups} calculates all \emph{relative} homology groups of the arrangement of hypersurfaces at a linear or quadratic simple pinch, generalizing results of \cite{FFLP}. To be fully self-contained, we also compute the classical Picard-Lefschetz formula for the variation of a quadric.
\item
\Cref{sec:relative-residues} explains the construction of the long exact sequences for partial boundaries and relative residues from \cite[Chapitre~2]{Leray:CauchyIII}, in the general case of transverse arrangements of smooth closed submanifolds with arbitrary codimensions. 
\item
\Cref{sec:verdier+intersection} gives a general definition and the key properties of relative intersection numbers and the duality swapping $A\leftrightarrow B$.
\item
\Cref{sec:fg-codim1} shows that the codimension one components $\ell$ of $L$ already determine the full fundamental group $\pi_1(T\setm L)$, and that the kernel of $\pi_1(T\setm L) \relbar\joinrel\twoheadrightarrow \pi_1(T)$ is generated by so-called \textit{simple loops} (\cref{def:small-simple}).
\end{itemize}

\subsection{Two examples}\label{ss:introexamples}
We illustrate our findings with two very simple examples: The first example observes that the non-trivial monodromy of the logarithm arises from a linear pinch; the second example illustrates the hierarchy principle for the bubble Feynman integral. In both cases, $Y=X \times T$ is a product with $\pi(x,t)=t$ and compact fibre $X=\PP^1$.

\subsubsection{The logarithm}\label{ss:Li1}
For $t\in \PP^1$ consider the integral
\begin{equation*}
    \log (1-t) = \int_{\sigma} \frac{t\,\td x}{tx-1}.
\end{equation*}
For $t\neq[1,\infty]$, the interval $\sigma=[0,1]$ defines a class in $H_1(X\setm A_t,B_t)$, where the hypersurfaces $\irrone{A}=\set{A_1,A_2}$ and $\irrone{B}=\set{B_1,B_2}$ of $Y$ are given by
\begin{equation*}
    A_1=\set{xt=1 },\ 
    A_2=\set{ x=\infty },\ 
    B_1=\set{x=0},\ 
    B_2=\set{ x=1 }.
\end{equation*}
These four hypersurfaces are smooth and transverse to each other. 
Each hypersurface submerses onto $T$. Thus, the critical strata are the points (codimension two) where two hypersurfaces in $\irrone{A}$, or one in $\irrone{A}$ and one from $\irrone{B}$, intersect (see \cref{fig:loga2}):
\begin{equation*}
\begin{cases}
    \{ xt=1 \} \cap \{ x=\infty \} = \set{(\infty,0)}, & \irrone{A}\cap \irrone{A} \\
    \{ xt=1 \} \cap \{ x=0 \} = \set{(0,\infty)}, & \irrone{A}\cap\irrone{B} \\
    \{ xt=1 \} \cap \{ x=1 \} = \set{(1,1)}, & \irrone{A}\cap \irrone{B}. \\
\end{cases}
\end{equation*}

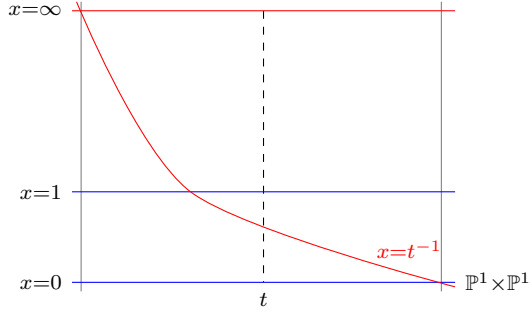
\begin{figure}
\begin{tikzpicture}[scale=1.2]
\coordinate (lu) at (-.1,0);
\coordinate (lo) at (-.1,3);
\coordinate (ru) at (4.1,0);
\coordinate (ro) at (4.1,3);
\coordinate (m) at (1,1);
\coordinate (ml) at (-.1,1);
\coordinate (mr) at (4.1,1);
\coordinate (mu) at (1,0);
\coordinate (mo) at (1,3);
\node[right] at (ru) {$\scriptstyle{ \mb P^1 \times \mb P^1 }$};
\node[left] at (lo) {$\scriptstyle{x=\infty}$};
\node[left] at (lu) {$\scriptstyle{x=0}$};
\node[left] at (ml) {$\scriptstyle{x=1}$};
\node[below] at (2,0) {$\scriptstyle{t}$};
\draw[gray] (0,-.1) -- (0,3.1);
\draw[gray] (3.95,-.1) -- (3.95,3.1);
\draw[dashed] (2,3) -- (2,0);
\draw[blue] (lu) -- (ru);
\draw[blue] (ml) -- (mr);
\draw[red] (lo) -- (ro);
\draw[red] plot [smooth] coordinates {(-.05,3.1) (1.2,1) (4.1,-0.05)} node[xshift=-.6cm, yshift=.5cm] {$\scriptstyle{x=\inv t}$};
\end{tikzpicture}
    \caption{The arrangement $D=A\cup B \subset \mb P^1 \times \mb P^1$ associated to the integral $\log (1-t) = \int_{\sigma_t} \frac{t\,\td x}{tx-1}$.}%
    \label{fig:loga2}%
\end{figure}

These critical points are called \emph{linear pinches} and their projection along $\pi$ determines the Landau variety $L=\{0,1,\infty\}$. At the points $t=1,\infty$, the function $\log(1-t)$ has logarithmic singularities. In contrast, at $t=0$, this function is smooth and takes values in $2\ipi\ZZ$. In particular, the variation at $t=0$ is zero. We will see in general that linear pinches of ``pure $\irrone{A}$-type'' (or ``pure $\irrone{B}$-type'') always have zero variation, whereas ``mixed $\irrone{A}$--$\irrone{B}$-type'' linear pinches have non-zero variation.

In order to describe the monodromy associated to small loops winding around $L$ we note that for generic $t$ the first relative Betti homology of the fiber is $H_1(Y\setm A,B)_t=H_1(X \setm A_t,B_t ) \cong \mb Z^2$.  It is spanned by the interval $\sigma=[0,1]$, and a small loop $\nu$ around the point $x=\inv{t}$. This cycle is constructed as follows (see \cref{fig:loga}): 
\begin{enumerate}
 \item Let $\ell\in L$. For $t$ close to $\ell$ we consider the pair of fibers $(X,A_t \cup B_t)$. The \textbf{vanishing cell} $\vcell_\ell$ is the unique (real) $n$-dimensional cell ($n=\dim_{\mb C}X$) in $X$ that is bounded by $A_t \cup B_t$ and vanishes as $t\to \ell$. 
 
 \noindent 
 In the present case there are three such cells, given by paths from $x$ to $\inv{t}$ for $x=0,1,\infty \in X$ (corresponding to $\ell=\infty,1,0$). Each represents an element in $H_1(Y,A \cup B)_t\cong \mb Z^3$. Note that, if we restrict to a small neighbourhood $U$ of a critical point $x=\inv \ell$, then each class is the unique generator of $H_1(U, (A_t \cup B_t) \cap U)\cong \mb Z$.
 
 \item The \textbf{vanishing cycle} $\vcyc_\ell$ is defined as the \emph{Leray coboundary} (see \cref{s:leray}) of the \emph{partial boundary} of the vanishing cell (both with respect to the relevant hypersurfaces of $\irrone{A}$-type), $\vcyc_\ell:=\delta_A \partial_A \vcell_\ell$.
 
 \noindent 
 Here $\nu_1$ and $\nu_\infty$ coincide, both given by small loops around $x=\inv{t}$. They represent classes $[\nu_1]\equiv [\nu_\infty]$ in $H_1(Y \setm A,B )_t$, unique after localisation. The vanishing cycle $\vcyc_0$ is zero, because $\partial_A \vcell_0= \partial_{A_2}\partial_{A_1}\vcell_0=0$.
\end{enumerate}
 
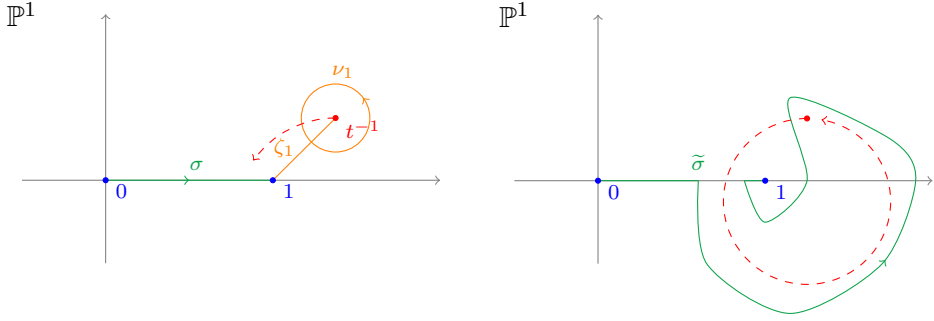
\begin{figure}
  \begin{tikzpicture}[scale=1.1]
\coordinate (lu) at (0,0);
\coordinate (lo) at (0,3);
\coordinate (ru) at (4,0);
\coordinate (ro) at (4,3);
\coordinate (m) at (1,1);
\coordinate (ml) at (0,1);
\coordinate (mr) at (5,1);
\coordinate (mu) at (1,0);
\coordinate (mo) at (1,3);
\coordinate (t) at (3.75,1.75);
\node at (lo) {$\mb P^1$};
\draw[black!50,->] (mu) -- (mo);
\draw[black!50,->] (ml) -- (mr);
\draw[olive!15!green,decoration={markings, mark=at position 0.5 with {\arrow{>}}},postaction={decorate}] (m) -- (3,1) node[above,xshift=-1cm] {$\scriptstyle{\sigma}$};
\draw[orange] (3,1) -- (t) node[left,xshift=-.4cm,yshift=-.4cm] {$\scriptstyle{\vcell_1}$};
\filldraw[blue] (m) circle (0.03) node[right,yshift=-.15cm] {$\scriptstyle{0}$};
\filldraw[blue] (3,1) circle (0.03) node[right,yshift=-.15cm] {$\scriptstyle{1}$};
\draw[orange,decoration={markings, mark=at position 0.1 with {\arrow{>}}},postaction={decorate}] (t) circle [radius = 0.41cm] node[above,xshift=.1cm,yshift=.4cm] {$\scriptstyle{\vcyc_1}$};
\filldraw[red] (t) circle (0.03) node[right,yshift=-.15cm] {$\scriptstyle{\inv t}$};
\draw[red,dashed,decoration={markings, mark=at position 1 with {\arrow{>}}},postaction={decorate}] (t) arc[radius = 1.2cm, start angle= 90, end angle= 145];
\end{tikzpicture}
\quad 
\raisebox{-.666cm}{
\begin{tikzpicture}[scale=1.1]
\coordinate (lu) at (0,0);
\coordinate (lo) at (0,3);
\coordinate (ru) at (4,0);
\coordinate (ro) at (4,3);
\coordinate (m) at (1,1);
\coordinate (ml) at (0,1);
\coordinate (mr) at (5,1);
\coordinate (mu) at (1,0);
\coordinate (mo) at (1,3);
\coordinate (t) at (3.5,1.75);
\coordinate (s1) at (2.2,1);
\coordinate (s2) at (2.75,1);
\node at (lo) {$\mb P^1$};
\draw[black!50,->] (mu) -- (mo);
\draw[black!50,->] (ml) -- (mr);
\draw[olive!15!green] (m) -- (s1) node[above] {$\scriptstyle{\ti \sigma}$};
\draw[olive!15!green] (s2) -- (3,1);
\draw [olive!15!green,decoration={markings, mark=at position 0.4 with {\arrow{>}}},postaction={decorate}] plot [smooth] coordinates {(s1) (2.3,0) (3.3,-.6) (4.4,0) (4.8,1) (4.5,1.5) (3.3,2) (3.5,1) (3,0.5) (s2)};
\filldraw[blue] (m) circle (0.03) node[right,yshift=-.15cm] {$\scriptstyle{0}$};
\filldraw[blue] (3,1) circle (0.03) node[right,yshift=-.15cm] {$\scriptstyle{1}$};
\filldraw[red] (t) circle (0.03) node[right,yshift=-.15cm] {};
\draw[red,dashed,decoration={markings, mark=at position 1 with {\arrow{>}}},postaction={decorate}] (t) arc[radius = 1cm, start angle= 90, end angle= 440];
\end{tikzpicture}}
  \caption{Geometry of   $\log(1-t)$: The left hand side shows the vanishing cell $\vcell_1$ and vanishing cycle $\vcyc_1$ for $t$ close to 1, the right hand side sketches the monodromy of $\sigma$ associated to a small loop winding around $t=1$; the resulting relative cycle $\ti \sigma$ is homologous to $\sigma + \vcyc_1$.}%
  \label{fig:loga}%
\end{figure}
 
The Picard-Lefschetz theorem states that the variation $\Var_\ell$ associated to a small loop winding around $\ell$ maps a homology class $h\in H_n(Y\setm A,B)_t$ to a multiple of the vanishing cycle $N [\nu_\ell]$ where $N\in \mb Z$ is determined by a (relative) \textit{intersection index} $\is {\partial_B h} {\partial_B \vcell_\ell}$
 with $\partial_B$ denoting the iterated partial boundary with respect to $B$.\footnote{We discuss intersection indices in \cref{sec:verdier+intersection}.}
 
In the present case we have $\nu_0=0$, so the variation associated to $0\in L$ vanishes. For $\ell=1,\infty$ one calculates $\is {\partial_B \sigma} {\partial_B \vcell_\ell}=\pm 1$, while $\is {\partial_B \nu_\ell} {\partial_B \vcell_\ell}=\is {0} {\partial_B \vcell_\ell}=0$. In summary, the variation of $[\sigma]$ is non-trivial around $\ell=1,\infty$, the variation of $[\nu_1]=[\nu_\infty]$ is trivial around all points of $L$. We recover thus the well-known analytic structure of the function $\log(1-t)$.

\subsubsection{The bubble Feynman integral}
\label{eg:bubble}
As a second example we consider the parametric Feynman integrals (cf.\ \cref{sec:feynman}) of the bubble graph:
\[
    \II(p^2,m_1^2,m_2^2) = \int_{ \sigma= \set{ [x_1:x_2] \colon x_i > 0}  } \frac{P\Omega}{\Upol^{a}\Fpol^{b} }
    = \int_0^{\infty} \left.\frac{P}{\Upol\Fpol}\right|_{x_1=1} \td x_2.
\]
Here $a$ and $b$ are integers, $P$ is any homogeneous polynomial in $x$ with degree $a+2b-2\geq 0$, the projectivized measure is $\Omega=x_1 \td x_2 - x_2\td x_1$, and
\begin{equation*}
    \Upol = x_1+x_2,\qquad
    \Fpol = (x_1+x_2)(m_1^2 x_1+m_2^2 x_2) - p^2 x_1 x_2.
\end{equation*}

The parameters are $T=\set{ (p^2,m_1,m_2) \in \mb C^3}$. 
The arrangements $\irrone{A}=\set{A_1,A_2}$ and $\irrone{B}=\set{B_1,B_2}$ consist of the boundary points $B_1=\set{[0:1]}$, $B_2=\set{[1:0]}$, a point $A_1=\set{\Upol=0}=\set{[1:-1]}$, and a pair of points
\begin{equation*}
    A_2=\set{\Fpol=0} = \set{Q_1,Q_2} 
    = \big\{\big[p^2-m_1^2-m_2^2\pm\sqrt{\Delta}:2m_1^2\big]\big\}
    %= \set{[1-z:z],[1-\bar{z}:\bar{z}]}.
\end{equation*}
where $\Delta = (m_1^2+m_2^2-p^2)^2-4m_1^2m_2^2=\Delta_+\Delta_-$ with $\Delta_{\pm}=p^2-(m_1\pm m_2)^2$.

The Landau variety $L=\ell_{1}\cup \ell_{2} \cup \ell_{\Delta} \cup \ell_p$ has four components, determined by collisions of the points $Q_{1,2}$---with each other, or with one of $B_1,B_2,A_1$:
\begin{itemize}
    \item a linear pinch of $A_2 \cap B_1$ ($Q_1=B_1$) over $\ell_2=\set{m_2^2=0}$,
    \item a linear pinch of $A_2 \cap B_2$ ($Q_2=B_2$) over $\ell_1=\set{m_1^2=0}$,
    \item a quadratic pinch of $A_2$ ($Q_1=Q_2$) over $\ell_{\Delta}=\set{\Delta=0}$,
    \item a linear pinch of $A_1\cap A_2$ ($A_1=Q_1$ or $A_1=Q_2$) over $\ell_p=\set{p^2=0}$.
\end{itemize}
In physics terms, $\ell_{1,2}$ are called ``reduced singularities'', whereas $\ell_{\Delta}$ and $\ell_p$ are ``leading singularities'' of the ``first type'' ($\ell_{\Delta}$) and ``second type'' ($\ell_p$). The components $\Delta_+$ and $\Delta_-$ of $\Delta$ are also referred to as ``normal threshold'' and ``pseudo threshold'', respectively.

The vanishing cells of $\ell_{1}$, $\ell_2$, and $\ell_\Delta$ are paths
\begin{equation*}
    \vcell_{1}=[ B_2, Q_1 ], \quad
    \vcell_{2}= [ B_1, Q_2 ], \quad
    \vcell_{\Delta}= [ Q_1 , Q_2 ]
\end{equation*}
where the branch of $\sqrt{\Delta}$ is chosen such that $Q_1\rightarrow B_2$ for $m_1\rightarrow 0$ and $Q_2\rightarrow B_1$ for $m_2\rightarrow 0$.
The respective vanishing cycles $\vcyc_{\bullet}=\fibSphere_{A_2} \partial_{A_2}\vcell_{\bullet}$ are
\begin{equation*}
    \vcyc_{1}
    = \fibSphere_{A_2} [Q_1], \quad
    \vcyc_{2}
    = \fibSphere_{A_2} [Q_2], \quad
    \vcyc_{\Delta}
    = \vcyc_2-\vcyc_1.
\end{equation*}

The second type singularity $p^2=0$ is a linear pinch on the intersection of two elements of $\irrone A$. In this ``pure $\irrone{A}$'' case, the vanishing cycle is zero and so\footnote{%
This shows that the integral $\II$ is meromorphic at $\ell_p$, but not necessarily smooth. For example, the integral $\int_{\sigma} x_1 \Omega/(\Upol\Fpol)$ with $P=x_1$ develops a pole at $\ell_{p}$ (on some sheets).}
\begin{equation*}
    \Var_{\ell_p} = 0.
\end{equation*}

The dual vanishing cycles $\dvcyc_\bullet=-\fibSphere_B\partial_B \vcell_{\bullet}$ over $\ell_1$ and $\ell_2$ have intersection numbers $\is{\dvcyc_1}{\sigma}=\is{\partial_{B} \vcell_1}{\partial_B \sigma} = \is{-[B_2]}{[B_1]-[B_2]} = 1$, $\is{\dvcyc_2}{\sigma}=-1$, and $\is{\dvcyc_i}{\vcyc_j}=0$ for all $i,j\in\set{1,2}$ because $\partial_B\vcyc_{j}=0$. So we determine
\begin{equation*}
    \Var_{\ell_1} [\sigma] = -\vcyc_1
    ,\qquad
    \Var_{\ell_2} [\sigma] = \vcyc_2
    ,\qquad
    \Var_{\ell_i} \vcyc_j = 0
\end{equation*}
from the Picard-Lefschetz formula \eqref{eq:intro-PL}. For example, if $m_1^2\rightarrow 0$, then $Q_1=[1:m_1^2/(p^2-m_2^2) + \asyO{m_1^4}]$ so that, in the $x_2/x_1$ chart, $Q_1$ encircles the origin (the starting point $B_2=[1:0]$ of the integration path) in the same sense as $m_1^2$ goes around the origin, say \emph{counter clockwise}. Deforming $\sigma$ to avoid $Q_1$ winds up a \emph{clockwise} loop around $Q_1$ (cf.\ \cref{fig:loga}).

Assume that $m_1$ and $m_2$ are real and positive. Then at the normal threshold, the degeneration of $A_2$ into $Q_1=Q_2=[m_2:m_1]$ pinches the straight path $\sigma$, whereas at the pseudo-threshold, the collision $Q_1=Q_2=[-m_2:m_1]$ happens away from $\sigma$:
\begin{equation*}
    \Var_{\ell_{\Delta_+}} [\sigma] = \vcyc_{\Delta}, \qquad
    \Var_{\ell_{\Delta_-}} [\sigma] = 0.
\end{equation*}
In both cases, $Q_1$ and $Q_2$ get swapped, $\fibSphere_{A_2}[Q_1] \leftrightarrow \fibSphere_{A_2}[Q_2]$, and thus
\begin{equation*}
    \Var_{\ell_{\Delta}} \vcyc_1 = -\Var_{\ell_{\Delta}} \vcyc_2 =  \vcyc_2 - \vcyc_1.
\end{equation*}

Since the group $H_1(Y\setm A,B)_t\cong\ZZ^3$ is generated by $\sigma$, $\nu_1$ and $\nu_2$, we have completely determined all variation operators.
The hierarchy on $L$ can be expressed in the following diagram: 
\begin{equation}\label{eq:posetbubble}\begin{aligned}
\xymatrix{
    \ell_{2}  \ar[d] \ar[dr] & \ell_{1}  \ar[d] \ar[dl]  & \ \ \text{ -- intersections $A_2 \cap B_i$} \\
    \ell_{\Delta_+} \ar@(ul,l) \ar[d]  \ar[r] & \ell_{\Delta_-} \ar@(ur,r) \ar[dl] \ar[l] & \text{ -- degenerations of $A_2$} \\
    \ell_{p} &  & \ \ \ \text{ -- intersections $A_2 \cap A_1$}
} 
\end{aligned}\end{equation}
Only iterated variations that follow arrows may be non-zero, for example
\begin{equation*}
   \Var_{\ell_{\Delta}}  \circ  \Var_{\ell_2} [\sigma] = \Var_{\ell_{\Delta}} \vcyc_2  = \vcyc_1 - \vcyc_2.
\end{equation*}
All other iterated variations vanish. For instance,
\begin{equation*}
    \Var_{\ell_i} \circ \Var_{\ell_{\Delta}} , \Var_{\ell_i} \circ \Var_{\ell_p},
    \Var_{\ell_i} \circ \Var_{\ell_i} , \Var_{\ell_p} \circ \Var_{\ell_p}
    = 0.
\end{equation*}

\subsection{Acknowledgments}
We thank Matteo Parisi and \"{O}mer G\"{u}rdo\u{g}an for discussions at an early stage of this project, and Holmfridur S.~Hannesdottir for bringing \cite{LandshoffOlivePolkinghorne:Hierarchical} to our attention.
Marko Berghoff thanks Max M\"{u}hlbauer for many illuminating discussions on the subject, Erik Panzer thanks Matija Tapu\v{s}kovi\'{c} for clarifying aspects of Feynman motives.
Erik Panzer is funded as a Royal Society University Research Fellow through grant {URF{\textbackslash}R1{\textbackslash}201473}.
Marko Berghoff was also supported through this grant, and by the European Research Council (ERC) under the European Union’s Horizon 2020 research and innovation programme (grant agreement No.\ 724638).

\section{Stratified maps and their Landau varieties} \label{sec:stratsandlandau}

\subsection{Stratified sets}
Let $Y$ be a smooth manifold and $D \subset Y$ a closed subset. We require that $D$ can be decomposed into smooth pieces as follows.
\begin{defn}\label{defn:strat}
A \textbf{stratification} of a topological space $D$ is a sequence
\begin{equation} \label{eq:strat}
    \varnothing = C_{-1} \subseteq C_0 \subseteq \ldots \subseteq C_k = D
\end{equation}
of closed sets $C_i$ such that
\begin{itemize}
    \item the successive complements $C_{i}\setm C_{i-1}$ are smooth manifolds of dimension $i$. Their connected components are called \textbf{strata}, and
    \item for each stratum $S$, its boundary $\partial S=\overline{S} \setm S$ is a union of strata of smaller dimension than $S$.
\end{itemize}
\end{defn}
We write $\mathfrak{S}=\bigsqcup_i \pi_0(C_i\setm C_{i-1})$ for the set of all strata. The stratification can be recovered from this set as $C_i=\bigcup_{S\in\mathfrak{S}\colon \dim S\leq i} S$.
\begin{defn}\label{def:transverse}
A finite collection $\set{D_1,D_2,\ldots}$ of closed smooth submanifolds $D_i\subseteq Y$ is called \textbf{transverse}\footnote{Some authors, e.g.\ \cite{FFLP,Pham:Singularities}, call it ``in general position''.}, if every intersection 
\begin{equation*}
    D^{i_1,\ldots, i_m}=D_{i_1}\cap\ldots\cap D_{i_m}
\end{equation*}
can be covered with charts
$(x_0,x_1,\ldots,x_m)\colon Y\supset U\rightarrow\RR^{r_0}\times\cdots\times\RR^{r_m}=\RR^n$
such that $D_{i_k}\cap U=\set{x_k=0}$ for each $k=1,\ldots,m$.
Here we denote $n=\dim Y$, $r_k=n-\dim D_{i_k}$ and $r_0=n-\sum_{k=1}^m r_k$. See \cref{defn:transverse} for a coordinate-free formulation.
\end{defn}
In this situation, all intersections $D^I=\bigcap_{i\in I}D_i$ are themselves closed submanifolds. Thus we obtain a canonical stratification of the union $D=\bigcup_i D_i$ by setting $C_d=\bigcup_{I\colon \dim D^I\leq d} D^I$.

We consider the special case where $Y$ is a complex manifold with complex dimension $n$, and each $D_i$ is a smooth complex hypersurface. Then ``$D$ is transverse'' is also called ``$D$ has \textbf{simple normal crossings}'' and means that locally there exist $m\leq n$ and a holomorphic chart $x\colon U\rightarrow \CC^n$ with $D\cap U=\set{x_1\cdots x_m=0}$. The canonical stratification is then given by
\begin{equation}\label{eq:canonical-strat}
    C_{2d} = \bigcup_{\abs{I}=n-d}  D^I \quad\text{with}\quad n=\dim_{\mb C} Y.
\end{equation}
The topological closure $S\mapsto\overline{S}$ defines a bijection between the strata of dimension $2d$ (connected components of $C_{2d}\setm C_{2d-2}$) and the irreducible components of the analytic variety $C_{2d}$.

If the hypersurfaces are not smooth, or smooth but not transverse, then \eqref{eq:canonical-strat} is not necessarily a stratification---and even if it happens to be a stratification, it may not be sufficiently regular in order to apply the isotopy theorems of the next section. This issue can be resolved in at least two ways:
\begin{enumerate}
    \item Apply resolution of singularities, that is, construct an iterated blow up $\ti{Y}\rightarrow Y$ such that the total transform $\ti D$ of $D$ has simple normal crossings in $\ti{Y}$. This preserves the complement $Y\setm D\cong\ti{Y}\setm\ti{D}$.
    \item Replace \eqref{eq:canonical-strat} by a more subtle stratification and substitute transversality by a weaker notion of regular incidence.
\end{enumerate}
We will use the first option and discuss this procedure specifically for Feynman integrals in \cref{sec:feynman}.
For completeness, we briefly address the second option.
\begin{defn}
    A stratification is called \emph{Whitney regular} if every pair $S,S'$ of strata that are incident ($S\subseteq \partial S'$) satisfies \emph{Whitney's conditions A and B of regular incidence}: Embed $S$ and $S'$ in some Euclidean space. Then at every point $x \in S$, we require that for all convergent sequences $x'\to x$ of points $x'\in S'$:
\begin{itemize}
    \item[(A)] if $T_{x'}S'$ converges as $x'\mapsto x$, then its limit $\tau\subseteq T_xY$ (understood in the Gra{\ss}mannian bundle of subspaces of $TY$) contains $T_x S \subseteq \tau$, 
    \item[(B)] for any local retraction $r\colon S' \to S$, the angle between $T_{x'} S'$ and the vector $\overrightarrow{x',r(x')}$ converges to zero as $x'\to x$.
\end{itemize}
\end{defn}
Condition (A) is implied by (B) and hence redundant. For details we refer to \cite{Pham:Singularities,Mather:Stability,GoreskyMacPherson:Stratified,Trotman:StratTheory}.
Whitney showed \cite{Whitney:Tangents} that every (analytic or algebraic) subvariety $D$ admits a Whitney regular stratification.

Note that a (Whitney regular) stratification of $D$ always gives rise to a (Whitney regular) stratification of $Y$. Simply add $Y \setm D$ as an additional stratum; conditions (A) and (B) are satisfied automatically, since for $S'=Y\setm D$ we have $T_{x'}S'=T_{x'}Y$.

In the following we will assume that $D$ is normal crossing in $Y$. In this case (and more generally for transverse arrangements of smooth submanifolds) the canonical stratification described above is automatically Whitney regular. In fact it is minimal in the sense that every other Whitney stratification is a refinement of it.

\subsection{Isotopy theorems}
Given a smooth map $\pi\colon Y\rightarrow T$ of manifolds, the fibres $Y_t=\pi^{-1}(t)$ vary with $t\in T$. Let $T$ be connected. Then \emph{Ehresmann's fibration theorem} gives a sufficient criterion for all fibres to be diffeomorphic:
%A submersion is always open, and a proper map between manifolds is always closed. Hence, for $T$ connected, a proper submersion $\pi\colon Y\rightarrow T$ is necessarily surjective.
\begin{thm}[Ehresmann]\label{thm:Ehresmann}
    Every proper submersion is a smooth fibre bundle.
\end{thm}
Hence, if a closed subset $D\subset Y$ is smooth and the restriction $\pi|_D$ a proper submersion, then $D$ is a smooth fibre bundle over $T$. Namely, each $t\in T$ has an open neighbourhood $U$ over which $D$ trivializes into a product of $U$ with the (smooth) fibre $D_t=D\cap \pi^{-1}(t)$. That is, we have commutative diagrams
\begin{equation}\label{eq:fibre-bundle}\begin{gathered}\xymatrix{
    D \ar[d]^{\pi|_D} & D \cap \pi^{-1}(U) \ar@{_{(}->}[l] \ar[r]^-{\phi}_-{\cong} & D_t \times U \ar[dl]^{p} \\
    T & \ar@{_{(}->}[l] U &
}\end{gathered}\end{equation}
where $\phi$ is a diffeomorphism and $p$ denotes the projection onto the factor $U$.

If $D$ is singular, we cannot apply \cref{thm:Ehresmann} to $D$ itself. Thus, stratify $D$, so that we may apply the theorem to each (smooth) stratum $S$. If $\pi|_S$ is a submersion for each $S$, we conclude that all strata are smooth fibre bundles.

However, it is not always possible to glue trivializations of the strata $S$ into a trivialization of $D$. \emph{Thom's first isotopy lemma} guarantees that this \emph{does} work, provided that the stratification is sufficiently nice.
\begin{thm}[Thom's first isotopy lemma]
\label{thm:Thom}
Let $\pi\colon Y \to T$ be a smooth map and $D\subseteq Y$ a closed subset with a Whitney regular stratification, so that
\begin{itemize}
    \item $\pi|_{D}$ is proper, and
    \item $\pi|_S$ is a submersion for each stratum $S$ of $D$.
\end{itemize}
Then $\pi|_D$ is a stratified fibre bundle.
\end{thm}
\begin{proof}
See \cite[\S1.5--1.7]{GoreskyMacPherson:Stratified} for a sketch, and \cite{Mather:Stability} for a detailed proof.
\end{proof}
By a \textbf{stratified fibre bundle}\footnote{also called \emph{locally trivial stratified fiber bundle} \cite{Pham:Singularities} or \emph{locally trivial stratified map} \cite{GoreskyMacPherson:Stratified}} we mean a topological fibre bundle $D\rightarrow T$ with the property that the trivializations $\phi$ in \eqref{eq:fibre-bundle} can be chosen such that they simultaneously trivialize all strata $S$ of $D$:
\begin{equation}\label{eq:stratified-bundle}
    \phi(S\cap\pi^{-1}(U))=S_t\times U,
\end{equation}
where $S_t=S\cap\pi^{-1}(t)$.
For connected $T$, it follows that all fibres $D_t$ are homeomorphic, in a stratum preserving way.

Thom's first isotopy lemma shows only that $\pi|_D$ is a \emph{topological} fibre bundle: the trivializations $\phi$ are homeomorphisms, but not necessarily smooth. To achieve smoothness, we need even stronger regularity of $D$. Transversality will be sufficient for our applications \cite[\S X, Lemma~2.1]{Pham:Singularities}:
\begin{prop}\label{prop:smoothom}
    Let $\pi\colon Y \to T$ be a proper submersion and $D=\bigcup_{i \in I} D_i$ for transverse smooth closed submanifolds $D_i\subset Y$. Suppose that the restriction of $\pi$ to $D^J= \bigcap_{j\in J}D_j$ is a submersion, for each $J\subseteq I$.
    
    Then $\pi$ is a \emph{smooth} stratified fibre bundle, which means that $Y$ has \emph{smooth} local trivializations $\phi\colon \pi^{-1}(U)\cong Y_t\times U$ with the property that\footnote{This condition is equivalent to \eqref{eq:stratified-bundle} for all strata $S$ of the canonical stratification of $D$.}
    \begin{equation*}
        \phi(D_i\cap \pi^{-1}(U))=(D_i\cap\pi^{-1}(t))\times U \quad\text{for all}\quad i\in I.
    \end{equation*}
\end{prop}
\begin{proof}
It follows from \cref{def:transverse} that $Y$ has a cover by charts $x^V\colon V\rightarrow \RR^n$ with the property that if $J\subseteq I$ denotes the indices $i$ with $D_i\cap V\neq\varnothing$, then
\begin{equation*}
    x^V(D_i\cap V)=\bigcap_{\alpha\in K_i} \set{x^V_\alpha=0}    
\end{equation*}
for each $i\in J$, with disjoint subsets $K_i\subseteq\set{1,\ldots,n}$. Permute the coordinates so that the first $r=\dim D^J$ of them parametrize $D^J$, whereas $\bigcup_{i\in J}K_i=\set{r+1,\ldots,n}$. Since $\pi|_{D^J}$ is a submersion, we can furthermore arrange that the first $k=\dim T \leq r$ coordinates form a chart of $T$. This way, every smooth vector field $u$ on $T$ lifts to a smooth vector field $u^V$ on $V$ such that:
\begin{itemize}
    \item $u^V$ is a lift of $u$: $\td \pi_y(u^V_y)=u_{\pi(y)}$ at every $y\in V$, and
    \item $u^V$ is a linear combination of $\partial/\partial x^V_1,\ldots,\partial/\partial x^V_k$.
\end{itemize}
The latter implies that $u^V$ is tangent to all $D_i$. Using a partition of unity, we can glue the $u^V$ to a smooth vector field $w$ on $Y$. This construction provides a smooth lift $w$ of $u$ that is tangent to each $D_i$. The flow $\Phi^w_\tau$ along such a vector field preserves the submanifolds: $\Phi^w_\tau(D_i)=D_i$

Now pick any point $p\in T$ and a chart $(t_1,\ldots,t_k)\colon U\rightarrow \RR^k$ with $t(p)=0$. Construct lifts $w_{\alpha}$ of the vector fields $u_\alpha=\partial/\partial t_{\alpha}$ on $U$, then we have
\begin{equation*}
    \Phi^{u_\alpha}_\tau \circ \pi= \pi \circ \Phi^{w_\alpha}_\tau
\end{equation*}
whenever the corresponding flows are defined. Since $\pi$ is proper, we can shrink $U$ and find $\epsilon>0$ so that $\Phi^{w_\alpha}_\tau(y)$ is defined for all $y\in \pi^{-1}(U)$ and $\abs{\tau}<\epsilon$. We obtain the desired trivialization $\phi\colon \inv{\pi}(U) \to \inv{\pi}(p) \times U$ by 
\begin{equation*}
y  \longmapsto \Big(    \left(\Phi^{w_k}_{-t_k} \circ \cdots \circ \Phi^{w_1}_{-t_1} \right)(y), \pi(y) \Big) 
\end{equation*}
where $(t_1,\ldots,t_k)=\pi(y)$.
\end{proof}

\subsection{Landau varieties}

Let $\pi\colon Y \to T$ be smooth and surjective, and $\mathfrak{S}$ a Whitney regular stratification of a subset $D\subset Y$.
The \textbf{critical set} $cS\subseteq S$ of a stratum $S \in \mathfrak{S}$ is the subset where the restriction $\pi|_S$ fails to be a submersion: $\rk {\td \pi\rest{S}} < \dim T$. Points in $cS$ and $\pi(cS)\subset T$ are called \emph{critical points} and \emph{critical values}, respectively.
\begin{defn}\label{def:landau-variety}
    The \textbf{Landau variety} is the set of all critical values,\footnote{In \cite{Brown:PeriodsFeynmanIntegrals}, the Landau variety is defined to consist only of the codimension one part.}
\begin{equation*}
    L =\bigcup_{S\in\mathfrak{S}} \pi(cS) \subset T.
\end{equation*}
\end{defn}
If $\pi$ is proper, then \cref{thm:Thom} guarantees that the restriction of $\pi$ to $\inv{\pi}(T\setm L)$ is a stratified fibre bundle over $T\setm L$.

Now let $Y$ and $T$ be complex manifolds, $\pi\colon Y \to T$ proper and analytic, and $D\subset Y$ stratified such that the $C_i$ in \eqref{eq:strat} are complex analytic varieties. Then the Landau variety is a complex analytic variety \cite{Pham:Singularities}. We care only for its irreducible components with complex codimension one, which we denote
\begin{equation*}
    \irrone{L} = \set{\ell\subseteq L\colon\  \text{$\ell$ is irreducible and $\codim \ell=1$}}
    .
\end{equation*}

The restriction to codimension one arises as follows: An integral in fibres of $\pi$, see \cref{sec:integrals}, defines a multivalued analytic function $\II\colon T\setm L\to \CC$. The Landau variety provides an upper bound on the singularities of $\II$. However, we also have Riemann’s Extension Theorem \cite[\S 7.1]{GrauertRemmert:CAS}:
\begin{thm}
Let $U\subset \CC^n$ be open and $A \subset U$ a closed analytic subset of complex codimension greater than one. Then every analytic function $f \colon  U \setm A \to \CC$ can be extended to an analytic function on $U$.
\end{thm}

Any irreducible component $\ell\subseteq L$ with codimension greater than one can therefore not arise as a singularity of $\II$.
Therefore, the variation of $\II$ around such $\ell$ is zero. In fact, this follows already from a merely topological fact, without any reference to a function: The fundamental groups
\begin{equation*}
    \pi_1(T\setm L)\cong \pi_1(T \setm \bigcup_{\ell\in\irrone{L}}\ell)
\end{equation*}
of the complements of the Landau variety $L$, and its codimension one part, are canonically isomorphic. This is reviewed in \cref{sec:fg-codim1}.

As a consequence, if $D$ has simple normal crossings, we only need to consider critical sets of strata that arise from at most $(n+1)$-fold intersections $D^I= \bigcap_{i\in I} D_i$, $\abs{I}\leq n+1$. All lower dimensional strata are entirely critical, $cS=S$, but their critical values $\pi(cS)$ have codimension greater than one and therefore they do not contribute to $\irrone{L}$.

\begin{rem}\label{rem:minimal-landau}
    As defined, the Landau variety $L=L(\mathfrak{S})$ depends on the choice of stratification. But local triviality of $D\rightarrow T$ is an open condition in $T$, hence it holds outside the intersection $L_{\text{min}}=\bigcap_{\mathfrak{S}} L(\mathfrak{S})$ over all possible Whitney stratifications $\mathfrak{S}$ of $D$. In our setup, where $D\subset Y$ are analytic, all such $\mathfrak{S}$ are refinements of a unique minimal Whitney stratification $\mathfrak{S}_{\text{min}}$ of $D$ \cite[Proposition~3.2 on p.~479]{Teissier:PolairesII}; hence $L_{\text{min}}=L(\mathfrak{S}_{\text{min}})$ gives a canonical definition of \emph{the} Landau variety of $D$. Whenever $D$ is transverse, the canonical stratification \eqref{eq:canonical-strat} is Whitney regular and minimal, hence it computes $L_{\text{min}}$.
\end{rem}

\subsection{Integrals in fibres}\label{sec:integrals}

Assume now that $D$ is a union of complex hypersurfaces in $Y$, and furthermore that the irreducible components $\irrone{D}=\irrone{A}\sqcup\irrone{B}$ of $D$ are partitioned into two finite sets $\irrone{A}=\{ A_1 , \ldots, A_r \}$ and $\irrone{B}=\{ B_1,\ldots, B_s \}$. In all our applications, $\pi\colon Y\rightarrow T$ and $D\subset Y$ will be \emph{algebraic}, but this additional structure will play no role. In \cref{sec:intro-integrals}, we considered integrals
\begin{equation*} %\label{eq:function}
    \II(t)=\int_{\sigma_t} \omega_t
\end{equation*}
with $\omega$ a holomorphic form on $Y\setm A$ of degree $n$, the complex dimension of the fibres $Y_t$. Then $\omega_t=\omega|_{Y_t\setm A_t}$ is a closed form, and furthermore zero when restricted to the subvariety $B_t$. Thus $[\omega_t]\in H^n_{\dR}(Y_t\setm A_t,B_t\setm A_t)$ defines a class in de Rham cohomology.\footnote{
    More generally, we could allow $\omega$ that are neither holomorphic, nor of degree $n$. What matters is that $\omega|_B=0$, $\omega$ is closed fibrewise ($\td(\omega_t)=0$), and further that in a holomorphic chart $(x,t)\colon Y\supset U\rightarrow \CC^{n+q}$ with $t=\pi(x,t)$, $\omega$ is smooth in $x$ and holomorphic in $t$.
} Our assumption of global holomorphicity of $\omega$ imposes that the family of integrands $\omega_t$ is single-valued; hence the analytic continuation of $\II(t)$ is determined completely by the monodromy of the integration chain $\sigma_t$. %We comment on multi-valued integrands in \cref{sec:ramified-integrands}.
\begin{rem}
    In every Whitney stratification $\mathfrak{S}$ of $D=A\cup B$, the irreducible components $\irrone{D}$ are in bijection with the strata $S\in\mathfrak{S}$ of codimension one. Therefore, each $D_i=\overline{S}\in\irrone{D}$ is a union of strata by \cref{defn:strat}~(2). It follows that the subvarieties $A\subseteq D$ and $B\subseteq D$ are also unions of strata.
\end{rem}
Take a Whitney stratification $\mathfrak{S}$ of $D$, with its Landau variety $L=L(\pi,\mathfrak{S})$. A continuous trivialization $\phi\colon \pi^{-1}(U)\cong Y_{t_0}\times U$ over $t_0\in U\subseteq T\setm L$ from \cref{thm:Thom} preserves all strata. Since $A$ and $B$ are unions of strata, $\phi$ restricts to trivializations of the subvariety $\pi^{-1}(U)\cap A\cong A_{t_0}\times U$ and similarly for $B$. Hence the homeomorphism $\phi$ restricts to a compatible pair
\begin{equation}\label{eq:pair-of-bundles}\begin{gathered}\xymatrix{
    (Y\setm A) \cap \pi^{-1}(U) \ar[rr]^{\phi|_{Y\setm A}} & & (Y_{t_0}\setm A_{t_0}) \times U \\
    (B\setm A) \cap \pi^{-1}(U) \ar@{^{(}->}[u] \ar[rr]^{\phi|_{B\setm A}} & & (B_{t_0}\setm A_{t_0}) \times U \ar@{^{(}->}[u]
}\end{gathered}\end{equation}
of trivializations of the two fibre bundles $(Y\setm A)|_{\pi^{-1}(T\setm L)}$ and $(B\setm A)|_{\pi^{-1}(T\setm L)}$. This is what we mean by a \textbf{pair of fibre bundles} under \textbf{(L)} in the introduction \cref{sec:intro-mon}. Taking homology of this pair, it follows that also the groups $H_{\bullet}(Y\setm A,B)_t$ are locally trivial (a local system) over $T\setm L$. Therefore, a relative cycle $\sigma_{t_0}$ extends to a multivalued section $\sigma_t$ and hence determines the integral $\II(t)$ as a multivalued function on $T\setm L$. The analyticity of $\II(t)$ is shown in \cite[Theorem~56]{Brown:PeriodsFeynmanIntegrals}.

\section{Picard-Lefschetz theory}\label{sec:PL}

Throughout this section, we consider a proper holomorphic submersion $\pi\colon Y\rightarrow T$ and a subset $D=\bigcup_i D_i\subset Y$ that is the union of transverse smooth complex hypersurfaces $\irrone{D}=\irrone{A}\sqcup\irrone{B}=\set{A_1,\ldots,A_r,B_1,\ldots,B_s}$. We denote by $L\subset T$ the Landau variety of the canonical stratification $\mathfrak{S}$ of $D$.
\begin{defn}\label{def:small-simple}
	A \textbf{small loop} is a closed curve $\gamma\colon[0,1]\rightarrow T\setm L$ such that there exists a holomorphic chart $t\colon U\rightarrow \CC^k$ of $T$ with $U\cap L=\set{t_1=0}$ and $\gamma(\tau)=(e^{2\ipi\tau},0,\ldots,0)$ is the counter-clockwise loop in the $t_1$-plane.
	A \textbf{simple loop} is a loop which is homotopic to the conjugate $\inv{\xi}\conc\eta \conc \xi $ of a small loop $\eta$ by a path $\xi$ (\cref{fig:simple-loops}).
\end{defn}
\begin{figure}
    \centering
    \includegraphics[scale=.333]{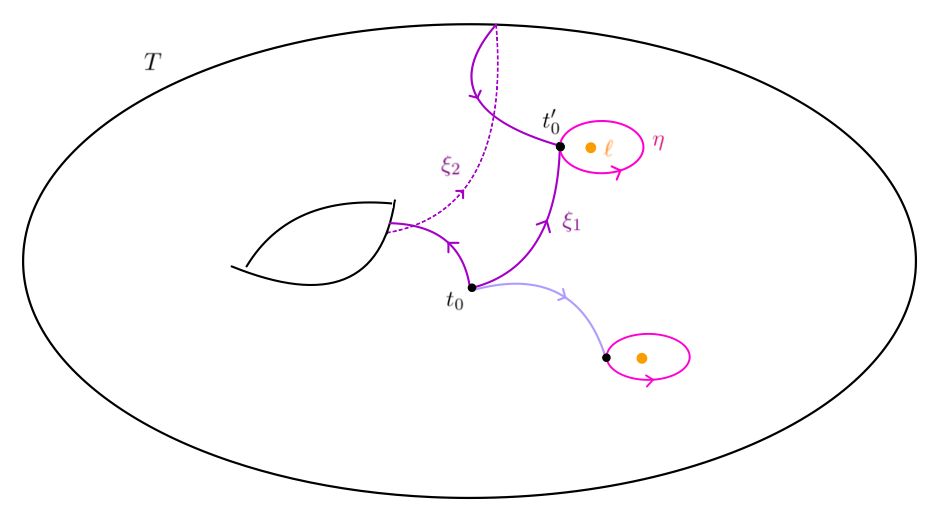}
    \caption{Simple loops around Landau components.}%
    \label{fig:simple-loops}%
\end{figure}
Note that simple loops become contractible under the inclusion $T\setm L\rightarrow T$. In fact, simple loops generate the entire kernel of the corresponding surjection
\begin{equation}\label{eq:fg-surjection}
    \pi_1(T\setm L,t_0) \twoheadrightarrow \pi_1(T,t_0)
\end{equation}
of fundamental groups of $T\setm L$, for any base point $t_0\in T\setm L$ (see \cref{sec:fg-codim1}).
Picard-Lefschetz theory determines the variation along small loops, but says nothing about loops that remain non-trivial in $\pi_1(T,t_0)$.\footnote{
	For this reason, \cite{Arnold:singularities2} and \cite{Pham:Singularities} assume that $T$ is simply connected.}
The latter emerge from the global topology of $T$, whereas the former arise from the degeneration of the fibres of $\pi$ over a critical point $t_c\in L$, and are in this sense local in $T$.
\begin{eg}
  In \cref{fig:simple-loops} the loops $\xi_1^{-1}\conc \eta\conc \xi_1$ and $\xi_2^{-1}\conc\eta \conc \xi_2$ are two non-homotopic simple loops. The loop $\xi_1^{-1}\conc \xi_2$ is not simple, because it is non-trivial in the fundamental group of the torus $T\cong \Sphere^1\times \Sphere^1$.
\end{eg}

Let $\irrone{L}$ denote the set of irreducible components $\ell\subseteq L$ of complex codimension one in $T$.
%By \cref{def:small-simple}, a small loop is localized near a \emph{smooth} point $t_c\in L$ of the Landau variety. That means that $t_c\in\ell$ is a smooth point of some a component $\ell\in\irrone{L}$, and away from all other components $\ell\cap\ell'$.
Since we are only concerned with small loops, we can replace $T$ with a disk $\oBallD=\set{\abs{t_1}<2}\subset\CC$ of complex dimension one in a suitable coordinate $t_1$. This disk embeds transverse to some component $\ell\in\irrone{L}$ such that $\oBallD\cap L=\set{t_c}$ is a single point of $\ell$. By \cref{def:small-simple}, $t_c$ is a smooth point of $\ell$, so that, in particular, $t_c\notin \ell\cap\ell'$ for any other component $\ell\neq \ell'\in\irrone{L}$.

Since $\pi$ is a proper submersion and hence $Y\rightarrow T$ a fibre bundle, this localization in $T$ also implies that we can assume that $Y\cong X\times T$ is trivial as a smooth manifold, with $X$ a smooth compact manifold and $\pi(x,t)=t$.

In fact, in this section we will only consider variations around Landau singularities $\ell\in\irrone{L}$ that emerge from fibrewise isolated critical points: We assume that over the point $t_c \in \ell$, the critical sets $cS \cap \pi^{-1}(t_c)$ consist of isolated points, for each stratum $S\in\mathfrak{S}$. In \cref{ss:loc} we show that this leads to a further localization, this time in $X$, to neighbourhoods of precisely those isolated points in the critical fibre. After this second localization, we may assume that $Y\cong X\times T$, with $\pi(x,t)=t$, also as a complex manifold.

\subsection{Monodromy and variation}
Let $\gamma\colon [0,1]\rightarrow T\setm L$ be a closed path based at $\gamma(0)=\gamma(1)=t_0$. Since $\pi \rest{\inv{\pi}(T \setm L)}$ is a pair of fibre bundles \eqref{eq:pair-of-bundles}, the pullback bundle $\gamma^* \pi$ over the contractible interval $[0,1]$ is trivial.\footnote{Glue together local trivializations over a sufficiently small subdivision of $\gamma$.} Hence one can find a continuous family $g_{\tau}\colon Y_{t_0} \rightarrow Y_{\gamma(\tau)}$ of homeomorphisms, which furthermore restrict to homeomorphisms $g_{\tau}|_S\colon S_{t_0}\rightarrow S_{\gamma(\tau)}$ of the fibres of every stratum $S\in\mathfrak{S}$ of $A\cup B$. The induced homeomorphisms
\begin{equation*}
(Y_{t_0}\setm A_{t_0},B_{t_0}\setm A_{t_0})
\longrightarrow 
(Y_{\gamma(\tau)}\setm A_{\gamma(\tau)},B_{\gamma(\tau)} \setm A_{\gamma(\tau)})    
\end{equation*}
of pairs are well-defined up to homotopy, and depend (up to homotopy) only on the homotopy class of $\gamma$.\footnote{For a homotopy $h\colon [0,1]^2\rightarrow T\setm L$ between paths $\gamma(t)=h(0,t)$ and $\eta(t)=h(1,t)$, the pullback bundle $h^*\pi$ on $[0,1]^2$ is also trivial ($[0,1]^2$ is simply connected). Any global trivialization of $h^*\pi$ provides a homotopy between $\gamma^* \pi=h^*\pi|_{\set{0}\times[0,1]}$ and $\eta^*\pi$.}
Hence on the level of the homology groups
\begin{equation*}
    H_\bullet(t) = H_\bullet( Y_t\setm A_t, B_t),
\end{equation*}
the pushforward of $g_\tau$ (restricted to $Y_{t_0}\setm A_{t_0}$) is independent of the choice of the trivialization $g$ of $\gamma^*\pi$. We denote this pushforward of $g_1$ by
\begin{equation*}
    \gamma_*\colon H_\bullet(t_0) \longrightarrow H_\bullet(t_0).
\end{equation*}
This map depends only on the homotopy class of $\gamma$. The assignment $\gamma \mapsto \gamma_*$ defines an action of the fundamental group $\pi_1(T \setm L,t_0)$ on the homology group $H_{\bullet}(t_0)$, called the \textbf{monodromy representation}.
\begin{defn}
    The \textbf{variation} operators on the homology group of a fibre measure the difference between the monodromy and the identity:
\begin{equation}\label{eq:Var}
    \Var_\gamma\colon H_\bullet(t_0) \longrightarrow H_\bullet(t_0),
    \quad h \longmapsto \gamma_* h - h.
\end{equation}
\end{defn}
\begin{rem}\label{rem:ambiso}
    If the fibres $Y_t=X$ are constant ($Y=X\times T$), then a trivialization $g$ of $\gamma^*\pi$ provides an \emph{ambient isotopy} $g_1$ of $(A \cup B)_{t_0}$ in $X$ that preserves each stratum individually.
    %$g_1|_{S_{t_0}}$ of $S_{t_0}$ in $X$, for every stratum $S\in\mathfrak{S}$.
    An ambient isotopy of a subspace $M\subseteq X$ is a homeomorphism $\phi\colon M\rightarrow M$ that can be extended to a continuous family $(g_{\tau})_{0\leq\tau\leq 1}$ of homeomorphisms $g_\tau\colon X \to X$, such that $g_0=\mathrm{id}$ and $g_1|_M=\phi$.
\end{rem}
 
Since we assume that $A\cup B$ has at worst simple normal crossing singularities, we can in fact find an ambient isotopy in the \emph{smooth} category (provided that the path $\gamma$ is smooth). This follows from \cref{prop:smoothom}, and means that the trivialization $g\colon X\times[0,1]\rightarrow X$ can be chosen to be smooth and such that each $g_{\tau}\colon X\rightarrow X$ is a diffeomorphism. Such an isotopy is also called a \emph{diffeotopy}  \cite{hirsch:dt}.

\begin{eg}\label{eg:pureA-nm=11}
    Consider projective space $X=\PP^1$ over the base $T=\CC$, with $B=\varnothing$ and remove $A=A_1\cup A_2$ where $A_1=\set{x=\infty}$ and $A_2=\set{x^2=t}$. All fibres $X\setm A_t=\CC\setm\set{\sqrt{t},-\sqrt{t}}$ over $t\neq 0$ are isomorphic via $x\mapsto x/\sqrt{t}$ to $X\setm A_{t=1}=\CC\setm\set{1,-1}$. This trivialization lifts the loop $t=\gamma(\tau)=e^{2\ipi\tau}$ around the Landau variety $L=\set{0}$ to the diffeotopy $g_{\tau}(x)=x\cdot e^{\ipi\tau}$. The latter ends in the antipode $g_1(x)=-x$, which swaps the punctures $\pm 1$ in the fibre over $t_0=1$. Let $\sigma_{\pm}$ be two loops around $\pm 1$. They generate $H_1(t_0)\cong\ZZ^2$ and the monodromy $\gamma_{*}$ swaps them, $\gamma_{*} (\sigma_{\pm})=\sigma_{\mp}$, see \cref{fig:quadr-global-isotopy}. Hence, the variation is
    \begin{equation*}
        \Var_{\gamma} (\sigma_+)=-\Var_{\gamma}(\sigma_-)
        = \sigma_{-}-\sigma_+.
    \end{equation*}
\end{eg}

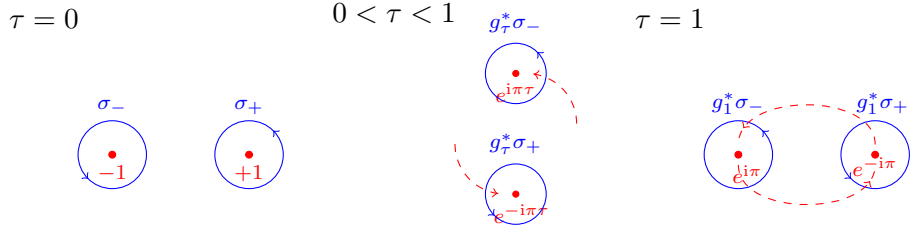
\begin{figure}
    \centering
 \begin{tikzpicture}[scale=.9]
\node[] (A) at (-2,2) {$\tau=0$};
\node[] (L) at (-1,0) {};
\node[] (R) at (1,0) {};
\node[] (O) at (0,1) {};
\node[] (U) at (0,-1) {};
\filldraw[red] (L) circle (0.05) node[below] {$\scriptstyle{-1}$};
\filldraw[red] (R) circle (0.05) node[below] {$\scriptstyle{+1}$};
\draw[blue,decoration={markings, mark=at position 0.625 with {\arrow{>}}},postaction={decorate}] (L) circle [radius=.5] node[above,yshift=.4cm] {$\scriptstyle{\sigma_-}$};
\draw[blue,decoration={markings, mark=at position 0.125 with {\arrow{>}}},postaction={decorate}] (R) circle [radius=.5] node[above,yshift=.4cm] {$\scriptstyle{\sigma_+}$};
\end{tikzpicture}
\quad
\begin{tikzpicture}[scale=.8]
\node[] (A) at (-2,2) {$0<\tau<1$};
\node[] (L) at (-1,0) {};
\node[] (R) at (1,0) {};
\node[] (O) at (0,1) {};
\node[] (U) at (0,-1) {};
\filldraw[red] (U) circle (0.05) node[below,xshift=.08cm] {$\scriptstyle{e^{-\ipi\tau}}$};
\filldraw[red] (O) circle (0.05) node[below] {$\scriptstyle{e^{\ipi\tau}}$};
\draw[blue,decoration={markings, mark=at position 0.625 with {\arrow{>}}},postaction={decorate}] (U) circle [radius=.5] node[above,yshift=.4cm] {$\scriptstyle{g_\tau^* \sigma_+ }$};
\draw[blue,decoration={markings, mark=at position 0.125 with {\arrow{>}}},postaction={decorate}] (O) circle [radius=.5] node[above,yshift=.4cm] {$\scriptstyle{g_\tau^* \sigma_-  }$};
\draw[red,dashed,decoration={markings, mark=at position 0.9 with {\arrow{>}}},postaction={decorate}] (L) to[out=-90,in=180] (U); 
\draw[red,dashed,decoration={markings, mark=at position 0.9 with {\arrow{>}}},postaction={decorate}] (R) to[out=90,in=0] (O); 
\end{tikzpicture}
\quad
\begin{tikzpicture}[scale=.9]
\node[] (A) at (-2,2) {$\tau=1$};
\node[] (L) at (-1,0) {};
\node[] (R) at (1,0) {};
\node[] (O) at (0,1) {};
\node[] (U) at (0,-1) {};
\filldraw[red] (R) circle (0.05) node[below,yshift=.1cm] {$\scriptstyle{e^{-\ipi}}$};
\filldraw[red] (L) circle (0.05) node[below,xshift=.12cm] {$\scriptstyle{e^{\ipi}}$};
\draw[blue,decoration={markings, mark=at position 0.625 with {\arrow{>}}},postaction={decorate}] (R) circle [radius=.5] node[above,yshift=.4cm,xshift=.1cm] {$\scriptstyle{g_1^* \sigma_+}$};
\draw[blue,decoration={markings, mark=at position 0.125 with {\arrow{>}}},postaction={decorate}] (L) circle [radius=.5] node[above,yshift=.4cm] {$\scriptstyle{g_1^* \sigma_- }$};
\draw[red,dashed,decoration={markings, mark=at position 0.9 with {\arrow{>}}},postaction={decorate}] (L) to[out=-90,in=-90] (R); 
\draw[red,dashed,decoration={markings, mark=at position 0.9 with {\arrow{>}}},postaction={decorate}] (R) to[out=90,in=90] (L); 
\end{tikzpicture}
    \caption{Isotopy for $A=\set{x^2=t}\subset\CC$ by rotation swaps the loops $\sigma_{\pm}$ (\cref{eg:pureA-nm=11}).}%
    \label{fig:quadr-global-isotopy}%
\end{figure}

\begin{eg} \label{eg:AB-nm=12}
    Consider $X=\PP^1$, $A=\set{x=t,\infty}$ and $B=\set{x=0,2}$ over the loop $t=\gamma(\tau)=e^{2\ipi \tau}$ (see \cref{fig:W-mono-lin}). The trivialization $(X\setm A,B)_{\gamma(\tau)}\cong (\CC\setm\set{1},\set{0,2})$ given by $g_{\tau}(x)=x e^{2\ipi\tau}$ yields the isotopy $g_{\gamma}(x)=x$. Let $t_0=1$. Choose a path $\sigma$ in $\CC \setm \{1\}$ from $0$ to $2$, and let $\delta\set{1}$ as above. Then $H(t_0)\cong \ZZ^2$ is generated by (the classes of) $\sigma$ and $\delta\set{1}$, and the monodromy $\gamma_*$ maps $\delta\set{1}$ to itself and $\sigma$ to $\sigma -\delta\set{1}$. Thus,
    \[
    \Var_\gamma (\delta\set{1})= 0, \qquad \Var_\gamma (\sigma)= - \delta\set{1}.
    \]
\end{eg}

\subsection{Localization} \label{ss:loc}
A crucial property of the variation operators is that both their domain and codomain can be localized near the critical points over $\ell$. 

In order to define a local variation operation, we need to show that if we are able to localize the critical sets $cS$ of all strata inside an open submanifold $W\subseteq X$, then the isotopy $g_\gamma$ can be chosen to be the identity outside $W$. Note that for fibrewise isolated critical points such a $W$ always exists by Milnor's theorem \cite{Milnor:SingComplex}---take for instance the union of small open balls around each critical point. If $\cup_{S \in \mathfrak{S}}cS$ is not isolated (fibrewise), but still smooth, then we can construct $W$ by taking small tubular neighbourhoods of the critical sets in a fibre.

\begin{lem}\label{lem:localization}
Let $W$ be an open submanifold of $X$ such that $\partial W$ is smooth and transverse to $A\cup B$, and such that $W$ contains the critical sets of all strata of $A_t\cup B_t$, for all $t$ in a small disk near $t_c \in L$. Then the ambient isotopy $g_\gamma\colon (X\setm A_{t},B_{t} \setm A_t) \to (X\setm A_{t},B_{t} \setm A_t)$ is homotopic to a map that is the identity outside of $W$.
\end{lem}

\begin{proof}
This is a special instance of the \emph{isotopy extension theorem} \cite[\S 8]{hirsch:dt}. 

Let $C$ denote the union of all critical sets $cS$ in the fiber over $t_c \in L$. Using a collar neighbourhood of $\partial W$ we can construct a closed submanifold $W'\subset W$ that contains $C$ and whose boundary $\partial W'$ is also transverse to all strata of $A\cup B$. This transversality implies that adding $\partial W$ and $\partial W'$ to the arrangement $A \cup B$ does not alter the associated Landau variety. From \cref{thm:Thom} and \cref{prop:smoothom} we obtain thus a smooth stratum preserving trivialization $h$ that gives rise to (see \cref{rem:ambiso}) a smooth ambient isotopy $g_\gamma\colon (X\setm A_{t},B_{t} \setm A_t) \to (X\setm A_{t},B_{t} \setm A_t)$. The crucial point here is that this map fixes $\partial W'$ (globally). 

Let $g_\tau$ be a realization of $g_\gamma$, that is, a family of (stratum preserving) diffeomorphisms $g_\tau \colon X \to X$ and subsets $S_\tau \subseteq X$ with $S_0=S_1=A_t \cup B_t \cup \partial W \cup \partial W'$, such that $g_0=\id_X$, $g_1=g_\gamma$ and $g_\tau|_{S_0}\colon S_0 \cong S_\tau$. The \emph{track} of $g$ is the diffeomorphism 
\begin{equation*}
G \colon I \times X \longrightarrow I \times X, \quad (\tau,x) \longmapsto (\tau ,g_\tau(x)).
\end{equation*}
Let $V$ be the \emph{velocity field} of $G$, the vector field on $I \times X$, tangential to the curves $\tau \mapsto (\tau, g_\tau(x))$. Identifying $T(\RR \times X)$ with $T(\RR)\oplus T(X)$ it has the form $V(\tau,x)=\frac{\partial}{\partial \tau} +  v(\tau,x)$ where $v(\tau,x)$ is a vector field on $X$. 
In this way every ambient isotopy gives rise to a time-dependent vector field on $X$. The reverse statement is also true if the time-dependent vector field is compactly supported. In our case $X$ is compact, so this condition is automatically satisfied.  
Moreover, observe that $V$ is tangential to the set $S=\{ (\tau,x)\mid x \in S_\tau \} \subseteq I \times X$, hence its flow (which is $G$) leaves $S$ invariant.

Since $C$ is contained in $W'$, the ambient isotopy $ g_\gamma$ will be homotopic to the identity map on its complement $X\setm W'$. Let $ V_\lambda(\tau,x)$ be the associated homotopy of vector fields on $I \times (X\setm W')$ (between $V_0(\tau,x)=V(\tau,x)|_{x\in X\setm W'}$ and the constant vector field $V_1(\tau,x)=\frac{\partial}{\partial \tau}$, whose flow induces the identity map). 
Choose a partition of unity $\{\rho_1,\rho_2\}$, subordinate to the open cover $U_1 \cup U_2=X$ with $U_1=W$ and $U_2=X\setm W'$, and define a vector field $F_\lambda$ on $I\times X$ by $F_\lambda(\tau,x)= \rho_1(x) V(\tau,x) + \rho_2(x) V_\lambda(\tau,x)$. Let $f_\lambda$ be the ambient isotopy induced by $F_\lambda$. Then
\begin{itemize}
    \item $f_\lambda|_{S_0} \colon S_0 \cong S_1$ for each $\lambda \in [0,1]$,
    \item $f_\lambda$ is a homotopy between $f_1$ and $f_0=g_\gamma$,
    \item $f_1$ is equal to $g_\gamma$ on $W'$,
    \item $f_1$ is the identity outside of $W$.
\end{itemize} 
\end{proof}

\begin{eg}\label{eg:pureA-nm=11-ambient}
    The isotopy $g_{\gamma}(x)=-x$ from \cref{eg:pureA-nm=11} is not the identity outside of $W=\set{\abs{x}<1}\subset\CC$. Pick any smooth function $\rho\colon \RR\rightarrow [0,1]$ with $\rho(r)=1$ for $r\geq 1$ and $\rho(r)=0$ for $r\leq 1/2$. Then the diffeomorphisms $\phi_{\tau}(x)=x\cdot e^{\ipi\tau\rho(\abs{x})}$ restrict to $x\mapsto x\cdot e^{\ipi\tau}$ outside $W$, and to the identity inside the smaller ball $W'=\set{\abs{x}<1/2}$. In particular, $\phi_{\tau}$ fixes the stratum $S=\set{x^2=t}\subseteq W'$ and thereby defines a diffeotopy of the pair $(\CC,S_t)$. This family $\phi_{\tau}$ is smooth in $\tau$, so $\phi_1$ is homotopic to $\phi_0=\id$. It follows that $g_{\gamma}\simeq f$ is homotopic the composition $f(x)=\phi_1 g_{\gamma} (x)=x \cdot e^{\ipi(1-\rho(\abs{x}))}$ which has the desired property that $f|_{W^{\setcompl}}=\id$ (\cref{fig:ambiso}).
\end{eg}
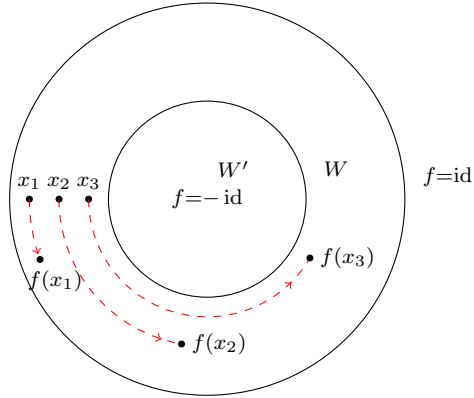
\begin{figure}
    \centering
  \begin{tikzpicture}[scale=1.3]
\node[] (m) at (0,0) {$\scriptstyle{ f= -\id}$};
\node[] (x1) at (-1.8,0) {};
\node[] (x2) at (-1.5,0) {};
\node[] (x3) at (-1.2,0) {};
\filldraw[] (x1) circle (0.03) node[above] {$\scriptstyle{x_1}$};
\filldraw[] (x2) circle (0.03) node[above] {$\scriptstyle{x_2}$};
\filldraw[] (x3) circle (0.03) node[above] {$\scriptstyle{x_3}$};
\draw (m) circle [radius=1] node[right,yshift=.4cm] {$\scriptstyle{ W' }$};
\draw (m) circle [radius=2] node[right,xshift=1.4cm,yshift=.4cm] {$\scriptstyle{ W }$} node[right,xshift=2.7cm,yshift=.3cm] {$\scriptstyle{ f=\id }$};
\draw[red,dashed,decoration={markings, mark=at position 0.9 with {\arrow{>}}},postaction={decorate}] (x1) arc[radius = 1.8cm, start angle= 180, end angle= 200] node[at end] (y1) {};
\filldraw[] (y1) circle (0.03) node[below,xshift=.2cm] {$\scriptstyle{f(x_1)}$};
\draw[red,dashed,decoration={markings, mark=at position 0.9 with {\arrow{>}}},postaction={decorate}] (x2) arc[radius = 1.5cm, start angle= 180, end angle= 260] node[at end] (y2) {};
\filldraw[] (y2) circle (0.03) node[right] {$\scriptstyle{f(x_2)}$};
\draw[red,dashed,decoration={markings, mark=at position 0.9 with {\arrow{>}}},postaction={decorate}] (x3) arc[radius = 1.2cm, start angle= 180, end angle= 330] node[at end] (y3) {};
\filldraw[] (y3) circle (0.03) node[right] {$\scriptstyle{f(x_3)}$};
\end{tikzpicture}
    \caption{\Cref{eg:pureA-nm=11-ambient} of an ambient isotopy $f=\phi_1 g_\gamma$ that is the identity outside of $W$ and homotopic to $g_\gamma(x)=-x$. Here $S=\{ x^2=t \}\subset \CC$ and $\phi_1|_{W'} = \id$, $\phi_1|_{W^\setcompl} = -\id$.}%
    \label{fig:ambiso}%
\end{figure}

By \cref{lem:localization}, the variation of a class $h\in H_{\bullet}(t_0)$ can be represented by a cycle contained in $W\setm A_{t_0}$ with boundary in $W \cap B_{t_0}$, and furthermore this variation $\Var_{\gamma} h$ depends only on how $h$ intersects the closed ball $\overline{W}$.
To make this precise, let $W^{\setcompl}=X\setm W$ denote the complement of the ball and consider the following %two maps, induced by inclusions of pairs:
map, induced by an inclusion of pairs:
\begin{equation}\label{eq:var-loc-dom-iso}
    %H_\bullet (X\setm A_t, B_t) \to H_\bullet (X\setm A_t, W^{\setcompl} \cup  B_t) \xleftarrow{\sharp} H_\bullet ( \overline{W} \setm A_t, \partial W \cup B_t ).
    H_\bullet ( \overline{W} \setm A_t, \partial W \cup B_t ) \longrightarrow H_\bullet (X\setm A_t, W^{\setcompl} \cup B_t).
\end{equation}
\begin{lem}\label{lem:var-loc-dom-iso}
    The map \eqref{eq:var-loc-dom-iso} is an isomorphism.
\end{lem}
\begin{proof}
    Extend the inclusion $\partial W\hookrightarrow W^{\setcompl}$ of the boundary to some collar neighbourhood $u\colon\partial W \times [0,1)\cong U \subset W^{\setcompl}$.
    \begin{figure}
        \centering
        \includegraphics[scale=0.23]{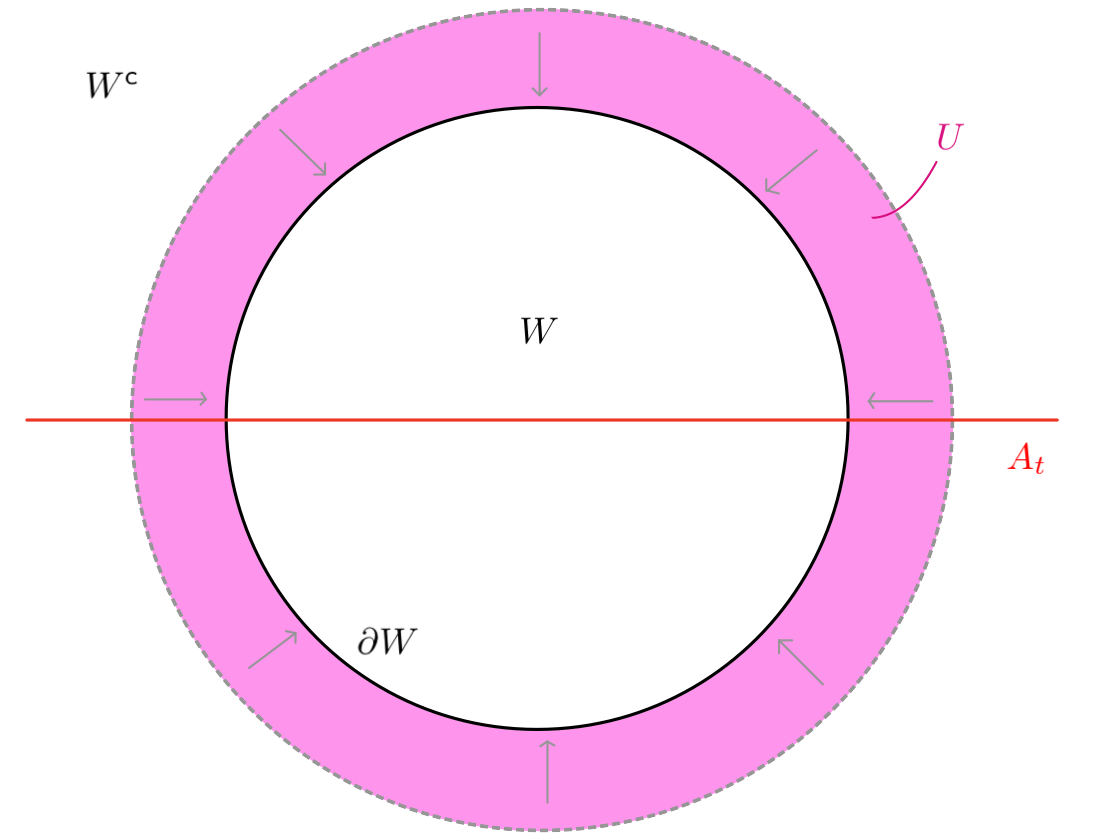}%
        \caption{Collar neighbourhood $U\cong [0,1) \times \partial W \subset W^\setcompl$ used for the excision argument in the proof of \cref{lem:var-loc-dom-iso}.}%
        \label{fig:excision1}%
    \end{figure}
    Adding the collar $U$, we can thicken $W$ into $W'=W\cup U$ and excise $W^{\setcompl}\setm (W'\cup A_t)$ from the pair $(X\setm A_t, W^{\setcompl} \cup  B_t\setm A_t)$, to replace it with the pair $(W'\setm A_t,U\cup B_t\setm A_t)$; see \cref{fig:excision1}. Since $\partial W\cup A_t\cup B_t$ intersect transversely, we can choose $u$ in such a way that it preserves $A$ and $B$ (see \cref{sec:relative-residues}). For the corresponding projection $r\colon U\rightarrow \partial W$, this means $r(U\cap A_t)=\partial W\cap A_t$ and $r(U\cap B_t)=\partial W\cap B_t$. Note that $r$ fixes $\partial W$ pointwise, so we can extend $r$ to a map $r\colon W'\rightarrow \overline{W}$ by setting $r|_W=\id_W$ and this still preserves $A_t$ and $B_t$. This $r$ descends to deformation retractions $(U\cup B_t)\setm A_t\rightarrow (\partial W\cup B_t)\setm A_t$ and $W'\setm A_t\rightarrow \overline{W}\setm A_t$.
\end{proof}
\begin{defn}
The \emph{$W$-trace} or \emph{trace-in-the-ball} \cite{FFLP} is the homomorphism
\begin{equation}\label{eq:wtrace}
    W_*\colon H_\bullet(X\setm A_t,B_t) \longrightarrow H_\bullet(\overline{W}\setm A_t,\partial W \cup B_t)
\end{equation}
given by the map of pairs $H_{\bullet}(X\setm A_t,B_t)\rightarrow H_{\bullet}(X\setm A_t,W^{\setcompl}\cup B_t)$ and the isomorphism \eqref{eq:var-loc-dom-iso}. We further write
%$\iota\colon H_{\bullet}(W\setm A_t,B_t)\longrightarrow H_{\bullet}(X\setm A_t,B_t)$ for the push-forward of the natural inclusion of pairs.
$\iota\colon (W\setm A_t,W\cap B_t\setm A_t)\rightarrow (X\setm A_t,B_t\setm A_t)$ for the natural inclusion of pairs, and denote its push-forward $\iota_*$.
\end{defn}
\begin{prop}\label{thm:loc-factorization}
    Suppose that $g_{\gamma}$ is the identity map outside of $W$. Then $g_{\gamma}$ canonically determines a homomorphism\footnote{Unlike $\Var_{\gamma}$, the map \eqref{eq:varlocal} may depend on the choice of $g_{\gamma}$ (not just $\gamma$ alone): if two realizations $g_{\gamma}$ and $g'_{\gamma}$ with $g_{\gamma}|_{W^{\setcompl}}=g_{\gamma}'|_{W^{\setcompl}}=\id$ are homotopic, but the homotopy does not fix $W^{\setcompl}$, then the corresponding variations $\var_{\gamma}$ could differ by elements in the kernel of $\iota_*$.}
    \begin{equation}\label{eq:varlocal}
        \var_\gamma \colon H_\bullet ( \overline{W} \setm A_t, \partial W \cup B_t ) \longrightarrow H_\bullet (W \setm A_t, B_t)
    \end{equation}
    such that the following diagram for $\Var_{\gamma}=g_{\gamma*}-\id$ commutes:
\begin{equation}\begin{gathered}\xymatrix{
    H_\bullet(X\setm A_t,B_t) \ar[d]^{W_{\ast}} \ar[rr]^{\Var_{\gamma}} & & H_\bullet(X \setm A_t,B_t) \\
    H_{\bullet}(\overline{W} \setm A_t,\partial W\cup B_t) \ar[rr]^{\var_{\gamma}} & & H_\bullet (W \setm A_{t},B_{t}) \ar[u]_{\iota_{\ast}} \\
}\end{gathered}\end{equation}
\end{prop}
\begin{proof}
    The triple
    $B_t \setm A_t \subset  (B_t\cup W^{\setcompl}) \setm A_t \subset X \setm A_t$
    gives a long exact sequence that identifies the kernel of $W_*$ with the image of
    \begin{equation*}
        j_*\colon H_{\bullet}((B_t\cup W^{\setcompl}) \setm A_t,B_t) \longrightarrow H_{\bullet}(X\setm A_t,B_t).
    \end{equation*}
    Analogously to the proof of \cref{lem:var-loc-dom-iso}, we can excise all of $\overline{W}\cap B_t\setm A_t$ from the domain of $j_*$, except for a layer $U\cap B_t\setm A_t\cong[0,1)\times(\partial W\cap B_t\setm A_t)$ in a collar neighbourhood of $\partial W$ inside $W$ (\cref{fig:excision2}).
\begin{figure}
    \centering
    \includegraphics[scale=.23]{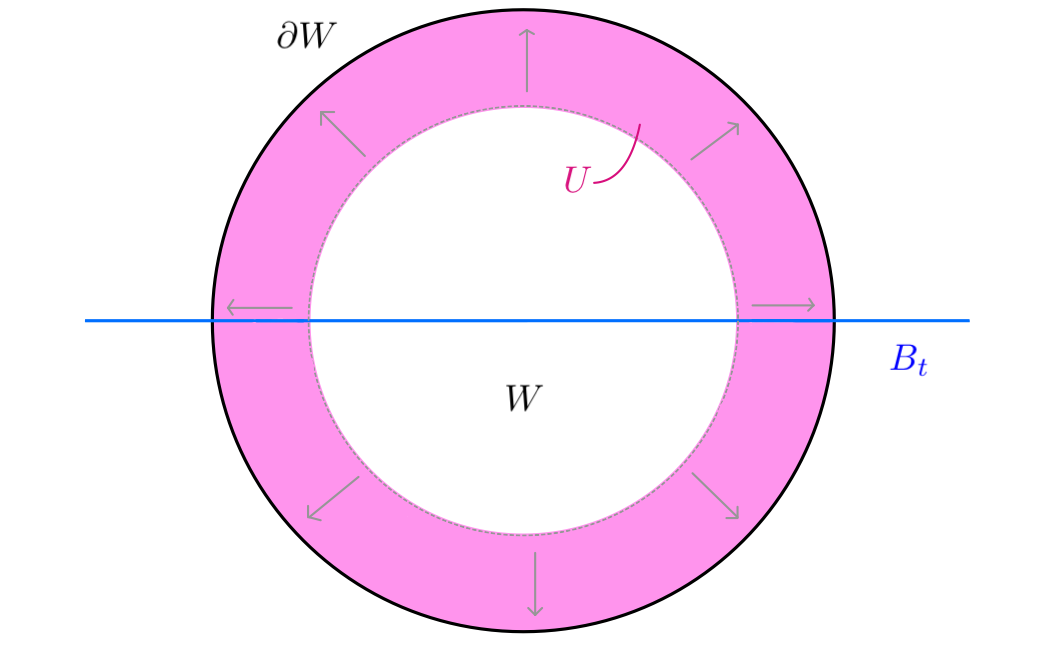}%
    \caption{Collar neighbourhood $U\cong[0,1)\times \partial W \subset \overline W$ used for the excision argument in the proof of \cref{thm:loc-factorization}.}%
    \label{fig:excision2}%
\end{figure}
    The deformation retraction 
    \begin{equation*}
        U \cap B_t \setm A_t \longrightarrow \partial W\cap B_t\setm A_t
    \end{equation*}
    then identifies the domain of $j_*$ with $H_{\bullet}(W^{\setcompl}\setm A_t,B_t)$. Every $h\in\ker W_*$ can thus be written as $h=j_* \sigma$ where $\sigma\in H_{\bullet}(W^{\setcompl}\setm A_t,B_t)$ has support inside $W^{\setcompl}$. This has variation
    \begin{equation*}
        \Var_{\gamma} h = (g_{\gamma*}-\id) j_* \sigma = j_* (g_{\gamma}|_{W^{\setcompl}\setm A_t}-\id)_*\sigma = j_* 0 \sigma = 0,
    \end{equation*}
    and we conclude that indeed, $\Var_{\gamma}$ factors through $W_*$. Now consider again the thicker $W'\supset W$ from \cref{lem:var-loc-dom-iso}. Since $X=W'\cup (X\setm \overline{W})$ is an open cover, we can represent every $h\in H_{\bullet}(X\setm A_t,B_t)$ by a chain $h=x+y$ with $x$ and $y$ supported in $W'\setm A_t$ and $X\setm(\overline{W}\cup A_t)$, respectively. Since $g_{\gamma*} y=y$, we see that $\Var_{\gamma} h=g_{\gamma*}x - x$ takes values in $H_{\bullet}(W'\setm A_t,B_t)$. But $H_{\bullet}(W\setm A_t,B_t)\rightarrow H_{\bullet}(W'\setm A_t,B_t)$ is an isomorphism, because the deformation retracts $W'\rightarrow \overline{W} \rightarrow \overline{W}\setm U\subseteq W$ (from \cref{lem:var-loc-dom-iso} and the above collar $U$ inside $W$) are compatible with $A_t$ and $B_t$. This proves the factorization of $\Var_{\gamma}$ through $\iota_*$.
\end{proof}

Here it is important to note that $\var_\gamma$ is not simply a localized version of $\Var_\gamma$; it does not carry the same information. The crucial point is that the latter \emph{factorizes} by means of $W_*$ and $\var_\gamma$, not that local and global variations are necessarily equal. See \cref{eg:pureA-nm=11}/\cref{fig:W-mono-quad} and \cref{eg:AB-nm=12}/\cref{fig:W-mono-lin} for examples where the two are indeed different. These highlight a general fact that holds for the \textit{simple pinch singularities} that we consider in the next section: For linear pinches the variation always vanishes in $H_\bullet(\overline{W}\setm A,\partial W \cup B)$, for quadratic pinches the variation vanishes if $n-m$ is odd.

\begin{figure}
   \centering%
\begin{tikzpicture}[scale=1]
\node[] (m) at (0,0) {};
\node[] (t1) at (0,1) {};
\node[] (t2) at (0,-1) {};
\draw[black,decoration={markings, mark=at position 0.5 with {\arrow{<}}},postaction={decorate}] (-2,0) -- (-1.5,0);
\draw[black,decoration={markings, mark=at position 0.5 with {\arrow{<}}},postaction={decorate}] (1.5,0) -- (2,0) node[right,xshift=-.3cm,yshift=0.2cm] {$\scriptstyle{h}$};
\draw[olive!15!green,decoration={markings, mark=at position 0.5 with {\arrow{<}}},postaction={decorate}] (-1.5,0) -- (1.5,0) node[left,xshift=0cm,yshift=0.2cm] {$\scriptstyle{W_*h}$};
\draw[blue] (m) circle [radius=1.5] node[above,xshift=1.5cm,yshift=1cm] {$\scriptstyle{ \partial W }$};
\filldraw[red] (t1) circle (0.03) node[below] {};
\filldraw[red] (t2) circle (0.03) node[above] {};
\draw[red,dashed,decoration={markings, mark=at position 1 with {\arrow{>}}},postaction={decorate}] (t1) arc[radius = 1cm, start angle= 90, end angle= 150] (t2) {};
\draw[red,dashed,decoration={markings, mark=at position 1 with {\arrow{>}}},postaction={decorate}] (t2) arc[radius = 1cm, start angle= 270, end angle= 330] (t1) {};
\end{tikzpicture}
\quad 
\raisebox{1cm}{$g_{\gamma_*}(h)=$}
\begin{tikzpicture}[scale=.7]
\node[] (m) at (0,0) {};
\node[] (t1) at (0,1) {};
\node[] (t2) at (0,-1) {};
\draw[black,decoration={markings, mark=at position 0.5 with {\arrow{<}}},postaction={decorate}] (-2,0) -- (-1.5,0);
\draw[black,decoration={markings, mark=at position 0.5 with {\arrow{<}}},postaction={decorate}] (1.5,0) -- (2,0);
\draw[olive!15!green,decoration={markings, mark=at position 0.5 with {\arrow{<}}},postaction={decorate}] plot [smooth] coordinates {(-1.5,0) (-1.3,-0.1) (-1,-.5) (-.5,-1) (0,-1.3) (.3,-1) (.5,-.5) (m) (-.5,.5) (-.3,1) (0,1.3) (.5,1) (1,.5) (1.3,0.1) (1.5,0)};
\draw[blue] (m) circle [radius=1.5];
\filldraw[red] (t1) circle (0.03) node[below] {};
\filldraw[red] (t2) circle (0.03) node[above] {};
\end{tikzpicture}
%%%%%%%%%%%%%%%%%%%%%%%%%%%%%%%%
\raisebox{1cm}{$\cong $ }
%%%%%%%%%%%%%%%%%%%%%%%%%%%%%%%%
\begin{tikzpicture}[scale=.7]
\node[] (m) at (0,0) {};
\node[] (t1) at (0,1) {};
\node[] (t2) at (0,-1) {};
\draw[black,decoration={markings, mark=at position 0.5 with {\arrow{<}}},postaction={decorate}] (-2,0) -- (-1.5,0);
\draw[black,decoration={markings, mark=at position 0.5 with {\arrow{<}}},postaction={decorate}] (1.5,0) -- (1.8,0);
\draw[olive!15!green,decoration={markings, mark=at position 0.5 with {\arrow{<}}},postaction={decorate}] (-1.5,0) -- (1.5,0);
\draw[blue] (m) circle [radius=1.5];
\filldraw[red] (t1) circle (0.03) node[below] {};
\filldraw[red] (t2) circle (0.03) node[above] {};
\draw[olive!15!green,decoration={markings, mark=at position 1 with {\arrow{>}}},postaction={decorate}] (.4,.9) arc[radius = .4cm, start angle= 0, end angle= 360];
\draw[olive!15!green,decoration={markings, mark=at position 1 with {\arrow{<}}},postaction={decorate}] (.4,-.9) arc[radius = .4cm, start angle= 0, end angle= 360];
\end{tikzpicture}
%%%%%%%%%%%%%%%%%%%%%%%%%%%%%%

\vspace{.3cm}
\hspace{3.7cm}
\raisebox{1cm}{$g_{\gamma_*}(\color{olive!15!green} W_*h \color{black})= \ $ }
\begin{tikzpicture}[scale=.7]
\node[] (m) at (0,0) {};
\node[] (t1) at (0,1) {};
\node[] (t2) at (0,-1) {};
\draw[olive!15!green,decoration={markings, mark=at position 0.5 with {\arrow{>}}},postaction={decorate}] (-1.5,0) --  (1.5,0);
\draw[blue] (m) circle [radius=1.5];
\filldraw[red] (t1) circle (0.03) node[below] {};
\filldraw[red] (t2) circle (0.03) node[above] {};
\end{tikzpicture}
%%%%%%%%%%%%%%%%%%%%%%%%%%%%%%%%
\raisebox{1cm}{ $\quad \cong \ $ }
%%%%%%%%%%%%%%%%%%%%%%%%%%%%%%%%
\begin{tikzpicture}[scale=.7]
\node[] (m) at (0,0) {};
\node[] (t1) at (0,1) {};
\node[] (t2) at (0,-1) {};
\draw[olive!15!green,decoration={markings, mark=at position 0.5 with {\arrow{<}}},postaction={decorate}] (-1.5,0) -- (1.5,0);
\draw[blue] (m) circle [radius=1.5];
\filldraw[red] (t1) circle (0.03) node[below] {};
\filldraw[red] (t2) circle (0.03) node[above] {};
\draw[olive!15!green,decoration={markings, mark=at position 1 with {\arrow{>}}},postaction={decorate}] (.4,.9) arc[radius = .4cm, start angle= 0, end angle= 360];
\draw[olive!15!green,decoration={markings, mark=at position 1 with {\arrow{<}}},postaction={decorate}] (.4,-.9) arc[radius = .4cm, start angle= 0, end angle= 360];
\end{tikzpicture}
    \caption{Global vs.\ local monodromy: For a quadratic pinch $A=\{x^2=t\}$ and $\gamma$ a small loop around $t=0$ the variation $\Var_\gamma$ of $h\in H_\bullet(X\setm A)$ is given by the vanishing cycle $\nu \in H_\bullet(W\setm A)$ (top right). The variation of $W_*h$ is $-2W_*h$ which is homologous to $\nu$ in $H_\bullet(\overline{W}\setm A,\partial W)$, but not in $H_\bullet(\overline{W}\setm A)$ (bottom right).}%
    \label{fig:W-mono-quad}%
\end{figure}
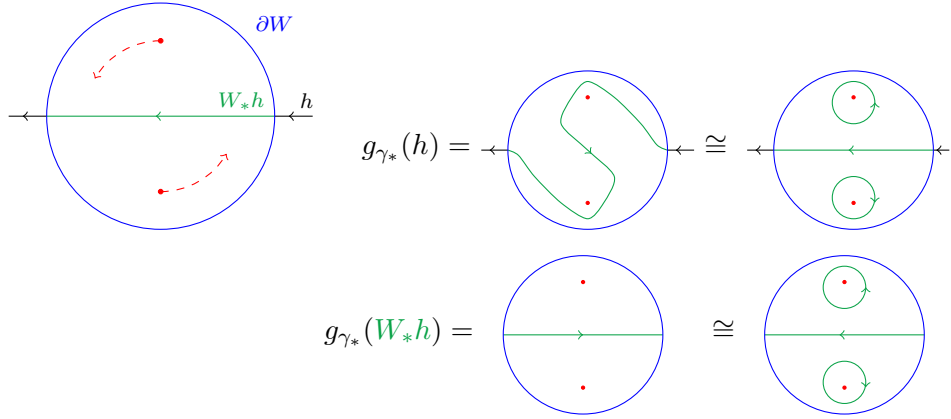

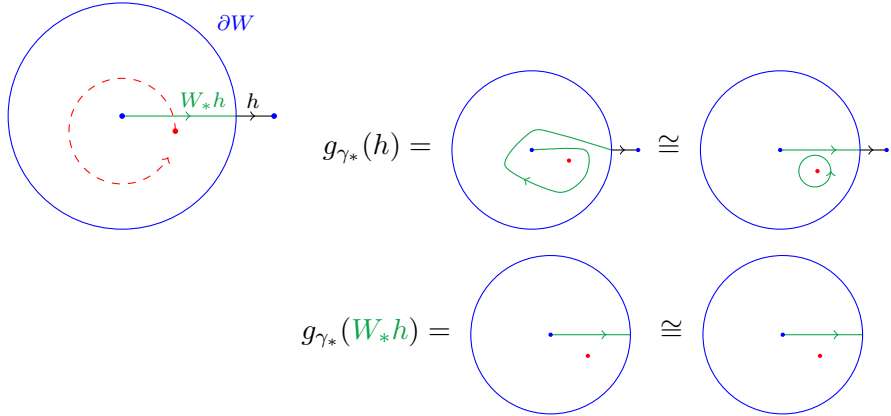
\begin{figure}
    \centering%
    \begin{tikzpicture}[scale=1]
\node[] (m) at (0,0) {};
\node[] (t1) at (.7,-.2) {};
\draw[black,decoration={markings, mark=at position 0.5 with {\arrow{>}}},postaction={decorate}] (1.5,0) -- (2,0) node[right,xshift=-.5cm,yshift=0.2cm] {$\scriptstyle{h}$};
\draw[olive!15!green,decoration={markings, mark=at position 0.6 with {\arrow{>}}},postaction={decorate}] (0,0) -- (1.5,0) node[left,xshift=0cm,yshift=0.2cm] {$\scriptstyle{W_*h}$};
\draw[blue] (m) circle [radius=1.5] node[above,xshift=1.5cm,yshift=1cm] {$\scriptstyle{ \partial W }$};
\filldraw[blue] (m) circle (0.03) node[below] {};
\filldraw[blue] (2,0) circle (0.03) node[below] {};
\filldraw[red] (t1) circle (0.03) node[above] {};
\draw[red,dashed,decoration={markings, mark=at position 1 with {\arrow{>}}},postaction={decorate}] (t1) arc[radius = .7cm, start angle= 0, end angle= 330] {};
\end{tikzpicture}
\quad 
\raisebox{1cm}{$g_{\gamma_*}(h)=$ }
\begin{tikzpicture}[scale=.7]
\node[] (m) at (0,0) {};
\node[] (t1) at (.7,-.2) {};
\draw[black,decoration={markings, mark=at position 0.5 with {\arrow{>}}},postaction={decorate}] (1.5,0) -- (2,0);
\draw[olive!15!green,decoration={markings, mark=at position 0.5 with {\arrow{>}}},postaction={decorate}] plot [smooth] coordinates {(m) (1,0) (1,-.5) (.5,-.8) (-.5,-.3) (0,.35) (.5,.3) (1.5,0) };
\draw[blue] (m) circle [radius=1.5];
\filldraw[blue] (m) circle (0.03) node[below] {};
\filldraw[blue] (2,0) circle (0.03) node[below] {};
\filldraw[red] (t1) circle (0.03) node[above] {};
\end{tikzpicture}
%%%%%%%%%%%%%%%%%%%%%%%%%%%%%%%%
\raisebox{1cm}{$\cong $ }
%%%%%%%%%%%%%%%%%%%%%%%%%%%%%%%%
\begin{tikzpicture}[scale=.7]
\node[] (m) at (0,0) {};
\node[] (t1) at (.7,-.4) {};
\draw[black,decoration={markings, mark=at position 0.5 with {\arrow{>}}},postaction={decorate}] (1.5,0) -- (2,0);
\draw[olive!15!green,decoration={markings, mark=at position 0.7 with {\arrow{>}}},postaction={decorate}] (0,0) -- (1.5,0);
\draw[olive!15!green,decoration={markings, mark=at position 1 with {\arrow{>}}},postaction={decorate}] (.95,-.4) arc[radius = .3cm, start angle= 0, end angle= 360];
\draw[blue] (m) circle [radius=1.5];
\filldraw[blue] (m) circle (0.03) node[below] {};
\filldraw[blue] (2,0) circle (0.03) node[below] {};
\filldraw[red] (t1) circle (0.03) node[above] {};
\end{tikzpicture}
%%%%%%%%%%%%%%%%%%%%%%%%%%%%%%

\vspace{.3cm}
\hspace{3.3cm}
\raisebox{1cm}{$g_{\gamma_*}(\color{olive!15!green} W_*h \color{black})=  $ }
\begin{tikzpicture}[scale=.7]
\node[] (m) at (0,0) {};
\node[] (t1) at (.7,-.4) {};
\draw[olive!15!green,decoration={markings, mark=at position 0.7 with {\arrow{>}}},postaction={decorate}] (0,0) -- (1.5,0);
\draw[blue] (m) circle [radius=1.5];
\filldraw[blue] (m) circle (0.03) node[below] {};
\filldraw[red] (t1) circle (0.03) node[above] {};
\end{tikzpicture}
%%%%%%%%%%%%%%%%%%%%%%%%%%%%%%%%
\raisebox{1cm}{ $\ \cong  $ }
%%%%%%%%%%%%%%%%%%%%%%%%%%%%%%%%
\begin{tikzpicture}[scale=.7]
\node[] (m) at (0,0) {};
\node[] (t1) at (.7,-.4) {};
\draw[olive!15!green,decoration={markings, mark=at position 0.7 with {\arrow{>}}},postaction={decorate}] (0,0) -- (1.5,0);
\draw[blue] (m) circle [radius=1.5];
\filldraw[blue] (m) circle (0.03) node[below] {};
\filldraw[red] (t1) circle (0.03) node[above] {};
\end{tikzpicture}
    \caption{Global vs.\ local monodromy: For a linear pinch $A=\{x=t\}$, $B=\{x=0,2\}$ and $\gamma$ a small loop around $t=0$ the variation $\Var_\gamma$ of $h\in H_\bullet(X\setm A,B)$ is given by the vanishing cycle $\nu \in H_\bullet(W\setm A)$ (top right). However, the variation of $W_*h$ in $H_\bullet(\overline{W}\setm A,\partial W \cup B)$ vanishes (bottom right).}
    \label{fig:W-mono-lin}%
\end{figure}

\begin{rem}
All of the above remains valid in the special cases $A=\varnothing$ or $B=\varnothing$. Furthermore, it holds also if we replace $X\setm A_{t_0}$ by other subsets of $X$ such as intersections, unions and complements of elements of $\irrone{A}$ and $\irrone{B}$. 
\end{rem}

\subsection{Simple pinch singularities} \label{sec:simplepinchsings}
The Picard-Lefschetz formulas from \cite{FFLP,Pham:Singularities} apply only where the degeneration of $A \cup B$ over $L$ is sufficiently benign. They require that the neighbourhoods of the critical points can be modelled by hypersurface arrangements of the following types (we comment on generalizations in \cref{sec:arbitrary-simple-pinch}).
\begin{defn}\label{defn:pinches}
Consider a critical point $p\in cS$ of the stratified map $\pi$, on some stratum $S\subseteq D_1\cap\ldots\cap D_{\PLm}$ with complex codimension $\PLm\leq n+1$. Then $p$ is called a \textbf{linear} ($m=n+1$) or \textbf{quadratic} ($m\leq n$) \textbf{simple pinch}\footnote{This is the terminology from \cite{Pham:Singularities}. In \cite{FFLP}, a simple pinch is called \emph{zero pinch}.} if there exist holomorphic local coordinates $t=(t_1,\ldots,t_q)$ at $\pi(p)\in T$ and $y=(x,t)\colon U\rightarrow \CC^{n+q}$ at $p\in U\subseteq Y$, trivializing $\pi(x,t)=t$, such that the hypersurfaces $D_i\cap U\cong S_i\cap y(U)$ take the form
\begin{equation}\label{eq:localcoords}
\begin{split}
    S_{i} & = \{ x_i =0 \}, \quad\text{for}\quad i=1, \ldots, \PLm-1 , \\ 
    S_{\PLm} &= \{ x_1 + \ldots + x_{\PLm-1} + x_{\PLm}^2 + \ldots + x_n^2=t_1\}.
\end{split}
\end{equation}
\end{defn}
These hypersurfaces $\set{S_1,\ldots,S_{\PLm}}$ are smooth and intersect transversely.
There are no critical points $c(S^I)=\varnothing$ on the intersections $S^I=\bigcap_{i\in I} S_i$ of a proper subset $I\subset \set{1,\ldots,\PLm}$.
The only critical stratum is the intersection of all $\PLm$ hypersurfaces. Its critical set
\begin{equation*}\label{eq:simple-pinch-crit}\tag{$\sharp$}
    c(S^{1,\ldots,\PLm}) = \set{x=0} \cap \set{t_1=0} = \set{(0,\ldots,0,t_2,\ldots,t_q)}
\end{equation*}
is smooth and projects isomorphically onto the Landau variety $\set{t_1=0}$; each fibre with $t_1=0$ has a unique critical point with coordinates $(0,t)$.

For a linear simple pinch, the fibres of $S^{1,\ldots,n+1}=c S^{1,\ldots,n+1}$ are empty over $t_1\neq 0$, and consist of a single point (the critical point) over $t_1=0$.

For a quadratic simple pinch, each fibre of the stratum $S^{1,\ldots,\PLm}$, over $t$ with $t_1\neq 0$, contains a real sphere %(as a deformation retract, see \cref{lem:vanishing-spheres})
\begin{equation*}
    \Sphere^{n-m} = \set{u\in\RR^{n-m+1}\colon \norm{u}=1},
\end{equation*}
embedded as $u\mapsto (0,\ldots,0,u\cdot \sqrt{t_1},t)$. As $t_1\rightarrow 0$, these spheres shrink and eventually collapse, over $t_1=0$, to a point (the critical point). They are therefore called \emph{vanishing spheres}.
\begin{defn}\label{defn:simplecomponent}
A \textbf{simple critical value} is a point $t\in L$ such that all critical points in the fibre $\pi^{-1}(t)$ are linear or quadratic simple pinches.

A \textbf{simple pinch component} of $L$ is a component $\ell\in \irrone{L}$ such that a generic point $t\in \ell$ is simple.
\end{defn}
Suppose $t_c\in L$ is a simple critical value, and denote $P=\pi^{-1}(t_c)\cap \bigcup_{S\in\mathfrak{S}} cS$ the corresponding simple pinches. Every $p\in P$ is an isolated critical point in the compact fibre $Y_{t_c}$, hence $P$ is finite. Pick a chart $y_p\colon U_p\rightarrow \CC^{n+q}$ of the form \eqref{eq:localcoords} for each $p$, centred at $y_p(p)=(0,0)$. By rescaling these coordinates, and shrinking the neighbourhoods $U_p$, we can ensure that (\cref{fig:localization}):
\begin{itemize}
    \item $U_p$ and $U_{p'}$ are disjoint when $p\neq p'$,
    \item the chart $y_p(U_p)$ contains the closure $\overline{W}\times\set{0}$ of the open ball
    \begin{equation*}
        W=\set{\norm{x}^2=\abs{x_1}^2+\ldots+\abs{x_n}^2<1} \subset \CC^n.
    \end{equation*}
    \item the only hypersurfaces $D_i\in\irrone{D}$ that intersect $U_p$ are those that contain $p$ (that is, those that correspond to the $S_i$).
\end{itemize}

\begin{figure}
    \centering
    \includegraphics[width=0.9\textwidth]{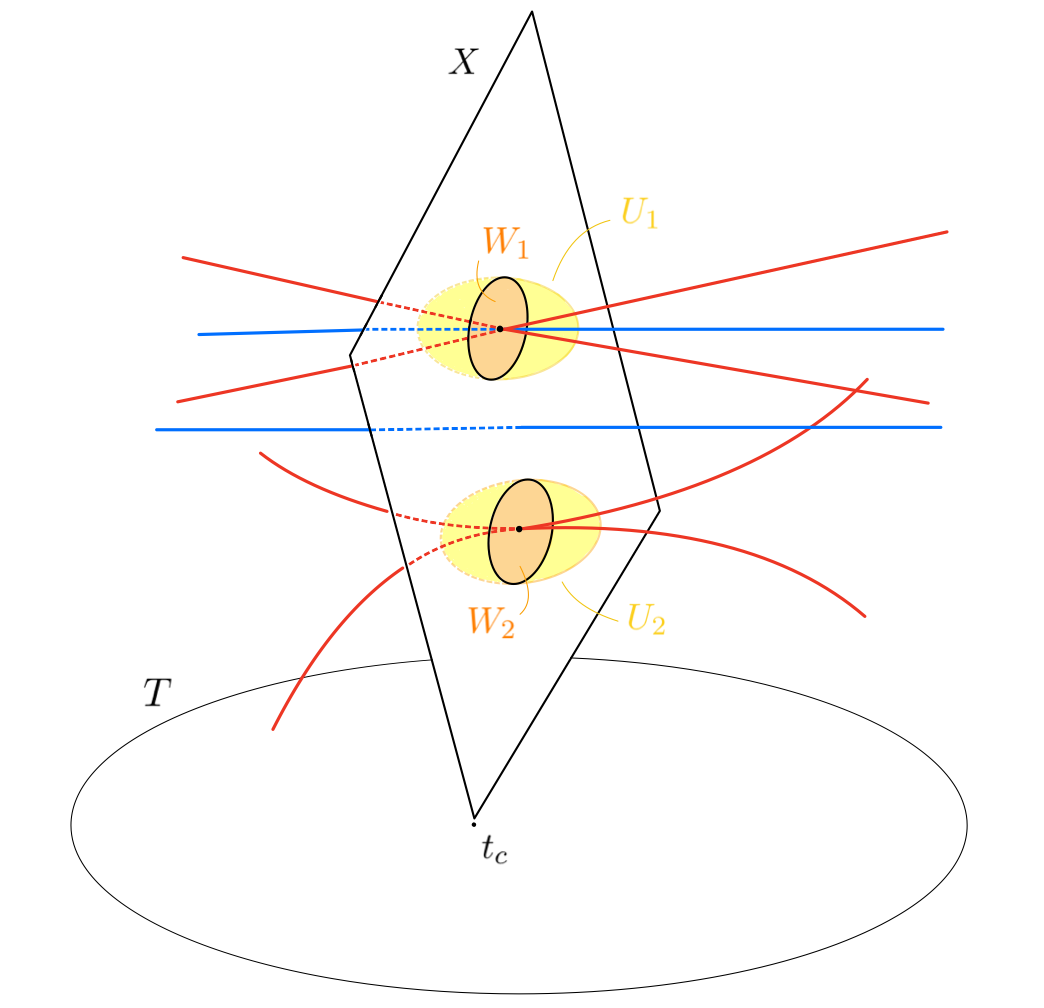}%
    \caption{Global and fibrewise localization.}%
    \label{fig:localization}%
\end{figure}

Consider now the corresponding open balls $W_p=y_p^{-1}(W\times\set{0}) \subset Y_{t_c}$ in the fibre over $t_c$. Observe that the unit sphere $\partial W=\Sphere^{2n-1}$ lies transverse to the arrangement \eqref{eq:localcoords}. It follows that the closure of the disjoint union
\begin{equation*}
    W_{P} \defas \bigsqcup_{p\in P} W_p \;\subset\; Y_{t_c}
\end{equation*}
embeds as a compact submanifold, whose boundary intersects $\irrone{D}$ transversely. Furthermore, $W_{P}\supset P$ contains all critical points over $t_c$. In fact, since $P$ is finite and the simple pinches \eqref{eq:simple-pinch-crit} depend continuously on $t$, we conclude that every local trivialization $\pi^{-1}(V)\cong X\times V$ of $Y$ at $t_c\in V\subseteq T$ can be restricted to an open $V'\subseteq V$, still containing $t_c$, such that $W_{P}\times V'$ contains all critical points over $V'$ (\cref{fig:localization}).

\begin{cor}\label{cor:sum-of-varp}
    Let $t_c \in L$ a simple critical value, and $[\gamma]\in\pi_1(T\setm L,t_0)$ a class that can be represented by a small loop near $t_c$. Then the variation along $\gamma$ decomposes into a sum
\begin{equation}\label{eq:sum-of-varp}
    \Var_{\gamma} = \sum_{p\in P} \Var_p
    \quad\text{where}\quad
    \Var_p = \iota_{p*} \var_p W_{p*}.
\end{equation}
\end{cor}
\begin{proof}
Apply \cref{lem:localization} to $W_P$ to find an ambient isotopy $g_{\gamma}$ of $D_{t_0}$ inside $X$, so that $g_{\gamma}$ is the identity outside $W_P$. By continuity, $g_{\gamma}(W_p)=W_p$ maps each ball to itself. Therefore, in the splitting
\begin{equation*}\xymatrix{
    \bigoplus_{p\in P} H_{\bullet}(\overline{W_p}\setm A,\partial W_p\cup B) \ar@{=}[r] \ar[d] & H_{\bullet}(\overline{W_P}\setm A,\partial W_P\cup B) \ar[d]^{\var_{\gamma}=\id-g_{\gamma*}} \\
    \bigoplus_{p\in P} H_{\bullet}(W_p\setm A, B) \ar@{=}[r] & H_{\bullet}(W_P\setm A, B)
}\end{equation*}
from the disjoint union $\overline{W_P}=\bigsqcup_p \overline{W_p}$, the map $\var_{\gamma}$ from \cref{thm:loc-factorization}, viewed as a matrix (arrow on the left), is diagonal: $\var_{\gamma}=\bigoplus_p \var_p$. The entry $\var_p$ is the restriction of $\id-g_{\gamma*}$ to $W_p$ and thus equal to the factorization of the variation $\Var_p = \id-g_{p*}$ induced by the homeomorphism $g_p\colon X\rightarrow X$ defined by $g_p(x)=g_{\gamma}(x)$ for $x\in W_p$ and $g_p(x)=x$ otherwise.
\end{proof}

For a simple pinch $p\in cS$ on a stratum $S$ of codimension $\PLm$, let $J\subseteq\irrone{A}$ and $K\subseteq\irrone{B}$ denote the $\PLm=\abs{J}+\abs{K}$ hypersurfaces in $A$ and $B$, respectively, that contain the pinch: $p\in D^{J\sqcup K}$. We chose $W_p$ so small that the pairs
\begin{align*}
    (W_p\setm A,B) &= (W_p\setm A_J,B_K) \quad\text{and} \\
    (\overline{W_p}\setm A,\partial W_p\cup B) &= (\overline{W_p}\setm A_J,\partial W_p\cup B_K)
\end{align*}
only meet precisely those hypersurfaces.\footnote{Here we write, as for the homology groups, simply $(W\setm A,B)$ for the pair $(W\setm A,W\cap B\setm A)$ etc.\ (the second element of a pair is always to be viewed as a subspace of the first).} We have thus reduced the computation of the variation $\Var_{\gamma}$, at least for small loops around Landau singularities $\ell\subseteq L$ with a simple critical value $t_c\in\ell$, to a calculation of $\var_p$ for the arrangement \eqref{eq:localcoords} within the unit ball $W\subset\CC^n$. The coordinates $y_p$ translate $J\sqcup K$ into a bipartition of $\set{S_1,\ldots,S_{\PLm}}$.

\subsection{Leray's residue and (co)boundary maps} \label{s:leray}
To state the relative Picard-Lefschetz theorem, we employ two operations on relative homology groups: the \textbf{partial boundary} $\partial_i$ and the \textbf{relative coboundary} $\fibSphere_i$. These were constructed by Leray in \cite{Leray:CauchyIII} and a review is provided in \cref{sec:relative-residues}.

In summary, we assume that $\irrone{D}=(D_1,D_2,\ldots)$ is a finite transverse family of smooth complex hypersurfaces $D_i \subset X$ inside a complex manifold $X$. For any index set $I$, we denote the corresponding union and intersection by
\begin{equation*}
    D_I = \bigcup_{i\in I} D_i
    \qquad\text{and}\qquad
    D^I = \bigcap_{i\in I} D_i.
\end{equation*}
Then for disjoint index sets $I\cap J=\varnothing$, and any $i\in J$, Leray's partial boundary comes from the usual boundary map in relative homology, but only keeps the piece that lies in $D_i\subseteq D_J$. This partial boundary is denoted
\begin{equation*}
    \partial_i\colon H_{\bullet}(X\setm D_I,D_J)
    \longrightarrow H_{\bullet-1}(D_i\setm D_I,D_{J-i})
\end{equation*}
and it fits into a long exact sequence of the form
\begin{equation}\label{eq:exact_boundary_seq}
    \cdots \rightarrow
    H_{\bullet}(X\setm D_I,D_{J-i}) \rightarrow
    H_{\bullet}(X\setm D_I,D_J) \xrightarrow{\partial_i}
    H_{\bullet-1}(D_i\setm D_I,D_{J-i}) \rightarrow \cdots
\end{equation}
where the other two maps are the natural inclusions of pairs. 

Leray's relative coboundary takes a chain in $D_i$ (for some $i\in I$) and fibres it in circles through a tubular neighbourhood around $D_i$. This yields a homomorphism
\begin{equation*}
\fibSphere_i \colon H_\bullet(D_i\setm D_{I-i},D_J) \longrightarrow H_{\bullet + 1 }(X\setm D_I,D_J)
\end{equation*}
which also fits into a long exact sequence
\begin{equation*}
    \cdots \rightarrow
    H_{\bullet}(X\setm D_{I-i},D_J) \xrightarrow{\varpi_i}
    H_{\bullet-2}(D_i\setm D_{I-i},D_J) \xrightarrow{\fibSphere_i}
    H_{\bullet-1}(X\setm D_I,D_J) \rightarrow \cdots
\end{equation*}
where the map $\varpi_i$ intersects a transverse chain with $D_i$.
For \emph{multiple} indices $K=\set{i_1,\ldots,i_r}\subseteq J$, we denote the corresponding \emph{iterated} boundaries and coboundaries as
\begin{align*}
    \partial_K &= \partial_{i_1}\!\cdots\partial_{i_r} \colon H_{\bullet}(X\setm D_I,D_{J}) \longrightarrow H_{\bullet-r}(D^K\setm D_I,D_{J-K})
    \quad\text{and} \\
    \fibSphere_K &= \fibSphere_{i_1}\!\cdots \fibSphere_{i_r} \colon H_{\bullet}(D^K\setm D_{I},D_{J-K})\longrightarrow H_{\bullet+r}(X\setm D_{I+K}, D_{J-K}).
\end{align*}
The ordering of these compositions only affects the overall sign, since for $i\neq j$, the maps $\partial_{i}\partial_j=-\partial_j\partial_i$ and $\fibSphere_i\fibSphere_j=-\fibSphere_j\fibSphere_i$ and $\partial_i\fibSphere_j=-\fibSphere_j\partial_i$ anticommute.
\begin{rem}
    On the level of differential forms (de Rham cohomology), the Leray coboundary is dual to Leray's residue map\footnote{In the category of mixed Hodge structures, a Tate twist factor $(-1)=H_{\dR}^1(\CC^{\times})=\CC$ appears on the right, to account for the weight of the factor $2\ipi$ in \eqref{eq:residue-theorem}.}
\begin{equation*}
    \Res_i \colon H^\bullet_{\dR}(X \setm D_I,D_J) \longrightarrow H_{\dR}^{\bullet - 1 }(D_i\setm D_{I-i},D_J).
\end{equation*}
The latter generalizes the Cauchy residue of a function: for any chain $\sigma$ in $D_i\setm D_{I-i}$ with boundary in $D_J$, and any closed smooth form $\omega\in\Omega^{\bullet}(X\setm D_I)$ that vanishes on $D_J$, their (co)homology classes $[\sigma]$ and $[\omega]$ integrate to
\begin{equation}\label{eq:residue-theorem}%\tag{$\natural$}
    \int_{\fibSphere_i [\sigma]} [\omega] = 2\ipi \cdot \int_{[\sigma]} \Res_i [\omega].
\end{equation}
\end{rem}

\subsection{Vanishing chains} \label{ss:vanishers}
 
The Picard-Lefschetz theorem describes the variation $\Var_p$ at a linear or quadratic simple pinch, in terms of certain homology classes of the arrangement \eqref{eq:localcoords}. The example in \cref{ss:Li1} illustrates these classes in dimension $n=1$; illuminating figures for higher dimensions can be found in \cite[Fig.~V.2]{Pham:Singularities} and \cite[Figs.~I.14--16]{Vassiliev:AppliedPL}.
 
From now on we employ the localization discussed in the context of small and simple loops (\cref{def:small-simple}), that is, we assume $T=\CC$ and $\ell=\set{0}$.
 
\begin{defn}\label{defn:van-cyc}
    Let $p\in cS$ denote a linear or quadratic simple pinch on a stratum $S\subseteq A^J\cap B^K$ of codimension $\PLm$. Let $A_j\in\irrone{A}$ ($j\in J$) and $B_k\in\irrone{B}$ ($k \in K$) denote the hypersurfaces containing $S$. Recall that, in coordinates of the form \eqref{eq:localcoords}, these define a bipartition $J\sqcup K=\set{1,\ldots,\PLm}$. Let $0<t<1$.
\begin{enumerate}
    \item The \textbf{vanishing cell} $\vcell \in H_n(W_p,A\cup B)_t \cong H_n(W,S_{1,\ldots,\PLm})_{t}$ is the class of the compact region in \emph{real} coordinates $x\in\RR^n$, bounded by all $\PLm$ hypersurfaces $S_{i}$, in the fibre over $t$:
    \begin{equation*}\quad\quad\quad
        \set{x_1,\ldots,x_{\PLm-1} \geq 0} \cap  \set{x_{\PLm}^2+\ldots+x_n^2\leq t-x_1-\ldots-x_{\PLm-1}} \subset \RR^n.
    \end{equation*}
    \item The \textbf{vanishing cycle} $\vcyc\in H_n(W_p\setm A,B)_t \cong H_n(W\setm S_J,S_K)_t$ is obtained from the vanishing cell by taking first the boundary in, and then the Leray coboundary around, each hypersurface $A_j$:
    \begin{equation}\label{eq:van-cycle}
        \vcyc=
        %\delta_{A_1}\partial_{A_1} \!\cdots \delta_{A_j}\partial_{A_j}
        \Big(\prod_{j\in J}\delta_{A_j}\partial_{A_j}\Big)\:
        %\Big(\prod_{j\in J}\delta_{j}\partial_{j}\Big)\: 
        \vcell
        %\in H_n(W_p\setm A,B)_t \cong H_n(W\setm S_J,S_K)_t
        %= (-1)^{\abs{J}(\abs{J}-1)/2} \fibSphere_J \partial_J \vcell
        .
    \end{equation}
    \item The \textbf{dual vanishing cycle} $\dvcyc\in H_n(W_p\setm B,A)_t \cong H_n(W\setm S_K,S_J)_t$ is obtained in the same way, but swapping the roles of $A$ and $B$:\footnote{The sign $(-1)^{\abs{K}}$ is included here to simplify the signs in the Picard-Lefschetz theorem.}
    \begin{equation*}
        \dvcyc=(-1)^{\abs{K}} \cdot \Big(\prod_{k\in K}\delta_{B_k}\partial_{B_k}\Big)\: \vcell
        %= (-1)^{\abs{K}(\abs{K}+1)/2} \fibSphere_K \partial_K \vcell
        .
    \end{equation*}
    \item Their images under the inclusion $\iota_p\colon W_p\hookrightarrow Y_t$ are denoted
    \begin{equation*}\quad\quad\quad\quad\quad
        \vcyc_p = \iota_{p*} \vcyc \in H_n(Y\setm A,B)_t
        \quad\text{and}\quad
        \dvcyc_p = \iota_{p*} \dvcyc \in H_n(Y\setm B,A)_t.
    \end{equation*}
\end{enumerate}
\end{defn}
\begin{rem}\label{rem:vcyc-ordering}
    The operators $\fibSphere_j \partial_j$ in \eqref{eq:van-cycle} are even, hence $\vcyc$ and $\dvcyc$ depend only on the orientation of $\vcell$, but not on any implicit ordering of $J$ or $K$. In fact, for \emph{any} ordering $J=\set{j_1,\ldots,j_r}$, we have $\vcyc=(-1)^{\abs{J}(\abs{J}-1)/2} \cdot \fibSphere_J \partial_J \vcell$ in terms of $\fibSphere_J=\fibSphere_{j_1}\!\cdots\fibSphere_{j_r}$ and $\partial_J=\partial_{j_1}\!\cdots\partial_{j_r}$.
\end{rem}
The vanishing cell $\vcell$ varies continuously with $t\in(0,1)$ and it literally vanishes as $t\rightarrow 0$; hence the name. 
For a linear pinch ($\PLm=n+1$), $\vcell$ is the simplex bounded by $x_i\geq 0$ (from $S_i$ for $i\leq n$) and $x_1+\ldots+x_n\leq t$ from
\begin{equation*}
    S_{n+1}=\{ x_1 + \ldots + x_n=t \}.
\end{equation*}
For a quadratic pinch ($\PLm\leq n$), $\vcell$ is a family of $n-\PLm+1$ dimensional disks (coordinates $x_{\PLm},\ldots,x_n$) over an $\PLm-1$ dimensional simplex $\set{x_1+\ldots +x_{\PLm-1}\leq t, x_i\geq 0}$, with the disks collapsing to a point over $x_1+\ldots+x_{\PLm-1}=t$. The maximal boundary of this cell,
\begin{equation*}
    %\partial_{J\sqcup K} 
    \partial_{1,\ldots,\PLm}
    \vcell = \pm[\Sphere^{n-\PLm}] 
    \in
    %H_{n-\PLm}(W_p\cap A^{J\sqcup K})_t =
    H_{n-m}(W\cap S^{1,\ldots,\PLm}_t),
\end{equation*}
is the fundamental class of the \emph{vanishing sphere}. This name refers to the embedded spheres $\Sphere^{n-m}\hookrightarrow S^{1,\ldots,\PLm}_t\subset\CC^n$, $u \mapsto (0,\ldots,0,u\cdot \sqrt{t})$.
See \cite[Fig.~V.2]{Pham:Singularities} and \cite[Figs.~I.14--16]{Vassiliev:AppliedPL} for figures of these vanishing chains for $n\leq 3$ and $m\leq 4$.

\begin{rem}
    We treat the zero-dimensional sphere $\Sphere^0=\set{1,-1} \subseteq \RR$ as an oriented manifold with the orientation induced from the interval $(-1,1)\subseteq \RR$. So its fundamental class is $[\Sphere^0]=[\set{1}]-[\set{-1}] \in H_0(\Sphere^0)$.
\end{rem}
The vanishing cell $\vcell$ is defined only up to a sign (orientation). There is no canonical choice, because the natural orientations on $\RR^n$ are not compatible between different parametrizations \eqref{eq:localcoords}; e.g.\ swap $x_1\leftrightarrow x_2$. This sign is the only ambiguity, because the vanishing cell is a generator of the cyclic group
\begin{align*}
    \vcell & \in  H_n (W_p, A\cup B)_t \cong \ZZ.
\end{align*}
This isomorphism is reviewed in \cref{prop:PL-isos-JKnon0}. The (dual) vanishing cycles $\vcyc$ and $\dvcyc$ also depend on the orientation of $\vcell$, but nothing more (\cref{rem:vcyc-ordering}).
Hence the sign ambiguity cancels in the quadratic expression $\is{\dvcyc}{h}\cdot \vcyc$ in the Picard-Lefschetz formula.

\begin{lem}\label{lem:vcell-monodromy}
    The monodromy of the vanishing cell of a simple pinch, along a simple loop $\gamma(\tau)=t\cdot e^{2\ipi\tau}$ based at $t\in(0,1)$, is
    \begin{equation*}
        \gamma_{*} \vcell = (-1)^{n-\PLm+1}\cdot \vcell.
    \end{equation*}
\end{lem}
\begin{proof}
    The trivialization $\phi_{\tau}(x)=(x_1 e^{2\ipi\tau},\ldots,x_{\PLm-1} e^{2\ipi\tau},x_{\PLm} e^{\ipi\tau},\ldots,x_n e^{\ipi\tau})$ identifies the fibres over $\gamma$. It lifts the vanishing cell to a section $\phi_{\tau*}(\vcell)\in H_n(W,S_{1,\ldots,\PLm})_{\gamma(\tau)}$.
    After one revolution, $\phi_{1*}$ acts by $(-1)^{n-\PLm+1}$ on the orientation $\td x_1\wedge\ldots\wedge \td x_n$, because $\phi_1(x)=(x_1,\ldots,x_{\PLm-1},-x_{\PLm},\ldots,-x_n)$.
\end{proof}
This simple computation determines the variation of the vanishing cycles, because variation commutes with (co)boundary maps:\footnote{This follows from \eqref{eq:[(co)boundary,pushforwad]} in \cref{sec:coboundary-comm-functor}, because $\Var_{\gamma}=g_{\gamma*}-\id$ is determined by a push-forward under a (stratified) diffeomorphism $g_{\gamma}$.}
\begin{equation}\label{eq:variation-vcyc}
    \Var_{\gamma} \vcyc_{p} = \left[(-1)^{n-\PLm+1}-1\right]\cdot \vcyc_{p}
    =\begin{cases}
        0 & \text{if $n-m$ odd, and} \\
        -2 \vcyc_{p} & \text{if $n-m$ even.} \\
    \end{cases}
\end{equation}
Beware this does not imply that $\Var_{\gamma}=0$ for $n-m$ odd. Consider the map of pairs $(W_p\setm A,B)\rightarrow (\overline{W}_p\setm A,\partial W_p\cup B)$, which induces
\begin{equation}\label{eq:W*i*}
    W_*\iota_*\colon H_n(W_p\setm A,B) \longrightarrow H_n(\overline{W}_p\setm A,\partial W_p\cup B).
\end{equation}
We will see later that, for odd $n-m\geq 1$, this map is zero. Hence $\Var_{\gamma}\vcyc_p=0$ simply reflects the fact that the $W$-trace of $\vcyc_{p}$ is zero; we learn nothing about $\var_{p}$ in \eqref{eq:sum-of-varp}. The linear pinch $n-m=-1$ is more subtle.

\subsection{Picard-Lefschetz theorem} \label{ss:PLthm}
Several variants of the Picard-Lefschetz theorem can be found for example in \cite{FFLP,Lamotke:TopVarLef,Lamotke:HomIsoSing}.
The following relative version for hypersurface arrangements $D=A\cup B$ is stated in \cite{Pham:Singularities}, albeit without a complete proof, and with an inaccurate interpretation for linear pinches.

\begin{thm}\label{thm:PL}
Let $p \in D$ be a linear or quadratic simple pinch over a smooth point $\pi(p)\in L$ of the Landau variety.
Then the local variation
\begin{equation*}
    \var_p\colon H_d(\overline{W}_p\setm A,\partial W_p\cup B)_t \longrightarrow H_d(W_p\setm A,B)_t
\end{equation*}
can be expressed in terms of the vanishing cycles $\vcyc\in H_n(W_p\setm A,B)_{t}$ and $\dvcyc\in H_n(W_p\setm B,A)_{t}$ from \cref{defn:van-cyc}:
\begin{itemize}
    \item In degree $d\neq n$, $\var_p=0$. %The variation in degrees other than $n$ is zero.
    \item In degree $d=n$, there exist integers $N(h)\in\ZZ$ such that%The variation of each $h\in H_n(\overline{W}_p\setm A,\partial W_p\cup B)_t$ is a multiple
\begin{equation} \label{eq:pic}
    \var_p(h) = N(h) \cdot \vcyc.
\end{equation}
    \item This identity holds for the intersection numbers
\begin{equation} \label{eq:lef}
    N(h) = (-1)^{(n+1)(n+2)/2} \cdot \is{\dvcyc}{h}.
\end{equation}
\end{itemize}
\end{thm}
\begin{rem}\label{rem:lin-pinch-types}
    The theorem covers linear and quadratic pinches simultaneously, but they behave slightly differently. At a quadratic pinch, the variation map $\var_p$ is always non-zero. At a linear pinch $p$, we distinguish three cases:
    \begin{itemize}
        \item If $p\in A\setm B$ (``$A$ type''), then $\vcyc_p=0$ and thus $\var_p=0$.
        \item If $p\in B\setm A$ (``$B$ type''), then $\dvcyc_p=0$ and thus $\var_p=0$.
        \item If $p\in A\cap B$ (``mixed type''), then $\var_p \neq 0$.
    \end{itemize}
    At an $A$-type linear pinch, \eqref{eq:pic} holds for \emph{any} choice of $N(h)$. For all other simple pinches, the integer $N(h)$ in \eqref{eq:lef} is the \emph{unique} solution to \eqref{eq:pic}.
\end{rem}
For a trace $h=W_{p*}\, \sigma$ of some cycle $\sigma \in H_n(Y\setm A,B)_{t}$, the intersection \eqref{eq:lef} inside the ball $W_p$ can also be thought of as taking place in the fibre $Y_t$,
\begin{equation*} 
    \is{\dvcyc}{W_{p*}\, \sigma}
    = \is{\dvcyc_p}{\sigma},
\end{equation*}
because $\partial W_p$ is transverse to $D$. The local formulas globalize with \eqref{eq:sum-of-varp}:
\begin{cor}
    For a small loop $\gamma$ around a simple pinch component $\ell\in\irrone{L}$, the variation $\Var_{\gamma}$ on $H_n(Y\setm A,B)_{t}$ can be written as:
\begin{equation}\label{eq:sum-of-varp-PL}
    \Var_{\gamma} (\sigma) = (-1)^{(n+1)(n+2)/2} \cdot \sum_{p} \is{\dvcyc_p}{\sigma} \cdot \vcyc_p.
\end{equation}
\end{cor}
\begin{rem}\label{rem:is-partialB}
    Let $A_j$ ($j\in J$) and $B_k$ ($k\in K$) denote the hypersurfaces that contain the pinch $p$. Iterating $\is{\fibSphere\dvcyc}{h}=(-1)^d\is{\dvcyc}{\partial h}$ from \eqref{eq:intersection-delta-partial}, with $d$ the degree of $h$, the intersection number reduces to
    \begin{equation*}
        \is{\dvcyc}{h}
        = (-1)^{n\abs{K}+\abs{K}(\abs{K}+1)/2}\cdot \is{\partial_K \vcell}{\partial_K h}.
    \end{equation*}
    The coefficient $N(h)$ can thus be computed on the manifold $\overline{M}=\overline{W_p}\cap B^K$ with its intersection pairing
    $
        H_{n-\abs{K}}(M,A') \times H_{n-\abs{K}}(\overline{M}\setm A',\partial M) \rightarrow \ZZ
    $
    for $A'=B^K\cap A_J$, after taking all $B_k$-boundaries of $\vcell$ and $h$.
\end{rem}

For the proof of \cref{thm:PL}, we fix a pinch $p$ (henceforth suppressing the subscript $_p$) and consider every disjoint triple $I\sqcup J\sqcup K\subseteq\set{1,\ldots,\PLm}$ of the hypersurfaces $S_1,\ldots,S_{\PLm}\subset\CC^n$ that contain $p$, as in \eqref{eq:localcoords}. The corresponding divisors $S_J\cap S^I\cap W$ and $S_K\cap S^I\cap W$ inside the manifold $W\cap S^I\subset \CC^n$ define pairs of fibre bundles with associated variations
\begin{equation}\label{eq:var-IJK}
    \var\colon H_{\bullet}(\overline{W}\cap S^I\setm S_J,\partial W\cup S_K)_t \longrightarrow H_{\bullet}(W\cap S^I\setm S_J,S_K)_t.
\end{equation}
\begin{rem}
    The configuration of a strict subset $I\sqcup J\sqcup K\subset \set{1,\ldots,\PLm}$ has no critical points and thus $\var=0$.
    Henceforth, we consider only the case of tripartitions $I\sqcup J\sqcup K=\set{1,\ldots,\PLm}$ of all $\PLm$ hypersurfaces.
\end{rem}

More precisely, we define the variations \eqref{eq:var-IJK} by choosing a \emph{single} isotopy $g_{\gamma}$ as in \cref{lem:localization}. We then obtain $\var=g'_*-\id$ via \cref{thm:loc-factorization}, using the corresponding restrictions $g'=g_{\gamma}|_{S^I\setm S_J}$ of $g_{\gamma}$.
Thus all variations arise from the same isotopy, and no confusion should arise from denoting all of them with the same symbol $\var$.

Furthermore, since (co)boundaries commute with push-forward, see \eqref{eq:[(co)boundary,pushforwad]}, the following diagram commutes for all $i\in I$:\footnote{This commutativity from \cref{sec:coboundary-comm-functor} requires that $\var$ is computed using a \emph{smooth} isotopy $g_{\gamma}$, hence our insistence on the smooth isotopy theorem (\cref{prop:smoothom}).}
\begin{equation}\label{eq:var-functorial}\xymatrixcolsep{12mm}\begin{gathered}\xymatrix{
    H_{\bullet}(\overline{W}\cap S^{I-i}\setm S_J,\partial W\cup S_{K+i})_t \ar[r]^-{\var} \ar[d]^{\partial_i} & H_{\bullet}(W\cap S^{I-i}\setm S_J,S_{K+i})_t \ar[d]^{\partial_i} \\
    H_{\bullet-1}(\overline{W}\cap S^I\setm S_J,\partial W\cup S_K)_t \ar[r]^-{\var} \ar[d]^{\fibSphere_i} & H_{\bullet-1}(W\cap S^I\setm S_J,S_K)_t \ar[d]^{\fibSphere_i} \\
    H_{\bullet}(\overline{W}\cap S^{I-i}\setm S_{J+i},\partial W\cup S_K)_t \ar[r]^-{\var} & H_{\bullet}(W\cap S^{I-i}\setm S_{J-i},S_K)_t \\
}\end{gathered}\end{equation}

Following \cite{FFLP}, the strategy to prove \cref{thm:PL} is to exploit these diagrams to reduce the calculation to a few special cases. For a quadratic pinch, we can always reduce to $J=K=\varnothing$ (\cref{sec:quadratic-pinch}). For a linear pinch (\cref{sec:linear-pinch}), we have to distinguish three different types, as in \cref{rem:lin-pinch-types}.

To spell out these reductions, we extend \cref{defn:van-cyc} and denote, for all $I\sqcup J\sqcup K=\set{1,\ldots,\PLm}$, the respective (dual) vanishing cycles as
\begin{equation*}\begin{aligned}
    \vcyc  &= \Big(\prod_{j\in J} \fibSphere_j \partial_j\Big) \partial_I \vcell
          & &\in H_{n-\abs{I}}(W\cap S^I\setm S_J,S_K)
    \quad\text{and}\\
    \dvcyc &= (-1)^{\abs{K}}\cdot \Big(\prod_{k\in K} \fibSphere_k \partial_k\Big) \partial_I \vcell
          & &\in H_{n-\abs{I}}(W\cap S^I\setm S_K,S_J)
    .
\end{aligned}\end{equation*}
Then \cref{thm:PL} arises as the special case $I=\varnothing$ of
\begin{thm}\label{thm:PL-IJK}
For every tripartition $(I,J,K)$, the variation \eqref{eq:var-IJK} is zero in all degrees other than $n-\abs{I}$. For a class $h$ in degree $n-\abs{I}$,
\begin{equation}\label{eq:PL-IJK}
    \var h =
    %(-1)^{(d+1)(d+2)/2}
    (-1)^{(n-\abs{I}+1)(n-\abs{I}+2)/2}
    \cdot \is{\dvcyc}{h} \cdot \vcyc.
\end{equation}
\end{thm}
Note that, for this formula to make any sense, it has to be compatible with the commutative diagram \eqref{eq:var-functorial}. This can be verified easily:
\begin{lem}
    Let $\tvar\colon H_{\bullet}(\overline{W}\cap S^I\setm S_J,\partial W\cup S_K)_t \longrightarrow H_{\bullet}(W\cap S^I\setm S_J,S_K)_t$ denote the map defined by the right-hand side of \eqref{eq:PL-IJK}. Then the diagram \eqref{eq:var-functorial}, with $\tvar$ in place of $\var$, is commutative.
\end{lem}
\begin{proof}
Consider the bottom square of \eqref{eq:var-functorial}, relating the variations of the triples $(I,J,K)$ and $(I',J',K')=(I-i,J+i,K)$. Let $\vcyc'$ and $\dvcyc'$ denote the vanishing cycles for $(I',J',K')$. Note that $\vcyc'=\epsilon\fibSphere_i \vcyc$ and $\dvcyc=\epsilon\partial_i\dvcyc'$, where the sign $\epsilon=\pm 1$ is determined by the order of $I$ such that $\partial_I=\epsilon \partial_i\partial_{I-i}$.
Thus,
\begin{align*}
    \fibSphere_i(\tvar h) &= (-1)^{(d+1)(d+2)/2}\cdot\is{\dvcyc}{h} \cdot \fibSphere_i \vcyc \\
    &= (-1)^{(d+2)(d+3)/2} \cdot \is{\dvcyc'}{\fibSphere_i h} \cdot \vcyc'
    = \tvar(\fibSphere_i h)
\end{align*}
follows from $\is{\dvcyc}{h}=\epsilon\is{\dvcyc'}{\fibSphere_i}\cdot(-1)^d$ by \eqref{eq:intersection-delta-partial}, where $d=n-\abs{I}$ denotes the degree of $h$ and $d+1=n-\abs{I'}$ is the degree of $\fibSphere_i h$. In the same way, one checks the upper square, where the vanishing cycles $\vcyc''$ and $\dvcyc''$ of the triple $(I-i,J,K+i)$ are related by $\vcyc=\epsilon\partial_i \vcyc''$ and $\dvcyc''=-\epsilon \fibSphere_i\dvcyc$.
\end{proof}
The proof of \cref{thm:PL-IJK} simplifies due to the following observation: if a vertical map ($\partial_i$ or $\fibSphere_i$) of the diagram \eqref{eq:var-functorial} is surjective in the domain (left column), or injective in the codomain (right column), then the validity of \eqref{eq:PL-IJK} for some partition implies that \eqref{eq:PL-IJK} also holds for a related partition.
\begin{cor}\label{cor:PL-IJK-reduction}
Subject to injectivity or surjectivity, each of the four vertical maps $\partial_i,\fibSphere_i$ in the diagram \eqref{eq:var-functorial} gives rise to an implication:
\begin{equation*}\xymatrixcolsep{20mm}\xymatrix{
    \fbox{\parbox{23mm}{\centering \eqref{eq:PL-IJK} holds for\\$(I-i,J+i,K)$}}
    \ar@{=>}@/^8mm/[r]^{\parbox{3cm}{\centering $\fibSphere_i$ injective\\ in codomain}}
    &
    \fbox{\parbox{23mm}{\centering \eqref{eq:PL-IJK} holds for $(I,J,K)$}}
    \ar@{=>}@/^8mm/[r]^{\parbox{3cm}{\centering $\partial_i$ injective\\ in codomain}} \ar@{=>}@/^8mm/[l]^{\parbox{3cm}{\centering $\fibSphere_i$ surjective\\ in domain}}
    &
    \fbox{\parbox{23mm}{\centering \eqref{eq:PL-IJK} holds for\\$(I-i,J,K+i)$}}
    \ar@{=>}@/^8mm/[l]^{\parbox{3cm}{\centering $\partial_i$ surjective\\ in domain}}
}\end{equation*}
\end{cor}
\begin{proof}
By the previous lemma, we can put the difference $\var-\tvar$ on the horizontal arrows in \eqref{eq:var-functorial} and obtain a commutative diagram. The validity of \eqref{eq:PL-IJK} is equivalent to $\var-\tvar=0$.
\end{proof}
In most cases, the (co)boundary maps are isomorphisms, and so the proof of \cref{thm:PL-IJK} reduces to a few boundary cases. 
For small $t\neq 0$, the groups
$
    H_{\bullet}(\overline{W}\cap S^I \setm S_J, \partial W \cup S_K)_t
$
and
$
    H_{\bullet}(W\cap S^I \setm S_J, S_K)_t
$
in the domain and codomain of the local variations \eqref{eq:var-IJK}, together with the (co)boundary maps between them, are computed in \cref{sec:homology-groups}. In particular, we find:
\begin{prop}\label{prop:PL-isos-JKnon0}
    For $K\neq\varnothing$, the codomain $H_{\bullet}(W\cap S^I \setm S_J, S_K)\cong \ZZ$ is concentrated in degree $n-\abs{I}$. The iterated (co)boundaries give isomorphisms
    \begin{equation*}
        \ZZ\cong H_n(W,S_{1,\ldots,\PLm})_t \xrightarrow[\delta_J \partial_{I\sqcup J}]{\cong} H_{n-\abs{I}}(W\cap S^I\setm S_J,S_K)_t.
    \end{equation*}
    For $J\neq\varnothing$, the domain $ H_{\bullet}(\overline{W}\cap S^I\setm S_J,\partial W\cup S_K)\cong\ZZ$ is concentrated in degree $n-\abs{I}$. The iterated (co)boundaries give isomorphisms
    \begin{equation*}
        H_{n-\abs{I}}(\overline{W}\cap S^I\setm S_J,\partial W\cup S_K)_t \xrightarrow[\delta_{I\sqcup K} \partial_K]{\cong} H_n(\overline{W}\setm S_{1,\ldots,\PLm},\partial W)_t\cong\ZZ.
    \end{equation*}
\end{prop}
\begin{proof}
See \cref{lem:codomain-Knon0-full} and \cref{ss:homgroups-duality}, in particular item (2).
\end{proof}
This result applies to linear and quadratic simple pinches. Note that in any degree other than $n-\abs{I}$, the domain or codomain, and hence the map $\var$, is zero (the first claim of \cref{thm:PL-IJK}).
Furthermore, with \cref{cor:PL-IJK-reduction} we learn that in order to prove \cref{thm:PL-IJK}, it is sufficient to confirm \eqref{eq:PL-IJK} for partitions with $\abs{J},\abs{K}\leq 1$.

\subsection{Linear pinch}\label{sec:linear-pinch}
For $J=K=\varnothing$, the identity \eqref{eq:PL-IJK} holds trivially, because the intersection $S^I=S^{1,\ldots,n+1}=\varnothing$ is empty. The factorization through the trivial group $0=H_{-1}(\varnothing)\ni\partial_{1,\ldots,n+1}\vcell$ also shows that:
\begin{itemize}
    \item If $J=\varnothing$, then $\dvcyc=\pm \fibSphere_K \partial_{1,\ldots,n+1} \vcell = 0$.
    \item If $K=\varnothing$, then $\vcyc=\pm \fibSphere_J\partial_{1,\ldots,n+1} \vcell = 0$.
\end{itemize}
So in these cases, \cref{thm:PL-IJK} amounts to $\var=0$. Let us first confirm this vanishing when $J\sqcup K=\set{k}$ is a singleton. Since $S^I=\bigcap_{i\neq k} S_i=\set{O_k}$ is then a single point---a corner of the vanishing cell---the groups
\begin{align*}
    H_{0}(\set{O_k})_t
    &=
    H_{0}(\overline{W} \cap S^{I}, \partial W\cup S_k )_t
    = H_{0}(W \cap S^{I}, S_k)_t \\
    &=
    H_{0}(\overline{W} \cap S^{I}\setm S_k ,\partial W)_t \,
    = H_{0}(W \cap S^{I}\setm S_k)_t
\end{align*}
are all \emph{canonically} identified $H_{0}(\set{O_k})_t\xrightarrow{\cong} \ZZ$, for all $t\neq 0$, by the augmentation map. Thus the class of this corner has no variation, and it also generates the domain of $\var$. Therefore we see that indeed $\var=0$.

For $J=\varnothing$, this vanishing extends to all $K$ via the commutative diagram
\begin{equation*}\xymatrixcolsep{15mm}\xymatrix{
H_{n-\abs{I}}(\overline{W}\cap S^I,\partial W\cup S_K) \ar[r]^-{\var} \ar[d]^{\partial_{K-k}} & H_{n-\abs{I}}(W\cap S^I,S_K) \ar[d]^{\cong}_{\partial_{K-k}} \\
H_0(O_k) \ar[r]^{\var=0}  & H_0(O_k)
}\end{equation*}
for any $k\in K$, because $\partial_{K-k}$ is an isomorphism in the codomain (\cref{prop:PL-isos-JKnon0}). Similarly, for $K=\varnothing$ and $j\in J$, the isomorphism $\fibSphere_{J-j}$ in the domain shows that $\var=0$ follows from the previous singleton case $\abs{J}=1$.

So far, we have proved $\var=0$ for all linear pinches of ``$A$ type'' ($K=\varnothing$) or ``$B$ type'' ($J=\varnothing$), as referred to in \cref{rem:lin-pinch-types}.

The remaining possibility is ``mixed type'', where both sets, $J$ and $K$, are non-empty. In this region, all boundaries and coboundaries are isomorphisms (\cref{prop:PL-isos-JKnon0}). Hence it suffices to verify the variation formula \eqref{eq:PL-IJK} for a single partition with $\abs{J}=\abs{K}=1$.\footnote{The only exception is at $n=1$, where there are no (co)boundary maps in the mixed region. Hence the two partitions $(I,J,K)=(\varnothing,\set{1},\set{2})$ and $(\varnothing,\set{2},\set{1})$ have to be considered separately. The calculations are the same, only $0$ and $t$ swap roles in \cref{fig:mixedlin-classes}.}
We pick $I=\set{2,\ldots,n}$, $J=\set{1}$ and $K=\set{n+1}$.
Inside the unit disk $W\cap S^I\cong\oBallD=\set{\abs{x}<1}\subset \CC$, the edge $\partial_I\vcell=[O_{1},O_{n+1}]$ of the vanishing cell is represented by a path from the corner $O_{1}=\set{x=t}$ to $O_{n+1}=\set{x=0}$. The groups
\begin{align*}
    H_1(W\cap S^I\setm S_J,S_K)_t &= H_1(\oBallD\setm\set{0},\set{t}) \cong\ZZ \quad\text{and} \\
    H_1(W\cap S^I\setm S_K,S_J)_t &= H_1(\oBallD\setm\set{t},\set{0}) \cong\ZZ
\end{align*}
are generated by the counter-clockwise vanishing cycles $\vcyc=\fibSphere_1\partial_1 \partial_I\vcell$ and $\dvcyc=-\fibSphere_{n+1}\partial_{n+1}\partial_I\vcell$, respectively (see \cref{fig:mixedlin-classes}).
The domain of the variation,
\begin{equation*}
    H_1(\overline{W}\cap S^I\setm S_J,\partial W\cup S_K)_t = H_1(\cBallD\setm\set{0},\partial \oBallD\cup \set{t})\cong\ZZ,
\end{equation*}
is generated by any path $h$ from the boundary circle $\partial\oBallD$ to $t$.
At the intersection of $h$ and $\dvcyc$, their tangent spaces $Th\wedge T\dvcyc$ are negatively oriented compared to the standard orientation $\td \Re x\wedge \td \Im x$ of $\CC$. The intersection number is therefore $\is{\dvcyc}{h}=-1$, and \cref{thm:PL-IJK} amounts to
%Its intersection number is $\is{\dvcyc}{h}=-1$, because at their intersection, the tangents $Th\wedge T\dvcyc$ are negatively oriented compared to the orientation $\td \Re x\wedge \td \Im x$ of $\CC$. Therefore, \cref{thm:PL-IJK} predicts
\begin{figure}
    \centering
  %  \includegraphics[width=0.6\textwidth]{mixedlin}
   % \quad
    %\includegraphics[width=0.2\textwidth]{mixedorient}
\begin{tikzpicture}[scale=1.4]
\node[] (m) at (0,0) {};
\node[] (t) at (1,0.5) {};
\node[] (b) at (2,0) {};
\coordinate (c1) at (4,1);
\coordinate (xr) at (4.5,1.5);
\coordinate (xl) at (3.5,0.5);
\coordinate (y1) at (3.5,1.5);
\coordinate (c2) at (4,-1.5);
\coordinate (x2) at (5,-1.5);
\coordinate (y2) at (4,-.5);
\draw[olive!15!green,decoration={markings, mark=at position 0.5 with {\arrow{>}}},postaction={decorate}] (2,0) -- (1,0.5) node[below right,xshift=.3cm,yshift=-.4cm] {$\scriptstyle{h}$};
\draw[orange,decoration={markings, mark=at position 0.5 with {\arrow{>}}},postaction={decorate}] (1,0.5) -- (0,0) node[above,xshift=.5cm,yshift=.45cm] {$\scriptstyle{\partial_I \vcell}$};
\draw[red,decoration={markings, mark=at position 0.9 with {\arrow{>}}},postaction={decorate}] (m) circle [radius = 0.41cm] node[left,xshift=-.5cm] {$\scriptstyle{\vcyc}$};
\draw[red,decoration={markings, mark=at position 0.1 with {\arrow{>}}},postaction={decorate}] (t) circle [radius = 0.41cm] node[above,yshift=.5cm] {$\scriptstyle{\dvcyc}$};
\draw[blue] (m) circle [radius=2] node[right,xshift=2.8cm] {$\scriptstyle{ \partial W }$};
\filldraw[] (0,0) circle (0.03) node[right,yshift=-.07cm] {$\scriptstyle{0}$};
\filldraw[] (t) circle (0.03) node[above] {$\scriptstyle{t}$};
\draw[red,decoration={markings, mark=at position 1 with {\arrow{>}}},postaction={decorate}] (xl) -- (xr) node[at end,yshift=-.4cm] {$\scriptstyle{T \dvcyc}$};
\draw[olive!15!green,decoration={markings, mark=at position 1 with {\arrow{>}}},postaction={decorate}] (c1) -- (y1) node[at end,yshift=.2cm] {$\scriptstyle{T h}$};
\draw[olive!15!green,decoration={markings, mark=at position 1 with {\arrow{>}}},postaction={decorate}] (c2) -- (x2) node[at end,yshift=.2cm] {$\scriptstyle{ \Re x }$};
\draw[red,decoration={markings, mark=at position 1 with {\arrow{>}}},postaction={decorate}] (c2) -- (y2) node[at end,xshift=.3cm] {$\scriptstyle{ \Im x}$};
\end{tikzpicture}
    \caption{The homology classes of a mixed linear pinch in degree $n-\abs{I}=1$.}%
    \label{fig:mixedlin-classes}%
\end{figure}
\begin{equation*}
    \var h = (-1)^{2\cdot 3/2} \cdot \is{\dvcyc}{h} \cdot \vcyc = \vcyc.
\end{equation*}
This finishes the proof of \cref{thm:PL-IJK} for linear pinches, since this is indeed the correct variation as computed by an isotopy (illustrated in \cref{fig:lin-pinch-proof}).

\begin{figure}
\begin{tikzpicture}[scale=1]
\node[] (m) at (0,0) {};
\node[] (t) at (1,0.25) {};
\node[] (b) at (2,0) {};
\coordinate (c1) at (4,1);
\coordinate (xr) at (4.5,1.5);
\coordinate (xl) at (3.5,0.5);
\coordinate (y1) at (3.5,1.5);
\coordinate (c2) at (4,-1.5);
\coordinate (x2) at (5,-1.5);
\coordinate (y2) at (4,-.5);
\draw[olive!15!green,decoration={markings, mark=at position 0.5 with {\arrow{>}}},postaction={decorate}] (2,0) -- (1,0.25) node[below right,xshift=-2.6cm,yshift=-.2cm] {$\scriptstyle{\gamma_*h}$};
\draw[olive!15!green,decoration={markings, mark=at position 0.5 with {\arrow{>}}},postaction={decorate}] (t) arc[radius = 1cm, start angle= 0, end angle= 330];
\draw[blue] (m) circle [radius=2] node[right,xshift=2cm,yshift=.4cm] {$\scriptstyle{ \partial W }$};
\filldraw[] (0,0) circle (0.03) node[right,yshift=-.07cm] {$\scriptstyle{0}$};
\filldraw[] (.85,-0.25) circle (0.03) node[right] {$\scriptstyle{t}$};
\draw[black,densely dotted,decoration={markings, mark=at position 1 with {\arrow{>}}},postaction={decorate}] (.85,-.25) arc[radius = 1cm, start angle= 330, end angle= 355];
\end{tikzpicture}
%%%%%%%%%%%%%%
\quad  \raisebox{2cm}{$\cong$} \qquad
%%%%%%%%%%%%%%
\begin{tikzpicture}[scale=1]
\node[] (m) at (0,0) {};
\node[] (t) at (1,0.25) {};
\node[] (b) at (2,0) {};
\coordinate (c1) at (4,1);
\coordinate (xr) at (4.5,1.5);
\coordinate (xl) at (3.5,0.5);
\coordinate (y1) at (3.5,1.5);
\coordinate (c2) at (4,-1.5);
\coordinate (x2) at (5,-1.5);
\coordinate (y2) at (4,-.5);
\draw[olive!15!green,decoration={markings, mark=at position 0.5 with {\arrow{>}}},postaction={decorate}] (2,0) -- (1,0.25) node[below right,xshift=.2cm,yshift=-.2cm] {$\scriptstyle{h}$};
\draw[olive!15!green,decoration={markings, mark=at position 0.9 with {\arrow{>}}},postaction={decorate}] (m) circle [radius = 0.41cm] node[left,xshift=-.4cm] {$\scriptstyle{\vcyc}$};
\draw[blue] (m) circle [radius=2] node[right,xshift=2cm,yshift=.4cm] {$\scriptstyle{ \partial W }$};
\filldraw[] (0,0) circle (0.03) node[right,yshift=-.07cm] {$\scriptstyle{0}$};
\filldraw[] (t) circle (0.03) node[above] {$\scriptstyle{t}$};
\end{tikzpicture}
\caption{Variation in the proof of \cref{thm:PL-IJK} for linear pinches.}%
\label{fig:lin-pinch-proof}%
\end{figure}

\subsection{Quadratic pinch}\label{sec:quadratic-pinch}
For a quadratic pinch ($\PLm\leq n$), even more of the (co)boundary maps are isomorphisms, at least in the degree $n-\abs{I}$ of interest. In addition to \cref{prop:PL-isos-JKnon0}, we show in \cref{sec:homology-groups}:
\begin{itemize}
    \item \Cref{cor:Kzero-isos}: For $K=\varnothing$, the coboundary maps are isomorphisms in the codomain,
\begin{equation*}
    %H_{n-\PLm}(\Sphere^{n-\PLm})\cong 
    \fibSphere_J\colon H_{n-\PLm}(W\cap S^{1,\ldots,\PLm}) \xrightarrow{\cong} H_{n-\abs{I}}(W\cap S^I\setm S_J).
\end{equation*}
    \item \Cref{ss:homgroups-duality}, (3): For $J=\varnothing$, boundaries are isomorphisms in the domain,
\begin{equation*}
    \partial_K\colon H_{n-\abs{I}}(\overline{W}\cap S^I,\partial W\cup S_K) \xrightarrow{\cong} H_{n-\PLm}(\overline{W}\cap S^{1,\ldots,\PLm},\partial W) %\cong H_0(\Sphere^{n-\PLm})
    .
\end{equation*}
\end{itemize}
Only two kinds of (co)boundary maps remain, that are not necessarily isomorphisms by the results above: the first coboundary ($J=\varnothing$, $i\in I$)
\begin{equation}\label{eq:first-coboundary}
    \fibSphere_i\colon H_{n-\abs{I}}(\overline{W}\cap S^I,\partial W\cup S_K)\longrightarrow H_{n-\abs{I}+1}(\overline{W}\cap S^{I-i}\setm S_i,\partial W\cup S_K)
\end{equation}
in the domain, and the last boundary ($K=\set{k}$) in the codomain:
\begin{equation}\label{eq:last-boundary}
    \partial_k\colon H_{n-\abs{I}}(W\cap S^I\setm S_J,S_k) \longrightarrow H_{n-\abs{I}-1}(W\cap S^{I+k}\setm S_J).
\end{equation}
\begin{lem}
    For a quadratic pinch, the first coboundaries \eqref{eq:first-coboundary} are surjective, and the last boundaries \eqref{eq:last-boundary} are injective.
\end{lem}
\begin{proof}
    Since all boundaries $\partial_K$ are isomorphisms in the domain and $\fibSphere_i\partial_K=(-1)^{\abs{K}}\partial_K\fibSphere_i$, it suffices to prove surjectivity of \eqref{eq:first-coboundary} for $K=\varnothing$. Then the term after \eqref{eq:first-coboundary} in the long exact residue sequence is zero,
    \begin{equation*}
        H_{n-\abs{I}+1}(\overline{W}\cap S^{I-i},\partial W)
        \cong H^{n-\abs{I}+1}(W\cap S^{I-i})
        = 0,
    \end{equation*}
    by Poincar\'{e}-Lefschetz duality and contractibility of $W\cap S^{I-i}$ (\cref{lem:intersect-contractible}). Similarly, for the injectivity of \eqref{eq:last-boundary}, it suffices to consider $J=\varnothing$, since the coboundaries $\fibSphere_J$ are isomorphisms in the codomain and $\partial_k\fibSphere_J=(-1)^{\abs{J}}\fibSphere_J\partial_k$.
    Then the group to the left of \eqref{eq:last-boundary} in the boundary sequence of $S_k$ is $H_{n-\abs{I}}(W\cap S^I)=0$, because $W\cap S^I$ is contractible and $\abs{I}<m\leq n$.
\end{proof}
\begin{rem}
    For linear pinch, one finds instead that \eqref{eq:first-coboundary} and \eqref{eq:last-boundary} are zero, because $\partial_k = (-1)^{\abs{J}}\fibSphere_J \partial_k \fibSphere_J^{-1} = 0$ and $\fibSphere_i=(-1)^{\abs{K}}\partial_K^{-1}\fibSphere_i \partial_K=0$ factor through the trivial homology group of $S^{1,\ldots,n+1}=\varnothing$.
\end{rem}

We conclude that, for a quadratic pinch, all coboundary maps are surjective in the domain, and all boundary maps are injective in the codomain.
Using \cref{cor:PL-IJK-reduction}, the proof of \cref{thm:PL-IJK} thus reduces to the single case $J=K=\varnothing$ of the complete intersection $I=\set{1,\ldots,\PLm}$. The first $\PLm-1$ coordinates are zero on this intersection $S^I$ and can be forgotten, so that
\begin{equation*}
    S^I \cong \set{x_1^2+\ldots+x_{r+1}^2=t} \subset \CC^{r+1}
\end{equation*}
where $r=n-\PLm$. The vanishing cycles $\vcyc=\dvcyc=[\Sphere^{r}]$ are equal and given by the fundamental class of the sphere $\Sphere^{r}=\set{u\in \RR^{r+1}\colon\norm{u}=1}$, embedded as
\begin{equation*}
    \Sphere^{r} \hookrightarrow W\cap S^{I},
    \quad
    u \mapsto x=u\cdot \sqrt{t}.
\end{equation*}
This vanishing sphere parametrizes the real part of $S^I$, and it is a deformation retract (\cref{lem:vanishing-spheres}) such that $H_{\bullet}(W\cap S^I) \cong H_{\bullet}(\Sphere^{r})$. The claim \eqref{eq:PL-IJK},
\begin{equation*}
    \var h = (-1)^{(r+1)(r+2)/2} \cdot \is{\vcyc}{h}\cdot \vcyc,
\end{equation*}
is the classical Picard-Lefschetz formula. A proof can be found e.g.\ in \cite{Lamotke:TopVarLef,Lamotke:HomIsoSing}; we include a proof in \cref{sec:quadratic-pinch-JKzero} for convenience of the reader.

\begin{rem}
    As a corollary of the Picard-Lefschetz theorem and \eqref{eq:variation-vcyc}, %\cref{lem:vcell-monodromy} 
    we can read off the well-known self-intersection number \cite{Fary:SelfSphere} of the real sphere $\Sphere^r$ inside the quadric $\set{x_1^2+\ldots+x_{r+1}^2=1}\subset \CC^{r+1}$:
    \begin{equation*}
        \is{[\Sphere^r]}{[\Sphere^r]} = (-1)^{(r+1)(r+2)/2}\cdot \left[(-1)^{r+1}-1\right]
        =\begin{cases}
            0 & \text{if $r$ odd, and} \\
            2\cdot (-1)^{r/2} & \text{if $r$ even.}\\
        \end{cases}
    \end{equation*}
\end{rem}

\subsection{Arbitrary simple pinch}\label{sec:arbitrary-simple-pinch}
The observations in \cref{sec:quadratic-pinch} generalize beyond quadratic pinches. Extending \cref{defn:pinches}, we say that a critical point $p\in cS$ on a stratum $S\subseteq D_1\cap\ldots\cap D_{\PLm}$ of codimension $\PLm$ is a \textbf{simple pinch} if there exist local coordinates $(x,t)$ near $p$ such that $\pi(x,t)=t$ and
\begin{equation}\label{eq:general-simple-pinch}\begin{aligned}
    D_1 &= \set{x_1=0},\quad\ldots,\quad D_{\PLm-1}=\set{x_{\PLm-1}=0},\\
    D_m &= \set{x_1+\ldots+x_{\PLm-1}+f(x_{\PLm},\ldots,x_n)=t}
\end{aligned}\end{equation}
for some polynomial $f\colon \CC^{r+1}\rightarrow \CC$ with an isolated critical point at the origin ($r=n-\PLm$). Milnor \cite{Milnor:SingComplex,MolinaSeade:MilnorRC} showed that the boundary $\partial W\cong\Sphere^{2n-1}$ of the ball $W=\set{\norm{x}<\epsilon}$, with $\epsilon$ small enough, intersects the arrangement $D_{1,\ldots,\PLm}$ transversely (hence the variation localizes) and moreover, the local fibres
\begin{equation}\label{eq:milnor-bouquet}
    S_t\cap W \simeq \Sphere^r\vee\ldots\vee\Sphere^r
\end{equation}
over sufficiently small $t\neq 0$ are homotopy equivalent to a bouquet (wedge sum) of finitely many spheres. The number of these spheres is called the \emph{Milnor number} $\nMil$. For a quadratic pinch, $\nMil=1$. For any proper subset $I\subsetneq\set{1,\ldots,\PLm}$, the arrangement $D_I$ has no critical point, $\nMil=0$, and thus
\begin{equation*}
    D^I_t\cap W \simeq \set{*}
\end{equation*}
is contractible, with the homology of a point. With this generalization of \cref{lem:intersect-contractible}, and \eqref{eq:milnor-bouquet} standing in for \cref{lem:vanishing-spheres}, the proofs in \cref{sec:homology-groups} still apply---we only need to replace $\ZZ$ by $\ZZ^{\nMil}$. In particular, for \emph{any} non-linear simple pinch ($\PLm\leq n$), it remains true that all boundary maps in the domain, and all coboundary maps in the codomain, are isomorphisms (in degree $n-\abs{I}$). Moreover, for any tripartition $I\sqcup J\sqcup K=\set{1,\ldots,\PLm}$, the maps
\begin{align*}
    H_{n-\abs{I}}(\overline{W}\cap D^I\setm D_J,\partial W\cup D_K)_t 
    &\stackrel{\fibSphere_{I\sqcup K}\partial_{K}}{\relbar\joinrel\relbar\joinrel\relbar\joinrel\relbar\joinrel\twoheadrightarrow} 
    H_n(\overline{W}\setm D_{1,\ldots,\PLm},\partial W)_t \cong\ZZ^{\nMil} \\
    \ZZ^{\nMil}\cong H_n(W, D_{1,\ldots,\PLm})_t 
    &\stackrel{\fibSphere_J\partial_{I\sqcup J}}{\lhook\joinrel\relbar\joinrel\relbar\joinrel\relbar\joinrel\relbar\joinrel\rightarrow} 
    H_{n-\abs{I}}(W\cap D^I\setm D_J,D_K)_t
\end{align*}
stay surjective in the domain, and injective in the codomain, of $\var_p$.
\begin{cor}\label{cor:varp-simple-factor-JK}
    The local variation of any simple pinch $p\in A^J\cap B^K$ of codimension $\PLm=\abs{J\sqcup K}$ factorizes through the boundary $\partial_K$ in the domain, and the coboundary $\fibSphere_J$ in the codomain:\footnote{For a linear pinch, these maps can vanish: $\fibSphere_J=0$ if $K=\varnothing$ ($A$ type) and $\partial_K=0$ if $J=\varnothing$ ($B$ type). Nevertheless, the factorization still holds, because $\var_p=0$ in these cases.}
\begin{equation}\label{eq:varp-simple-factor-JK}\begin{gathered}\xymatrixcolsep{17mm}\xymatrix{
    H_n(\overline{W}\setm A,\partial W\cup B)_t \ar[r]^-{\var_p} \ar@{->>}[d]^{\fibSphere_K\partial_K} & H_n(W\setm A,B)_t \\
    H_n(\overline{W}\setm D_{1,\ldots,\PLm},\partial W)_t \ar[d]^{\cong} & H_n(W,D_{1,\ldots,\PLm})_t \ar@{^{(}->}[u]^{\fibSphere_J\partial_J}  \\
    \ZZ^{\nMil} \ar[r] & \ZZ^{\nMil} \ar[u]^{\cong} \\
}\end{gathered}\end{equation}
\end{cor}

The fundamental classes $\vcyc_i=[\Sphere^r]\in\widetilde{H}_{n-r}(W\cap D^{1,\ldots,\PLm})_t \cong \ZZ^{\nMil}$ define a basis of the reduced homology of the bouquet \eqref{eq:milnor-bouquet}, to which corresponds a basis $\vcell_1,\ldots,\vcell_\nMil\in H_n(W,D)_t \cong\ZZ^{\nMil}$ of vanishing cells such that $\vcyc_i=\partial_{1,\ldots,\PLm} \vcell_i$. By Poincar\'{e}-Lefschetz duality, the intersections $h\mapsto\is{\vcell_k}{h}$ yield a basis of the dual of $H_n(\overline{W}\setm D_{1,\ldots,\PLm},\partial W)\cong\ZZ^{\nMil}$, and we can write
\begin{equation*}
    \var_p(h) = \sum_{i,j=1}^{\nMil} M_{ij} \cdot \is{\vcell_i}{\fibSphere_K\partial_K h} \cdot \fibSphere_J\partial_J \vcell_j
\end{equation*}
for some matrix $M\in M_{\nMil \times \nMil}(\ZZ)$. Hence, if we define $\dvcyc_j=\pm \fibSphere_K\partial_K \sum_i M_{ij}\vcell_i$ and $\vcyc_i=\pm \fibSphere_J\partial_J \vcell_i$, with suitable signs as in \cref{defn:van-cyc}, then we can write a Picard-Lefschetz formula for an arbitrary simple pinch in the form
\begin{equation}\label{eq:PL-any-simple-pinch}
    \var_p(h) = (-1)^{(n+1)(n+2)/2}\cdot\sum_{i=1}^{\nMil} \is{\dvcyc_i}{h}\cdot \vcyc_i.
\end{equation}
Note that the same formula \eqref{eq:sum-of-varp-PL} describes $\nMil$ quadratic pinches, and it can indeed be derived by considering a deformation of the singularity \cite[\S 6]{Lamotke:HomIsoSing}.

As in the quadratic case, the precise formula (the matrix $M$) is fully determined by the variation of the complete intersection $D^{1,\ldots,\PLm}_t=\set{f=t}$. The monodromy of such isolated hypersurface singularities is a vast subject; as an example let us only note Pham's contribution \cite{Pham:PLgeneralisees} of the highly structured case $f=x_m^{a_m}+\ldots+x_n^{a_n}$ with arbitrary integer exponents $a_i\geq 2$.

\section{The hierarchy principle}\label{sec:hierarchy}

In the Picard-Lefschetz theorem, the partition $\irrone{D}=\irrone{A}\sqcup\irrone{B}$ of the components of the divisor $D=A\cup B$ is important: vanishing cycles are coboundaries (tubes) around $A$ with boundary in $B$, whereas the dual vanishing cycles are coboundaries around $B$ with boundary in $A$. The \emph{hierarchy principle} refers to the vanishing of iterated variations for simple set-theoretic reasons, due to merely keeping track of which elements of $\irrone{A}$ and $\irrone{B}$ support a simple pinch.

\begin{defn}
    We say that a stratum $S$ of the canonical stratification of $D$ has \textbf{type} $(J,K)$ if $S$ is a connected component of
    \begin{equation*}\label{eq:stratum-type-JK}\tag{$\sharp$}
        A^J\cap B^K \setm \bigcup_{i\notin J\sqcup K} D_i.
    \end{equation*}
\end{defn}
The type of a stratum $S$ lists the components of $A$ and $B$ that contain $S$: $J=\set{j\colon S\subseteq A_j}$ and $K=\set{k\colon S\subseteq B_k}$. Since every point lies on a unique stratum, we also say that ``$p$ has type $(J,K)$'' if the stratum containing $p$ has this type, that is, if $p$ is an element of the set \eqref{eq:stratum-type-JK}.

\subsection{A preorder on simple pinches}\label{sec:pinch-preorder}
For every simple pinch $p$, we defined a local variation endomorphism $\Var_p = \iota_{p*} \var_p W_{p*}$ of $H_n(Y\setm A,B)_t$ in \cref{cor:sum-of-varp}, for $t\in T\setm L$ sufficiently close to the critical value $\pi(p)$. For a second simple pinch $p'$, the local variation $\Var_{p'}$ is defined canonically for $t'$ near $\pi(p')$. To make sense of the iterated variation $\Var_{p'}\circ\Var_p$, we need to identify
\begin{equation*}
    \eta_*\colon H_n(Y\setm A,B)_t \xrightarrow{\ \cong\ } H_n(Y\setm A,B)_{t'}
\end{equation*}
by parallel transport. Hence the iterated variation should more precisely be written as $\Var_{p'}\circ\Var_p=\eta_*^{-1}\Var_{p'} \eta_* \Var_p$ and it depends on the choice of a path $\eta$ in $T\setm L$ from $t$ to $t'$. However, the vanishing results derived below apply for \emph{any} choice of $\eta$, hence we write just $\Var_{p'}\circ \Var_p$.
\begin{figure}
    \centering
    \includegraphics[scale=0.25]{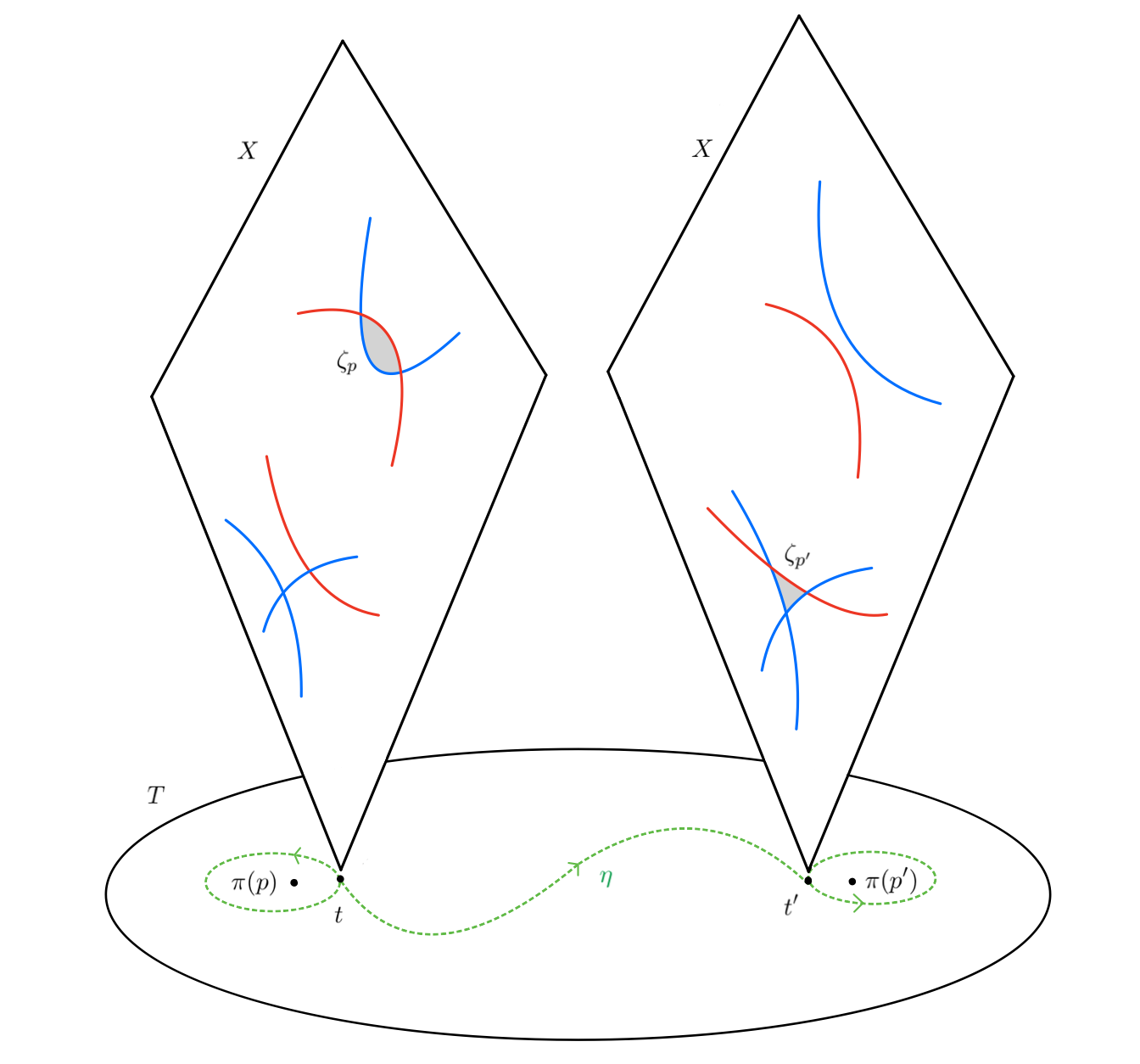}%
    \caption{Iterated variation around two critical values.}%
    \label{fig:double_var}%
\end{figure}

By the Picard-Lefschetz theorem, $\Var_p$ maps any class in $H_n(X\setm A,B)_t$ to a multiple of the vanishing cycle $\vcyc_p$. So if $\Var_{p'} (\eta_* \vcyc_p)=0$, it follows that the iterated variation $\Var_{p'} \circ \Var_p$ is zero.
Rewriting the intersection number
\begin{equation*}
    \is{\dvcyc_{p'}}{\vcyc_p} 
    = \pm \is{\partial_{K'}\vcell_{p'}}{\partial_{K'}\vcyc_p}
    = \pm \is{\partial_{J}\dvcyc_{p'}}{\partial_{J}\vcell_p}
\end{equation*}
as in \cref{rem:is-partialB},\footnote{More precisely, $\is{\dvcyc_{p'}}{\eta_*\vcyc_p} 
    = \pm \is{\iota_{p'*}\partial_{K'}\vcell_{p'}}{\eta_*\partial_{K'}\vcyc_p}
    = \pm \is{\partial_{J}\dvcyc_{p'}}{\eta_*\iota_{p*}\partial_{J}\vcell_p}$.} we see that such vanishing is implied by $\partial_{K'}\vcyc_p=0$ or $\partial_J \dvcyc_{p'}=0$. This observation leads directly to the hierarchy in \eqref{eq:pinch-hierarchy} below. However, intersection theory is not required to explain this phenomenon; what matters more fundamentally is the factorization of $\Var_p$---through $\partial_K$ in the domain, and through $\fibSphere_J$ in the codomain (\cref{thm:loc-factorization,cor:varp-simple-factor-JK}).

The cycle $\vcyc_p$ is supported in a small ball $W_p$ that intersects only those $A_j$ and $B_k$ where $j\in J$ and $k\in K$ are in the type $(J,K)$ of $p$. For any \emph{other} boundary component ($k\notin K$), we have $W_p\cap B_k=\varnothing$ and $\partial_{B_k} \vcyc_p=0$. After continuation to $t'$, the cycle $\eta_* \vcyc_p$ need not be contained in any small ball anymore---but it does stay away from $B_k$ and so $\partial_{B_k} \eta_* \vcyc_p=0$, for $k\notin K$. This follows from the diagram
\begin{equation*}\xymatrixcolsep{15mm}\xymatrix{
    H_n(Y\setm A,B)_t \ar[r]^-{\partial_{B_k}} \ar[d]^{\eta_*} & H_{n-1}(B_k\setm A, B')_t \ar[d]^{\eta_*} \\
    H_n(Y\setm A,B)_{t'} \ar[r]^-{\partial_{B_k}} & H_{n-1}(B_k\setm A,B')_{t'} \\
}\end{equation*}
where $B'=\bigcup_{i\neq k} B_i$. The diagram commutes (see \cref{sec:coboundary-comm-functor}), because an isotopy of the stratified set $D$ descends to compatible local trivializations of the fibre bundles $(Y\setm A,B)$ and $(B_k\setm A,B')$ over $T\setm L$, and furthermore the isotopy can be chosen to be smooth.

As pointed out in \cref{rem:is-partialB} and \cref{cor:varp-simple-factor-JK}, the variation $\Var_{p'}$ factors through the iterated boundary $\partial_{K'}$ of the second simple pinch $p'$ with type $(J',K')$. So if $K'\setm K$ is non-empty, and hence $\partial_{K'}\eta_*\vcyc_p = 0$ as explained above, it follows that $\Var_{p'} (\eta_* \vcyc_p )= 0$.
We conclude that
\begin{equation}\label{eq:itvar_B}
    \Var_{p'} \circ \Var_p = 0
    \quad \text{if $K'\nsubseteq K$}.
\end{equation}

Now consider some $j\in J$ and let $A'=\bigcup_{i\neq j} A_i$. As a coboundary around $A_j$, the vanishing cycle $\vcyc_p$ becomes zero under the inclusion $\kappa\colon Y\setm A\hookrightarrow Y\setm A'$, because the composition $\kappa_*\fibSphere_j=0$ vanishes as part of the residue sequence
\begin{equation*}
    H_{n-1}(A_j\setm A',B)_t \xrightarrow{\fibSphere_j} H_{n}(Y\setm A,B)_t \xrightarrow{\kappa_*} H_{n}(Y\setm A',B)_t.
\end{equation*}
As discussed above, such vanishing $\kappa_*\vcyc_p=0$ persists by parallel transport also at $t'$, whence $\kappa_*(\eta_*\vcyc_p)=0$. If $A_j$ does not contain $p'$, then $A_j$ does not intersect the ball $W_{p'}$ and so $W_{p'}\setm A=W_{p'}\setm A'$. Then the trace map $W_{p'*}$ factors through $\kappa_*$, and thus annihilates the vanishing cycle. It follows that
\begin{equation}\label{eq:itvar_A}
    \Var_{p'} \circ \Var_p = 0
    \quad \text{if $J'\nsupseteq J$}.
\end{equation}

\begin{defn}\label{def:simple-pinch-relation}
    We define a relation on the set of simple pinches as follows:
    For simple pinches $p$ and $p'$ with type $(J,K)$ and $(J',K')$, we set
    \begin{equation*}
        p'\CV p \quad\text{if}\quad\text{$J'\supseteq J$ and $K'\subseteq K$}.
    \end{equation*}
\end{defn}
This relation is transitive and reflexive (a preorder). We can then summarize the observations made above as
\begin{equation}\label{eq:pinch-hierarchy}
    \Var_{p'} \circ \Var_p = 0
    \quad \text{whenever}\quad p'\nCV p.
\end{equation}
\begin{rem}
    These conclusions about variations at simple pinches hold also in the presence of further, non-simple critical points.
\end{rem}

\subsection{Linear and quadratic refinements}\label{sec:hierarchy_refinements}
For a linear pinch $p$, the boundary $\partial_{K} \vcyc_p = \pm \iota_{p*}\fibSphere_J \partial_{1,\ldots,n+1} \vcell = 0$ of the vanishing cell (simplex) is empty.
In addition to the vanishing \eqref{eq:itvar_B} for $K'\setm K\neq\varnothing$, we therefore also have
\begin{equation} \label{eq:itvar_B_linear}
    \Var_{p'} \circ \Var_p = 0
    \quad \text{if}\quad \big(\text{$K'=K$ and $p$ linear}\big).%$K'\not\subseteq K$}.
\end{equation}
Similarly, if $p'$ is a linear pinch, then $\partial_{J'}\dvcyc_{p'}=0$ extends the vanishing \eqref{eq:itvar_A} by the additional constraint
\begin{equation} \label{eq:itvar_A_linear}
    \Var_{p'} \circ \Var_p = 0
    \quad \text{if}\quad\big(\text{$J'=J$ and $p'$ linear}\big). % {$I'\not\supset I$}.
\end{equation}
\begin{cor}\label{rem:simple-nilpotent}
    If all critical points are linear pinches, then $\Var_{p'} \circ \Var_p$ can only be non-zero if $J'\supset J$ and $K'\subset K$ are strict inclusions. Hence in this situation, the monodromy representation is unipotent: every sequence of $n$ or more variations is zero. We discuss this further in \cref{sec:Aomoto}.
\end{cor}
From \eqref{eq:variation-vcyc}, we also know that for a linear or quadratic simple pinch $p$ on a stratum with codimension $\PLm=\abs{J}+\abs{K}$, the repeated variation vanishes,
\begin{equation} \label{eq:itvar_n-m_odd}
   \Var_p\circ\Var_p = 0 \quad\text{if}\quad \text{$n-\PLm$ is odd.}
\end{equation}
This applies in particular to linear pinches ($\PLm=n+1$).
\begin{rem}
    We can incorporate the additional constraints \eqref{eq:itvar_B_linear}--\eqref{eq:itvar_n-m_odd} into \eqref{eq:pinch-hierarchy} by removing the relations $p'\CV p$ among those particular pairs of pinches. The resulting relation is still transitive, but not necessarily reflexive anymore.
\end{rem}

\subsection{A hierarchy on simple pinch components}\label{sec:hierarchy_spc}
Over every irreducible component $\ell\subseteq L$ of the Landau variety, we may have \emph{several} critical strata. Suppose that $\ell$ is a simple pinch component (\cref{defn:simplecomponent}), that is, all critical points over a generic point $t_c\in\ell$ are simple pinches.\footnote{We can allow here the more general notion \eqref{eq:general-simple-pinch} of simple pinch, not just the linear and quadratic ones from \cref{defn:pinches}.} Let $P_{\ell} \subset \pi^{-1}(t_c)$ denote these pinches.
Recall from \cref{cor:sum-of-varp} that 
\begin{equation*}
    \Var_{\ell} = \sum_{p \in P_{\ell}} \Var_p,
\end{equation*}
where we write $\Var_{\ell}$ for the variation $\Var_{\gamma}$ along a simple loop around $t_c\in \ell$. The operator $\Var_{\ell}$ does depend on the choice of $\gamma$, but this choice is irrelevant for the following vanishing results. Furthermore, we write $\Var_{\ell'}\circ\Var_\ell$ to abbreviate $\eta_*^{-1}\Var_{\ell'} \eta_* \Var_\ell$, as above.
\begin{defn}
    Define a binary relation $\CV$ on the set of simple pinch components of the Landau variety as follows:
    %\item If both $\ell$ and $\ell'$ are simple pinch components, then $\ell'\CV\ell$ if and only if there exist $p\in P_{\ell}$ and $p' \in P_{\ell'}$ such that $p'\CV p$.
    %Define a preorder on the simple pinch components of the Landau variety as follows:
    \begin{equation*}
       \ell' \CV \ell \quad \text{if there exist $p\in P_{\ell}$ and $p' \in P_{\ell'}$ such that $p'\CV p$.}
    \end{equation*}
\end{defn}
\begin{cor}\label{corr:simple-povar}
    For any pair $\ell,\ell'\subseteq L$ of simple pinch components, we have
    \begin{equation*}
        \Var_{\ell'} \circ \Var_{\ell} = 0 \quad\text{if}\quad \ell' \nCV \ell.
    \end{equation*}
\end{cor}
Due to the arbitrariness of the intermediate path $\eta$ in \cref{sec:pinch-preorder}, we have actually shown that for arbitrary paths $\gamma\in\pi_1(T\setm L,t_0)$,
\begin{equation*}
    \Var_{\ell'}\circ \Var_{\gamma} \circ \Var_{\ell} = 0
    \quad\text{if}\quad \ell'\nCV\ell.
\end{equation*}

\subsection{Generalization}\label{sec:hierarchy-general}
Above, we derived the hierarchy principle as a consequence of the explicit Picard-Lefschetz formula for (arbitrary) \emph{simple pinches}, as characterized by the local description \eqref{eq:general-simple-pinch}. What is ``simple'' about these critical points $p\in cS$ is not just that they are (fibrewise) isolated, but furthermore that the exclusion of any of the supporting hypersurfaces from the arrangement removes the critical point. We will demonstrate below that the latter property is sufficient to explain the hierarchy principle---without any recourse to the Picard-Lefschetz formula. This discussion will lead to a generalization of the hierarchy principle to ``non-simple'' critical points.

Consider some generic critical value $t_c\in \ell\subseteq L$. The corresponding critical points $P=\bigcup_{S\in\mathfrak{S}} cS\cap\pi^{-1}(t_c)$ need not be isolated; they form some closed subvariety $P\subset X$. In any case, we expect some form of the localization from \cref{ss:loc}: For every connected component $p\in\pi_0(P)$, suppose we can find a smooth manifold $W_p\subset X$ containing $p\subset W_p$, not intersecting any other components of $P$, and such that its boundary $\partial W_p$ is smooth and transverse to $D$. Then the proof of \cref{cor:sum-of-varp} still applies, so that
\begin{equation*}
    \Var_{\ell} = \sum_{p\in\pi_0(P)} \Var_p
    \quad\text{with}\quad
    \Var_p=\iota_{p*} \var_p W_{p*}
\end{equation*}
where $\var_p$ is defined just as in \cref{thm:loc-factorization}. For example, if $p\in\pi_0(P)$ is a smooth component, we can construct $W_p$ as a tubular neighbourhood of $p$. For an isolated point $p$, this reproduces the balls used in \cref{sec:simplepinchsings}.
\begin{defn}
    The \textbf{type} $(J,K)$ of a critical component $p\in\pi_0(P)$ is the set of those hypersurfaces $A_j$ ($j\in J$) and $B_k$ ($k\in K$) which intersect $p$.
\end{defn}
Note that if we choose $W_p$ small enough, then it will intersect only the hypersurfaces in the type of $p$: $W_p\cap A=W_p\cap A_J$ and $W_p\cap B=W_p\cap B_K$. Then the trace $W_{p*}$ factors through the inclusion $X\setm A\hookrightarrow X\setm A_J$, and $\iota_{p*}$ factors through the map of pairs $(X\setm A,B_K)\rightarrow (X\setm A,B)$:
\begin{equation}\label{eq:varp-factor-JK}\begin{gathered}\xymatrix{
    H_n(X\setm A,B)_t \ar@/_5pc/[dd]_{W_{p*}} \ar[d] \ar[r]^{\Var_p} & H_n(X\setm A,B)_t \\
    H_n(X\setm A_J,B)_t \ar[d] & H_n(X\setm A, B_K)_t \ar[u] \\
    H_n(\overline{W_p}\setm A_J, B_K) \ar[r]^{\var_p} & H_n(W_p\setm A_J,B_K) \ar[u] \ar@/_5pc/[uu]_{\iota_{p*}} \\
}\end{gathered}\end{equation}

The hierarchy constraints arise from hypersurfaces in the type that are indispensable to make $p$ critical. We call those hypersurfaces \emph{simple}.
\begin{defn}\label{def:simple-type}
    Given a critical component $p\in\pi_0(P)$ of type $(J,K)$, we call an element $i\in J\sqcup K$ and the corresponding hypersurface $D_i$ \textbf{simple} if the union $\bigcup_{k\neq i} D_k$ of all other components of $D$ has no critical points in $p$. We denote the simple elements of the type as $\simple{J}\subseteq J$ and $\simple{K}\subseteq K$.
\end{defn}
\begin{rem}\label{rem:simple-pinch-everything-simple}
    For a simple pinch, omitting \emph{any} hypersurface $D_1,\ldots,D_{\PLm}$ in \eqref{eq:general-simple-pinch} results in a non-critical arrangement. Therefore, the entire type is simple: $(\simple{J},\simple{K})=(J,K)$. In fact, this property characterizes simple pinches topologically: Every isolated critical point with $(\simple{J},\simple{K})=(J,K)$ admits local coordinates of the form \eqref{eq:general-simple-pinch}.
\end{rem}

Non-simple situations can arise even for isolated critical points, for example if one starts with a simple pinch $p$ and then enlarges the arrangement by an additional hyperplane that contains $p$.
\begin{eg}\label{eg:simple-nonsimple}
    Consider $Y=X\times T=\PP^1\times\CC$ with the transverse arrangement $D=D_1\cup D_2$ formed by the smooth hypersurfaces $D_1=\set{x^2=t}$ and $D_2=\set{x=0}$ in coordinates $[x:1]\in\PP^1$. The pair $(Y,D)$ has a single critical point, $c(D_1\cap D_2)=\set{(p,t_c)}$ over $t_c=0$, namely $p=[0:1]$.
    \begin{itemize}
        \item $D_1$ is simple: The remaining pair $(Y,D_2)\cong (\PP^1,\set{p})\times\CC$ is a (trivial) fibre bundle, without any critical points.
        \item $D_2$ is not simple: Even after removing $D_2$, the remaining pair $(Y,D_1)$ still has $c(D_1)=\set{(p,t_c)}$ as a critical point.
    \end{itemize}
\end{eg}
For simple pinches, the hierarchy follows from the factorization \eqref{eq:varp-simple-factor-JK} of $\var_p$ through each $\partial_k$ in the domain, and each $\fibSphere_j$ in the codomain. For general critical sets, such a factorization applies only to the simple indices. Given an index $k\in K$ or $j\in J$, consider the natural map of pairs
\begin{alignat*}{2}
    \rho_k\colon &\quad & \textstyle (X\setm A,\bigcup_{i\neq k} B_i) &\rightarrow(X\setm A,B),\\
    \kappa_j\colon & \quad & (X\setm A,B)  & \textstyle\rightarrow(X\setm \bigcup_{i\neq j}A_i,B).
\end{alignat*}
\begin{lem}\label{lem:simple-type-factorization}
    Let $p\in\pi_0(P)$ with type $(J,K)$ and simple part $(\simple{J},\simple{K})$. Then for every $k\in\simple{K}$, and for every $j\in\simple{J}$, we have
\begin{equation}\label{eq:Varp-JK-simple-0}
    \kappa_{j*}\Var_p = 0
    \qquad\text{and}\qquad
    \Var_p \rho_{k*} = 0
    .
\end{equation}
    Equivalently, the local variation at $p$ factors through the boundary map $\partial_k$ in the domain, and through the Leray coboundary $\fibSphere_j$ in the codomain.
\end{lem}
\begin{proof}
    Let $A'=\bigcup_{i\neq j} A_i$ and $B'=\bigcup_{i\neq k} B_i$ denote the arrangements $A$ and $B$, with one of the simple indices omitted. By \cref{def:simple-type}, the corresponding pairs
    $(W_p\setm A',B)$ and $(W_p\setm A, B')$ are locally trivial even over the critical value $\set{t_c}=\pi(P)$. Hence they have trivial monodromy and zero variation.
    
   % In the domain, this implies that the local variation $\Var_p$ factors through the image of the push-forward of the map $\rho\colon(X\setm A,B')\rightarrow(X\setm A,B)$. By the long exact boundary sequence \eqref{eq:exact_boundary_seq}, $\partial_k$ embeds this quotient into $H_{n-1}(B_k\setm A,B')_t$:
    
In the domain, this implies that the local variation $\Var_p$ vanishes on the image of the push-forward of the map $\rho\colon(X\setm A,B')\rightarrow(X\setm A,B)$. By the long exact boundary sequence \eqref{eq:exact_boundary_seq}, this image equals the kernel of $\partial_k$, thus $\Var_p$ factors through $\partial_k$:
    
    \begin{equation*}\xymatrix{
        H_n(X\setm A,B')_t\ar[r]^{\rho_*} \ar[d]^{\Var_p} & H_n(X\setm A,B)_t \ar[d]^{\Var_p} \ar[r]^-{\partial_k} & H_{n-1}(B_k\setm A,B')_t \ar@{-->}[dl] \\
        0 \ar[r]^{\rho_*} & H_n(X\setm A,B)_t & \\
    }\end{equation*}
    
    In the codomain, the zero variation of $(W_p\setm A',B)$ implies that $\Var_p$ takes values in the kernel of $\kappa_*$, where $\kappa\colon (X\setm A,B)\rightarrow (X\setm A',B)$:
    \begin{equation*}\xymatrix{
        & H_n(X\setm A,B)_t \ar@{-->}[dl] \ar[d]^{\Var_p} \ar[r]^{\kappa_*} & H_n(X\setm A',B)_t \ar[d]^{\Var_p} \\
        H_{n-1}(A_j\setm A',B)_t \ar[r]^-{\fibSphere_j} & H_n(X\setm A,B)_t \ar[r]^{\kappa_*} & 0 \\
    }\end{equation*}
    By the long exact residue sequence, the kernel of $\kappa_*$ is the image of $\fibSphere_j$.
\end{proof}
\begin{figure}
  %  \centering\begin{tabular}{cc}
  %      \includegraphics[width=0.45\textwidth]{simple-nonsimple} & \includegraphics[width=0.45\textwidth]{nonsimple-simple} \\
  %      $\Var_p h = \Graph[0.05]{simple-nonsimple-var}$ & $\Var_p h=\Graph[0.05]{nonsimple-simple-var}$
 %   \end{tabular}%
\scalebox{.9}{
\begin{tikzpicture}[scale=1]
\node[] (m) at (0,0) {};
\node[] (t1) at (1,0) {};
\node[] (t2) at (-1,0) {};
\draw[olive!15!green,decoration={markings, mark=at position 0.5 with {\arrow{>}}},postaction={decorate}] (.75,1.299) -- (.75,-1.299) node[left,xshift=0cm,yshift=0.2cm] {$\scriptstyle{h}$};
\draw[blue] (m) circle [radius=1.5] node[right,xshift=1.5cm,yshift=.2cm] {$\scriptstyle{ \partial W }$};
\filldraw[blue] (0,0) circle (0.03) node[right,yshift=-.07cm] {$\scriptstyle{0}$};
\filldraw[red] (t1) circle (0.03) node[below] {$\scriptstyle{\sqrt{t}}$};
\filldraw[red] (t2) circle (0.03) node[above] {$\scriptstyle{-\sqrt{t}}$};
\draw[red,dashed,decoration={markings, mark=at position 1 with {\arrow{>}}},postaction={decorate}] (t1) arc[radius = 1.5cm, start angle= 0, end angle= 20] (t2) {};
\draw[red,dashed,decoration={markings, mark=at position 1 with {\arrow{>}}},postaction={decorate}] (t2) arc[radius = 1.5cm, start angle= 180, end angle= 195] (t1) {};
\end{tikzpicture}
%%%%%%%%%%%%%%%%%%%%%%%%%%%%%%%%
\raisebox{1.3cm}{$\overset{g_*}{\longmapsto}\ $ }
%%%%%%%%%%%%%%%%%%%%%%%%%%%%%%%%
\begin{tikzpicture}[scale=1]
\node[] (m) at (0,0) {};
\node[] (t1) at (1,0) {};
\node[] (t2) at (-1,0) {};
\draw [olive!15!green,decoration={markings, mark=at position 0.5 with {\arrow{>}}},postaction={decorate}] plot [smooth] coordinates {(.75,1.299) (-.75,.7) (-1.3,0) (-.75,-.7) (-.2,0.2) (.75,.7) (1.2,0) (.75,-1.299)} node[left,xshift=0cm,yshift=0.2cm] {$\scriptstyle{g_*h}$};
\draw[blue] (m) circle [radius=1.5] node[right,xshift=1.5cm,yshift=.2cm] {$\scriptstyle{ \partial W }$};
\filldraw[blue] (0,0) circle (0.03) node[right,yshift=-.07cm] {$\scriptstyle{0}$};
\filldraw[red] (t1) circle (0.03) node[below] {};
\filldraw[red] (t2) circle (0.03) node[below] {};
\end{tikzpicture}

\vspace{.5cm}
\raisebox{1.3cm}{$ \Rightarrow  \Var_p h=$\raisebox{-.5cm}{
\begin{tikzpicture}[scale=1]
\node[] (m) at (0,0) {};
\node[] (t1) at (.5,0) {};
\node[] (t2) at (-.5,0) {};
\draw[olive!15!green,decoration={markings, mark=at position 0.9 with {\arrow{<}}},postaction={decorate}] (t1) circle [radius = 0.3cm];
\draw[olive!15!green,decoration={markings, mark=at position 0.1 with {\arrow{>}}},postaction={decorate}] (t2) circle [radius=0.3cm];
\filldraw[blue] (0,0) circle (0.03) node[below,yshift=-.07cm] {$\scriptstyle{D_2}$};
\filldraw[red] (t1) circle (0.03) node[below,yshift=-.23cm] {$\scriptstyle{D_1}$};
\filldraw[red] (t2) circle (0.03) node[below,yshift=-.23cm] {$\scriptstyle{D_1}$};
\end{tikzpicture}}
}}
%%%

\vspace{.7cm}

%%%
\scalebox{.9}{
\begin{tikzpicture}[scale=1]
\node[] (m) at (0,0) {};
\node[] (t1) at (1,0) {};
\node[] (t2) at (-1,0) {};
\draw[olive!15!green,decoration={markings, mark=at position 0.5 with {\arrow{>}}},postaction={decorate}] (0,1.5) -- (1,0) node[left,xshift=-.2cm,yshift=0.2cm] {$\scriptstyle{h}$};
\draw[blue] (m) circle [radius=1.5] node[right,xshift=1.5cm,yshift=.2cm] {$\scriptstyle{ \partial W }$};
\filldraw[red] (0,0) circle (0.03) node[right,yshift=-.07cm] {$\scriptstyle{0}$};
\filldraw[blue] (t1) circle (0.03) node[below] {$\scriptstyle{\sqrt{t}}$};
\filldraw[blue] (t2) circle (0.03) node[above] {$\scriptstyle{-\sqrt{t}}$};
\draw[blue,dashed,decoration={markings, mark=at position 1 with {\arrow{>}}},postaction={decorate}] (t1) arc[radius = 1.5cm, start angle= 0, end angle= 20] (t2) {};
\draw[blue,dashed,decoration={markings, mark=at position 1 with {\arrow{>}}},postaction={decorate}] (t2) arc[radius = 1.5cm, start angle= 180, end angle= 195] (t1) {};
\end{tikzpicture}
%%%%%%%%%%%%%%%%%%%%%%%%%%%%%%%%
\raisebox{1.3cm}{$\overset{g_*}{\longmapsto}\ $ }
%%%%%%%%%%%%%%%%%%%%%%%%%%%%%%%%
\begin{tikzpicture}[scale=1]
\node[] (m) at (0,0) {};
\node[] (t1) at (1,0) {};
\node[] (t2) at (-1,0) {};
\draw[blue] (m) circle [radius=1.5] node[right,xshift=1.5cm,yshift=.2cm] {$\scriptstyle{ \partial W }$};
\filldraw[red] (0,0) circle (0.03) node[right,yshift=-.07cm] {$\scriptstyle{0}$};
\draw[olive!15!green,decoration={markings, mark=at position 0.5 with {\arrow{>}}},postaction={decorate}] (0,1.5) -- (-1,0) node[above,xshift=-0.1cm,yshift=0.2cm] {$\scriptstyle{g_*h}$};
\filldraw[blue] (t1) circle (0.03) node[below] {};
\filldraw[blue] (t2) circle (0.03) node[above] {};
\end{tikzpicture}
%%%%%%%%%%%
\raisebox{1.3cm}{$ \Rightarrow \Var_p h=$\raisebox{-.5cm}{
\begin{tikzpicture}[scale=1]
\node[] (m) at (0,0) {};
\node[] (t1) at (.5,0) {};
\node[] (t2) at (-.5,0) {};
\draw [olive!15!green,decoration={markings, mark=at position 0.5 with {\arrow{>}}},postaction={decorate}] plot [smooth] coordinates {(.5,0) (0,.2) (-.5,0)};
\filldraw[red] (0,0) circle (0.03) node[below,yshift=-.07cm] {$\scriptstyle{D_2}$};
\filldraw[blue] (t1) circle (0.03) node[below,yshift=-.23cm] {$\scriptstyle{D_1}$};
\filldraw[blue] (t2) circle (0.03) node[below,yshift=-.23cm] {$\scriptstyle{D_1}$};
\end{tikzpicture}}}}
    \caption{Non-vanishing variations at a non-simple isolated singularity (\cref{eg:simple-nonsimple-var}).}%
    \label{fig:simple-nonsimple}%
\end{figure}
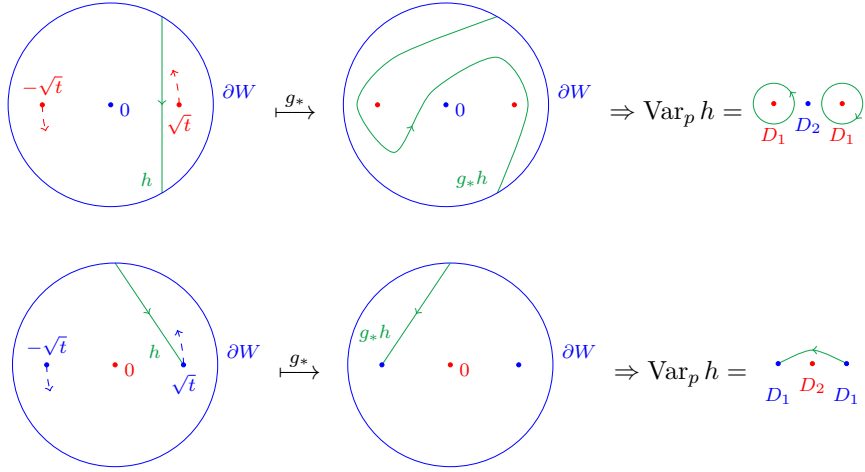

\begin{eg}\label{eg:simple-nonsimple-var}
    Consider \cref{eg:simple-nonsimple} with the partition $(A,B)=(D_1,D_2)$. As illustrated in \cref{fig:simple-nonsimple}, the variation $\Var_p h$ is a coboundary around $D_1$ as expected ($1\in\simple{J}$). But it is non-zero, despite $\partial_2 h=0$, so $\Var_p$ does not factor through $\partial_2$ in the domain. The case $(A,B)=(D_2,D_1)$ demonstrates that the variation $\Var_p h$ need not be a coboundary around the non-simple hypersurface $D_2$.
\end{eg}

\begin{rem}
    If there are two simple types $j,j'\in\simple{J}$, then \cref{lem:simple-type-factorization} states that $\vcyc=\Var_p h$ can be written as a coboundary $\vcyc=\fibSphere_j \alpha$ around $A_j$, and also as a coboundary $\vcyc=\fibSphere_{j'} \alpha'$ around $A_{j'}$. This does \emph{not} necessarily imply that $\vcyc=\fibSphere_j\fibSphere_{j'} \alpha''$ is an iterated coboundary \cite{Yuzhakov:CoConLeray}. Similarly, for two simple types $k,k'\in\simple{K}$, it is not necessarily the case that $\Var_p h$ depends only on the iterated boundary $\partial_k\partial_{k'} h$.
\end{rem}

\begin{defn}\label{def:simple_compatible}
    Suppose that $p\subset \pi^{-1}(t_c)$ and $p'\subset\pi^{-1}(t_c')$ are connected components of the critical sets in two fibres. Let $(J,K),(J',K')$ denote the types of $p,p'$ and $(\simple{J},\simple{K}),(\simple{J'},\simple{K}')$ their simple parts. We then declare $p'$ and $p$ to be \emph{compatible},
    \begin{equation*}
        p' \CV p,
        \qquad\text{if}\qquad
        \simple{K'} \subseteq K
        \quad\text{and}\quad
        \simple{J} \subseteq J'.
    \end{equation*}
\end{defn}
For simple pinches, this relation is just \cref{def:simple-pinch-relation}, due to \cref{rem:simple-pinch-everything-simple}.
\begin{cor}\label{corr:povar_upstairs}
    If $p'\nCV p$, then $\Var_{p'}\circ\Var_p = 0$.
\end{cor}
\begin{proof}
    Combine the factorization \eqref{eq:varp-factor-JK} with the vanishing \eqref{eq:Varp-JK-simple-0}. For example, if $k\in\simple{K}'\setm K$, then $\iota_{p*}$ factors through $\rho_{k*}$, which is annihilated by $\Var_{p'}$.
\end{proof}

We can now formulate the general hierarchy principle. It enforces the vanishing of certain iterated variations, and it follows from the previous constraints on $\Var_{p'}\circ\Var_p = 0$, together with the localization $\Var_{\ell} = \sum_{p} \Var_p$ discussed at the beginning of \cref{sec:hierarchy-general}.
\begin{defn} \label{def:rel_landau_components}
    Let $\irrone{L}$ denote the set of Landau singularities, that is, the set of irreducible components of the Landau variety that have complex codimension one. Define a binary relation $\CV$ on $\irrone{L}$ as follows: For $\ell\in\irrone{L}$, let $P_{\ell}\subset\pi^{-1}(t_c)$ denote the critical set in a fibre over a generic point $t_c\in\ell$. Then set 
    $\ell'\CV\ell$ if and only if there exist $p\in P_{\ell}$ and $p' \in P_{\ell'}$ such that $p'\CV p$.
\end{defn}
\begin{cor}\label{corr:povar}
    For any pair $\ell,\ell'\in\irrone{L}$ of Landau singularities, we have
    \begin{equation*}
        \Var_{\ell'} \circ \Var_{\ell} = 0 \quad\text{if}\quad \ell' \nCV \ell.
    \end{equation*}
\end{cor}

\section{Polylogarithms}
\label{sec:polylogs}

In this section, we illustrate the hierarchy principle with polylogarithms. The most generic situation, Aomoto polylogarithms, is the simplest---all critical points are linear simple pinches. In \cref{sec:dilog}, we also discuss the hierarchy for the dilogarithm, which involves non-isolated critical sets.

\subsection{Aomoto polylogarithms}\label{sec:Aomoto}
A collection $Q=Q_0\cup\ldots\cup Q_n$ of $n+1$ hyperplanes in complex projective space $\PP^n$ is called a simplex. We say that it is non-degenerate if the intersection $Q_0\cap\ldots\cap Q_n$ is empty. 

To a pair of non-degenerate simplices $Q$ and $R$, we associate:
\begin{itemize}
\item
a differential form $\omega_Q$ with simple poles on each $Q_i$ that generates $H_{\dR}^n(\PP^n\setm Q)\cong\CC$. In terms of linear equations $f_i=0$ for $Q_i$,
\begin{equation}\label{eq:aomoto-integrand}
	\omega_Q = \td \log \left(\frac{f_1}{f_0} \right) \wedge \ldots \wedge \td \log \left(\frac{f_n}{f_0}\right).
\end{equation}
\item
a generator $\sigma_R\in H_n(\PP^n,R)\cong \ZZ$ corresponding to the real standard simplex $\set{\sum_i u_i=1}\subset\RR_{\geq0}^{n+1}$, embedded as $\set{[u_0:\cdots:u_n]\colon u_i \geq 0}\subset\PP^n$ with boundary in $R$, in coordinates $u$ where $R_i=\set{u_i=0}$.
\end{itemize}
Suppose that $\sigma_R$ can be deformed to avoid $Q$, so that $\sigma_R$ has a preimage $\sigma_R'\in H_n(\PP^n\setm Q,R)$ under the inclusion of pairs $(\PP^n\setm Q,R\setm Q)\rightarrow (\PP^n,R)$.

The \textit{Aomoto polylogarithm} \cite[\S2]{Aomoto:AdditionAbelHyper} of weight $n$ is the integral
\begin{equation*}\label{eq:aomoto}
	\Lambda_n(Q,R) = \int_{\sigma_R'} \omega_Q.
\end{equation*}
It depends on the choice of $\sigma_R'$, hence $\Lambda_n$ is a multivalued function on the space $T=(\PP^n)^{2n+2}$ of pairs $t=(Q_0,\ldots,Q_n,R_0,\ldots,R_n)$ of simplices. We identify hyperplanes $Q_i\subset X=\PP^n$ with points $[Q_{0i}:\cdots:Q_{ni}]$ in the dual $\PP^n$, and we can represent points $t\in T$ as $(n+1) \times (2n+2)$ matrices with column vectors $Q_i=(Q_{0i},\ldots,Q_{ni})^\Transpose,R_j=(R_{0j},\ldots,R_{nj})^\Transpose$.
We obtain an arrangement $A\cup B$ in the total space $Y=X\times T$, with the components
\begin{align*}
	\irrone{A} 
	& = \set{A_0,\ldots,A_n}, \quad A_i= \big\{ x_0 Q_{0i}+\ldots+x_n Q_{ni}=0 \big\} 
	\subset Y , \\
	\irrone{B} 
	& = \set{B_0,\ldots,B_n}, \quad B_j= \big\{ x_0 R_{0j} + \ldots +x_n R_{nj}  =0 \big\} 
	\subset Y.
\end{align*}
\begin{lem}
The divisor $A\cup B$ has simple normal crossings.
\end{lem}
\begin{proof}
	At every point $[p_0:\cdots:p_n]\in X$ we can find $k$ such that $p_k\neq 0$. If $(p,t)\in A_i$, there must be an index $r\neq k$ such that $Q_{ri}\neq 0$, for otherwise, the equation of $A_i$ would collapse at $p$ to $p_k Q_{ki}=0$ and imply $Q_i=0\notin \PP^n$. So after relabelling, we may assume that $p_0=Q_{1i}=1$.
	It follows that at $p$, and hence in a small neighbourhood, the functions
	\begin{equation*}
		z^{(i)}
		=(z_1^{(i)},\ldots,z_n^{(i)})
		=(
		Q_{0i}+p_1+p_2 Q_{2i}+\ldots+p_n Q_{ni},
		Q_{2i},\ldots,Q_{ni}
		)
	\end{equation*}
	form a coordinate chart of the $\PP^n$ that parametrizes the hyperplane $Q_i$. This way, we can locally choose coordinates $(z^{(0)},\ldots,z^{(2n+1)})$ on the product $T=(\PP^n)^{2n+2}$ such that $A_i=\set{z_1^{(i)}=0}$ and $B_j=\set{z_1^{(n+1+j)}=0}$.
\end{proof}
The strata $S$ of the canonical stratification of $A\cup B$ are indexed by pairs $I,J\subseteq\set{0,\ldots,n}$ so that $\overline{S}=A^I\cap B^J=\bigcap_{i\in I}A_i \cap \bigcap_{j\in J}B_j$.\footnote{The actual strata are the subsets $S=A^I\cap B^J\setm(\bigcup_{i\notin I}A_i\cup\bigcup_{j\notin J}B_j)$.} It follows from considerations similar to \cite[\S IV.4.4]{Pham:Singularities} that the Landau variety can be computed from the closures of the strata:
\begin{equation*}
    L=\bigcup_{S\in\mathfrak{S}} \pi(cS)
    = \bigcup_{S\in \mathfrak{S}} \pi(c\overline{S})
    = \bigcup_{I,J} \pi(c(A^I\cap B^J)).
\end{equation*}
Let $t_{IJ}=(Q_{i\in I},R_{j\in J})$ denote the submatrix of $t$, consisting of the columns in $I$ and $J$. Write $V_{IJ}\subset \CC^{n+1}$ for the vector space spanned by these columns.
\begin{lem}\label{lem:aomoto-critical}
    A point $(x,t)$ of the submanifold $A^I\cap B^J\subset Y$ is critical precisely when the columns $t_{IJ}=(Q_{i\in I},R_{j\in J})$ of $t$ are linearly independent.
\end{lem}
\begin{proof}
    The submanifold $S=A^I\cap B^J$ has complex dimension $\dim T+n-k$ with $k=\abs{I}+\abs{J}$. The fibre at $t$ is the projectivization $S_t=\PP(V_{IJ}^{\perp})\subset\PP^n$ of
    \begin{equation*}
        V_{IJ}^{\perp}
        =\bigcap_{v\in V_{IJ}}\set{y\in\CC^{n+1}\colon x_0v_0+\ldots +x_n v_n=0}
        \subset \CC^{n+1},
    \end{equation*}
    the annihilator of $V_{IJ}$. Therefore, at any $t\in T$, the fibre $S_t=A^I\cap B^J\cap \pi^{-1}(t)$ is a complex manifold with complex dimension
    \begin{equation}\label{eq:aomoto-fibredim}
        \dim S_t =  \dim V_{IJ}^{\perp} - 1 = n-\dim V_{IJ}.
    \end{equation}
    The kernel of the differential $\td (\pi|_S)$ is the intersection (in $TY=TX\oplus TT$) of the vectors that are tangent to $S$ ($TS$) and vertical ($TX$). In other words, $\ker\td(\pi|_S)=T(S_t)$ is the tangent bundle of the fibre. Hence the image of $\td(\pi|_S)$ has dimension $\dim S-(n-\dim V_{IJ})=\dim T-(k-\dim V_{IJ})$ and we have a submersion precisely when $\dim V_{IJ}=k$.
\end{proof}
The fibre of $A^I\cap B^J$ over $t$ has dimension $n-\dim V_{IJ}$, see \eqref{eq:aomoto-fibredim}. In particular, the fibre is empty, unless $\dim V_{IJ}\leq n$.
\begin{cor}
    The Landau variety of $\pi\colon (Y,A\cup B)\rightarrow T$ is
\begin{equation*}
    L = \bigcup_{I,J} \big\{\rk {t_{IJ}}<\min\set{\abs{I}+\abs{J},n+1}\big\}.
\end{equation*}
\end{cor}
For $k=\abs{I}+\abs{J}\leq n+1$, the subvariety $\set{\rk {t_{IJ}}<k}\subset T$ has codimension $n+1-(k-1)$: instead of $n+1$ arbitrary entries, some column is a linear combination of the other $k-1$ columns. Similarly, for $k\geq n+1$ the codimension of the subvariety $\set{\rk {t_{IJ}}<n+1}\subset T$ is $k-n$.
\begin{cor}
    The components of $L$ with codimension one in $T$ are in bijection with pairs $I,J\subseteq\set{0,\ldots,n}$ so that $\abs{I}+\abs{J}=n+1$. Their form is
\begin{align*}
   \ell_{IJ} =\big\{\det (t_{IJ}) =0  \big\}
   \subset T.
\end{align*}
\end{cor}
In fact, $L$ has no irreducible components of higher codimension. For if $k=\abs{I}+\abs{J}\leq n$, then $\set{\rk{t_{IJ}}<k}\subset\ell_{I'J'}$ for any $I'\supseteq I$, $J'\supseteq J$ such that $\abs{I'}+\abs{J'}=n+1$. Similarly, if $k>n+1$, then any $I'\subseteq I$, $J'\subseteq J$ with $\abs{I'}+\abs{J'}=n+1$ provide $\set{\rk{t_{IJ}}<n+1}\subset\ell_{I'J'}$. It follows that
\begin{equation}\label{eq:aomoto-Landau}
    L= \bigcup_{I,J\colon \abs{I}+\abs{J}=n+1} \ell_{IJ}.
\end{equation}

At a generic point of $\ell_{IJ}$, we have $\rk {t_{IJ}}=n$ and the one-dimensional kernel of $t_{IJ}$ defines the unique critical point $x$ in the fibre (\cref{lem:aomoto-critical}). It is a linear pinch of type $(I,J)$.
The vanishing cell $\vcell_{IJ}$ is a real $n$-simplex bounded by the hyperplanes $Q_i$ ($i\in I$) and $R_j$ ($j\in J$).

The cases where $I=\varnothing$ or $J=\varnothing$ have zero variation (\cref{sec:linear-pinch}) and the corresponding loci $\ell_{\varnothing,0\ldots n}$ and $\ell_{0\ldots n,\varnothing}$ in $T$ parametrize where $R$, respectively $Q$, become degenerate.
\begin{prop}\label{lem:aomoto-var}
Let $\varnothing\neq I,J \subseteq\set{0,\ldots,n}$ with $\abs{I}+\abs{J}=n+1$. Then with the intersection number $N=\is{\partial_J\vcell_{IJ}}{\partial_J \sigma_R'}\in \ZZ$, we have
\begin{equation}\label{eq:aomoto-var}
  \Var_{\ell_{IJ}} \Lambda_n(Q,R) = \pm N (2\ipi)^{|I|} \Lambda_{n-|I|}(\partial_I Q, R_J \cap Q^I).
\end{equation}
\end{prop}
\begin{proof}
The iterated boundary $\partial_I \vcell_{IJ}$ is a simplex in $Q^I$ bounded by $R_J=\bigcup_{j\in J}R_j$. By the Picard-Lefschetz and residue formulas \cref{eq:pic,eq:residue-theorem},
\begin{equation*}
  \Var_{\ell_{IJ}} \Lambda_n(Q,R) = \pm N \int_{\delta_I \partial_I \vcell_{IJ}} \omega_Q = \pm N \cdot (2\ipi)^{|I|} \int_{\partial_I \vcell_{IJ}} \Res_I \omega_Q
\end{equation*}
is an Aomoto polylogarithm of lower weight. Namely, $\Res_I$ denotes the iteration of all residues $\Res_{A_i}$ with $i\in I$, and we see from \eqref{eq:aomoto-integrand} that $\Res_{A_i} \omega_Q=(-1)^{i-1}\omega_{\partial_i Q}$.
\end{proof}

\Cref{rem:simple-nilpotent} states that iterated variations $\Var_{\ell_{I'J'}} \circ \Var_{\ell_{IJ}}$ are zero, unless $I' \supset I$ and $J' \subset J$ are strict inclusions. Together with the vanishing of $\Var_{\ell_{IJ}}$ when $I=\varnothing$ or $J=\varnothing$, it follows that any iteration $\Var_{\ell_{I_k J_k}}\circ\cdots\circ\Var_{\ell_{I_1J_1}}$ of more than $n$ variations of $\Lambda_n$ is zero. This is also clear from \eqref{eq:aomoto-var}.

Thus, the maximal iterated variations are $n$-fold, and correspond to an increasing sequence $I_k=\set{\sigma(0),\ldots,\sigma(k-1)}$ and a decreasing sequence $J_k=\set{\tau(k),\ldots,\tau(n)}$. Such sequences encode permutations $\sigma,\tau$ of $\set{0,\ldots,n}$.
\begin{lem}
    The $n$-fold iterated variations for all permutations $\sigma,\tau$ are
\begin{equation}\label{eq:aomoto-var-flag}
%    \Var_{\ell_{I_nJ_n}} \circ \cdots \circ \Var_{\ell_{I_1 J_1}} \Lambda_n = \pm (2\ipi)^{n+1}.
    \Var_{\ell_{\sigma(0)\ldots\sigma(n-1),\tau(n)}} \circ \cdots \circ \Var_{\ell_{\sigma(0),\tau(1)\ldots\tau(n)}} \Lambda_n = \pm (2\ipi)^{n}.
\end{equation}
\end{lem}
\begin{proof}
    Iterating \cref{lem:aomoto-var} shows that the iterated variation is an integral multiple of $(2\ipi)^n$. It remains to show that all intersection numbers are $\pm 1$.
    
    For the first intersection, note that the $n$-fold boundaries $\partial_{J_1}$ of $\sigma_R'$ and the first vanishing cell $\vcell_{I_1 J_1}$ are, up to sign, the fundamental class of the (0-dimensional) corner $R^{J_1}$ that both simplices share. In particular, in this case where $\abs{J}=n$, the number $N=\pm 1$ in \eqref{eq:aomoto-var} only depends on the orientation of $\sigma_R$, but not on the choice of lift $\sigma_R'$.
    
    The subsequent intersections are between the vanishing cycles $\pm\fibSphere_{I_k}\partial_{I_k}\vcell_{I_k J_k}$ of the $k$'th variation, and the $(k+1)$'th dual vanishing cycle. Again,
\begin{equation*}
    \is{\partial_{J_{k+1}}\vcell_{I_{k+1}J_{k+1}}}{\partial_{J_{k+1}}\fibSphere_{I_k}\partial_{I_k}\vcell_{I_kJ_k}}
    =\pm\is{\partial_{J_{k+1}\sqcup I_K}\vcell_{I_{k+1}J_{k+1}}}{\partial_{J_{k+1}\sqcup I_K}\vcell_{I_kJ_k}}
\end{equation*}
    is the self-intersection of the shared corner $Q^{I_k}\cap R^{J_{k+1}}$ of both simplices.
\end{proof}
From their differential equations, obtained in \cite[\S2]{Aomoto:AdditionAbelHyper}, it follows that Aomoto polylogarithms are iterated integrals of the differential forms
\begin{equation*}
    \alpha_{IJ}=\td \log \big(\det (t_{IJ})\big).
\end{equation*}
This means that for any base point $t_0\in T\setm L$, there exists a linear combination $W$ of words (tensors) in $\alpha_{IJ}$'s, such that
\begin{equation}\label{eq:aomoto-II}
    \Lambda_n(t)=\Lambda_n(t_0) + \int_{t_0}^t W.
\end{equation}
Here the integral of a word is defined in terms of a path $\gamma$ from $t_0$ to $t$, as
\begin{equation*}
    \int_{t_0}^t \alpha_k \otimes \cdots \otimes \alpha_1
    %= \int_{0<\lambda_1<\cdots<\lambda_k<1} \alpha_1(\gamma(\lambda_1))\cdots \alpha_k(\gamma(\lambda_k))
    =  \int_0^1 \gamma^*\alpha_k(\lambda_k) \int_0^{\lambda_k} \gamma^*\alpha_{k-1}(\lambda_{k-1}) \,\cdots \int_0^{\lambda_2} \gamma^*\alpha_1(\lambda_1).
\end{equation*}

In the representation \eqref{eq:aomoto-II}, the variations around $\ell_{IJ}$ can be computed using path concatenation with simple loops. The vanishing of any iteration of more than $n$ variations, dictated by the hierarchy principle as explained above, shows that the words in $W$ must have lengths $k\leq n$.
\begin{cor}
    The part of $W$ including all words with length $n$ is
\begin{equation}\label{eq:aomoto-symbol}
	\pm \sum_{\sigma,\tau} (\sgn \sigma)(\sgn \tau)\,
	\alpha_{\sigma(0)\ldots\sigma(n-1),\tau(n)} \otimes\cdots\otimes \alpha_{\sigma(0),\tau(1)\ldots\tau(n)}.
\end{equation}
\end{cor}
\begin{proof}
By \eqref{eq:aomoto-var-flag}, the words for all pairs of permutations of $\set{0,\ldots,n}$ must have the coefficient $+1$ or $-1$. Since $\Lambda_n(Q,R)$ is odd under swapping two columns in $Q$ or $R$, the relative signs are determined by the signs of the permutations. The overall sign depends on the orientation of $R$.
\end{proof}
The expression \eqref{eq:aomoto-symbol} is is called the \emph{symbol} of $\Lambda_n$ in \cite[(B.13)]{ArkaniHamedYuan:SphericalPlanesQuadrics} and was computed, in a slightly different form, already in \cite[Proposition~\S 2.2]{Aomoto:AdditionAbelHyper}.

\subsection{Dilogarithm}\label{sec:dilog}
The dilogarithm is defined by $\Li{2}(t)=\sum_{n=1}^{\infty} t^n/n^2$ for $\abs{t}<1$. It extends to a multivalued holomorphic function on $\CC\setm\set{0,1}$, for example via the integral representation
\begin{equation}\label{eq:Li2}
    \Li{2}(t) = \int_\sigma \omega_t = \int_{[0,1]^2} \frac{t\td x_1 \td x_2}{1-tx_1x_2}.
\end{equation}
To interpret this integral as in \cref{sec:integrals}, we can compactify the fibres to $X=(\PP^1)^2$ and set $Y=X\times T$ over the parameter space $T=\PP^1$. In coordinates $([x_1:1],[x_2:1],[t:1])\in Y$, the divisor $D=A\cup B$ consists of:\footnote{In homogeneous coordinates $([x_1:y_1],[x_2:y_2],[t:s])$, the equations for $A$ are $A_i=\set{y_i=0}$ and $A_3=\set{tx_1x_2=sy_1y_2}$.}
\begin{itemize}
\item two hyperplanes $A_{1,2}$ and a quadric $A_3$ where the integrand has poles:
\begin{equation*}
    A_i = \set{x_i=\infty}
    \quad\text{and}\quad
    A_3 = \set{t x_1 x_2=1}
\end{equation*}
\item four hyperplanes that bound the integration domain $[0,1]^2\subset X$:
\begin{equation*}
    B_i=\set{x_i=0}
    \quad\text{and}\quad
    B_{2+i}=\set{x_i=1}
\end{equation*}
\end{itemize}
\begin{figure}
  \includegraphics[width=\textwidth]{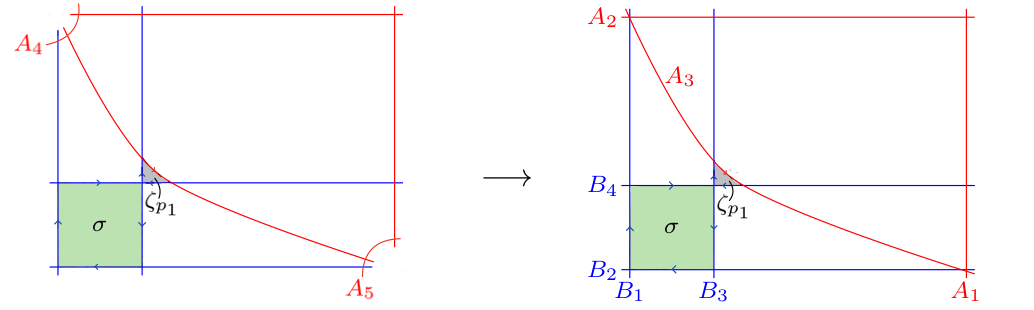}
  \caption{Blowup $X\rightarrow (\PP^1)^2$ and divisors $A$ (red) and $B$ (blue) for the dilogarithm \eqref{eq:Li2}.}%
  \label{fig:Li2-divisors}%
\end{figure}

When $\abs{t}<1$, the square $[0,1]^2$ does not intersect $A$, hence it defines a homology class $\sigma\in H_2(Y\setm A,B)_t$ in those fibres. However, the projection
\begin{equation*}
    \pi\colon Y=(\PP^1)^3 \longrightarrow T=\PP^1
\end{equation*}
onto the third factor and the subset $D\subset Y$ are not immediately amenable to Picard-Lefschetz theory, due to two problems:
\begin{enumerate}
    \item The divisor $A\cup B$ is not normal crossing, e.g.\ $B_1\cap A_2 \subset A_3$.
    \item The critical points are not always isolated, e.g.\ the entire stratum $B_1\cap A_3\setm (A_2\cup B_2\cup B_4)=\set{x_1=0}\times \PP^1\setm \set{x_2=0,1,\infty} \times \set{t=\infty}$ is contained in a single fibre, so each point on this stratum is critical.
\end{enumerate}

To solve the first problem, we separate the triple intersections by replacing $(\PP^1)^2$ with the blow up $X\rightarrow (\PP^1)^2$ of the points $(0,\infty),(\infty,0)\in (\PP^1)^2$, see \cref{fig:Li2-divisors}. We add the exceptional divisors $A_{4,5}$ to the arrangement $\irrone A$, so that
\begin{equation*}
    X \setm \ti A_t \cong (\PP^1)^2 \setm A_t
    \quad\text{and}\quad
    \ti B \setm \ti A\cong B\setm A
\end{equation*}
where $\ti A$ is the total transform of $A$ and $\ti B$ the strict transform of $B$. Having replaced the pair $( (\PP^1)^3\setm A,B\setm A)$ with $(Y\setm \ti{A},\ti{B}\setm \ti{A})$, we have achieved that $\ti{D}=\ti{A}\cup\ti{B}$ has normal crossings. Therefore, we can use the canonical stratification of $\ti{D}$ to compute the Landau variety. We find $L=\set{0,1,\infty}$:
\begin{itemize}
    \item Over $t=0$, the quadric $A_3$ (we omit $\sim$ in the following) degenerates into a union of two lines, $\set{x_1=\infty}\cup \set{x_2=\infty}$. Their intersection point $p_0=(\infty,\infty,0)\in Y$ would define a quadratic simple pinch,\footnote{Write $x_1=u+\iu v$ and $x_2=u-\iu v$ to bring $A_3=\set{x_1x_2=1/t}$ into the form \eqref{eq:localcoords}} were it not for the fact that these two lines are precisely $A_1\cup A_2$. Therefore, the critical set over $t=0$ is not just the point $p_0$, but rather the entire union
    \begin{equation*}
        P_0 \cong (A_1\cup A_2)\cap A_3 = (A_1\cup A_2) \cap \pi^{-1}(0).
    \end{equation*}
    \item Similarly, the critical set over $t=\infty$ is the union
    \begin{equation*}
        P_{\infty}\cong(B_1\cup B_2)\cap A_3 = (B_1\cup B_2) \cap \pi^{-1}(\infty).
    \end{equation*}
    \item Over $t=1$, we have linear simple pinch $p_1=(1,1,1)$ at the triple intersection $B_3 \cap B_4 \cap A_3=\set{p_1}$, with vanishing cell $\vcell_{p_1}$ (see \cref{fig:Li2-divisors}).
\end{itemize}
\begin{figure}
    \centering
    \includegraphics[width=0.95\textwidth]{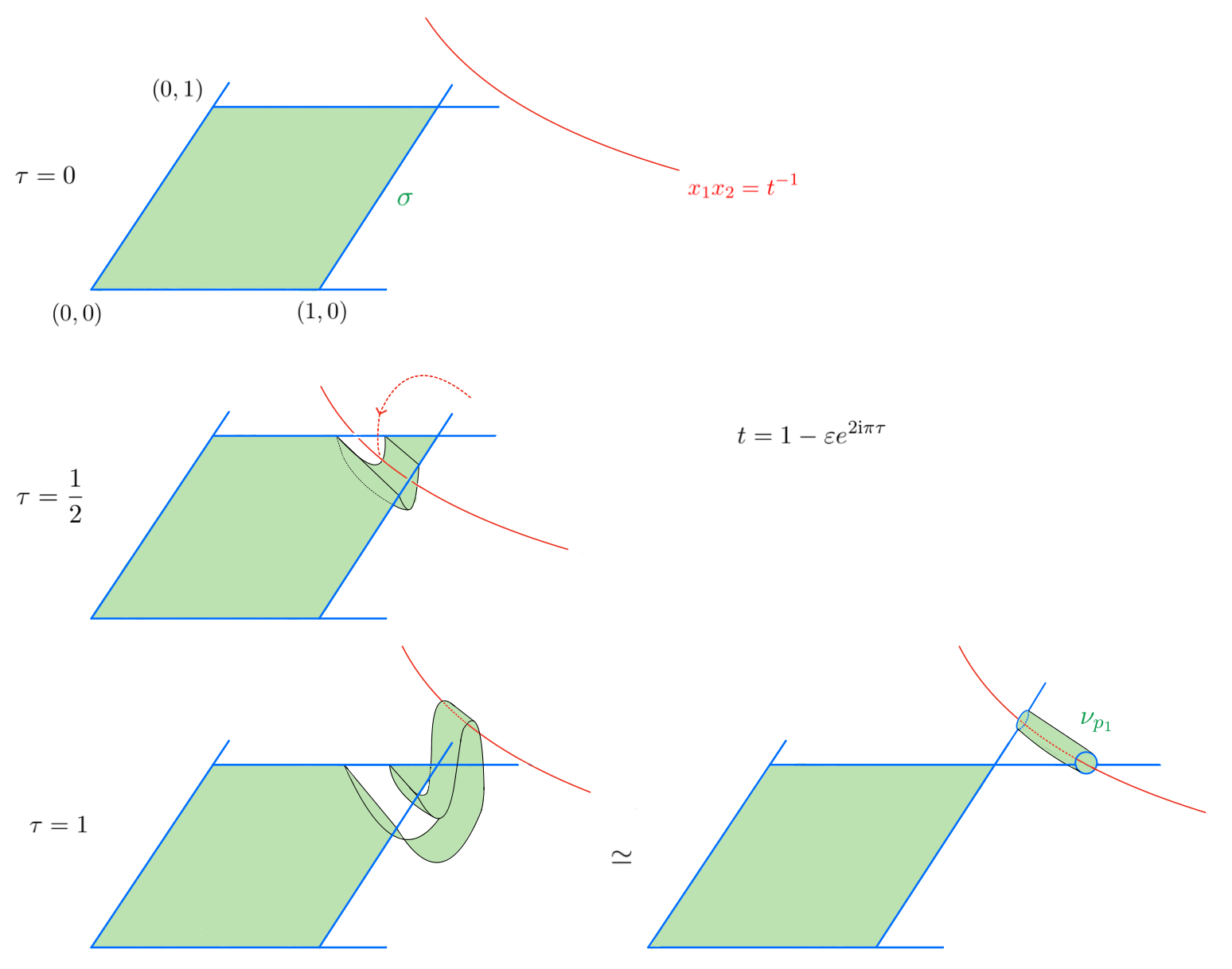}%
    \caption{As $t$ encircles the critical value 1, the relative cycle $\sigma=[0,1]^2$ is deformed to avoid the quadric $\{ x_1 x_2 = \inv t \}$. After a full circle, the quadric has ``cut out a strip'' from $\sigma$, the resulting cycle is now homologous to $\sigma + \vcyc_{p_1}$ (cf.\ \cref{fig:loga}).}%
    \label{fig:dilog-wurst}%
\end{figure}

The Picard-Lefschetz theorem applies only to the linear pinch at $t=1$ (the vanishing cycle is a cylinder, see \cref{fig:dilog-wurst}), but not to the non-isolated critical sets over $t=0,\infty$. However, we can still conclude constraints, e.g.
\begin{lem}\label{lem:li2-var.var1=0}
    $\Var_1 \circ \Var_0 = \Var_1\circ\Var_1=\Var_1 \circ \Var_{\infty} = 0$.
\end{lem}
\begin{proof}
Consider the terminology from \cref{sec:hierarchy-general}. The critical set $P_0$ over $t=0$ has type $(\set{A_1,A_2,A_3,A_4,A_5},\set{B_3,B_4})$, because it intersects all divisors except $B_{1,2}$. However, the simple type is just $(\set{A_3},\varnothing)$: removing $A_3$ from the arrangement makes all critical points disappear, whereas removing any other divisor still leaves critical points behind (especially the quadratic pinch point $p_0\in P_0$ of $A_3$ on its own). So even without an explicit description of $\Var_0$, we conclude from \cref{lem:simple-type-factorization} that it is a coboundary $\fibSphere_{A_3} \mu$ of some
\begin{equation}\label{eq:li2-H1(A3)}%\tag{$\sharp$}
    \mu\in H_1(A_3\setm (A_4\cup A_5),B_3\cup B_4)_t.
\end{equation}
Since $p_1$ is a linear pinch, the variation at $t=1$ factors through the iterated boundary in $B_3$ and $B_4$ (\cref{rem:is-partialB}). But $A_3\cap B_3\cap B_4\cap \pi^{-1}(t)=\varnothing$ for all $t\neq 1$, so $\partial_{B_3}\partial_{B_4} \mu \in H_{-1}(A_3\cap B_3\cap B_4)_t=0$ and thus $\Var_1\circ\Var_0=0$.\footnote{It is trivial that $H_{-1}=0$ (negative degree), but the argument as given here shows $H_{\bullet}(A_3\cap B_3\cap B_4)_t=0$ in all degrees and hence could generalize to other applications.} A very similar argument shows that $\Var_1\circ\Var_{\infty}=0$. Furthermore, $\Var_1\circ\Var_1=0$ holds for every linear pinch (\cref{sec:hierarchy_refinements}).
\end{proof}

Of course, the monodromy of the dilogarithm is well known: with respect to the loops $\gamma_0(\tau) = e^{2\ipi\tau}/2$ and $\gamma_1(\tau)=1-\gamma_0(\tau)$ based at $t_0=1/2$ that generate $\pi_1(T\setm L,t_0)=\pi_1(\CC\setm\set{0,1},t_0)$, one has
\begin{equation*}
    \Var_0 \Li{2}(t) = 0,\quad
    \Var_1 \Li{2}(t) = -2\ipi \log (t),\quad
    \Var_0 \log (t)=2\ipi.
\end{equation*}
From this, it follows easily that $\Var_1\circ \Var_{\ell} \Li{2}(t)=0$ for $\ell=0,1,\infty$. The point of \cref{lem:li2-var.var1=0} is that such constraints can be derived from purely topological considerations.
\begin{rem}\label{rem:li2-residuethm}
Consider the cells $\sigma$ and $\vcell_{p_1}$ with the orientation $\td x_1\wedge \td x_2$, then their intersection in the corner $B_3\cap B_4=\set{(1,1)}$ gives
\begin{equation*}
\Var_1\sigma=\is{\dvcyc_{p_1}}{\sigma} \vcyc_{p_1}=-\is{\partial_{B_3B_4}\vcell_{p_1}}{\partial_{B_3B_4}\sigma}\vcyc_{p_1}=-\vcyc_{p_1}
\end{equation*}
by \cref{rem:is-partialB}.
So with the additional input $\omega_t=\td \log(1-tx_1x_2)\wedge \td \log (x_1)$ of the actual integrand, Picard-Lefschetz reproduces $\Var_1 \Li{2}(t)=-2\ipi\log(t)$ through the residue theorem \eqref{eq:residue-theorem}:
\begin{equation*}
    \int_{\vcyc_{p_1}} \omega_t = 2\ipi \int_{\partial_{A_3} \vcell_{p_1}} \Res_{A_3} \omega_t
    =2\ipi \int_{(1/t,1)}^{(1,1/t)} \td \log (x_1)
    =2\ipi\log(t).
\end{equation*}
\end{rem}

Despite the non-isolated critical set $P_0$ over $t=0$, the variation $\Var_0$ behaves like the variation of the quadratic simple pinch $p_0\in A_3$, as if the hyperplanes $A_1\cup A_2$ would not be part of the divisor $D$. Namely, for the quadratic pinch, we would expect that $\Var_0 h = N(h)\vcyc_{p_0}$ is an integer multiple of the vanishing cycle $\vcyc_{p_0}=\fibSphere_{A_3}\eta$, where $\eta\in H_1(A_3\setm (A_4\cup A_5),B_3\cup B_4)_t$ is the class of the vanishing sphere $\Sphere^1$, embedded as the closed loop
\begin{equation*}
    \eta(\tau)=(x_1(\tau),x_2(\tau))=\big(e^{2\ipi\tau}/\sqrt{t},e^{-2\ipi\tau}/\sqrt{t}\big).
\end{equation*}
The vanishing cycle $\vcyc_{p_0}\in H_2(Y\setm A,B)_t$ is the fundamental class of an embedded torus $\Sphere^1\!\times\Sphere^1\hookrightarrow (\PP^1)^3\setm A$, which we can parametrize explicitly as
\begin{equation*}
    (x_1(\tau,\rho),x_2(\tau,\rho)) = \big(  (1+e^{2\ipi\rho}/2) e^{2\ipi \tau}/\sqrt{t},\,e^{-2\ipi\tau}/\sqrt{t}\big).
\end{equation*}
For an isolated quadratic pinch, we also would have $\Var_0\circ\Var_0=0$ by \eqref{eq:itvar_n-m_odd}.
\begin{lem}\label{lem:Li2-Var0}
    For the dilogarithm homology $H_{\bullet}(Y\setm A,B)_t$, the variation $\Var_0$ is an integer multiple of $\vcyc_{p_0}$, and $\Var_0\circ\Var_0=0$.
\end{lem}
\begin{proof}
    As a quadric in $\PP^1\!\times\PP^1$, we have $A_3\cap \pi^{-1}(t) \cong \PP^1$, so the homology group in \eqref{eq:li2-H1(A3)} is the group of $\PP^1$ punctured at two points, relative to two other points. This group has rank 2, generated by a path (in $A_3$) from $B_4$ to $B_3$ ($\partial_{A_3} \vcell_{p_1}$) and a closed loop ($\eta$). Therefore,
    \begin{equation*}
        \Var_0 h=N(h) \vcyc_{p_0}+ M(h) \vcyc_{p_1}.
    \end{equation*}
    To see that $M(h)=0$, note that $H_1(B_3\setm (A_2\cup A_3),B_2\cup B_4)_t$ has zero variation at $t=0$ (linear pinch of type $(\set{A_2,A_3},\varnothing)$). So $0=\Var_0 \partial_{B_3}h=\partial_{B_3}\Var_0 h=M(h) \alpha$, where the circle $\alpha$ at the $B_3$-end of the cylinder $\vcyc_{p_1}$ is non-zero in $H_1(B_3\setm (A_2\cup A_3),B_2\cup B_4)_t$. One way to see that $\Var_0 \vcyc_{p_0}=0$ is to compute $\int_{\vcyc_{p_0}}\omega_t = (2\ipi)^2$ with the residue theorem (see \cref{rem:li2-residuethm}).
\end{proof}
\begin{rem}
The homology classes $\nu_{p_0}$, $\nu_{p_1}$, and $\sigma$ constitute a basis of $H_2(Y\setm A,B)_t \cong \ZZ^3$ at every $t\neq 0,1,\infty$.
\end{rem}

\section{Feynman integrals} \label{sec:feynman}

We apply the results of \cref{sec:stratsandlandau,sec:PL,sec:hierarchy} to the study of Feynman integrals in the parametric representation \cite[\S 1]{Brown:FeynmanAmplitudesGalois}. Such integrals are associated to a graph $G$ (assumed to be connected) with $n+1$ edges $E(G)$, and each edge $e\in E(G)$ labels a coordinate $x_e$ (called Schwinger/Feynman parameter) of $X=\PP^n=\PP(\CC^{E_G})$. The real simplex
$\sigma=\set{[x_0:\cdots:x_n]\colon x_e>0} \subset \PP^n$
defines a relative cycle $\sigma \in H_n(X,B)$ with boundary in the divisor
\begin{equation*}
    B=B_0\cup\ldots\cup B_n,
    \quad\text{where}\quad
    B_e=\set{x_e=0}.
\end{equation*}
The first and second \emph{Symanzik polynomials} \cite{BognerWeinzierl:GraphPolynomials} of a graph, $\Upol\in\ZZ[x]$ and $\Fpol\in\ZZ[x,t]$, are homogeneous in $x$ of degrees $h_1(G)$ and $h_1(G)+1$, respectively.\footnote{In the notation of \cite{Brown:FeynmanAmplitudesGalois}, $\Upol=\Psi_G$ and $\Fpol=\Xi_G$.} Here $h_1(G)=\rk{ H_1(G)}$ is called the \emph{loop number} of $G$, and it equals $|E(G)|-|V(G)|+1$ in terms of the number $|V(G)|$ of vertices. We denote the hypersurfaces defined by the Symanzik polynomials as
\begin{equation*}
    A=A_1 \cup A_2,
    \quad\text{where}\quad
    A_1=\set{\Upol=0}
    \quad\text{and}\quad
    A_2=\set{\Fpol=0}.
\end{equation*}
The parameters $t$ in $\Fpol$ have interpretations in physics as masses and scalar products of momenta of particles. These parameters range over a complex variety $T$. We set $Y=X\times T$ and denote $\pi\colon Y\rightarrow T$ the canonical projection.
\begin{defn}\label{def:L(G)}
    The Landau variety $L_G$ of a Feynman graph $G$ is the minimal (\cref{rem:minimal-landau}) Landau variety (\cref{def:landau-variety}) of $D=A\cup B \subset Y$.
\end{defn}
Our definition follows \cite{Brown:PeriodsFeynmanIntegrals}, but other authors \cite{MizeraTelen:LandauDisc,HMSV:SeqDiscOnShell} use the phrase ``Landau variety of $G$'' for the subset $L'\subseteq L_G$ that arises as the Landau variety of the arrangement $A_2\cup B$ without $A_1$ (so-called \emph{first type singularities}).\footnote{In dimensions $D<2d$, the singularities of $\II(t)$ reduce from $L$ to just $L'$, because the form \eqref{eq:FI-integrand} has then no pole on $A_1$.}

As explained in \cref{sec:integrals}, the Landau variety $L_G$ is an upper bound on the singularities of integrals $\II(t) = \int_{\sigma} \omega$, where
\begin{equation}\label{eq:FI-integrand}
    \omega = \frac{P(x,t)\, \Omega}{\Upol(x)^{D/2-d}\Fpol(x,t)^d}
    \in H_{\dR}^n(Y\setm A,B)
\end{equation}
is a rational differential $n$-form with poles contained in $A_1\cup A_2$. The exponents in the denominator are called \textit{superficial degree of convergence} $d$ and \emph{dimension of space-time} $D$. We assume that $D/2$ and $d$ are integers, so that $\omega$ is rational and in particular single-valued. The $n$-form
\begin{equation*}
    \Omega=\sum_{e}(-1)^{e} x_e \bigwedge_{k\neq e} \td x_k
\end{equation*}
is the projectivized Lebesgue measure, and $P(x,t)$ is a polynomial, homogeneous in $x$ of degree $d-E(G)+h_1(G) D/2$, so that $\omega$ is homogeneous of degree zero. The precise form of $P$ is irrelevant for us; it depends on the nature of the particles and interaction vertices encoded by the graph.

To analyse the singularities of an integral $\II(t)=\int_{\sigma}\omega_t$ with the topological method (\cref{sec:integrals}), we must in the first step deform the simplex $\sigma$ away from $A$, in order to define a class in $H_n(Y\setm A,B)_t$. The resulting integral is then necessarily convergent, and it makes sense to discuss its singularities. We do not consider divergent Feynman integrals; those need to be renormalized first \cite{BrownKreimer:AnglesScales}.

However, convergence is not sufficient---indeed, it often happens that the boundary of $\sigma$ intersects $A$, so that $\sigma$ cannot define a class in $H_n(Y\setm A,B)_t$.
This problem is solved by replacing the fibre $\PP^n$ with an iterated blowup $\rho\colon X\rightarrow \PP^n$. We will illustrate this construction explicitly in several examples (\cref{ssec:massivesunrise,ssec:masslesstriangle}) and refer to \cite{Brown:FeynmanAmplitudesGalois} for details.\footnote{These blowups were first described by Boyling \cite[\S 3]{Boyling:HomParaFeyn}, but it appears that this work was largely ignored or unknown. They were rediscovered in \cite{BlochEsnaultKreimer:MotivesGraphPolynomials} and generalized in \cite{Brown:FeynmanAmplitudesGalois}.} In summary:
\begin{itemize}
    \item The centres of the blowups are linear subspaces $\bigcap_{e \in \gamma} B_e$ indexed by certain subgraphs $\gamma\subset E(G)$, called the \emph{motic subgraphs}.
    \item The total transform $\ti{B}=\rho^{-1}(B)$ is a simple normal crossings divisor. Its components are (strict transforms of) the original $B_e$, plus the exceptional divisors (denoted $B_{\gamma}$, one for each motic subgraph).
    \item In $X$, the closure $\ti{\sigma}=\overline{\rho^{-1}(\sigma\setm B)}$ of the preimage of the interior of $\sigma$, is a manifold with corners, called the \emph{Feynman polytope}.\footnote{This is usually not the same as the ``Feynman polytope'' $\mathbf{S}_G$ from \cite{HamedHillmanMizera:FeynmanPolytopes}: the face $\ti{\sigma}\cap B_{\gamma}$ of $\ti{\sigma}$ has codimension one for \emph{every} motic $\gamma$, but the face of $\mathbf{S}_G$ corresponding to $\gamma$ can have higher codimension, e.g.\ if $\gamma$ is not biconnected.}
\end{itemize}
This construction achieves that, for generic parameters $t$ and convergent integrands \eqref{eq:FI-integrand}, the Feynman polytope and the pullback of $\omega$ define (co)homology classes
\begin{equation*}
    [\ti{\sigma}] \in H_n(Y\setm \ti{A},\ti{B})_t
    \qquad\text{and}\qquad
    [\rho^{\ast}\omega] \in H^n_{\dR}(Y\setm \ti{A},\ti{B})_t.
\end{equation*}
Here $\ti{A} \subset Y=X\times T$ denotes the strict transform of $A$. We can now indeed apply \cref{sec:integrals} to study the monodromy of $\II(t)=\int_{\ti{\sigma}}\rho^{\ast} \omega$, purely in terms of the homology groups (called the \emph{Feynman motive})
\begin{equation}\label{eq:Mot(G)}
    H_n(Y\setm \ti{A}, \ti{B})_t
    = H_n(X\setm \ti{A}_t,\ti{B}).
\end{equation}

To apply our machinery from \cref{sec:stratsandlandau,sec:PL} straightforwardly, we require that $\ti{D}=\ti{A}\cup\ti{B}$ is transverse. Although $\ti{B}$ is transverse on its own, this property is easily lost once we add $\ti{A}$ to the arrangement. Hence, further blowups become necessary, like for the dilogarithm (\cref{sec:dilog}). There is no general recipe for this step. We will discuss examples where $\ti{A}_2$ is smooth but not transverse to $\ti{B}$ (\cref{ssec:massivesunrise,sec:icecream}).

We begin this section with a discussion of the simplest case, massive one loop Feynman integrals. Here, no blowups whatsoever are needed, and all critical points are simple pinches. Hence we can apply the results of \cref{ss:PLthm,sec:linear-pinch,sec:quadratic-pinch} and find strong hierarchies according to \cref{sec:pinch-preorder,sec:hierarchy_refinements,sec:hierarchy_spc}.

After that, we consider two rather degenerate examples: the massive sunrise and the massless triangle (\cref{fig:triangle+sunrise}). We discuss the construction of their Feynman motives and demonstrate that several Landau singularities arise from \emph{non-isolated critical points} (\cref{sec:hierarchy-general}), and that often the simple type is only $(\set{A_2},\varnothing)$. As a consequence, a more subtle analysis is needed to derive hierarchy constraints, and we give examples of such.

In \cref{sec:icecream} we study the ice cream cone graph and we prove hierarchy constraints in this notoriously subtle case. Note that these subtleties are a general phenomenon rather than an exception, as was already noted in \cite{Boyling:HomParaFeyn}. We suggest a modification of the Feynman motive by additional blowups, to resolve non-transverse intersections (\cref{rem:mod_motive}).

\subsection{Massive one loop integrals}\label{ssec:oneloop}
We consider the geometry associated to generic (fully massive) 1-loop integrals with $n+1$ external and internal edges. The case $n=1$ was discussed in \cref{eg:bubble}.

The parameters $t=(m,s)\in T=\CC^{(n+1)(n+2)/2}$ consist of $n+1$ masses $m_i\in\CC$ of internal edges, and $n(n+1)/2$ invariants $s_{ij}=s_{ji}=-(p_i+\ldots+p_{j-1})^2\in\CC$ formed from the incoming momenta $p_i\in\CC^D$ at each vertex.\footnote{We treat the invariants $s_{ij}$ as unconstrained parameters, that is, we assume that $D\geq n$ so that there are no relations among the $s_{ij}$ other than $s_{ii}=0$ and $s_{ij}=s_{ji}$.} We can identify $T$ with the space of all symmetric $(n+1)\times(n+1)$ matrices $M$, with entries
\begin{equation*}
    M_{ij} = \begin{cases} m_i^2 & \text{ if } i=j, \\
     \frac{m_i^2+m_j^2+s_{ij}}{2} & \text{ if } i\neq j.
    \end{cases} 
\end{equation*}
The two Symanzik polynomials are $\Upol=x_1+ \ldots + x_{n+1}$ and $\Fpol=x^{\Transpose} M x$. They define a hyperplane $A_1=\set{\Upol=0}$ and a quadric $A_2=\set{\Fpol=0}$. Together with the $(n+1)$ coordinate hyperplanes $B_i=\set{x_i=0}$, we obtain a simple normal crossings divisor $A\cup B\subset \PP^{n} \times T$.

The quadric $A_{2 t}=A_1\cap\pi^{-1}(t) \subset \PP^{n}$ is smooth precisely when $\det M \neq 0$. The critical values of $\pi|_{A_2}$ form the \emph{leading Landau singularity of the first type},
\begin{equation*}
    \ell_G^{\Fpol}= \set{\det M=0}\subset T.
\end{equation*}
Over a generic point $t$ of $\ell_G^{\Fpol}$, the rank of $M$ is $n$, and the one-dimensional kernel $\CC p$ of $M$ defines the unique singular point $[p]\in\PP^n$ of $A_{2 t}$. This point is a quadratic simple pinch, and the corresponding vanishing cycle is an $\Sphere^1$-bundle over the vanishing sphere $\cong \Sphere^{n-1} \subset A_{2 t}$ (cf.\ \cref{ss:vanishers}).

The \emph{leading singularity of the second type} refers to the critical values of the intersection $A_{12}=A_1\cap A_2$. Its fibre over $t$ is
\begin{equation*}
    A_{12}\cap\pi^{-1}(t)
    = \set{x_1+\ldots+x_{n+1}=0} \cap \set{ x^{\Transpose} S x = 0 }
    \subset \PP^n
\end{equation*}
and depends only on the momentum part $S=M|_{m=0}=(s_{ij}/2)$ of $M$, but not the masses. A point $[x]\in A_{12}$ is critical if either, the normal vectors $\mathbf{1}=(1,\ldots,1)^{\Transpose}$ and $Sx$ of $A_1$ and $A_2$ become collinear, or if $\set{ x^{\Transpose} S x = 0 }$ is singular. So either $\det S=0$, or $S$ is invertible and $x\propto S^{-1}\mathbf{1}$. In the latter case, for $x$ to lie on $A_{12}$, we need $\mathbf{1}^{\Transpose} S^{-1}\mathbf{1}=0$.
Thus, using the \emph{Cayley matrix}
\begin{equation*}
    S'
    = \begin{pmatrix}
        0 & \mathbf{1}^{\Transpose} \\
        \mathbf{1} & S \\
    \end{pmatrix}
\end{equation*}
with determinant $(\mathbf{1}^{\Transpose} S^{-1} \mathbf{1})\cdot \det S$, the leading second type singularity is\footnote{$\det S'$ is proportional to the Gram determinant \cite[Appendix A.1]{Melrose:RedFeyn}.}
\begin{equation*}
    \ell_G^{\Fpol\Upol} = \left\{\det S'=0  \right\} \subset T.
\end{equation*}
Over a generic point of this component, there is a unique critical point. It is a quadratic simple pinch of type $(\set{A_1,A_2},\varnothing)$.
The corresponding vanishing cycle is an $\Sphere^1\!\times \Sphere^1$-bundle over the vanishing sphere $\cong \Sphere^{n-2} \subset A_{12}\cap\pi^{-1}(t)$.

All other strata lie in $B$. Since $\Upol|_{x_e=0}$ and $\Fpol|_{x_e=0}$ are the Symanzik polynomials of the graph $G/e$ where the edge $e$ is contracted, the critical values of $A_2\cap B^I$ and $A_{12}\cap B^I$ are the leading singularities of the quotient graph $G/I$, where $B^I=\bigcap_{e\in I}B_e$ for any subset $I$ of edges. To summarize:
\begin{prop}
The full Landau variety of a generic one-loop graph $G$ is
\begin{equation}\label{eq:L(oneloop)}
    L_G 
    = \bigcup_{ I \subsetneq E(G) } \ell_{G/I}^\Fpol \cup \ell_{G/I}^{\Fpol\Upol}.
\end{equation}
All components of $L_G$ are simple pinch components (\cref{defn:simplecomponent}).
\end{prop}
If we denote by $M_I$ and $S'_I$ the submatrices without the rows and columns indexed by $I$, then $\ell^{\Fpol}_{G/I}=\set{\det M_I=0}$ and $\ell^{\Fpol\Upol}_{G/I}=\set{\det S_I'=0}$. We exclude $\abs{I}=n+1$ in \eqref{eq:L(oneloop)} because $B_1\cap\ldots\cap B_{n+1}=\varnothing$. In the case $\abs{I}=n$, the graph $G/I$ consists of a single edge, say $i$, and we have no 2nd type singularity: $\ell^{\Fpol\Upol}_{G/I}=\varnothing$ because $A_1\cap B^I$ is empty. The first type singularity is
\begin{equation}\label{eq:type1(tadpole)}
    \ell^{\Fpol}_{G/I}=\set{m_i^2=0}
    \quad\text{for}\quad
    I=E(G)\setm\set{i}.
\end{equation}
For $\abs{I}=n-1$, the quotient $G/I$ is a graph with two edges (c.f.\ \cref{eg:bubble}). In particular, the second type singularity is
\begin{equation}\label{eq:type2(bubble)}
    \ell_{G/I}^{\Fpol\Upol}=\set{s_{ij}=0}
    \quad\text{for}\quad
    I=E(G)\setm\set{i,j}.
\end{equation}
The components \eqref{eq:type1(tadpole)} and \eqref{eq:type2(bubble)} come from strata that are $(n+1)$-fold intersections, hence linear simple pinches. All other components of $L_G$ have quadratic simple pinches.

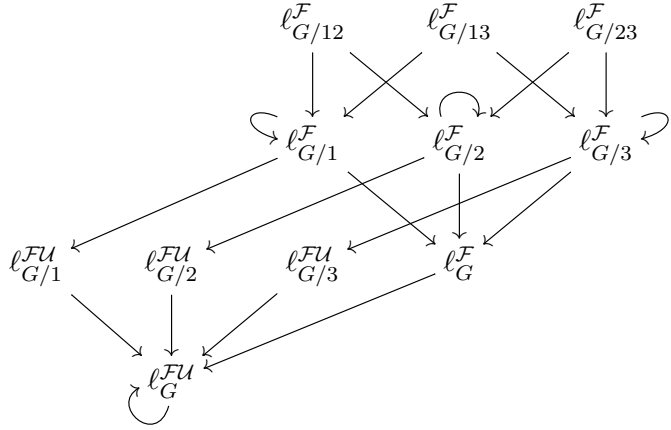
\begin{figure}
\begin{equation*}\xymatrix{
& & \ell_{G/12}^{\Fpol} \ar[d] \ar[dr] & \ell_{G/13}^{\Fpol} \ar[dl] \ar[dr] & \ell_{G/23}^{\Fpol} \ar[dl] \ar[d] \\
& & \ell_{G/1}^{\Fpol} \ar[dr] \ar[dll] \ar@(ul,l) & \ell_{G/2}^{\Fpol} \ar[d]\ar[dll] \ar@(ul,ur) & \ell_{G/3}^{\Fpol} \ar[dl] \ar[dll] \ar@(ur,r) \\
\ell^{\Fpol\Upol}_{G/1} \ar[dr] & \ell^{\Fpol\Upol}_{G/2} \ar[d] & \ell^{\Fpol\Upol}_{G/3}\ar[dl] & \ell_G^{\Fpol} \ar[dll] & \\
 & \ell^{\Fpol\Upol}_G \ar@(d,l) & & & \\
}\end{equation*}
    \caption{The hierarchy of the massive triangle integral. An arrow $\ell\rightarrow \ell'$ indicates the compatibility $\ell' \CV \ell$ of the corresponding variations. If there is no directed path from $\ell$ to $\ell'$, then $\Var_{\ell'}\circ\Var_{\ell}=0$ (\cref{sec:hierarchy_spc}).}%
    \label{fig:poset-triangle}%
\end{figure}
The hierarchy relation defined in \cref{sec:hierarchy} is antisymmetric in this case, and illustrated in \cref{eq:posetbubble} for $n=1$ and in \cref{fig:poset-triangle} for $n=2$. The reflexive elements (indicated by a self-loop) are $\ell_{G/I}^{\Fpol}$ with $n-\abs{I}$ odd, and $\ell_{G/I}^{\Fpol\Upol}$ with $n-\abs{I}$ even, according to \eqref{eq:itvar_n-m_odd}.
This hierarchy implies:
\begin{enumerate}
    \item The variation around a singularity of $G/I$, followed by a variation of a graph $G/J$ that contracts an edge that is not in $I$, must vanish:
\begin{equation*}
    \left(\Var_{\ell_{G/J}^{\Fpol \Upol} } \ \text{or}\ \Var_{\ell_{G/J}^\Fpol} \right)
    \circ
    \left(\Var_{\ell_{G/I}^{\Fpol \Upol} } \ \text{or}\ \Var_{\ell_{G/I}^\Fpol} \right)
    = 0
    \quad\text{whenever}\quad
    J \not \subseteq I.
\end{equation*}
    \item The variation of any second type singularity, followed by the variation of any first type singularity, must vanish:
\begin{equation*}
    \Var_{\ell_{G/J}^\Fpol} \circ \Var_{\ell_{G/I}^{\Fpol \Upol} } 
    = 0
    \quad\text{for all}\quad
    I,J\subsetneq E(G).
\end{equation*}
    \item If $n-\abs{I}$ is even, then $\Var_{\ell^{\Fpol}_{G/I}} \circ \Var_{\ell^{\Fpol}_{G/I}} = 0$.
    If $n-\abs{I}$ is odd, then $\Var_{\ell^{\Fpol\Upol}_{G/I}} \circ \Var_{\ell^{\Fpol\Upol}_{G/I}} = 0$.
\end{enumerate}
We hence recover the results of \cite[eq.'s (2.7), (2.12--14)]{Boyling:HomParaFeyn}. For an analysis of one-loop integrals in (compactified) \emph{momentum space}, see \cite{Boyling:VanishingHyperspheres,Muehlbauer:Cutkosky}. The fundamental group of $T\setm L$ is studied in \cite{PonzanoReggeSpeerWestwater:Monodromy1}.
\begin{rem}
    The integrals $\int_{\sigma} \Omega/\Fpol^{D/2}$ with $D=n+1$ compute the volume of a spherical simplex and have been studied in great detail. Since $\Upol$ is absent, the singularities of 2nd type disappear. It was shown in \cite{Aomoto:AnalyticSchlaefli} that these integrals can be written as iterated integrals of logarithmic one-forms. As mentioned in \cref{sec:Aomoto}, the corresponding symbol encodes the maximal iterated variations. We verified that the symbol \cite[eq.~(3.15)]{Aomoto:AnalyticSchlaefli} is compatible with the hierarchy. In fact, as for the Aomoto polylogarithm, the variations can be computed explicitly using Picard-Lefschetz theory, see \cite[\S 6]{ArkaniHamedYuan:SphericalPlanesQuadrics}.
\end{rem}

\subsection{Massless triangle}\label{ssec:masslesstriangle}
Setting a mass to zero is a singularity of the generic one-loop integral. The corresponding linear pinch \eqref{eq:type1(tadpole)} indicates that the integration simplex $\sigma$ now intersects the quadric $A_2=\set{\Fpol=0}$ in a corner.

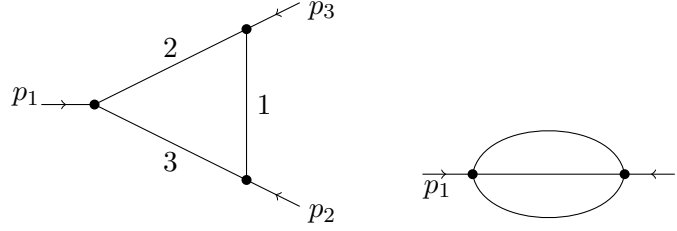
\begin{figure}
  \begin{tikzpicture}
   \coordinate (v0) at (0,0);
   \fill[black] (v0) circle (.066cm);
   \coordinate  (v1) at (2,-1);
   \fill[black] (v1) circle (.066cm);
   \coordinate (v2) at (2,1);
   \fill[black] (v2) circle (.066cm); 
   \coordinate (p3) at (2.7,-1.35);
   \coordinate (p4) at (2.7,1.35);
   \coordinate (p1) at (-0.7,0);
   \draw[decoration={markings, mark=at position 0.625 with {\arrow{<}}},postaction={decorate}] (v1) -- (p3) node [xshift=0.3cm,yshift=-0.1cm] {$p_2$};
   \draw[decoration={markings, mark=at position 0.625 with {\arrow{<}}},postaction={decorate}] (v2) -- (p4) node [xshift=0.3cm,yshift=-0.1cm] {$p_3$};
   \draw[decoration={markings, mark=at position 0.625 with {\arrow{<}}},postaction={decorate}] (v0) -- (p1) node [xshift=-0.2cm,yshift=0.1cm] {$p_1$};
   \draw (v0) -- (v1) node [midway,below] {$3$};
   \draw (v0) -- (v2) node [midway,above] {$2$};
   \draw (v2) -- (v1) node [midway,right] {$1$};
  \end{tikzpicture}
  \qquad
  \begin{tikzpicture}
   \coordinate (v0) at (0,0);   
   \coordinate  (v1) at (2,0); 
   \draw (v0) -- (v1);
   \draw (v0) to[out=80,in=100] (v1);
   \draw (v0) to[out=-80,in=-100] (v1);  
   \fill[black] (v0) circle (.066cm);
   \fill[black] (v1) circle (.066cm);
   \coordinate (p1) at (-0.66,0);
   \coordinate (p2) at (2.66,0);
   \draw[decoration={markings, mark=at position 0.625 with {\arrow{<}}},postaction={decorate}] (v0) -- (p1) node [xshift=.2cm,yshift=-0.2cm] {$p_1$};
   \draw[decoration={markings, mark=at position 0.625 with {\arrow{<}}},postaction={decorate}] (v1) -- (p2) node [xshift=-0.2cm,yshift=-0.2cm] {};
  \end{tikzpicture}
\caption{The triangle graph (left) and the sunrise graph (right).}%
\label{fig:triangle+sunrise}%
 \end{figure}

For the massless triangle graph (\cref{fig:triangle+sunrise}) with parameters $t=(p_1^1,p_2^2,p_{3}^2)\in T=\CC^3$ and the Symanzik polynomials
\begin{equation*}
    \Upol=x_1 +x_2 +x_3
    \quad\text{and}\quad
    \Fpol=-p_1^2 x_2x_3 - p_2^2x_1x_3 - p_3^2x_1x_2,
\end{equation*}
this situation is illustrated in \cref{fig:massless-triangle-divisors}. The Feynman motive is obtained from the blowup $X\rightarrow \PP^2$ of the three intersection points\footnote{The three subgraphs consisting of two edges are the ``motic graphs'' from \cite{Brown:FeynmanAmplitudesGalois}.}
\begin{equation*}
    A_2\cap B = (B_2\cap B_3) \cup (B_1\cap B_3) \cup (B_1\cap B_2) = \set{[1:0:0],[0:1:0],[0:0:1]}
\end{equation*}
of $B=B_1\cup B_2\cup B_3$. The total transform $\ti{B}$ of $B$ includes the exceptional divisors $B_{ij}\cong \PP^1$ above $x_i=x_j=0$ and thus has 6 components in total. It bounds the Feynman polytope $\ti{\sigma}$, which is a hexagon, see \cref{fig:massless-triangle-divisors}.
\begin{figure}
    \centering
    \includegraphics[width=\textwidth]{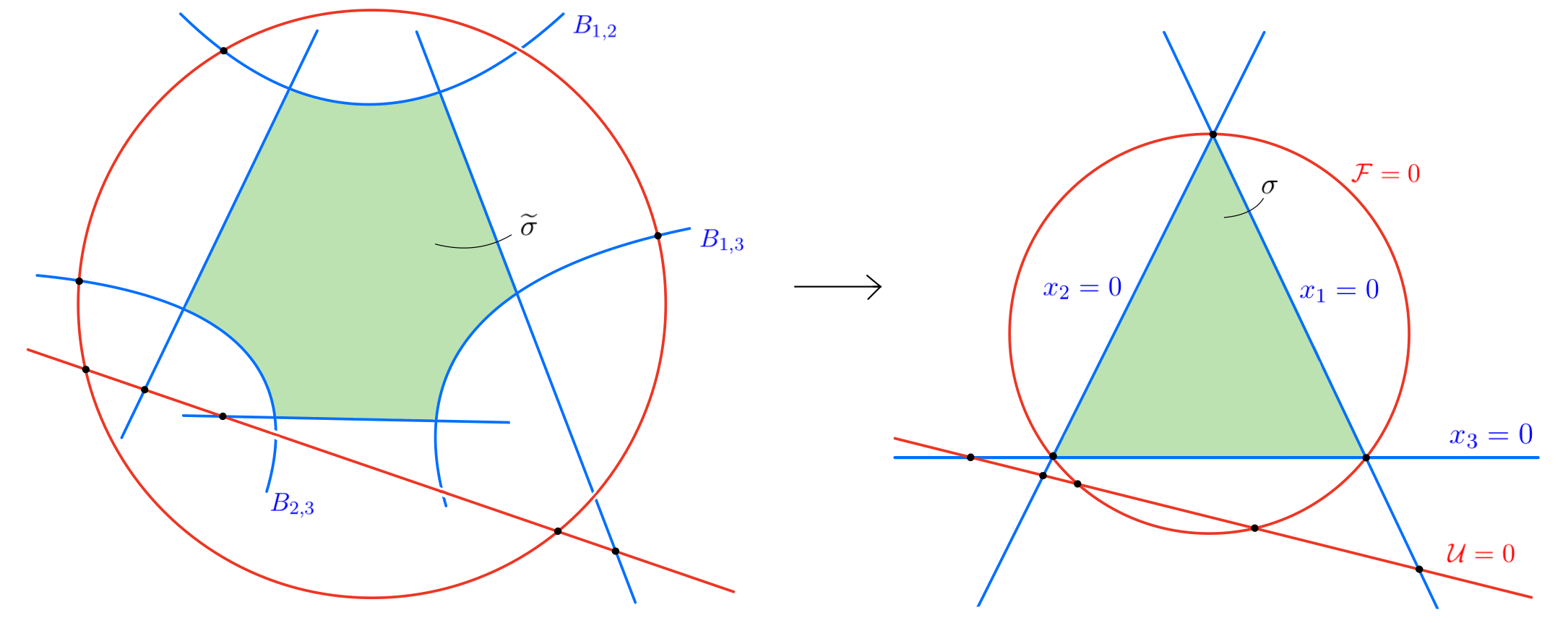}%
\caption{The geometry of the Feynman motive of the massless triangle, in a generic fibre.}%
    \label{fig:massless-triangle-divisors}%
\end{figure}
The blowups do not affect the existing divisor components, and we denote their strict transforms in $Y=X\times T$ still as $B_i$ and $A_j$. All divisor components are smooth, $A_i\cong B_j\cong B_{jk} \cong \PP^1$, and intersect transversely.

An inspection of the projection $\pi\colon Y\rightarrow T$ and the canonical stratification of $D=A_1\cup A_2\cup \ti{B}$ reveals a Landau variety with 4 components:
\begin{equation}\label{eq:L(massless-tri)}
    L_G = \ell_1\cup \ell_2 \cup \ell_3 \cup \ell_{\delta}.
\end{equation}
The first three components $\ell_i=\set{p_i^2=0}$ constitute a singularity of the 1st type. Namely, when $p_i^2=0$, then $A_2\cap Y_t=B_i \cup \set{p_k^2 x_j+p_j^2 x_k}$ degenerates into a union of two lines, where $\set{i,j,k}=\set{1,2,3}$. Since one of these lines is $B_i$, the critical set over $\ell_i$ is not just the simple quadratic pinch point of $A_2$ (e.g.\ $[0:p_2^2:-p_3^2]$ for $i=1$), but the entire line $B_i$. Note that $A_2$ and $B_i$ do not intersect for generic $t$, only when $p_i^2=0$. The critical set over $\ell_i$ is therefore not isolated and can be denoted as
\begin{equation*}
    P_i = c(A_2\cap B_i) = B_i\cap \ell_i = \set{x_i=0} \times \set{p_i^2=0} \subset X\times T.
\end{equation*}

The 2nd type singularity $\ell_{\delta}=\set{\delta=0}$ arises from the coincidence of the two intersection points
$(A_1 \cap A_2)_t = \set{Q_1,Q_2} $.
Explicitly, it is defined by
\begin{equation*}%\label{eq:triangle-2nd-type}
    \delta = p_1^4+p_2^4+p_3^4-2p_1^2p_2^2-2p_1^2p_3^2-2p_2^2p_3^2,
\end{equation*}
because $ Q_{1,2} = [-p_1^2+p_2^2-p_3^2 \pm \sqrt{\delta} : p_1^2-p_2^2 - p_3^2 \mp \sqrt{\delta} : 2p_3^2] $ degenerate to the same point when $\delta=0$.\footnote{In terms of the notation from \cref{ssec:oneloop}, $\ell_{\delta}=\ell^{\Fpol\Upol}_G$ and $\delta=\det S'$.} This is a quadratic simple pinch, with the vanishing sphere $\set{Q_1,Q_2}\cong\Sphere^0$.

The Picard-Lefschetz formula does not apply due to the non-isolated critical set $P_i$, but the variation still localizes in the sense of \cref{lem:localization} into some tubular neighbourhood $W$ of $B_i$. The type of $P_i$ is $(\set{A_1,A_2}, \set{B_i,B_{ij},B_{ik}})$, but only $A_2$ is simple in the sense of \cref{def:simple-type} (c.f.\ \cref{eg:simple-nonsimple}). So without further analysis, all components of \eqref{eq:L(massless-tri)} have the same simple type, and we do not get any hierarchy constraints from \cref{def:rel_landau_components}.

However, with a more detailed analysis, we can (as in \cref{lem:li2-var.var1=0} and \cref{lem:Li2-Var0}) nevertheless still derive non-trivial constraints:
\begin{lem}\label{lem:tri0-hierarchy}
    $\Var_{\ell_i}\circ \Var_{\ell_i}=\Var_{\ell_i}\circ\Var_{\ell_{\delta}}=0$ for all $i\in\set{1,2,3}$.
\end{lem}
\begin{proof}
    The type of $P_i$ shows that the localization in $W$ factors $\Var_{\ell_i}$ through $H_2(Y\setminus A,B_i\cup B_{ij} \cup B_{ik})_t$. Since $A_2$ is simple, and $A_{2t}=A_2\cap Y_t$ does not intersect $B_i$, we know that $\Var_{\ell_i}$ must be a coboundary $\fibSphere_{A_2}$ of some class in
\begin{equation*}
    H_1(A_2\setminus A_1,B_{ij} \cup B_{ik})_t \cong \ZZ^2.
\end{equation*}
    This group is generated by a path $\gamma_i$ in $A_{2t}\setm A_1$ between the points $B_{ij}\cap A_{2t}$ and $B_{ik}\cap A_{2t}$, and a closed loop $\eta=\fibSphere_{A_1} Q_1=-\fibSphere_{A_1} Q_2$. Note that
    \begin{equation*}
        \vcyc_{\delta}=  \delta_{A_2,A_1} \partial_{A_2,A_1} \vcell_{\delta} = 2\fibSphere_{A_2} \fibSphere_{A_1} Q_1 =2\fibSphere_{A_2} \eta 
        \in H_2(Y\setminus A)_t
    \end{equation*}
    is the vanishing cycle of the quadratic pinch over $\ell_{\delta}$. %Set $\vcyc_i\defas \fibSphere_{A_2} \gamma_i$, so 
    The image of $\Var_{\ell_\delta}$ and $\Var_{\ell_i}$ is thus contained in the span of $\vcyc_\delta/2$ and the cylinder
    \begin{equation*}
        \vcyc_i\defas \fibSphere_{A_2} \gamma_i
        \in H_2(Y\setminus A,B_{ij}\cup B_{ik} \cup B_{jk})_t.
    \end{equation*}
    Since $B_i$ is now absent, the critical point over $\ell_i$ is an isolated, quadratic simple pinch $q_i$ (e.g.\ $q_1=[0:p_2^2:-p_3^2]$). The claim, $\Var_{\ell_i}\vcyc_{\delta}=\Var_{\ell_i}\vcyc_i=0$, thus follows from the vanishing of the intersection numbers
    \begin{equation*}
        \is{\vcell_{i}}{\vcyc_{\delta}} = -2 \is{\partial_{A_2} \vcell_i}{\eta}
        \quad\text{and}\quad
        \is{\vcell_i}{\vcyc_i} = -\is{\partial_{A_2}\vcell_i}{\gamma_i}
    \end{equation*}
    with the vanishing cell $\vcell_i$ of the quadratic pinch. This cell is localized in a neighbourhood of the critical point $q_i$. The intersection with $\vcyc_{\delta}$ vanishes, because the latter can be localized near $Q_1$ or $Q_2$, and even over a critical value $t_c\in\ell_i$ (and hence, by continuity, for nearby $t\notin L$), neither point is equal to $q_i$. For example, $q_1=[0:p_2^2:-p_3^2]$ is not equal to
    \begin{equation*}
        Q_1|_{t_c} = [0:1:-1]
        \quad\text{or}\quad
        Q_2|_{t_c} = [p_2^2-p_3^2:-p_2^2:p_3^2],
    \end{equation*}
    unless $t_c$ lies in the intersection $\ell_1\cap \ell_{\delta}=\set{p_1^2=0}\cap\set{p_2^2=p_3^2}$.
\begin{figure}
    \centering
    \begin{tabular}{ccc}
    \includegraphics[height=4cm]{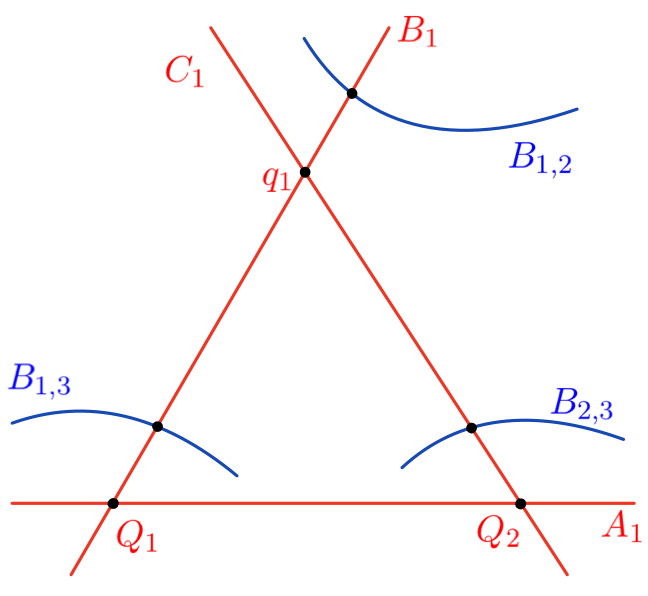} & \includegraphics[height=4cm]{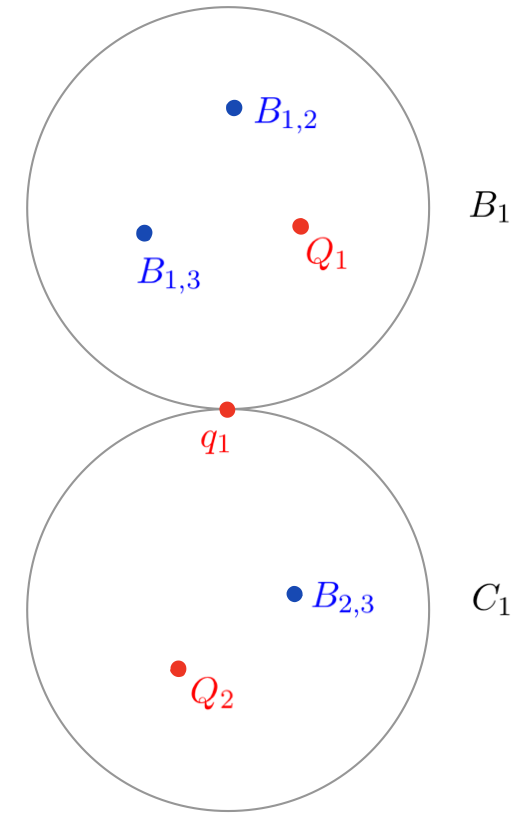} & \includegraphics[height=4cm]{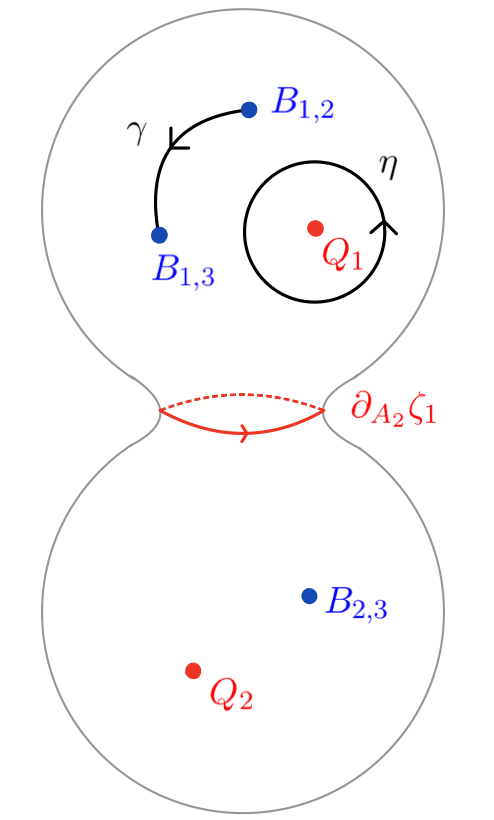} \\
    $(A\cup B_{12}\cup B_{13}\cup B_{23})_{t\in\ell_1}$  & $A_2\cap (A_1\cap \ti{B})_{t\in\ell_1}$ & $A_2\cap(A_1\cap\ti{B})_{t\notin\ell_1}$ \\
    \end{tabular}%
    \caption{The reduced divisor (without $B$) of the massless triangle in a critical fibre (left), the degenerate quadric $A_{2t}=B_1\cup C_1$ (middle), the quadric $A_{2t}$ in a nearby fibre (right).}%
    \label{fig:tri0-crit}%
\end{figure}

    The embedded vanishing sphere $\Sphere^1\subset A_{2t}$ with class $[\Sphere^1]=\partial_{A_2}\vcell_i$ cuts the quadric $A_{2t}\cong\Sphere^2$ over $t\notin L$ into two hemispheres. Over $\ell_i$, these hemispheres degenerate into the two components $B_i$ and $C_i=\set{p_k^2 x_j+p_j^2 x_k=0}$ of $A_{2t_c}$. Since both endpoints of $\gamma_i$ lie on the same component ($B_i$) in the critical fibre, they must lie on the same hemisphere for $t$ near, but not on, $\ell_i$. Hence the path $\gamma_i$ can be homotoped to lie within that hemisphere and thus does not intersect the vanishing sphere. This is illustrated in \cref{fig:tri0-crit}.
\end{proof}

In summary, for the massless triangle, the double variations can be computed using Picard-Lefschetz theory: after the first variation, no cycle has boundary in $B$, hence we can drop $B$ from the arrangement such that for the second variation, the critical points are quadratic simple pinches.

\Cref{lem:tri0-hierarchy} covers all constraints of the form $\Var_{\ell'}\circ\Var_{\ell}=0$. In particular,
\begin{equation}\label{eq:tri0-vari-varj}
    \Var_{\ell_i} \circ \Var_{\ell_j} \neq 0
\end{equation}
for distinct $i,j\in\set{1,2,3}$. Namely, $\is{\vcell_i}{\vcyc_j}=\pm 1$ and so $\Var_{\ell_i}\vcyc_j=\pm\vcyc_{\delta}/2$.\footnote{The path $\gamma_j$ must cross the vanishing sphere, because the endpoints of $\gamma_j$ lie in different hemispheres, see \cref{fig:tri0-crit} and the last part of the proof of \cref{lem:tri0-hierarchy}.}
Furthermore, although this does not follow directly from Picard-Lefschetz theory, we can choose $\gamma_i$ so that $\Var_{\ell_i} \ti{\sigma}=\vcyc_i$, and therefore
\begin{equation*}
    \Var_{\ell_i}\circ \Var_{\ell_j} \ti{\sigma} = \pm \vcyc_{\delta}/2.
\end{equation*}
\begin{rem}
If one interprets $\Var_{\ell_i}$ as only the singularity of the massless bubble quotient graph $G/i$, then \eqref{eq:tri0-vari-varj} would violate the hierarchy principle of absorption integrals \cite{Pham:DiffusionMultiple,HMSV:SeqDiscOnShell}. However, there is no contradiction here, because the latter does not apply to Feynman integrals to begin with (the integration domain $\ti{\sigma}$ has boundary), and also because the masses would be required to be non-zero. In addition, the Landau components $\ell_i$ also parametrize the leading singularity of the full triangle graph, hence $\Var_{\ell_i}$ cannot be attributed to $G/i$ alone (in our terminology, $B_i$ is not simple).
\end{rem}
\begin{rem}
The relations from \cref{lem:tri0-hierarchy} and \eqref{eq:tri0-vari-varj} agree with explicit computations of the triangle integral in terms of polylogarithms \cite[\S 6.2.2]{BHMLSV:SeqDiscMon}.
\end{rem}

\subsection{Massive sunrise} \label{ssec:massivesunrise}

The sunrise graph with unequal masses depends on four parameters $t=(p^2,m_1,m_2,m_3) \in T=\CC^4$. Its Symanzik polynomials
\begin{equation*}
    \Upol=x_2x_3+x_1x_3+x_1x_2, 
    \quad 
    \Fpol
    %=-p^2 x_1x_2x_3 + \Upol (m_1^2x_1+m_2^2x_2+ m_3^2x_3)
    =(m_1^2x_1+m_2^2x_2+ m_3^2x_3)\Upol-p^2 x_1x_2x_3
\end{equation*}
define a smooth quadric and cubic, respectively, in $\PP^2$. Both of them intersect in the corners $x_i=x_j=0$ of the simplex $\set{x_1 x_2 x_3=0}$. The Feynman motive is defined by the same blowup $X'\rightarrow \PP^2$ as in \cref{ssec:masslesstriangle}. It achieves the separation of the hexagon $\ti{\sigma}$ (Feynman polytope), which has boundary in
\begin{equation*}
    B'=B'_1\cup B'_2 \cup B'_3 \cup B'_{12} \cup B'_{13} \cup B'_{23},
\end{equation*}
from the strict transforms $A_1',A_2'\subset X'\times T$ of $\set{\Upol=0}$ and $\set{\Fpol=0}$. Here we denote $B_i'$ the strict transform of $\set{x_i=0}$ and $B'_{ij}$ the exceptional divisors over the points $\set{x_i=x_j=0}\in\PP^2$.

\begin{figure}
    \centering
    $\vcenter{\hbox{\includegraphics[height=4.17cm]{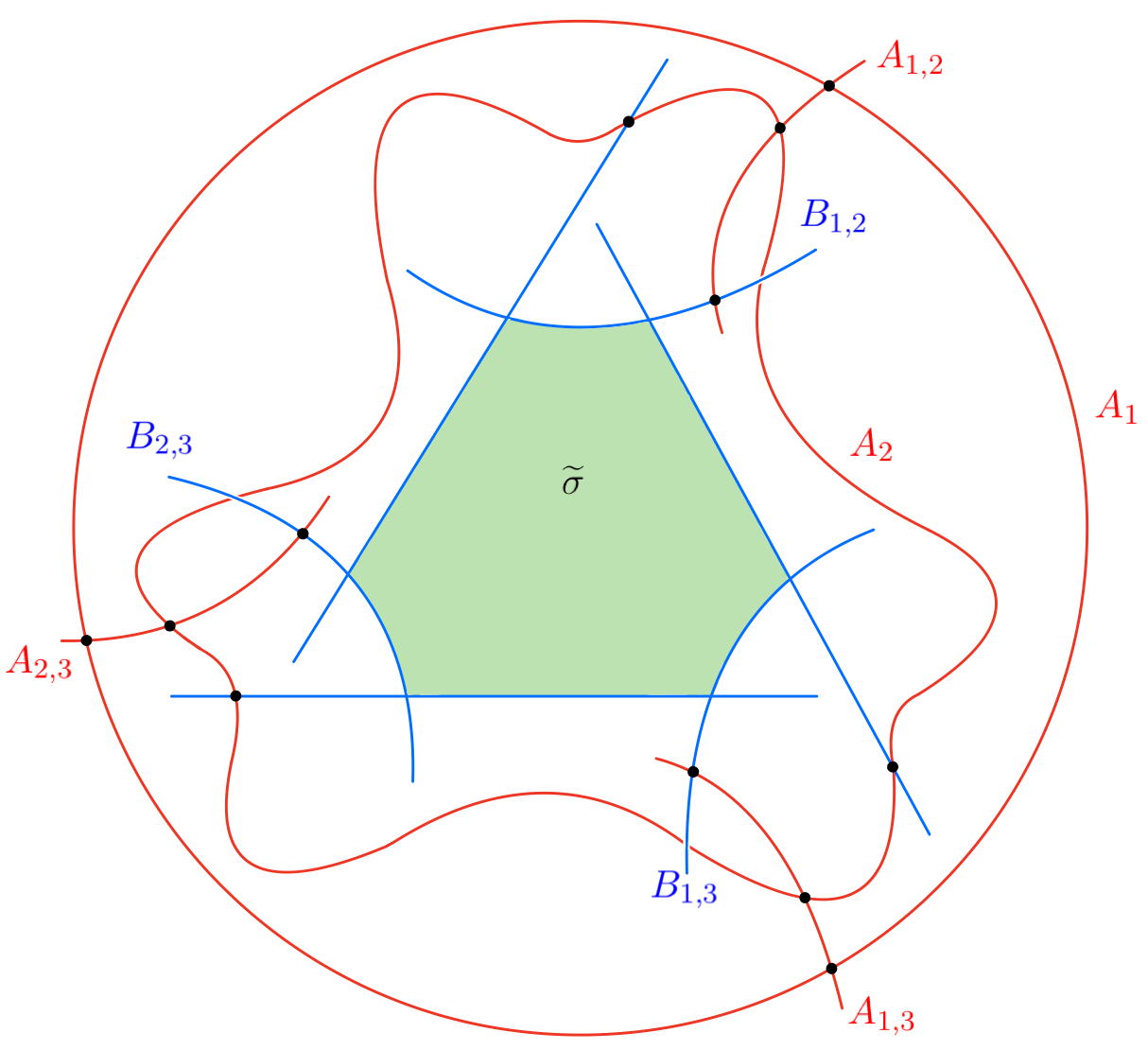}}}$
    $\subset X \quad \longrightarrow \quad X' \supset$
    $\vcenter{\hbox{\includegraphics[height=4.17cm]{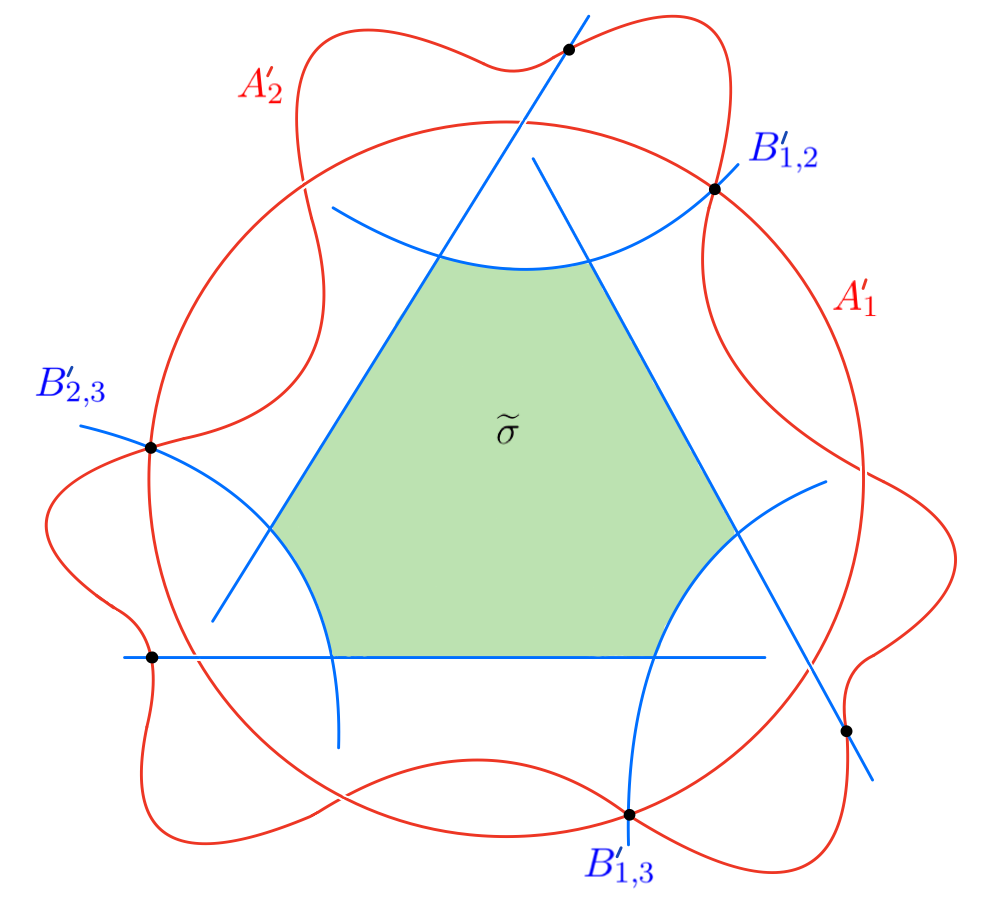}}}$
    \caption{Divisors $A\cup B\subset X$ and $A'\cup B'\subset X'$ of the sunrise graph.}%
    \label{fig:massive-sunrise-divs}%
\end{figure}

However, the divisor $A'\cup B' \subset X'\times T$ is still not normal crossing, because on each $B'_{ij}$, the intersections with $A_1'$ and $A_2'$ coincide (see \cref{fig:massive-sunrise-divs}). For example, in local coordinates $x=[u:uv:1]$ where $B'_{12}=\set{u=0}$, one finds
\begin{equation*}
    A_1' \cap B'_{12} = A_2'\cap B'_{12} = \set{v=-1} \subset B'_{12}.
\end{equation*}
This is resolved by an additional blowup $X\rightarrow X'$ of the fibre in the three points $A_1' \cap B'_{ij}$. Let $A_{ij}$ denote the corresponding exceptional divisors, and write $A_i$, $B_i$, $B_{ij}$ for the strict transforms of $A_i'$, $B_i'$, and $B_{ij}'$. We set
\begin{equation*}
    A=A_1\cup A_2\cup A_{12}\cup A_{13}\cup A_{23}
\end{equation*}
and analyse the normal crossings divisor $A\cup B \subset Y=X\times T$ with respect to the projection $\pi \colon Y \to T$. This is illustrated in \cref{fig:massive-sunrise-divs}. Note in particular that, in a generic fibre, $A_1$ and $A_2$ do not intersect.
\begin{lem}
    Let $\ell_p=\set{p^2=0}$, $\ell_i=\set{m_i^2=0}$ for $i\in\set{1,2,3}$, and for any pair of signs $\alpha,\beta\in\set{1,-1}$ set $\ell_{\alpha\beta}=\set{p^2=(m_1+\alpha m_2+\beta m_3)^2}$.\footnote{$\ell_{++}$ is the \textit{leading threshold}, whereas $\ell_{+-}$, $\ell_{-+}$, $\ell_{--}$ are called \textit{pseudo-thresholds}.}
    
    The Landau variety of the sunrise graph is
    \begin{equation*}
        L_G = \ell_p \cup \ell_1\cup \ell_2\cup \ell_3 \cup \ell_{++} \cup \ell_{+-} \cup \ell_{-+} \cup \ell_{--}.
    \end{equation*}
\end{lem}
\begin{proof}
    The Landau variety contains in particular the critical values of the cubic $A_2$. Those are determined by the parameters $t$ such that the fibre $A_{2t}=A_2\cap \pi^{-1}(t)$ is singular. This discriminant of $\Fpol$ is
    \begin{equation*}
        p^2 m_1^2 m_2^2 m_3^2 \prod_{\alpha,\beta=\pm} \big(p^2-(m_1 +\alpha m_2 +\beta m_3)^2 \big)
    \end{equation*}
    and produces $L_G$. It remains to verify that the intersections of $A_2$ with (one or several) of the other components $A_1,A_{ij},B_i,B_{ij}$ do not generate additional components. This is a routine calculation.
\end{proof}

In order to describe the variations, we need to determine the critical sets $P\subset\pi^{-1}(t_c)$ over a generic critical value $t_c \in L$. These are as follows:
\begin{itemize}
    \item Over $\ell_{\alpha\beta}$, $P$ is a quadratic simple pinch of type $(\set{A_2},\varnothing)$. This point is the unique node (double point) of the cubic $A_{2t}$. In coordinates:
    \begin{equation*}
        P = \set{[\alpha\beta m_2m_3:\beta m_1m_3:\alpha m_1m_2]}.
    \end{equation*}
    \item Over $\ell_i$, the cubic $A_{2t}$ acquires two nodes and splits into two $\PP^1$. For example, in the chart $x=[u:uv:1]$ with $B_{12}=\set{u=0}$,
    \begin{equation*}
        A_{2t} \cap B_{12} = \set{(1+v) m_3^2 = 0}
    \end{equation*}
    is generically a point; but for $t_c\in \ell_3$, it contains the entire hyperplane $B_{12}\subset A_{2t_c}$. The other factor of $A_{2t_c}$ (solve $\Fpol$ for $x_3$) is the image of
    \begin{equation*}
        \PP^1\ni[y_1:y_2] \mapsto \left[ y_1:y_2:
        \tfrac{ y_1 y_2(m_1^2 y_1+m_2^2 y_2)}{p^2 y_1 y_2-(y_1+y_2)(m_1^2y_1+m_2^2y_2)}
        \right].
    \end{equation*}
    The two components of $A_{2t_c}$ meet in the two roots $v=y_2/y_1$ of $p^2 v=(1+v)(m_1^2+m_2^2v)$. These double points are quadratic simple pinches of $A_2$; however, since the factor $B_{12}$ is in the arrangement, the critical set is not isolated:
    \begin{equation*}
        P=B_{jk}\cap \pi^{-1}(t_c)
        \quad\text{with type}\quad (\set{A_1,A_2,A_{jk}},\set{B_j,B_k,B_{jk}}).
    \end{equation*}
    \item Over $t_c\in\ell_p$, the cubic also degenerates into two components, $A_{2t_c}=(A_1\cup\set{x_1m_1^2+x_2m_2^2+x_3m_3^2=0})_{t_c}$. These components intersect in two quadratic simple pinches of $A_2$, and the critical set is
    \begin{equation*}
        P = A_1\cap\pi^{-1}(t_c)
        \quad\text{with type}\quad
        (\set{A_1,A_2,A_{12},A_{13},A_{23}},\varnothing).
    \end{equation*}
\end{itemize}

The simple type of every critical set is just $(\set{A_2},\varnothing)$, because the divisor $A_2\subset Y$ alone already produces the entire Landau variety. As for the massless triangle, we thus get no hierarchy constraints from \cref{def:rel_landau_components}, at least not between different variations. For the simple pinches, we do get that
\begin{equation}
    \Var_{\ell_{\alpha\beta}}\circ \Var_{\ell_{\alpha\beta}} = 0
\end{equation}
for all $\alpha,\beta\in\set{1,-1}$, from \eqref{eq:itvar_n-m_odd}.
Furthermore, with a very similar argument as in the proof of \cref{lem:tri0-hierarchy}, one can show that
\begin{equation}\label{eq:sunrise-vari-vari}
    \Var_{\ell_i}\circ\Var_{\ell_i} = 0
\end{equation}
for all $i\in\set{1,2,3}$. Note in particular that, by the localization near the critical set, the variations of any class $h\in H_2(Y\setm A,B)_t$ take values in
\begin{align*}
    \Var_{\ell_{\alpha\beta}} h, \Var_{\ell_p} h & \in \fibSphere_{A_2} H_1(A_2\setm (A_{12}\cup A_{13}\cup A_{23}),\varnothing)_t \quad\text{and} \\
    \Var_{\ell_i} h & \in \fibSphere_{A_2} H_1(A_2\setm (A_{12}\cup A_{13}\cup A_{23}),B_j\cup B_k)_t.
\end{align*}
So after any first variation, no cycle has boundary in $B_{12}\cup B_{13}\cup B_{23}$, so we can drop those from the arrangement. Thus the critical points over $\ell_i$ for the second variation are just the two nodes, hence subject to the Picard-Lefschetz theorem, and we can proceed as in \cref{lem:tri0-hierarchy} to confirm \eqref{eq:sunrise-vari-vari}.

\subsection{Ice cream cone}\label{sec:icecream}

\begin{figure}
  \begin{tikzpicture}
   \coordinate (v0) at (0,0);
   \fill[black] (v0) circle (.066cm);
   \coordinate  (v1) at (2,-1);
   \fill[black] (v1) circle (.066cm);
   \coordinate (v2) at (2,1);
   \fill[black] (v2) circle (.066cm); 
   \coordinate (p3) at (2.7,-1.35);
   \coordinate (p4) at (2.7,1.35);
   \coordinate (p1) at (-0.75,0);
   \draw[decoration={markings, mark=at position 0.625 with {\arrow{<}}},postaction={decorate}] (v1) -- (p3) node [xshift=0.3cm,yshift=-0.1cm] {$p_2$};
   \draw[decoration={markings, mark=at position 0.625 with {\arrow{<}}},postaction={decorate}] (v2) -- (p4) node [xshift=0.3cm,yshift=-0.1cm] {$p_3$};
   \draw[decoration={markings, mark=at position 0.625 with {\arrow{<}}},postaction={decorate}] (v0) -- (p1) node [xshift=-0.2cm,yshift=0.1cm] {$p_1$};
   \draw (v0) -- (v1) node [midway,below] {$4$};
   \draw (v0) -- (v2) node [midway,above] {$3$};
   \draw (v2) to[out=-135,in=135] (v1) node [left,xshift=0cm,yshift=1cm] {$1$};
   \draw (v2) to[out=-45,in=45] (v1) node [right,xshift=0.33cm,yshift=1cm] {$2$};  
  \end{tikzpicture}%
  %%%%%%%%%%
  \qquad 
  %%%%%%%%%%
   \raisebox{0.6cm}{\begin{tikzpicture}
   \coordinate (v1) at (0,0);
   \fill[black] (v1) circle (.066cm);
   \coordinate  (v2) at (0,2);
   \fill[black] (v2) circle (.066cm);
   \draw (v2) to[out=-135,in=135] (v1);
   \draw (v2) to[out=-45,in=45] (v1);  
  \end{tikzpicture}}%
  %%%%%%%%%%%%%
  \qquad
  %%%%%%%%%%%
   \raisebox{.6cm}{\begin{tikzpicture}
   \coordinate (v0) at (0,0);
   \fill[black] (v0) circle (.066cm);
   \coordinate  (v1) at (2,0);
   \fill[black] (v1) circle (.066cm);
   \coordinate (p3) at (2.5,-.6);
   \coordinate (p4) at (2.5,.6);
   \coordinate (p1) at (-0.75,0);
   \draw[decoration={markings, mark=at position 0.625 with {\arrow{<}}},postaction={decorate}] (v1) -- (p3) node [xshift=0.3cm,yshift=-0.1cm] {$p_2$};
   \draw[decoration={markings, mark=at position 0.625 with {\arrow{<}}},postaction={decorate}] (v1) -- (p4) node [xshift=0.3cm,yshift=-0.1cm] {$p_3$};
   \draw[decoration={markings, mark=at position 0.625 with {\arrow{<}}},postaction={decorate}] (v0) -- (p1) node [xshift=-0.2cm,yshift=0.1cm] {$p_1$};
   \draw (v0) to[out=70,in=110] (v1);
   \draw (v0) to[out=-70,in=-110] (v1);
  \end{tikzpicture}}%
\caption{The ice cream cone $G$, the subgraph $\gamma=\set{1,2}$, and the quotient graph $G/\gamma$.}\label{fig:icecream}%
 \end{figure}
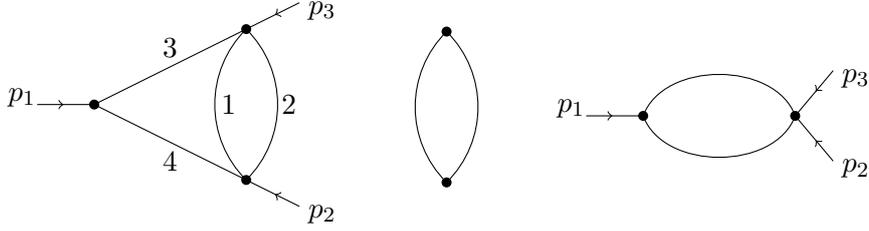

The graph $G$ with four edges in \cref{fig:icecream} depends on seven parameters $t=(p_1^2,p_2^2,p_3^2,m_1,m_2,m_3,m_4)\in T=\CC^7$. The Symanzik polynomials are
\begin{align*} 
  \Upol &= x_1x_2 + (x_1+x_2)(x_3+x_4) \quad\text{and} \\
  \Fpol &= x_1x_2(-p_2^2 x_4 - p_3^2 x_3) - p_1^2 (x_1 + x_2)x_3 x_4 + \Upol \sum_{1\leq i\leq 4} m_i^2x_i.
\end{align*}

The Feynman motive is obtained from the blowup $X'\rightarrow \PP^3$ in the two points $\set{[0:1:0:0],[1:0:0:0]}$ and the line $\set{x_1=x_2=0}$.\footnote{The 3 motic (bridgeless) subgraphs of $G$ have the edge sets $\set{1,2}$, $\set{1,3,4}$, $\set{2,3,4}$.} Its boundary divisor $B'$ has 7 components: exceptional divisors $B'_{134}$, $B'_{234}$, $B'_{12}$ and strict transforms $B'_i$ of $\set{x_i=0}$.
The strict transforms $A_1'$ and $A_2'$ of $\set{\Upol=0}$ and $\set{\Fpol=0}$ are smooth surfaces in $Y'=X'\times T$.
The Feynman polytope (\cref{fig:icecream-polytope}) has facets in all 7 components of $B'$ and defines a class
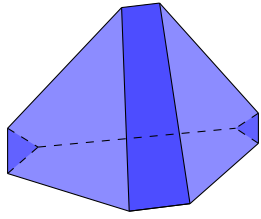
\begin{figure}
  \begin{tikzpicture}%
   \coordinate (v0) at (0,0);
   \coordinate (v01) at (0,-0.6);
   \coordinate (v02) at (0.4,-0.25);
   \coordinate  (v1) at (1.5,1.6); 
    \coordinate (v11) at (2,1.66);
   \coordinate (v2) at (1.6,-1.1);
   \coordinate (v21) at (2.4,-1);
   \coordinate (v3) at (3,0);
   \coordinate (v31) at (3.3,0.2);
   \coordinate (v32) at (3.3,-0.2);
    \fill[fill=blue!45] (v0) -- (v01) -- (v2) -- (v21) -- (v11) -- (v1) -- (v0);
   \fill[fill=blue!45] (v2) -- (v21) -- (v32) -- (v31) -- (v11);
   \fill[fill=blue!70,opacity=0.66] (v0) -- (v01) -- (v02) -- (v0);
   \fill[fill=blue!70,opacity=0.66] (v3) -- (v31) -- (v32) -- (v3);
   \draw[dashed] (v0) -- (v02);
   \draw[dashed] (v01) -- (v02);
   \draw[dashed] (v32) -- (v3);
   \draw[dashed] (v31) -- (v3);
   \fill[fill=blue!70] (v1) -- (v11) -- (v21) -- (v2) -- (v1);
   \draw (v0) -- (v01);
   \draw (v01) -- (v2);
   \draw (v2) -- (v21);
   \draw (v21) -- (v11);
   \draw (v0) -- (v1);
   \draw (v1) -- (v11);
   \draw (v11) -- (v31);
   \draw (v31) -- (v32);
   \draw (v32) -- (v21);
   \draw (v21) -- (v2); 
   \draw (v2)--(v1);
   \draw[dashed] (v02) -- (v3);
  \end{tikzpicture}%
\caption{The Feynman polytope $\ti \sigma$ of the ice cream cone.}\label{fig:icecream-polytope}%
\end{figure}
\begin{equation*}
    \ti{\sigma}\in H_3(Y'\setminus A',B')_t.
\end{equation*}

Although $B'$ is transverse and all components of $D'=A'\cup B'$ are smooth, the arrangement $D'$ is not transverse. In particular, $A_1'$ and $A_2'$ are not transverse to $B_{12}'$, because the intersections $A_1'\cap B_{12}'$ and $A_2'\cap B_{12}'$ are singular: Note that in coordinates $x=[u:uv:x_3:1]$, the divisor
\begin{equation}\label{eq:ice-B12p}
    B_{12}'=\set{u=0}=\set{([1:v],[x_3:1]}=\PP^1\times\PP^1
\end{equation}
is canonically identified with the product of the coordinates $[1:v]=[x_1:x_2]$ of the subgraph $\gamma=\set{1,2}$ and $[x_3:x_4]$ of the quotient graph $G/\gamma$.
The Symanzik polynomials vanish to first order on $B_{12}'$, with
\begin{equation}\label{eq:UF-factors}\begin{aligned}
    \Upol &\equiv \Upol_{\gamma}\Upol_{G/\gamma}\equiv (x_1+x_2)(x_3+x_4) \quad\text{and}\\
    \Fpol &\equiv \Upol_{\gamma}\Fpol_{G/\gamma}\equiv (x_1+x_2)( -p_1^2 x_3 x_4+(x_3+x_4)(m_3^2 x_3+m_4^2 x_4) )
\end{aligned}\end{equation}
modulo terms of order $u^2$ or higher. The factorization of the leading order in $u$ into Symanzik polynomials of $\gamma$ and $G/\gamma$ is well-known \cite{Brown:FeynmanAmplitudesGalois}. Therefore,
\begin{equation*}\begin{aligned}
    %A_1' \cap B_{12}' &= A_{1\gamma}'\times\PP^1 \cup \PP^1\times A_{1G/\gamma}' \quad\text{and}\\
    %A_2' \cap B_{12}' &= A_{1\gamma}'\times\PP^1 \cup \PP^1\times A_{2G/\gamma}'
    A_1' \cap B_{12}' &= \set{[1:-1]}\times\PP^1 \cup \PP^1\times \set{[1:-1]} \quad\text{and}\\
    A_2' \cap B_{12}' &= \set{[1:-1]}\times\PP^1 \cup \PP^1\times \set{\Fpol_{G/\gamma}=0}
\end{aligned}\end{equation*}
 are unions of lines (\cref{fig:icecream-B12}). In a generic fibre, they have three singularities, the intersection points $\set{[1:-1]}\times\set{[1:-1],Q_1,Q_2}$ where $Q_{1,2}\in \PP^1$ denote the roots of $\Fpol_{G/\gamma}=0$. Note also that $A_1'\cap A_2'\cap B_{12}'=\set{[1:-1]}\times\PP^1$ intersect not in a point but in a line---another witness of non-transversality.
\begin{figure}
    \centering
    \includegraphics[width=\textwidth]{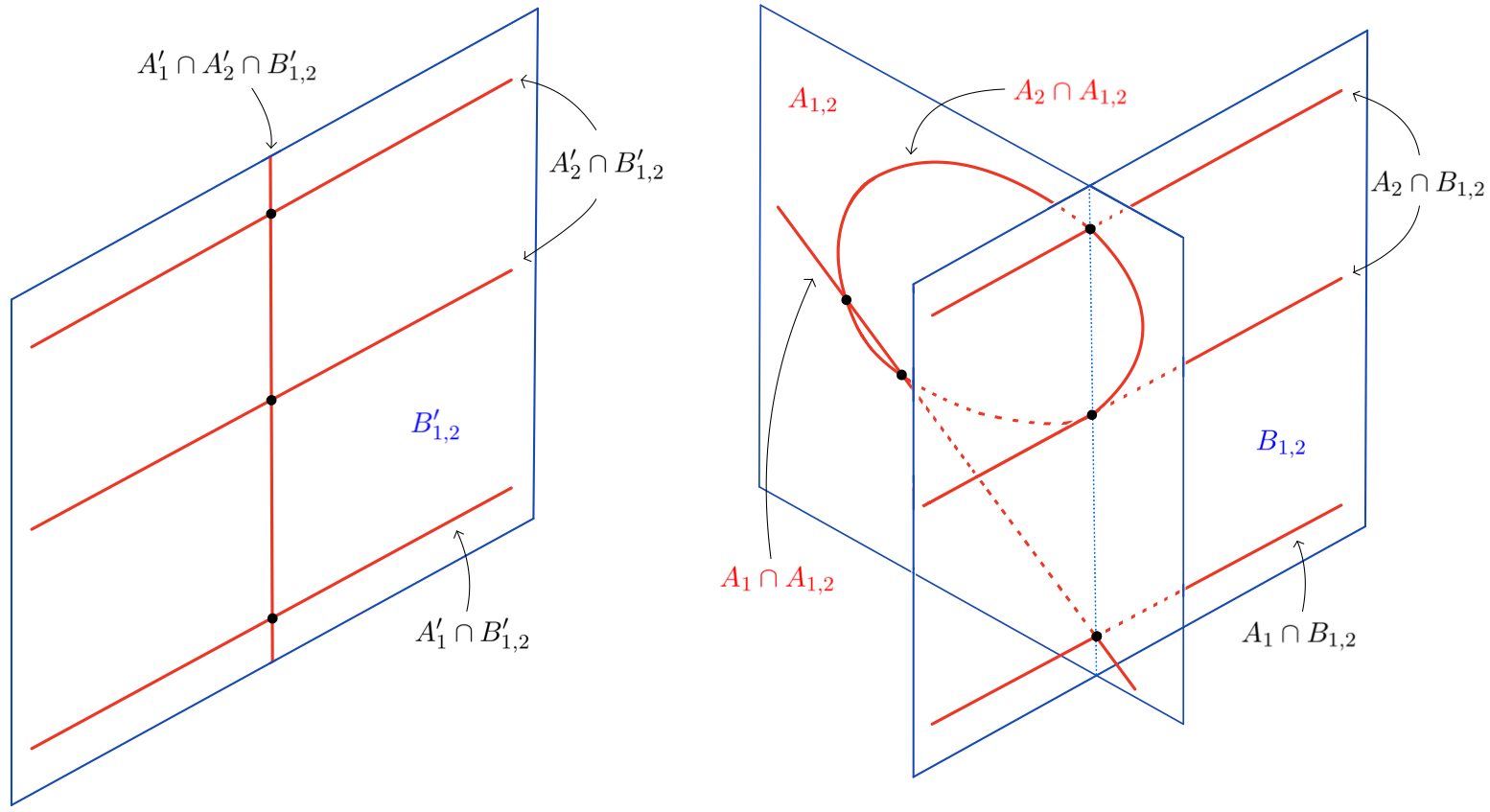}
    \caption{Singular intersection $A'\cap B_{12}'$ (left) and transverse arrangement $A\cup B_{12}$ (right) after blowup.}%
    \label{fig:icecream-B12}%
\end{figure}

Similarly, we identify two further non-transverse intersections. Instead of a point, $A_1'\cap A_2' \cap B_{134}'\cong \PP^1$ intersect in the line
\begin{equation*}
    \set{\Upol_{134}=0}=\set{x_1+x_3+x_4=0} \subset B_{134}'\cong\PP^2,
\end{equation*}
and analogously for $B_{234}'$. To rectify all these problems, we make
\begin{defn}
    We define the fibre $X$ of the ice cream cone as the blowup $X\rightarrow X'$ of the Feynman motive $X'$ in the three lines (each is $\cong\PP^1$)
    \begin{equation}\label{eq:U-lines}
        \set{\Upol_{12}=0}\subset B_{12}',\ \set{\Upol_{134}=0} \subset B_{134}', \ \set{\Upol_{234}=0} \subset B_{234}'.
    \end{equation}
    We denote the corresponding exceptional divisors as $A_{12}$, $A_{134}$, $A_{234}$ and otherwise write $A_{\bullet}$, $B_{\bullet}$ for the strict transforms of $A_{\bullet}'$ and $B_{\bullet}'$. We set $Y=X\times T$ and obtain a divisor $A\cup B\subset Y$.
\end{defn}
\begin{rem}
    The component $B_3'$ is the blowup of $\PP^2\cong \set{x_3=0}\subset \PP^3$ in the 3 corners of $\set{x_1x_2x_4=0}$. This is precisely the fibre denoted $X'$ in \cref{ssec:massivesunrise}, corresponding to the quotient sunrise graph $G/3$.
    Each of the 3 lines \eqref{eq:U-lines} intersects $B_3'$ in a point, so $B_3$ is the blowup of $B_3'$ in 3 further points, that is, $B_3$ is isomorphic to the fibre denoted $X$ in \cref{ssec:massivesunrise}. This means that the operadic structure of Feynman motives $X'$ from \cite{Brown:FeynmanAmplitudesGalois} extends to the spaces $X\rightarrow X'$ with our additional blowups.
\end{rem}
\begin{lem}\label{lem:ice-snc}
    The divisor $D=A\cup B\subset Y$ has simple normal crossings.
\end{lem}
\begin{proof}
    In coordinates $x=[u:uv:x_3:1]$ where $B_{12}'=\set{u=0}$, we are blowing up $\set{u=0}\cap\set{v=-1}$. Hence in the chart $x=[\lambda w:\lambda w(\lambda-1):x_3:1]$, we see $B_{12}=\set{w=0}$ and $A_{12}=\set{\lambda=0}$. The strict transform 
    \begin{equation*}
        A_1 = \set{1+x_3-w+w\lambda=0}
    \end{equation*}
    intersects $B_{12}$ in the horizontal line $\set{x_3=-1}$, whereas the previously vertical line has moved onto $A_{12}\cap A_1 = \set{w=1+x_3}$. In the same coordinates, the strict transform $A_2$ of $A_2'$ reads
    \begin{equation*}
        0=-p_1^2 x_3+(1+x_3-w)(m_3^2 x_3+m_4^2)+w(p_2^2+x_3 p_3^2) + w\lambda\cdot(\cdots).
    \end{equation*}
    We see that $A_2\cap B_{12}$ is smooth, namely the union of two non-intersecting lines (parametrized by $\lambda$) with $w=0$ and $x_3$ a root of $\Fpol_{G/\gamma}$. The other intersection is also smooth, namely the quadric\footnote{We note that the coefficient of $w$ is equal to $\Fpol_{34} = x_3(m_3^2-p_3^2)+x_4(m_4^2-p_2^2)$, the second Symanzik polynomial of the subgraph $\set{3,4}$.}
    \begin{equation}\label{eq:ice-A2-A12}
        A_{2t}\cap A_{12} = \set{\lambda=0} \cap \set{w(m_3^2 x_3+m_4^2-p_2^2-p_3^2 x_3)=\Fpol_{G/\gamma}}\cong\PP^1.
    \end{equation}
    This geometry is illustrated in \cref{fig:icecream-B12}. Similarly, one can verify the transversality of $D$ near $A_{134}$, $B_{134}$, $A_{234}$, and $B_{234}$.
\end{proof}
Since we blow up subsets of $A_2'$, these blowups do not change the complements $Y\setm A\cong Y'\setm A'$ and thus we have a canonical identification
\begin{equation*}
    H_3(Y\setm A,B)_t = H_3(Y'\setm A',B')_t.
\end{equation*}
Because $A\cup B$ has normal crossings, we can now apply our results and analyse the Landau variety and variation constraints using the canonical stratification of $A\cup B$.
Instead of a complete analysis, we discuss two particularly interesting cases, stemming from only 3 components of $D$:
\begin{equation*}
    A_2\cup A_{12} \cup B_{12}.
\end{equation*}

In physics language, the critical values of $A_2\cup A_{12}$ would be called \emph{leading singularities} (absence of boundary $B$) of \emph{first type} (absence of $A_1$). They constitute singularities of integrals $\int_{\sigma} \omega_t$ where $\sigma\in H_3(Y\setminus (A_{12}\cup A_2),\varnothing)_t$ has no boundary, and the integrand $\omega$ has only $\Fpol$ in the denominator (no $\Upol$). Note that $\Fpol$ vanishes on $A_{12}$, hence we cannot discard $A_{12}$ (the differential form $\omega$ can have poles on $A_{12}$).\footnote{To get \emph{all} leading singularities of first type, we have to consider $A_{12}\cup A_{134} \cup A_{234} \cup A_2$. However, we are only interested in the singularities near $\set{x_1=x_2=0}$.} Critical values of $B_{12}\cap A_2$ might be called \emph{reduced} singularities, associated with the quotient graph $G/\gamma$. Instead, we will use the terminology of type and simple type (\cref{sec:hierarchy-general}).

In the coordinates $(w,x_3)$ on $A_{12}$ from the proof of \cref{lem:ice-snc}, the singular points of the quadric \eqref{eq:ice-A2-A12} constitute the critical set
\begin{equation}\label{eq:c(A2+A12)}
    c(A_2\cap A_{12}) = \Big\{w=\frac{\partial_{x_3} \Fpol_{G/\gamma}}{m_3^2-p_3^2}\Big\}
    \cap \Big\{x_3=-\frac{m_4^2-p_2^2}{m_3^2-p_3^2}\Big\}
    \cap \Big\{\Fpol_{G/\gamma}=0\Big\}.
\end{equation}
This set has a non-empty fibre only when $x_3=-(m_4^2-p_2^2)/(m_3^2-p_3^2)$ is a root of $\Fpol_{G/\gamma}$ from \eqref{eq:UF-factors}. The corresponding discriminant defines a component $\ell_{A_{12}}\subset L_G$ of the Landau variety, concretely
\begin{equation}\label{eq:ice-L(A12)}
    \ell_{A_{12}} = \big\{p_1^2(m_3^2-p_3^2)(m_4^2-p_2^2)+(m_3^2p_2^2-m_4^2p_3^2)(m_3^2-m_4^2+p_2^2-p_3^2)=0\big\}.
\end{equation}
Over a generic point $t$ of $\ell_{A_{12}}$, the critical point \eqref{eq:c(A2+A12)} is not a critical point for $A_2$, that is, $A_{2t}$ remains smooth over $\ell_{A_{12}}$. This means that the type and simple type of $\ell_{A_{12}}$ are the same: $(\set{A_{12},A_2},\varnothing)$. More precisely, \eqref{eq:c(A2+A12)} is a quadratic simple pinch, where $A_2\cap A_{12}$ degenerates into two lines,
\begin{equation}\label{eq:ice-A212-sing-fibre}
    A_{2t}\cap A_{12} = \Big\{x_3=-\frac{m_4^2-p_2^2}{m_3^2-p_3^2}\Big\}
    \cup \Big\{w= \frac{m_3^2 x_3}{m_3^2-p_3^2}+\frac{m_4^2}{m_4^2-p_2^2}\Big\}
    .
\end{equation}

Now consider the submanifold $A_2\cap B_{12} = \PP^1\times \set{\Fpol_{G/\gamma}=0}$. In a generic fibre, it consists of two lines (coordinate $\lambda=v+1$), which degenerate into a single line when the two roots of $\Fpol_{G/\gamma}$ coincide. Hence, the corresponding contribution $\ell_{B_{12}}\subset L_G$ to the Landau variety of the ice cream cone is
\begin{equation}\label{eq:ice-bubble-sing}
    \ell_{B_{12}} = \ell^{\Fpol}_{G/\gamma} = \set{p_1^2=(m_3\pm m_4)^2}
\end{equation}
in the notation from \cref{ssec:oneloop}, the leading singularity of the bubble graph $G/\gamma$ (\cref{eg:bubble}). The critical set over $t\in\ell_{B_{12}}$ is not isolated:
\begin{equation*}
    c(A_{2} \cap B_{12}) \cap \pi^{-1}(t) = \PP^1\times\big\{[x_3:x_4]=[m_4:\pm m_3]\big\}.
\end{equation*}
This line intersects also $B_1=\set{\lambda=\infty}$, $B_2=\set{\lambda=1}$ and $A_{12}=\set{\lambda=0}$. The component $B_{12}$ is \emph{not} simple over $t\in\ell_{B_{12}}$, because two points\footnote{$P_{1,2}$ are the roots of $v^2 m_2^2+v(m_1^2+m_2^2+m_3m_4)+m_1^2=v(m_3 p_2^2\pm m_4 p_3^2)/(m_3\pm m_4)$}
\begin{equation}\label{eq:icecream-c(A2)}
    c(A_2) \cap \pi^{-1}(t)=\set{P_1,P_2} \times \set{[m_4:\pm m_3]} \subset c(A_2\cap B_{12})
\end{equation}
on this line are critical points of $A_2$ alone.
This fact, that $B_{12}$ is not simple for the singularity $\ell_{B_{12}}$, was noted already in \cite[\S 4]{Boyling:HomParaFeyn} and \cite[\S 2]{LandshoffOlivePolkinghorne:Hierarchical}.
\begin{cor}
    The bubble singularity \eqref{eq:ice-bubble-sing} of the ice cream cone graph has type $(\set{A_2,A_{12}},\set{B_1,B_2,B_{12}})$ and simple type $(\set{A_2},\varnothing)$. The type and simple type of the singularity \eqref{eq:ice-L(A12)} is $(\set{A_2,A_{12}},\varnothing)$.
\end{cor}

Let $\irrone{L}_G$ denote the codimension one components of the full Landau variety $L_G$ of the ice cream cone $G$, and $\irrone{L}_{G/\gamma}$ correspondingly for the bubble quotient. Note that $\irrone{L}_{G/\gamma}\subset \irrone{L}_G$.
If $B_{12}$ were simple for $\ell_{B_{12}}$, then \cref{sec:hierarchy} would give
\begin{equation}\label{eq:ice-B12-hierarchy}
    \Var_{\ell_{B_{12}}} \circ \Var_{\ell} = 0
    \quad\text{for all}\quad
    \ell\in\irrone{L}_G\setm \irrone{L}_{G/\gamma}.
\end{equation}
Namely, a compatible variation requires that $\ell$ has $B_{12}$ in its type, meaning that the critical set $P$ over $\ell$ intersects $B_{12}$, and thus $P\cap B_{12}$ are critical points of $B_{12}$---showing that the restriction of the arrangement $A\cup B$ to $B_{12}$ is critical over $\ell$ (hence $\ell\in\irrone{L}_{G/\gamma}$). Anyway, $B_{12}$ is \emph{not} simple, and indeed, \eqref{eq:ice-B12-hierarchy} does not hold \cite[eq.~(14)]{LandshoffOlivePolkinghorne:Hierarchical}.

However, we \emph{do} have hierarchy constraints in the $\irrone{A}$-direction:
\begin{equation}\label{eq:ice-A12-hierarchy}
    \Var_{\ell} \circ \Var_{\ell_{A_{12}}} = 0,
\end{equation}
unless $\ell$ has $A_{12}$ and $A_2$ in its type. In fact, since $\ell_{A_{12}}$ is a simple pinch component, its vanishing cycle is an \emph{iterated} coboundary $\vcyc=\fibSphere_{A_2}\fibSphere_{A_{12}} \eta$ around the vanishing sphere
\begin{equation}\label{eq:ice-VarA12-cell}
    \eta=\partial_{A_{12}} \partial_{A_2} \vcell=[\Sphere^1]\in H_1(A_2\cap A_{12}\setm A_1)_t\cong \ZZ.
\end{equation}
This group has rank one, because $A_{2t}\cap A_{12} \cong \PP^1$ from \eqref{eq:ice-A2-A12} is punctured in two points $A_{2t}\cap A_{12} \cap A_1=\set{R_1,R_2}$. The vanishing sphere $\eta$ separates the two punctures, because over the critical fibre, the line $A_1\cap A_{12}$ intersects each of the two components \eqref{eq:ice-A212-sing-fibre} once. Therefore, $\eta=\fibSphere_{R_1} 1=-\fibSphere_{R_2} 1$ is a generator of \eqref{eq:ice-VarA12-cell}.

By \eqref{eq:itvar_n-m_odd}, $\Var_{\ell_{A_{12}}} \vcyc=0$. The only other singularity of $A_2\cap A_{12}\setm A_1$ is the quadratic simple pinch of $A_2\cap A_{12}\cap A_1$, when $R_1=R_2$. This happens over a singularity $\ell_{\delta}\in\irrone{L}_G$ that is familiar from the triangle graph:
\begin{equation}
    \ell_{\delta} = \big\{p_1^4+p_2^4+p_3^4-2p_1^2p_2^2-2p_1^2p_3^2-2p_2^2p_3^2=0\big\} \subset L_G.
\end{equation}
\begin{cor}
    For all $\ell\in\irrone{L}_G\setm\set{\ell_{\delta}}$, we have $\Var_{\ell}\circ\Var_{\ell_{A_{12}}}=0$.
\end{cor}

In fact, the singularity $\ell_{\delta}$ gives another example of non-simple types. Traditionally, this singularity is thought of as ``second type'', that is, associated with the critical points of $A_1\cap A_2$. And indeed, this intersection is critical over $\ell_{\delta}$: if we parametrize the smooth quadric $A_1\cong\PP^1\times\PP^1\times T$ as
\begin{equation*}
([y],[z]) \mapsto  [y_1(y_1+y_2)(z_1+z_2):y_2(y_1+y_2)(z_1+z_2):-y_1y_2z_1:-y_1y_2z_2],
\end{equation*}
then $A_1\cap A_{2t}=\PP^1\times\set{z_1^2p_3^2-z_1z_2(p_1^2-p_2^2-p_3^2)+z_2^2p_2^2=0}$ is generically a union of two non-intersecting lines. Over $\ell_{\delta}$, the two solutions for $z$ coincide and the lines merge into one, $c(A_1\cap A_2)\cap \pi^{-1}(t)\cong\PP^1$. As for $B_{12}$ above, there are two points on this line that are critical points of $A_2$ alone. Hence the simple type of $\ell_{\delta}$ is only $(\set{A_2},\varnothing)$, like for a ``first type'' singularity.
\begin{rem}
    The vanishing sphere of $A_2\cap A_{12}\cap A_1$ over $\ell_{\delta}$ is $\set{R_1,R_2}\cong\Sphere^0$, hence the corresponding vanishing cycle is $2\vcyc$. It follows that $\Var_{\ell_{\delta}} \vcyc=-\vcyc$ and $\Var_{\ell_{\delta}}\circ\Var_{\ell_{A_{12}}}=-\Var_{\ell_{A_{12}}}$. If $A_1$ were simple for $\ell_{\delta}$, then the hierarchy would in contradiction predict that this double variation is zero.
\end{rem}

To summarize: the additional blowup $X\rightarrow X'$ not only makes the divisor $A\cup B$ normal crossing, but the new divisor $A_{12}$ also provides a new singularity \eqref{eq:ice-L(A12)} with hierarchy constraints \eqref{eq:ice-A12-hierarchy}.

\begin{rem}\label{rem:mod_motive} 
    As already noted in \cite[\S 4]{Boyling:HomParaFeyn}, the non-transversality of $A'\cup B'$ is a completely general phenomenon for Feynman motives. As soon as a graph $G$ has a subgraph $\gamma$ such that $\gamma$ and $G/\gamma$ have each at least one loop and at least two edges, the factorizations \eqref{eq:UF-factors} lead to singular intersections of $A_1'$ and $A_2'$ with $B_{\gamma}'$.
    
    Our calculations suggests that an understanding of hierarchy constraints for variations of Feynman graphs with more than one loop might be obtained by a study of modified Feynman motives $X\rightarrow X'$, where $X$ is the blowup of $X'$ in the divisors $\set{\Upol_{\gamma}=0}\subset B_{\gamma}'$.
\end{rem}

\bibliographystyle{JHEPsortdoi}
\bibliography{refs}

\appendix

\section{Homology of simple pinches}\label{sec:homology-groups}

Consider a simple pinch on a stratum $S_1\cap \ldots \cap S_{\PLm}$ in local coordinates on a complex open ball $W=\oBall_{n}=\set{\abs{x_1}^2+\ldots+\abs{x_n}^2<1} \subset \CC^n$, as in eq.~\eqref{eq:localcoords}:
\begin{align*}
    S_{i} & = \{ x_i =0 \}, \quad i=1, \ldots, \PLm-1 , \\ 
    S_{\PLm} &= \{ x_1 + \ldots + x_{\PLm-1} + x_{\PLm}^2 + \ldots + x_n^2=t\} .
\end{align*}
For any subset $I\subseteq\set{1,\ldots,\PLm}$, we denote union and intersection as
\begin{equation*}
    S_I = \bigcup_{i \in I} S_i
    \quad\text{and}\quad
    S^I = \bigcap_{i \in I} S_i.
\end{equation*}
In this appendix we generalize results of \cite{FFLP} and calculate the homology groups
\begin{equation*}
    H_{\bullet}(W\cap S^I \setm S_J, S_K)_t
    \quad\text{and}\quad
    H_{\bullet}(\overline{W}\cap S^I \setm S_J, \partial W \cup S_K)_t
\end{equation*}
for all distributions $I\sqcup J\sqcup K \subseteq\set{1,\ldots,\PLm}$ of the hypersurfaces. As we only compute the groups in a fixed fibre, we suppress the subscript $_t$ from now on.

We begin with a review of the well-known case $J=K=\varnothing$, and then the rest follows from the boundary and residue sequences of \cref{sec:relative-residues}.
\begin{lem}\label{lem:intersect-contractible}
    If $\abs{t}<1$ and $\abs{I}<\PLm$, then $W\cap S^I$ is contractible.
\end{lem}
\begin{proof}
    Coordinate hyperplanes $i\in I\setm\set{\PLm}$ just lower the ambient dimension $\oBall_n\cap S_i\cong \oBall_{n-1}$, hence the claim reduces to $I\subseteq\set{\PLm}$. The ball $I=\varnothing$ is contractible, so only $I=\set{\PLm}$ remains.
    We have $\PLm\geq 2$ because $I$ is a proper subset. Set $\bar{x}=(x_1+\ldots+x_{\PLm-1})/(\PLm-1)$ and apply the linear homotopy
    \begin{equation*}
        H_{\lambda}\colon x=(x_1,\ldots,x_n) \mapsto (1-\lambda) x+\lambda(\underbrace{\bar{x},\ldots,\bar{x}}_{\text{$\PLm-1$ times}}, x_{\PLm},\ldots,x_n)
        ,\ \lambda \in [0,1].
    \end{equation*}
    This preserves $S_{\PLm}$, $H_{\lambda}(S_{\PLm})\subseteq S_{\PLm}$, and furthermore stays within the unit ball, $H_{\lambda}(\oBall_n)\subseteq \oBall_n$ for all $\lambda\in[0,1]$, because (Cauchy-Schwarz)
    \begin{equation*}
        \norm{H_{\lambda}(x)}^2
        =\norm{x}^2-\lambda(2-\lambda)\left(
            %\abs{x_1}^2+\ldots+\abs{x_{\PLm-1}}^2-\frac{\abs{x_1+\ldots+x_{\PLm-1}}^2}{k-1}
            %\sum_{i=1}^{k-1} \abs{x_i}^2-\frac{1}{\PLm-1}\left|\sum_{i=1}^{\PLm-1}x_i \right|^2
            \sum_{i=1}^{\PLm-1} \abs{x_i}^2-(\PLm-1)\abs{\bar{x}}^2
        \right)
        \leq \norm{x}^2.
    \end{equation*}
    We obtain a strict deformation retraction $H_1$ from $\oBall_n\cap S_{\PLm}$ onto the subspace where $x_1=\ldots=x_{\PLm-1}=\bar{x}$ are equal. In this subspace of codimension $\PLm-2$, we have $\norm{x}^2=\abs{\bar{x}\sqrt{\PLm-1}}^2+\abs{x_{\PLm}}^2+\ldots+\abs{x_n}^2$.
    This gives an isometry $H_1(\oBall_n\cap S_{\PLm}) \cong \oBall_{n-\PLm+2} \cap S'$ where
    \begin{equation*}
        S'=\set{x_1 \sqrt{\PLm-1} + x_2^2+\ldots+x_{n-\PLm+2}^2=t}.
    \end{equation*}
    For the linear pinch case $\PLm=n+1$, this is just a point $x_1=t/\sqrt{n}$ and we are done. In the quadratic case $\PLm\leq n$ let now $x=(x_1,\ldots,x_{n-\PLm+2})$ and apply a further homotopy
    \begin{equation*}
        F_{\lambda}\colon x \mapsto \left( 
        (1-\lambda) x_1+\lambda \frac{t}{\sqrt{\PLm-1}},
        x_2\sqrt{1-\lambda},\ldots, x_{n-\PLm+2}\sqrt{1-\lambda}
        \right)
    \end{equation*}
   which preserves $S'$, $F_{\lambda}(S')\subseteq S'$. 
   %For $x \in S'$ w
   A calculation shows%We find that
    \begin{equation*}
        \norm{F_{\lambda}(x)}^2 =
        %\lambda^2\norm{x}^2+(1-\lambda^2) \frac{\abs{t}^2}{\PLm-1} - \lambda^2(1-\lambda^2)\left|x_1-\frac{t}{\sqrt{\PLm-1}}\right|^2
        (1-\lambda)\norm{x}^2+\lambda \frac{\abs{t}^2}{\PLm-1} - \lambda(1-\lambda)\left|x_1-\frac{t}{\sqrt{\PLm-1}}\right|^2
    \end{equation*}
    and therefore $\norm{F_{\lambda}(x)}^2\leq (1-\lambda)\norm{x}^2+\lambda\abs{t}^2<1$, provided that $\abs{t}<1$ and $\norm{x}<1$. In conclusion, $F$ shows that the point $F_1(x)=(t/\sqrt{\PLm-1},0,\ldots,0)$ is a deformation retract of $\oBall_{n-\PLm+2} \cap S'$.
\end{proof}

\begin{lem}\label{lem:vanishing-spheres}
    Let $I=\set{1,\ldots,\PLm}$ and $0<\abs{t}<1$. If $\PLm=n+1$ (linear pinch), then $W\cap S^I=\varnothing$ is empty. If $\PLm\leq n$ (quadratic pinch), then $W\cap S^I$ is homotopy equivalent to a real sphere of dimension $n-\PLm$.
\end{lem}
\begin{proof}
    The case $\PLm=n+1$ is clear from \eqref{eq:localcoords}. For $\PLm\leq n$, we get $W\cap S^I \cong \oBall_{n-\PLm+1}\cap S'$ with the complex sphere
    \begin{equation*}
        S' = \set{x_1^2+\ldots+x_{n-\PLm+1}^2=t} \subset \CC^{n-\PLm+1}.
    \end{equation*}
    After a unitary transformation $x\mapsto x \sqrt{t/\abs{t}}$, we may assume that $t=\abs{t}>0$ is real. In terms of the real and imaginary parts $u,v\in\RR^{n-\PLm+1}$ of $x=u+\iu v$,
    \begin{equation*}
        S' = \set{\norm{u}^2=t+\norm{v}^2} \cap \set{u\cdot v=0}
    \end{equation*}
    is diffeomorphic to the tangent bundle $T\Sphere^{n-\PLm}$ of the real sphere via $x\mapsto(u/\norm{u},v)$ with inverse $(p,v)\mapsto p\sqrt{t+\norm{v}^2}+\iu v$. The contraction to the zero section can be realized by a homotopy $F\colon S'\times [0,1]\longrightarrow S'$ of the form
    \begin{equation*}
        F_{\lambda} \colon x=u+\iu v \mapsto
        u\sqrt{1-\lambda+\lambda t/\norm{u}^2} + \iu v \sqrt{1-\lambda}.
    \end{equation*}
    Due to $\norm{F_\lambda(x)}^2=(1-\lambda)\norm{x}^2+\lambda t<1$, this homotopy restricts to a contraction of $S' \cap \oBall_{n-\PLm+1}$ onto the real sphere $F_1(S'\cap \oBall_{n-\PLm+1})=\set{u\sqrt{t}/\norm{u}}$ with radius $\sqrt{t}$.
\end{proof}
\begin{lem}\label{lem:H(S^I-S_J.S_K)=0}
    For all disjoint triples of subsets $I\sqcup J \sqcup K \subset \set{1,\ldots,\PLm}$ that do not cover all of the $\PLm$ hypersurfaces, the homology groups are
    \begin{equation*}
        H_{\bullet}(W\cap S^I \setm S_J, S_K) \cong 0 \quad\text{whenever $K\neq \varnothing$.}
    \end{equation*}
\end{lem}
\begin{proof}
    Let $J=\varnothing$ and pick any $p\in K$. Consider the boundary in the submanifold $W\cap S^I\cap S_p$ inside the pair $(W\cap S^I,W\cap S^I \cap S_K)$. By induction over $\abs{K}$ and \cref{lem:intersect-contractible}, we may assume that $H_{\bullet}(W\cap S^{I},S_{K-p})$ is concentrated in degree zero. Hence the boundary sequence breaks into isomorphisms
    \begin{align*}
        H_{\bullet}(W\cap S^I,S_K) &\xrightarrow{\cong} H_{\bullet-1}(W\cap S^{I+p},S_{K-p}) \quad\text{for $\bullet\geq 2$ and}
        \\
        H_{\bullet}(W\cap S^I,S_K) &\cong 0 \quad\text{for $\bullet=0,1$,}
    \end{align*}
    noting for the latter that $H_0(W\cap S^{I+p},S_{K-p})\rightarrow H_0(W\cap S^I,S_{K-p})$ is an isomorphism. This concludes the induction step and hence the proof for $J=\varnothing$. For $J\neq\varnothing$, apply induction over $\abs{J}$, using the residue sequence
    \begin{multline*}
        \cdots \rightarrow H_{\bullet}(W\cap S^I\setm S_{J-p},S_K) \rightarrow H_{\bullet-2}(W\cap S^{I+p}\setm S_{J-p},S_K) \\ \rightarrow H_{\bullet-1}(W\cap S^I\setm S_J,S_K) \rightarrow H_{\bullet-1}(W\cap S^I\setm S_{J-p},S_K)\rightarrow \cdots
    \end{multline*}
    with respect to some $p\in J$: all terms with $J-p$ are zero by induction, so the groups with $J$ vanish as well, by exactness.
\end{proof}
\begin{cor}\label{lem:codomain-Knon0-full}
    Take any partition $I\sqcup J \sqcup K = \set{1,\ldots,\PLm}$ with $K\neq \varnothing$. Then all boundary and Leray coboundary maps $\partial_p,\fibSphere_p$ with $p\in I$ are isomorphisms:
    \begin{equation*}\xymatrix{
        H_{\bullet-1}(W\cap S^{I}\setm S_{J},S_K) \ar[d]^{\fibSphere_p}_{\cong}
        & H_{\bullet}(W\cap S^{I-p}\setm S_J,S_{K+p}) \ar[l]^{\partial_p}_{\cong}
        \\
        H_{\bullet}(W\cap S^{I-p}\setm S_{J+p},S_K) 
        &
    }\end{equation*}
    By iterations of these maps, all these groups (with $K\neq \varnothing$) are isomorphic and in fact free of rank one, concentrated in degree $n-\abs{I}$:
    \begin{equation*}
        \delta_J\partial_{I\sqcup J}\colon H_{n}(W,S_{1,\ldots,\PLm})\xrightarrow{\cong} H_{n-\abs{I}}(W\cap S^I\setm S_J,S_K) \cong \ZZ.
    \end{equation*}
\end{cor}
\begin{proof}
    Due to the vanishing of $H_{\bullet}(W\cap S^{I-p}\setm S_J,S_K)$ from \cref{lem:H(S^I-S_J.S_K)=0}, the boundary and residue exact sequences with respect to $S_p$ break into the isomorphisms $\partial_p$ and $\fibSphere_p$ as claimed. It remains to identify the group for a convenient choice. Pick $I=\set{1,\ldots,\PLm-1}$, $J=\varnothing$, and $K=\set{\PLm}$. For a linear pinch ($\PLm=n+1$), this yields a single point $W\cap S^I=\set{0}$ (the origin) relative to $W\cap S^I \cap S_K=\varnothing$, hence
    \begin{equation*}
        H_{\bullet}(W\cap S^{1,\ldots,n},S_{n+1}) = H_{\bullet}(\set{0}) \cong \ZZ%^{\delta_{\bullet,0}}
    \end{equation*}
    is indeed concentrated in degree $\bullet=n-\abs{I}=0$. For a quadratic pinch ($\PLm\leq n$), the intersection $W\cap S^I\cap S_K=W\cap S^{1,\ldots,\PLm}$ is non-empty and contributes an augmentation $H_0(W\cap S^{I+\PLm})\twoheadrightarrow H_0(W\cap S^I)\cong \ZZ$ in the boundary sequence for $S_{\PLm}$. Hence we get an identification
    \begin{equation*}
        \partial_{\PLm}\colon H_{\bullet}(W\cap S^{1,\ldots,\PLm-1},S_\PLm) \xrightarrow{\cong} \widetilde{H}_{\bullet-1}(W\cap S^{1,\ldots,\PLm}) \cong \ZZ%^{\delta_{\bullet-1,n-k}}
    \end{equation*}
    with the reduced homology of the vanishing sphere from \cref{lem:vanishing-spheres}, concentrated in degree $\bullet-1=n-\PLm$.
\end{proof}
\begin{lem}\label{lem:H(W-SJ)-decomposition}
    Take any disjoint pair $I\sqcup J \subset \set{1,\ldots,\PLm}$ of less than $\PLm$ indices. Then the inclusion induces an isomorphism $H_{\bullet}(W\cap S^I\setm S_J)\cong H_{\bullet}(W\setm S_J)$.
    Furthermore, the sum $\sum_L \delta_L$ of the iterated Leray coboundaries
    \begin{equation*}
        \fibSphere_L\colon \ZZ \cong H_{0}(W\cap S^{L}\setm S_{J-L})\rightarrow H_{\abs{L}}(W\setm S_{J})
    \end{equation*}
    over indices $L\subseteq J$, provides canonical isomorphisms
    \begin{equation}\label{eq:H(W-SJ)-decomposition}
        \ZZ^{\binom{\abs{J}}{d}}\cong
        \bigoplus_{L\subseteq J, \abs{L}=d} H_{0}(W\cap S^{L}\setm S_{J-L}) \xrightarrow[\sum_L \fibSphere_L]{\cong} H_{d}(W\setm S_J).
    \end{equation}
\end{lem}
\begin{proof}
    For any $p\in I$, the boundary in $S_p$ sequence breaks into isomorphisms
    \begin{equation*}
        \cdots\rightarrow 0 \rightarrow H_{\bullet}(W\cap S^{I}\setm S_J)
        \xrightarrow{\cong} H_{\bullet}(W\cap S^{I-p}\setm S_J) \rightarrow 0 \rightarrow \cdots
    \end{equation*}
    due to $H_{\bullet}(W\cap S^{I-p}\setm S_J,S_p)=0$ from \cref{lem:H(S^I-S_J.S_K)=0}. 
    To show that \eqref{eq:H(W-SJ)-decomposition} are isomorphisms, we apply induction over $\abs{J}$. Since $J=\varnothing$ forces $L=\varnothing$, these maps are zero in degree $d\neq 0$ in agreement with $H_d(W)=0$, and in degree $d=0$ we get the identity map on $\ZZ\cong H_0(W)$. For the induction step, let $p\in J\neq\varnothing$ and consider the residue sequence with respect to $S_p$. Its maps
    \begin{equation*}
        \varpi\colon H_{d}(W\setm S_{J-p}) \longrightarrow H_{d-2}(W\cap S_p\setm S_{J-p})
    \end{equation*}
    are zero, because every generator $\delta_L 1_L \in H_d(W\setm S_{J-p})$, from $L\subseteq J-p$ with $1_L\in H_0(W\cap S^L\setm S_{J-p-L})\cong \ZZ$, maps to $\varpi \delta_L 1_L=\delta_L \varpi 1_L=0$ due to $\varpi 1_L\in H_{-2}(W\cap S^{L+p}\setm S_{J-p-L})=0$. Therefore, the residue sequence breaks into short exact sequences. We obtain a commutative diagram
    \begin{equation*}\xymatrix@C-5pt{
        0\ar[r] & H_{d-1}(W\cap S_p\setm S_{J-p}) \ar[r]^(0.6){\fibSphere_p} & H_{d}(W\setm S_J) \ar[r] & H_d(W\setm S_{J-p}) \ar[r] & 0 \\
        0 \ar[r] &
        \ZZ^{\binom{J-p}{d-1}}
        \ar[u]^{\cong} 
        \ar[r]^{\alpha}
        &
        \ZZ^{\binom{J}{d}}
        \ar[u]
        \ar[r]^{\beta}
        &
        \ZZ^{\binom{J-p}{d}}
        \ar[u]^{\cong} \ar[r]
        & 0
    }\end{equation*}
    with \eqref{eq:H(W-SJ)-decomposition} on the vertical arrows, $\alpha(1_L)=1_{L+p}$, $\beta(1_L)=1_L$ for all $L\subseteq J-p$ and $\beta(1_L)=0$ otherwise ($p\in L$). The bottom row is split exact and thus the five lemma concludes the induction step.
\end{proof}
\begin{prop}
    Take any partition $I\sqcup J = \set{1,\ldots,\PLm}$ of all $\PLm$ indices. Then the iterated Leray coboundaries induce canonical splittings
    \begin{equation}\label{eq:H(WSI-SJ)-decomposition}
        H_{d-\abs{J}}(W\cap S^{I+J}) \oplus \bigoplus_{L\subset J,\abs{L}=d} H_{0}(W\cap S^{I+L}\setm S_{J-L}) \xrightarrow{\cong} H_{d}(W\cap S^I\setm S_J).
    \end{equation}
\end{prop}
\begin{proof}
    The claim is trivial for $J=\varnothing$, and generalizes to $\abs{J}>0$ by induction using the residue sequence.
    The same reasoning from the proof of \cref{lem:H(W-SJ)-decomposition} applies, except for the reduction to $I=\varnothing$ (which no longer holds). In particular, we still have $\varpi=0$, since its domain $H_{\bullet}(W\cap S^I \setm S_{J-p})$ involves less hypersurfaces and is thus generated by coboundaries \eqref{eq:H(W-SJ)-decomposition} of points.
\end{proof}
This decomposition implies the following two corollaries, which describe $H_{n-\abs{I}}(W\cap S^I\setm S_J)$; c.f.\ the \emph{marginal cases} in \cite[Footnote~10 on page~93]{Pham:Singularities}.
\begin{cor}\label{cor:Kzero-isos}
    For every bipartition $I\sqcup J=\set{1,\ldots,\PLm}$ of a quadratic pinch ($\PLm\leq n$), the iterated Leray coboundary gives an isomorphism
    \begin{equation*}
        \fibSphere_J\colon
        H_{n-\PLm}(W\cap S^{1,\ldots,\PLm})\xrightarrow{\cong}
        H_{n-\abs{I}}(W\cap S^I\setm S_J)
        .
    \end{equation*}
\end{cor}
\begin{proof}
    The classes $\delta_L 1_L$ from points $1_L\in H_0(W\cap S^{I+L}\setm S_{J-L})\cong\ZZ$ in \eqref{eq:H(WSI-SJ)-decomposition} contribute only to degrees $\abs{L}<\abs{J}=\PLm-\abs{I}\leq n-\abs{I}$.
\end{proof}
\begin{cor}
    For a linear pinch ($\PLm=n+1$) and $I\sqcup J=\set{1,\ldots,n+1}$, the $\abs{J}$ iterated Leray coboundaries $\fibSphere_{J-p}\colon \ZZ\cong H_0(S^{I+J-p})\rightarrow H_{n-\abs{I}}(W\cap S^I\setm S_J)$ provide a canonical isomorphism
    \begin{equation*}
        \ZZ^{\abs{J}} \cong \bigoplus_{p\in J} H_0(S^{I+J-p})\xrightarrow[\sum\limits_{p\in J}\fibSphere_{J-p}]{\cong} H_{n-\abs{I}}(W\cap S^I\setm S_J).
    \end{equation*}
\end{cor}
\begin{proof}
    The set $S^{I+J}=\varnothing$ is empty (\cref{lem:vanishing-spheres}), thus the first contribution in \eqref{eq:H(WSI-SJ)-decomposition} is absent and only $L\subset J$ with $\abs{L}=n-\abs{I}=\abs{J}-1$ remain. These $n$-fold intersections $W\cap S^{I+J-p}\setm S_p = S^{I+J-p}$ are points: precisely those corners of the simplex which lie opposite to the facets $S_p$ with $p\in J$.
\end{proof}

\subsection{Duality} \label{ss:homgroups-duality}
Since $W\cap S^I$ is contractible for $\abs{I}<\PLm$ (\cref{lem:intersect-contractible}), Poincar\'{e}-Lefschetz duality implies that
\begin{equation*}
    H_{\bullet}(\overline{W}\cap S^I, \partial W)
    \cong H^{2n-2\abs{I}-\bullet}(W\cap S^I)
    \cong \ZZ
\end{equation*}
is concentrated in degree $\bullet=2n-2\abs{I}$, spanned by the fundamental class of the manifold $\overline{W}\cap S^I$ with boundary $\partial W\cap S^I$. For a quadratic pinch with its retraction to the vanishing sphere, the duality shows that
\begin{align*}
    H_{\bullet}(\overline{W}\cap S^{1,\ldots,\PLm}, \partial W) 
    &\cong H^{2n-2\PLm-\bullet}(W\cap S^{1,\ldots,\PLm}) \\
    &\cong H^{2n-2\PLm-\bullet}(\Sphere^{n-\PLm})
    \cong H_{\bullet-n+\PLm}(\Sphere^{n-\PLm})
\end{align*}
is concentrated in degrees $n-\PLm$ and $2n-2\PLm$.

From this starting point, all proofs above formally dualize upon swapping boundaries $\partial$ and coboundaries $\fibSphere$, swapping $J$ with $K$, replacing the pairs $(W\cap S^I,\varnothing)$ with $(\overline{W}\cap S^I,\partial W)$, and shifting degrees $\bullet$ to $2n-2\abs{I}-\bullet$. In particular, one finds:
\begin{enumerate}
    \item For all strict subsets $I\sqcup J\sqcup K\subset \set{1,\ldots,\PLm}$ with $J\neq\varnothing$,
    \begin{equation*}
        H_{\bullet}(\overline{W}\cap S^I\setm S_J,\partial W\cup S_K) = 0.
    \end{equation*}
    \item For all $I\sqcup J\sqcup K=\set{1,\ldots,\PLm}$ with $J\neq\varnothing$, the groups are concentrated in degree $n-\abs{I}$ and iterated (co)boundaries yield isomorphisms
    \begin{equation*}
        \fibSphere_{I\sqcup K} \partial_K\colon H_{n-\abs{I}}(\overline{W}\cap S^I\setm S_J,\partial W\cup S_K) \xrightarrow{\cong} H_n(\overline{W}\setm S_{1,\ldots,\PLm},\partial W)\cong \ZZ.
    \end{equation*}
    \item For $I\sqcup K=\set{1,\ldots,\PLm}$ in a quadratic pinch ($\PLm\leq n$), the boundaries are isomorphisms in the lowest degree $n-\abs{I}$:
    \begin{equation*}
        \partial_K\colon H_{n-\abs{I}}(\overline{W}\cap S^I,\partial W\cup S_K) \xrightarrow{\cong} H_{n-\PLm}(\overline{W}\cap S^{1,\ldots,\PLm},\partial W).
    \end{equation*}
\end{enumerate}
\begin{proof}
The first claim above is dual to \cref{lem:H(S^I-S_J.S_K)=0}. That proof via long exact sequences used only that $H_{0}(W\cap S^{I+p})\rightarrow H_{0}(W\cap S^{I})$ is an isomorphism. So dually, one only has to check that the intersection 
$\varpi\colon H_{\bullet}(\overline{W}\cap S^{I},\partial W)\rightarrow H_{\bullet-2}(\overline{W}\cap S^{I+p},\partial W)$
with $S_p$ is an isomorphism. This is clear, since these groups $\cong\ZZ$ are generated by the respective fundamental classes.

The second claim is the dual of \cref{lem:codomain-Knon0-full}; the isomorphisms follow just from the existence of the exact sequences. All that is left is to identify the groups at a single slot, like $I=\set{1,\ldots,\PLm-1}$, $J=\set{\PLm}$ and $K=\varnothing$. For a linear pinch, $\overline{W}\cap S^{1,\ldots,n}\setm S_{n+1}=\set{0}$ is a point with homology $\ZZ$. For a quadratic pinch, consider the residue sequence for $S_m$. The intersection
\begin{equation*}
    \varpi_m\colon \ZZ\cong H_{\bullet+2}(\overline{W}\cap S^{1,\ldots,\PLm-1},\partial W) \hookrightarrow H_{\bullet}(\overline{W}\cap S^{1,\ldots,\PLm},\partial W)
\end{equation*}
is injective, sending the generator (fundamental class) to the (non-zero) fundamental class $[\overline{W}\cap S^{1,\ldots,\PLm}]$. Hence the coboundary
\begin{equation*}
    \fibSphere_m\colon H_{\bullet}(\overline{W}\cap S^{1,\ldots,\PLm},\partial W) \twoheadrightarrow H_{\bullet+1}(\overline{W}\cap S^{1,\ldots,\PLm-1}\setm S_m,\partial W)
\end{equation*}
is surjective and identifies the image with the quotient of $H_{\bullet-n+\PLm}(\Sphere^{n-\PLm})$ by $\ZZ [\overline{W}\cap S^{1,\ldots,\PLm}]$, which cancels degree $\bullet=2n-2\PLm$ and leaves only $H_0(\Sphere^{n-m})\cong\ZZ$ in degree $\bullet=n-\PLm$ for $\PLm<n$. The case $\PLm=n$ where $\Sphere^0=\set{P_+,P_-}$ consists of two points is slightly different: Here we reduce $\ZZ^2\cong H_0(\Sphere^0)$ modulo the sum $\varpi_{\PLm}[\overline{W}\cap S^{1,\ldots,\PLm-1}]=[P_+]+[P_-]$, leaving $\ZZ$.

The third claim is a very similar dualization, of \cref{cor:Kzero-isos}.
\end{proof}
\begin{rem}
    Modulo torsion, these claims follow immediately from a generalization of Poincar\'{e}-Lefschetz duality to transverse arrangements, spelled out in \cref{thm:verdier-duality}.
\end{rem}

\subsection{Variation of a quadric}\label{sec:quadratic-pinch-JKzero}
The well-known classical Picard-Lefschetz formula (e.g.\ \cite[\S(5.3.3)]{Lamotke:HomIsoSing} or \cite[\S(6.3.3)]{Lamotke:TopVarLef}) computes the variation of a quadratic pinch in the case $J=K=\varnothing$, i.e., for the complete intersection of all $I=\set{1,\ldots,\PLm}$. We can ignore the coordinates $x_1=\ldots=x_{\PLm-1}=0$, so
\begin{equation*}
    S^I_t \cong Q_t = \set{x_1^2+\ldots+x_{r+1}^2=t} \subset \CC^{r+1}
\end{equation*}
is a family of complex quadrics $Q_t$ with dimension $r=n-\PLm$. They are localized $W\cap S^I_t\cong \oBall\cap Q_t$ in the ball $\oBall=\set{\abs{x_1}^2+\ldots+\abs{x_{r+1}}^2<1}\subset \CC^{r+1}$.
For $0<\abs{t}<1$, the embedded vanishing sphere
\begin{equation}\label{eq:vanishing-sphere}
    \Sphere^{r} \hookrightarrow \oBall\cap Q_t
    ,\quad
    u \mapsto x=u\cdot \sqrt{t}
\end{equation}
is a deformation retract $H_{\bullet}(\oBall\cap Q_t) \cong H_{\bullet}(\Sphere^{r})$, and its fundamental class $\vcyc=\dvcyc = \pm[\Sphere^r]\in H_r(\oBall\cap Q_t)$ provides the vanishing cycle.
In polar coordinates $t=\rho\cdot e^{2\ipi\tau}$, the embedding \eqref{eq:vanishing-sphere} extends (for $\rho>0$) to a diffeomorphism
\begin{equation}\label{eq:tangent-bundle-vanishing-sphere}\begin{split}
    T\Sphere^{r}=\set{(u,v)\in\Sphere^{r}\times\RR^{r+1}\colon u\cdot v=0}
    \xrightarrow{\quad\cong\quad} Q_t \subset \CC^{r+1}, \\
    (u,v)
    \mapsto
    %\sqrt{\tfrac{t}{\abs{t}}}
    %\sqrt{t/\abs{t}}
    %\cdot \Big(u\sqrt{\abs{t}+\norm{v}^2\cdot \tfrac{1-\abs{t}}{2}}+\iu v\sqrt{\tfrac{1-\abs{t}}{2}}\Big)
    e^{\ipi\tau}
    \cdot \Big(u\sqrt{\rho+\norm{v}^2\cdot \tfrac{1-\rho}{2}}+\iu v\sqrt{\tfrac{1-\rho}{2}}\Big)
\end{split}\end{equation}
of the quadric with the tangent bundle of the vanishing sphere. Restricted to the disk bundle $\overline{V}=\set{\norm{v}\le 1} \subset T\Sphere^{r}$ and its boundary, the sphere bundle $\partial V=\set{\norm{v}=1}$, we obtain local trivializations of the pairs
\begin{equation}\label{eq:quadric-thom-iso}
    (\cBall\cap Q_t,\partial\oBall\cap Q_t)\cong (\overline{V},\partial V).
\end{equation}

The corresponding isotopy $g_{\tau}\colon Q_{\rho}\cong Q_{t}, x \mapsto e^{\ipi\tau}\cdot x$ over the path $\gamma(\tau)=\rho\cdot e^{2\ipi\tau}$ ends up in the antipode $g_1(x)=-x$, which is not the identity on $\partial \oBall$. Hence we compose it with the diffeomorphisms $\Phi_{\tau}\colon T\Sphere^r\longrightarrow T\Sphere^r$ defined by
\begin{equation*}
    \Phi_{\tau}(u,v) =\big(u\cos\theta+(v/\norm{v})  \sin\theta,\; -u\,\norm{v}\sin\theta+v\cos\theta\big)\big|_{\theta=\pi\tau\cdot \norm{v}}
\end{equation*}
which rotate in the plane spanned by $u$ and $v$. By continuity in $\tau$, the map $\Phi_{\tau}$ is homotopic to the identity $\Phi_0=\id$, and therefore $g_{\tau}$ is homotopic to $g'_{\tau}=g_{\tau}\circ \Phi_{\tau}$. The angle $\theta$ of rotation decreases from $\pi\tau$ on the boundary $\partial V$, down to $0$ at the zero section $\Sphere^{r}=\set{v=0}$.
The resulting modified isotopy (\cref{fig:spheremap})
\begin{equation*}
    g_1'(u,v) = -\big(u\cos\theta+(v/\norm{v})\sin\theta,\; -u\norm{v}\sin\theta+v\cos\theta\big)\big|_{\theta=\pi\cdot\norm{v}}
\end{equation*}
restricts to the identity map on $\partial V=\set{\norm{v}=1}$. 
\begin{figure}
    \centering
    \includegraphics[width=0.7\textwidth]{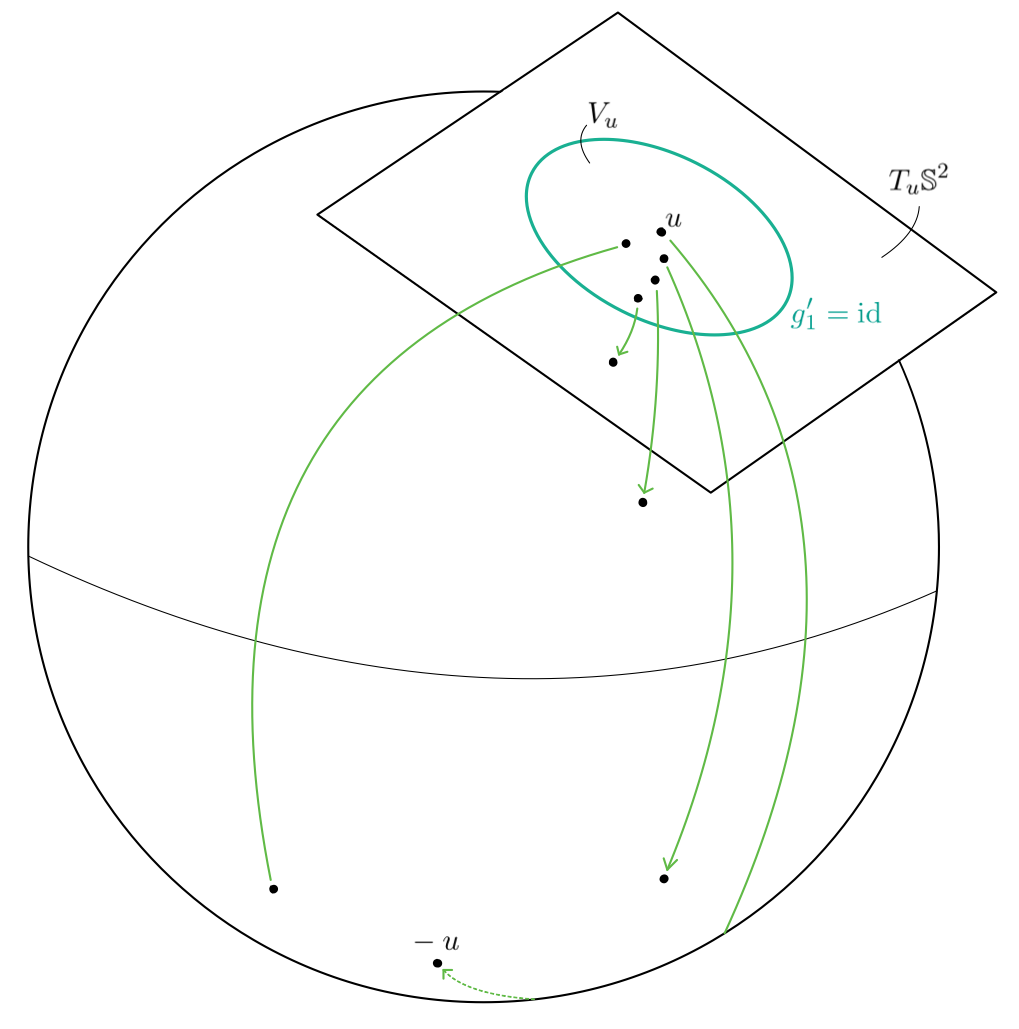}%
    \caption{Illustration of the map $v\mapsto g'_1(u,v)$ for $r=2$.}%
    \label{fig:spheremap}%
\end{figure}
In fact, $g_{\tau}'$ trivializes the sphere bundle globally: The isomorphism $g_{\tau}'|_{\partial V} \colon \partial V\rightarrow Q_t\cap\partial\oBall$ is given by
\begin{align*}
    (u,v) &\mapsto 
    e^{\iu\theta}\Big( [u\cos\theta+v\sin\theta]\sqrt{\tfrac{1+\rho}{2}}+\iu [v\cos\theta-u\sin\theta]\sqrt{\tfrac{1-\rho}{2}}\Big)\Big|_{\theta=\pi\tau}
    \\
    %&= \frac{u+\iu v}{\sqrt{2}}\sum_{k=0}^{\infty} \binom{1/2}{2k}\rho^{2k} + \frac{u-\iu v}{\sqrt{2}} t \sum_{k=0}^{\infty}\binom{1/2}{2k+1} \rho^{2k}
    %&= \frac{u+\iu v}{\sqrt{2}}\frac{\sqrt{1+\rho}+\sqrt{1-\rho}}{2} + \frac{u-\iu v}{\sqrt{2}} \frac{t}{\sqrt{1+\rho}+\sqrt{1-\rho}} %t \frac{\sqrt{1+\rho}-\sqrt{1-\rho}}{2\rho}
    &= \left(\frac{u+\iu v}{\sqrt{2}} \frac{R}{2} + \frac{u-\iu v}{\sqrt{2}} \frac{t}{R}\right)_{R=\sqrt{1+\rho}+\sqrt{1-\rho}}
\end{align*}
and therefore well-defined and smooth also over $t=0$, where it identifies $\partial V\cong Q_0\cap\partial\oBall$ as $(u,v)\mapsto (u+iv)/\sqrt{2}$.
Hence, we can use $g_1'$ to compute the localized variation. With the identification \eqref{eq:quadric-thom-iso}, the domain of
\begin{equation*}
    \var=(g_1')_*-\id \colon\quad H_r(\overline{V},\partial V)\rightarrow H_r(V)
\end{equation*}
is $H_{r}(\overline{V},\partial V)\cong H_0(\Sphere^{r})$ by the Thom isomorphism of the bundle $\mu\colon \overline{V}\rightarrow \Sphere^r$, $(u,v)\mapsto u$. So for $r\geq 1$, $H_{r}(\overline{V},\partial V)\cong\ZZ$ is generated by the class of a fibre $\mu^{-1}(p)$. Over the point $p=(0,\ldots,0,1)\in\Sphere^r$, this class $h=j_* [\cBallD]$ is represented by the $r$-dimensional disk $\cBallD=\set{\norm{w}\leq 1}\subseteq \RR^{r}$, embedded as
\begin{equation*}%}\label{eq:disk-fibre-embedding}\tag{$\natural$}
    j\colon
    (\cBallD,\partial\oBallD) \hookrightarrow (\overline{V},\partial V),\quad
    w \mapsto (p,(w,0)) \in T_{p}\Sphere^{r}.
\end{equation*}
Because the retraction $\mu$ induces an isomorphism $\mu_*\colon H_r(V)\cong H_r(\Sphere^r)\cong\ZZ$, we can compute the variation by pushing forward along $\mu$:
\begin{equation*}
    \mu_* (\var h) = \mu_* (g_1')_* j_* [\cBallD] - \mu_* j_* [\cBallD].
\end{equation*}
Since $\mu j\colon \cBallD\rightarrow \Sphere^r$ is the constant map $w\mapsto p$, the second term is zero in $H_r(\Sphere^r)\cong H_r(\Sphere^r,\set{p})$. The remaining term is the push-forward under
\begin{equation*}
    \mu g_1' j\colon (\cBallD,\partial\oBallD)\rightarrow (\Sphere^r,\set{p}),\quad 
    w \mapsto 
    %\frac{w}{\norm{w}}\sin(\pi\norm{w})-p\cos(\pi\norm{w})
    \left(-\frac{w}{\norm{w}}\sin(\pi\norm{w}),\;-\cos(\pi\norm{w})\right).
\end{equation*}
This map is bijective on the interior $\oBallD\cong \Sphere^r\setm\set{p}$ and realizes a homeomorphism $\cBallD/\partial\oBallD\cong\Sphere^r$. Therefore, we proved that
\begin{equation}\label{eq:pic-quadric}
    \var\colon H_r(\overline{V},\partial V)\cong\ZZ \longrightarrow H_r(V)\cong \ZZ,\quad h \mapsto \pm\vcyc
\end{equation}
is an isomorphism. To fix the sign, we endow $\oBallD\subset\RR^r$ with the orientation $\td w_1\wedge\ldots\wedge \td w_r$. Under the differential of $u=\mu g_1' j(w)$ at $w=0$, this becomes $(-1)^r\td u_1\wedge\ldots\wedge\td u_r$ at $T_{-p}\Sphere^r$. In terms of the standard orientation\footnote{The standard orientation is induced by the ambient volume form $\td u_1\wedge\ldots\wedge \td u_{r+1}$ of $\RR^{r+1}$, by contraction with the outward normal vector. At $p$, the outward normal is $\partial /\partial u_{r+1}$, hence the standard orientation of $T_p \Sphere^r$ is $(-1)^r \td u_1\wedge\ldots\wedge \td u_r$. At $-p$, the outward normal is $-\partial/\partial u_{r+1}$ and thus the standard orientation of $T_{-p}\Sphere^r$ is $(-1)^{r+1}\td u_1\wedge\ldots\wedge \td u_r$.} on the vanishing sphere $\vcyc=[\Sphere^r]$, this gives $\var h=-\vcyc$ in \eqref{eq:pic-quadric}.
The disk $j(\cBallD)$ and the vanishing sphere $\Sphere^r=\set{v=0}$ intersect in precisely one point, namely $(p,0)\in T_p\Sphere^r$. According to \cref{sec:verdier+intersection}, the intersection number
\begin{equation*}
    \is{\vcyc}{h}=\is{[\Sphere^r]}{[j(\cBallD)]} = \epsilon_p = \pm 1
\end{equation*}
is determined by comparing the orientation of the complex manifold $Q_t$ with the orientation of $T_{(p,0)} j(\cBallD)\oplus T_{(p,0)} \Sphere^r$. With our choices above, the latter is
\begin{equation*}\label{eq:quadric-is-orient-sum}\tag{$\sharp$}
    \td v_1\wedge\ldots\wedge\td v_r \wedge (-1)^r \td u_1\wedge\ldots\wedge\td u_r.
\end{equation*}
In terms of real and imaginary parts $x_k=\Re x_k+\iu \Im x_k$, the orientation of $T_{p\sqrt{t}} Q_t$ induced from the complex structure is $\td \Re x_1\wedge\td \Im x_1\wedge\ldots\wedge\td\Re x_r\wedge\td\Im x_r$. In the parametrization \eqref{eq:tangent-bundle-vanishing-sphere}, this is equal to
\begin{equation*}\label{eq:quadric-is-orient-complex}\tag{$\natural$}
    \td u_1\wedge \td v_1\wedge\ldots\wedge\td u_r\wedge \td v_r \cdot \big(\sqrt{\abs{t}(1-\abs{t})/2}\big)^r.
\end{equation*}
The orientations \eqref{eq:quadric-is-orient-sum} and \eqref{eq:quadric-is-orient-complex} differ by the sign $\epsilon_p=(-1)^r(-1)^{r(r+1)/2}$. We therefore arrive at the Picard-Lefschetz formula
\begin{equation*}
    \var h = -\vcyc = -(-1)^r(-1)^{r(r+1)/2}\is{\vcyc}{h} \vcyc
    =(-1)^{(r+1)(r+2)/2}\is{\vcyc}{h}\vcyc.
\end{equation*}
\begin{rem}
    If $r=0$, then $V=\overline{V}=\Sphere^0=\set{p,-p}$ consists of two points and $\partial V=\varnothing$. In this case, $\var\colon H_0(\Sphere^0)\rightarrow H_0(\Sphere^0)$ is a map from $\ZZ^2$ to $\ZZ^2$. The variations $\var [p] = -\nu$ and $\var[-p]=\nu=[p]-[-p]$ from \cref{eg:pureA-nm=11} are indeed compatible with $\var h = (-1)^{1\cdot 2/2}\is{\nu}{h} \nu$.
\end{rem}

\section{Partial boundaries and relative residues}
\label{sec:relative-residues}

\subsection{Partial boundaries}
Given any two topological subspaces $B,S\subseteq X$, we can consider the triple $B\subseteq B\cup S \subseteq X$ and its exact sequence
\begin{equation*}
    \cdots \rightarrow H_{\bullet}(X,B) \rightarrow H_{\bullet}(X,B\cup S) \xrightarrow{\partial} H_{\bullet-1}(B\cup S,B)\rightarrow H_{\bullet-1}(X,B)\rightarrow\cdots
\end{equation*}
The inclusion of pairs $(S,B\cap S) \rightarrow (B\cup S,B)$ induces a morphism
\begin{equation}\label{eq:(S,B)->(B+S,B)}%
    H_{\bullet}(S,B\cap S) \longrightarrow H_{\bullet}(B\cup S, B).
\end{equation}
\begin{defn}\label{defn:transverse}
    A finite collection $(S_i)_{i \in I}$ of smooth closed submanifolds $S_i\subset X$ is \emph{transverse} if, for every subset $J\subseteq I$ and at every point $p\in S^J$ of the intersection $S^J=\bigcap_{i\in J}S_i$, the canonical projections $N_p S^J \longrightarrow N_p S_i$ assemble into isomorphisms $N S^J \cong \bigoplus_{i \in J} N S_i|_{S^J}$ of the normal bundles.
\end{defn}
By the implicit function theorem, a transverse collection admits local coordinates $(x,(y_i)_{i \in J})\colon U\longrightarrow \RR^d\times\prod_{i\in J} \RR^{r_i}$ near any point $p\in S^J$ such that
$S_i \cap U = \set{y_i=0}$ for all $i\in J$, where $r_i=\codim S_i$ and $d=\dim S^J$. If each $S_i$ is a complex hypersurface ($r_i=2$) in a complex manifold, a transverse collection is also called \emph{simple normal crossings divisor}.
\begin{prop}\label{lem:partial-boundary-iso}%
    If $S=S_1$ and $B=\bigcup_{i>1} S_i$ form a transverse collection of smooth closed submanifolds $S_i \subset X$, then \eqref{eq:(S,B)->(B+S,B)} is an isomorphism.%
\end{prop}
\begin{proof}
    Pick a tubular neighbourhood $u\colon N\hookrightarrow X$ of $S$, that is, a smooth embedding of the normal bundle $N=TX|_S/TS$ such that $u\iota=\id_S$ for the zero section $\iota\colon S\hookrightarrow N$. Let $U=u(N)$, then excise $B\setm U$ to get  
    \begin{equation*}
        H_{\bullet}(B\cup S,B) \cong H_{\bullet}(S\cup (B\cap U),B\cap U).
    \end{equation*}
    Contraction $F_{\lambda}(v)=(1-\lambda)v$ of the fibres $N_{\pi(v)}$ shows that the projection $\pi=F_1\colon U\cong N \longrightarrow S$ is a deformation retract.
    Using the transversality assumption we may construct $u$ in such a way that $F_{\lambda}(U\cap B) \subseteq U\cap B$ for all $\lambda$, see \cite[Lemma~4.3]{TehraniMcLeanZinger:NCsymp} or \cite[Corollary~4.1]{Guadagni:SymplecticCrossing}. Then $F_{\lambda}$ restricts to homotopies on $B\cap U$ and $S\cup (B\cap U)$, which proves that
    \begin{equation*}
        \pi|_{B\cap U}\colon B\cap U \longrightarrow B\cap S
        \quad\text{and}\quad
        \pi|_{S\cup(B\cap U)} \colon S\cup(B\cap U)\longrightarrow S
    \end{equation*}
    are deformation retracts and induce isomorphisms $H_{\bullet}(B\cap U)\cong H_{\bullet}(B\cap S)$ and $H_{\bullet}(S\cup(B\cap U)) \cong H_{\bullet}(S)$. To conclude, apply the five lemma to the long exact sequences of the pairs $(S\cup(B\cap U),B\cap U)$ and $(S,B\cap S)$.
\end{proof}

In this proof, we exploited the fact that one may choose tubular neighbourhoods $u_i\colon TX|_{S_i}/T S_i \hookrightarrow X$ for all $S_i$ in a mutually compatible way, such that $u_i^{-1}(S_j)=T S_j|_{S_i\cap S_j}/T S_i|_{S_i\cap S_j}$.
We will repeatedly use such compatible systems of tubes in this appendix. They are constructed in the references given in the proof above, see also \cite[\S 7]{Mather:Stability}, and illustrated in \cref{fig:tubes}.
\begin{defn}
    In the situation of \cref{lem:partial-boundary-iso}, the \emph{boundary in $S$} is the map $\partial_S\colon H_{\bullet}(X,B\cup S)\longrightarrow H_{\bullet-1}(S,B\cap S)$ obtained from the connecting morphism $\partial$ of the triple $B\subseteq B\cup S \subseteq X$ under the isomorphism \eqref{eq:(S,B)->(B+S,B)}:
    \begin{equation*}\xymatrix{
        H_{\bullet}(X,B\cup S) \ar[r]^{\partial} \ar[dr]_{\partial_S} & H_{\bullet-1}(B\cup S,B) \\
        & H_{\bullet-1}(S,B\cap S) \ar[u]_{\cong}
    }\end{equation*}
\end{defn}
Now the exact sequence of the triple $B\subseteq B\cup S \subseteq X$ takes the form
\begin{equation*}
    \cdots \rightarrow H_{\bullet}(X,B) \rightarrow H_{\bullet}(X,B\cup S) \xrightarrow{\partial_S} H_{\bullet-1}(S,B\cap S)\rightarrow H_{\bullet-1}(X,B)\rightarrow\cdots
\end{equation*}

\subsection{Relative residues}
If we replace $S$ by its complement $S^\setcompl=X\setm S$ and consider instead the triple $B\subseteq B\cup S^\setcompl \subseteq X$, we get the sequence
\begin{equation*}
    \cdots
    \rightarrow H_{\bullet}(X,B)
    \rightarrow H_{\bullet}(X,B\cup S^\setcompl) %(X,B\cup(X\setm S))
    \xrightarrow{\partial} H_{\bullet-1}(B\cup S^\setcompl,B) %(B\cup(X\setm S),B)
    \rightarrow H_{\bullet-1}(X,B)
    \rightarrow \cdots
\end{equation*}
The inclusion of pairs $(S^\setcompl, B\setm S) \rightarrow (B\cup S^\setcompl, B)$ induces a morphism
\begin{equation}\label{eq:(X-S,B-S)->(X-S+B,B)}
    H_{\bullet}(X\setm S, B\setm S) \longrightarrow H_{\bullet}(B\cup S^\setcompl, B).
\end{equation}
\begin{prop}\label{prop:(X-S.B-S)excision}
    If $S=S_1$ and $B=\bigcup_{i\neq1} S_i$ form a transverse collection of smooth closed submanifolds $S_i\subset X$, then \eqref{eq:(X-S,B-S)->(X-S+B,B)} is an isomorphism.
\end{prop}
\begin{proof}
The zero section $S_i \hookrightarrow U_i$ of any tubular neighbourhood $N S_i \cong U_i\subseteq X$ induces an isomorphism $H_{\bullet}(S_i) \cong H_{\bullet}(U_i)$. Following \cite{TehraniMcLeanZinger:NCsymp,Guadagni:SymplecticCrossing} we may assume the compatibility that every intersection $U^J = \bigcap_{i \in J} U_i$ is a tubular neighbourhood of $S^J$ (\cref{fig:tubes}), so that $H_{\bullet}(S^J)\cong H_{\bullet}(U^J)$. 
\begin{figure}
    \centering
    \includegraphics[width=0.8\textwidth]{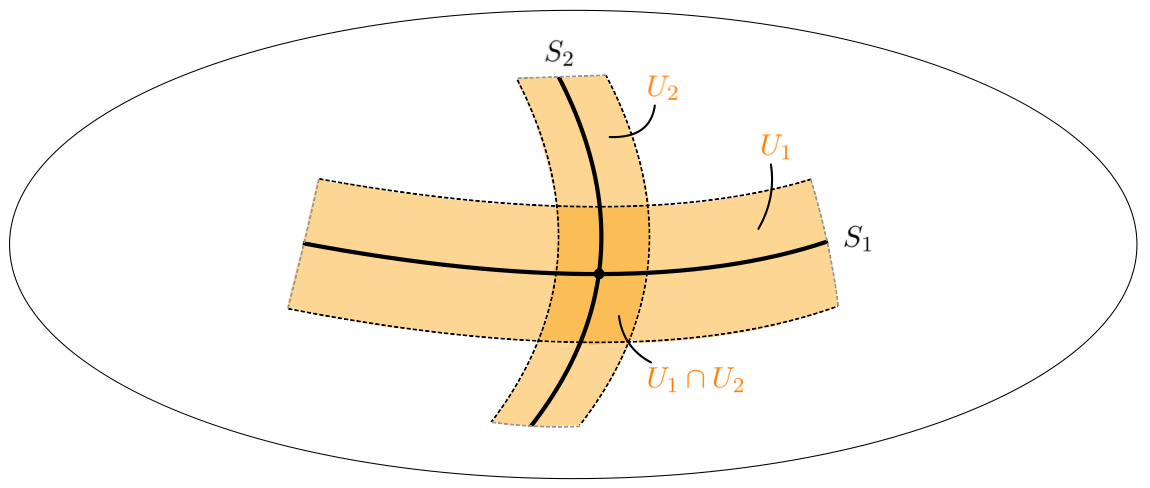}%
    \caption{Tubular neighbourhoods of transverse submanifolds.}%
    \label{fig:tubes}%
\end{figure}
Then also the inclusion of $B$ into $U \defas \bigcup_{i\neq1} U_i$ induces an isomorphism
$ %$\begin{equation*}
H_{\bullet}(B)\cong H_{\bullet}(U),    
$ %\end{equation*}
which follows from the Mayer-Vietoris sequence and the five lemma via induction over the number $\abs{I}-1$ of components of $B$. Analogously, the inclusions $B\setm S\hookrightarrow U\setm S$ and $B\hookrightarrow B\cup(U\setm S)$ induce isomorphisms
\begin{equation*}
    H_{\bullet}(B\setm S) \cong H_{\bullet}(U\setm S)
    \quad\text{and}\quad
    H_{\bullet}(B)\cong H_{\bullet}(B\cup(U\setm S))
    ,
\end{equation*}
because compatibility with $S_1$ ensures that for $1\notin J$, the projections $U^J\rightarrow S^J$ map $U^J\cap S$ to $S^J\cap S$ and restrict to deformation retracts $U^J\setm S\rightarrow S^J\setm S$ and $U^J\cap(B\cup S^\setcompl) \rightarrow S^J$. These thickenings of $B\setm S$ and $B$ induce $H_{\bullet}(S^\setcompl, B \setm S)\cong H_{\bullet}(S^\setcompl,U\setm S)$ and
$
    H_{\bullet}(B\cup S^\setcompl, B)
    \cong H_{\bullet}(B\cup S^\setcompl, B\cup(U\setm S))
$,
and these isomorphisms identify \eqref{eq:(X-S,B-S)->(X-S+B,B)} with the canonical map
\begin{equation*}
    H_{\bullet}(S^\setcompl, U \setm S)
    \rightarrow
    H_{\bullet}(B\cup S^\setcompl, B\cup(U\setm S)).
\end{equation*}
The interior of $B\cup(U\setm S)$ in $B\cup S^\setcompl$ contains $B\cap S$, thus the map above is an isomorphism by excision.
\end{proof}
To describe the group $H_{\bullet}(X,B\cup S^\setcompl)$, let $r=\codim S$ and consider a subbundle $\mu\colon V\rightarrow S$ of a tubular neighbourhood $NS\cong U\rightarrow S$ such that the fibres of $V\subset U$ are  isomorphic to closed discs $\cBall_r=\set{x\in \RR^r\colon\norm{x}\leq 1}$. Then we can excise the complement of $V$ to localize near $S$,
\begin{equation*}\label{eq:Gysin-excision}\tag{$\ast$}
    H_{\bullet}(X,B\cup S^\setcompl) \cong H_{\bullet}(V,(B\cap V) \cup (V\setm S)).
\end{equation*}
For a transverse collection $S\cup B$, recall that we can choose $U\rightarrow S$ such that it preserves $B$. Then $\mu^{-1}(B\cap S)=V\cap B=V|_{B\cap S}$ is itself a disc bundle, namely the restriction of $V$ to $B\cap S$.
\begin{defn}
    Suppose that the disc bundle $V\rightarrow S$ is oriented. Then the pullback bundle $\sigma^{\ast} V \cong \cBall_r\times \Delta$ over a singular simplex $\sigma\colon \Delta\rightarrow S$ inherits an orientation, and hence descends to a well-defined map on homology:\footnote{The maps \eqref{eq:fibre-in-discs} and \eqref{eq:fibre-in-spheres} only depend on the orientation of the normal bundle $NS$ (equivalently the orientation of $V$), but not on the choice of tubular neighbourhoods, since tubular neighbourhoods are unique up to isotopy.}
    \begin{equation}\label{eq:fibre-in-discs}
        \mu_S^{\ast}\colon H_{\bullet}(S,B\cap S) \longrightarrow
        %H_{\bullet+r}(V, (B\cap V) \cup (V\setm S))
        H_{\bullet+r}(X, B \cup S^\setcompl)
        .
    \end{equation}
    Similarly, the pullback of the sphere bundle $\partial V \rightarrow S$ induces a map
    \begin{equation}\label{eq:fibre-in-spheres}
        \fibSphere_S\colon H_{\bullet}(S,B\cap S) \longrightarrow H_{\bullet+r-1}(S^\setcompl, B\setm S)
    \end{equation}
    called the \emph{Leray coboundary}, such that $\partial (\mu_S^{\ast}[\sigma]) = \fibSphere_S[\sigma] + (-1)^r \mu_S^{\ast} (\partial[\sigma])$.
\end{defn}
\begin{prop}\label{prop:relative-thomiso}
    If $S=S_1$ and $B=\bigcup_{i\neq1} S_i$ form a transverse collection of smooth closed submanifolds $S_i\subset X$, and $S$ and $X$ are oriented, then the fibration-by-discs map \eqref{eq:fibre-in-discs} is an isomorphism.
\end{prop}
\begin{proof}
    The orientations of $S$ and $X$ induce an orientation of the disc bundle $V$, and the complement $V'\defas V\setm S$ of the zero section deformation retracts onto the sphere bundle $\partial V$. So if $B=\varnothing$, then the excision \eqref{eq:Gysin-excision} reduces the claim to the Thom isomorphism $H_{\bullet}(S)\rightarrow H_{\bullet+r}(V,V')$.\footnote{
        For this homological version of the Thom isomorphism, see e.g.\ \cite[Theorem~18.1.2]{tomDieck:AT} or \cite[Corollary~\S 10.7]{MilnorStasheff:CC}.
    } 
    %that follows from the relative Serre spectral sequence for the pair of fibre bundles $\partial V\subseteq V$ over $S$.
 
    For non-empty $B$, let $V_B\defas V\cap B=V|_{B\cap S}$ denote the restriction of $V$ to $B\cap S$ and consider the commutative diagram of the long exact sequences of the triple $V'\subseteq V'\cup V_B\subseteq V$ and the pair $B\cap S \subseteq S$:
    \begin{equation*}\xymatrix{
        \cdots \rightarrow H_{\bullet}(V,V') \ar[r] & H_{\bullet}(V,V' \cup V_B) \ar[r] & H_{\bullet-1}(V' \cup V_B,V') \rightarrow \cdots \\
        \cdots \rightarrow H_{\bullet-r}(S) \ar[r] \ar@<-15pt>[u]_{\cong} & H_{\bullet-r}(S,B\cap S) \ar[r] \ar[u]_{\mu_S^{\ast}} & H_{\bullet-1-r}(B\cap S) \ar[u]_{\nu} \rightarrow \cdots
    }\end{equation*}
    The vertical maps are induced by the pullbacks of the disc bundle $V\rightarrow S$ and its restriction $V_B\rightarrow B\cap S$. The claim follows from the five lemma because $\nu$ is an isomorphism. Indeed, by construction $\nu$ factors into
    \begin{equation*}
        H_{\bullet}(B\cap S)\rightarrow H_{\bullet+r}(V_B,V_B') \rightarrow H_{\bullet+r}(V_B\cup V', V')
    \end{equation*}
    where $V_B'\defas V_B\setm S = V_B\cap V'$ and the first map is the Thom isomorphism for the bundle $V_B$. The second map is an isomorphism because $V'$ and $V_B$ constitute an excisive couple, by the thickening principle from the proof of \cref{prop:(X-S.B-S)excision}: replace $V_B$ by an open neighbourhood $V\cap U \cap (V'\cup V_B)$ inside $V'\cup V_B$ with the same homology as $V_B$, obtained from a union $U$ of compatible tubular neighbourhoods of the components of $B$.
\end{proof}
Combining this result with \cref{prop:(X-S.B-S)excision}, we can write the long exact sequence of the triple $B\subseteq B\cup S^\setcompl \subseteq X$ for a transverse collection of oriented smooth closed submanifolds as the \emph{relative residue sequence}
\begin{equation*}
    \cdots \rightarrow H_{\bullet}(X,B)\xrightarrow{\varpi_S} H_{\bullet-r}(S,B\cap S) \xrightarrow{\fibSphere_S} H_{\bullet-1}(X\setm S,B\setm S)\rightarrow H_{\bullet-1}(X,B)\rightarrow \cdots
\end{equation*}
where the map $\varpi_S$ can be interpreted as ``transverse intersection with $S$''. It is defined by the commutative diagram
\begin{equation*}\xymatrix{
    %H_{\bullet}(X,B) \ar[r] \ar[dr]_{\varpi} & H_{\bullet}(X,B\cup S^\setcompl) \\ & H_{\bullet-r}(S,B\cap S) \ar[u]_{\mu^{\ast}_S}^{\cong}
    H_{\bullet}(X,B) \ar[r] \ar[dr]_{\varpi_S} & H_{\bullet}(X,B\cup S^\setcompl) \ar[r]^{\partial} & H_{\bullet-1}(B\cup S^\setcompl,B) \\ & H_{\bullet-r}(S,B\cap S) \ar[u]_{\mu^{\ast}_S}^{\cong} \ar[r]^{\fibSphere_S} & H_{\bullet-1}(X\setm S,B\setm S) \ar[u]^{\cong}_{\eqref{eq:(X-S,B-S)->(X-S+B,B)}} 
}\end{equation*}

\subsection{Functoriality and commutativity}\label{sec:coboundary-comm-functor}
Let $r_i$ denote the codimension of a submanifold $S_i\subset X$ in a transverse family. Then the maps $(\partial_{S_i},\fibSphere_{S_i},\varpi_{S_i})$ have homological degrees $(-1,r_i-1,-r_i)$, and for $i\neq j$ they are graded commutative with the same degrees:
\begin{equation}\begin{aligned}
    \partial_{S_i} \partial_{S_j}     &= -\partial_{S_j} \partial_{S_i} &
    \fibSphere_{S_i} \fibSphere_{S_j} &= (-1)^{(r_i-1)(r_j-1)} \fibSphere_{S_j} \fibSphere_{S_i} \\
    \partial_{S_i} \fibSphere_{S_j}   &= (-1)^{r_j-1} \fibSphere_{S_j} \partial_{S_i} &
    \fibSphere_{S_i} \varpi_{S_j}     &= (-1)^{(r_i-1)r_j} \varpi_{S_j} \fibSphere_{S_i} \\
    \partial_{S_i} \varpi_{S_j}       &= (-1)^{r_j} \varpi_{S_j} \partial_{S_i} &
    \varpi_{S_i} \varpi_{S_j}         &= (-1)^{r_i r_j} \varpi_{S_j} \varpi_{S_i}
\end{aligned}\end{equation}
The first identity amounts to $\partial^2=0$\footnote{Consider a single simplex $\Delta$. Its orientation induces orientations on its facets, which in turn induce orientations on their respective facets. These respective orientations must be opposite, in order to ensure $\partial^2=0$.} and the other signs compare the induced orientations on the disc and sphere bundles, see also \cite[Proposition~39]{Leray:CauchyIII}. Furthermore, suppose that a smooth map $F\colon X\longrightarrow X'$ restricts to submersions $S_i\longrightarrow S_i'\defas F(S_i)$ onto a transverse family $(S_i')_{i\in I}$ in $X'$. Then the maps commute with pushforward:
\begin{equation}\label{eq:[(co)boundary,pushforwad]}
    F_{\ast} \partial_{S_i} = \partial_{S_i'} F_{\ast},\quad
    F_{\ast} \fibSphere_{S_i} = \fibSphere_{S_i'} F_{\ast},\quad
    F_{\ast} \varpi_{S_i} = \varpi_{S_i'} F_{\ast}.
\end{equation}

\section{Relative intersection numbers}
\label{sec:verdier+intersection}

Consider a smooth compact oriented manifold $\overline{M}$ of dimension $n$, with boundary $\partial M=\overline{M}\setm M$ (possibly empty) and interior $M$. The cap product with the fundamental class $[M]\in H_n(\overline{M},\partial M)$ defines a homomorphism
\begin{equation*}
    \Psi\colon H^d(M)\longrightarrow H_{n-d}(\overline{M},\partial M)
    ,\quad \alpha\mapsto [M]\frown \alpha
%    \quad\text{and}\quad
%    H^d(\overline{M},\partial M) \longrightarrow H_{n-d}(M)
\end{equation*}
from singular cohomology, to homology in complementary degree. Poincar\'{e}-Lefschetz duality states that this is an isomorphism \cite[Theorem~6.25, p.~171]{Vick:HomIntroAT}. The duality $H_d(M)\times H^d(M)\rightarrow \ZZ$ can thus be viewed as a bilinear map
\begin{equation*}
    \is{\cdot}{\cdot}\colon H_{d}(M)\times H_{n-d}(\overline{M},\partial M) \longrightarrow \ZZ, 
    \quad
    \is{x}{y} = (\Psi^{-1}(y))(x).
\end{equation*}

These \emph{intersection numbers} can be computed as follows: If $x=[Q]$ and $y=[R]$ are fundamental classes of compact oriented submanifolds $Q$ and $R$ ($\partial Q=\varnothing$ and $\partial R\subseteq \partial M$) that intersect transversally, then $Q\cap R\subseteq M$ is also a compact oriented submanifold---in fact, it is a finite set of points. Then
\begin{equation*}
    \is{[Q]}{[R]} = \sum_{p \in Q\cap R} \epsilon_p,
\end{equation*}
with signs $\epsilon_p=\pm 1$ determined by the orientations. We follow the convention in \cite{Pham:Singularities}, namely that $\epsilon=+1$ precisely when the orientation of $M$ at $p$ agrees with the orientation of first $R$, then $Q$, under the identification
\begin{equation*}\label{eq:transverse-tangent-sum}\tag{$\flat$}
    T_pM=T_pR\oplus T_pQ.
\end{equation*}

When $x$ and $y$ are not presented by transverse submanifolds, one can perturb the singular simplices slightly to obtain homologous chains with transverse simplices, and apply the same recipe; cf.\ \cite[\S 73]{SeifertThrelfall:topology} and \cite[\S VII.4]{Dold:LecturesAT}.

If $\partial M=\varnothing$, we can swap $x$ and $y$. Then from \eqref{eq:transverse-tangent-sum} it follows that
\begin{equation}\label{eq:is-swap}
    \is{x}{y} = (-1)^{d(n-d)} \cdot \is{y}{x}.
\end{equation}
\begin{rem}
    The sign convention is more intuitive in terms of the normal bundles $NR=TM|_R/TR$ and $NQ=TM|_Q/TQ$. From \eqref{eq:transverse-tangent-sum} we see $N_p R\cong T_p Q$ and $N_p Q\cong T_p R$, hence $\epsilon_p=+1$ if the orientations on the normal bundles agree, under the canonical isomorphism
    \begin{equation*}\label{eq:transverse-normal-sum}\tag{$\natural$}
        N_p (Q\cap R) \cong N_p Q \oplus N_p R
    \end{equation*}
    with the normal bundles of \emph{first $Q$, then $R$} (the same order as in ``$\is{[Q]}{[R]}$'').
    Note that \eqref{eq:transverse-normal-sum} holds for transverse intersections of arbitrary dimensions $i=\dim Q$ and $j=\dim R$, whereas \eqref{eq:transverse-tangent-sum} only applies for $i+j=n$.
    The orientation of $Q\cap R$ induced by \eqref{eq:transverse-normal-sum} generalizes the intersection numbers to the \emph{intersection product} $[Q]\bullet [R]=[Q\cap R]\in H_{i+j-n}(M)$, see \cite[\S VIII.13]{Dold:LecturesAT}.
\end{rem}

\subsection{Transverse arrangements}
Inside a smooth compact oriented manifold $\overline{M}$ as above, consider a transverse arrangement $D=A\cup B=\bigcup_i D_i$ of smooth oriented closed submanifolds $D_i$. If $\partial M\neq\varnothing$, we also require that $D$ is transverse to $\partial M$. Then the intersection numbers define a bilinear pairing
\begin{equation}\label{eq:intersection-pairing}
    \is{\cdot}{\cdot}\colon H_{d}(M\setm B,A) \times H_{n-d}(\overline{M}\setm A,\partial M\cup B) \longrightarrow \ZZ.
\end{equation}

For transverse submanifolds $Q\subseteq M\setm B$, $R\subseteq\overline{M}\setm A$ with $\partial Q\subseteq A$ and $\partial R\subseteq B\cup\partial M$, their intersection is a finite set of points inside $M\setm D$, and $\is{[Q]}{[R]}$ can be computed as above. In general, the intersection pairing is constructed in \cite[\S VII.4]{Dold:LecturesAT}. In fact, we can define \eqref{eq:intersection-pairing} as the composition of:
\begin{itemize}
\item the K\"{u}nneth morphism for the product $Y=(M\setm B)\times(\overline{M}\setm A)$, relative to $D'=A\times (\overline{M}\setm A)\cup (M\setm B)\times(\partial M\cup B)$:
\begin{equation*}
    H_d(M\setm B,A)\otimes H_{n-d}(\overline{M}\setm A,\partial M\cup B)
    \longrightarrow 
    H_n(Y,D')
\end{equation*}
\item intersection with the diagonal $\Delta= \set{(p,p)\colon p\in M}\cap Y \cong M\setm D$:
\begin{equation*}
    H_n(Y,D') \xrightarrow{\varpi_{\Delta}} H_0(\Delta,D'\cap \Delta)
    \cong H_0(M\setm D) \cong \ZZ.
\end{equation*}
\end{itemize}

Here note that $D'\cap \Delta=\varnothing$, and the map $\varpi_{\Delta}$ from the Gysin/residue sequence is constructed in \cref{sec:relative-residues}: work relative to $D'\cup (Y\setm \Delta)$, excise all but a tubular neighbourhood of $\Delta$, and apply the Thom isomorphism.

\subsection{Duality}
Consider a submanifold $S\in\irrone{B}$ of real codimension $r$. Let $B'\subseteq B$ denote the union of the remaining $B_i\neq S$. Then the intersection pairings on $M$ and $S$ are compatible in the following sense: If $x\in H_{d}(S\setm B',A)$ and $y\in H_{n-d-r+1}(\overline{M}\setm A,\partial M\cup B)$, then
\begin{equation}\label{eq:intersection-delta-partial-general}
    \is{\fibSphere_S\, x}{y} =
    (-1)^{1+(r-1)(n-d)}
    %= (-1)^{r+(r-1)\abs{y}} 
    \cdot \is{x}{\partial_S\, y}.
\end{equation}
This is a simple generalization of \cite[\S III.2.3]{Pham:Singularities}. Equivalently, swapping the roles of $x$ and $y$ as in \eqref{eq:is-swap} when $S\in\irrone{A}$, we have for $x\in H_{d}(M\setm B,A)$ and $y\in H_{n-d-r+1}(\overline{S}\setm A',\partial S\cup B)$ the identity
\begin{equation*}
    \is{\partial_S\, x}{y}
    = (-1)^{1+d+n} % = (-1)^{r+\abs{y}}
    \cdot \is{x}{\fibSphere_S\, y}.
\end{equation*}
Hence in the case of a complex hypersurface $S\subset M$ ($r=2$ and $n$ even), we obtain for homology classes $x$ of degree $d$ that
\begin{equation}\label{eq:intersection-delta-partial}\begin{aligned}
    \is{\fibSphere_S\, x}{y}
    &= (-1)^{1+d} %(-1)^{\abs{\tau}}
    \cdot \is{x}{\partial_S\, y}
    \quad\text{and}\\
    \is{\partial_S\, x}{y}
    &= (-1)^{1+d} \cdot \is{x}{\fibSphere_S\, y}
    .
\end{aligned}\end{equation}

For an Abelian group $G$, let $G^{\vee}=\Hom(G,\ZZ)$ denote its dual. The pairing \eqref{eq:intersection-pairing} defines a map
$
    H_d(M\setm B,A) \rightarrow H_{n-d}(\overline{M}\setm A,\partial M\cup B)^{\vee}
$. Extend to rational coefficients, then the codomain is canonically isomorphic to singular cohomology, by the universal coefficient theorem. This defines a linear map
\begin{equation}\label{eq:verdier-duality}
    H_d(M\setm B,A; \QQ) \longrightarrow
    H^{n-d}(\overline{M}\setm A,\partial M\cup B;\QQ).
\end{equation}
\begin{thm}\label{thm:verdier-duality}
    The intersection pairing \eqref{eq:intersection-pairing} is perfect modulo torsion. In other words, \eqref{eq:verdier-duality} is an isomorphism.
    %\begin{equation*}
      %  H_d(M\setm B,A; \QQ) \xrightarrow{\cong} 
        %\Hom_{\QQ}( H_{n-d}(\overline{M}\setm A,\partial M\cup B;\QQ),\QQ).
     %   H^{n-d}(\overline{M}\setm A,\partial M\cup B;\QQ).
        %H_d(M\setm B,A) \longrightarrow \Hom_{\ZZ}( H_{n-d}(\overline{M}\setm A,\partial M\cup B),\ZZ)
    %\end{equation*}
\end{thm}
\begin{proof}
    When $A=B=\varnothing$, this is just a restatement of the Poincar\'{e}-Lefschetz theorem. %$H_d(M)\cong \Hom_{\ZZ}(H_{n-d}(\overline{M},\partial M),\ZZ)$ through the universal coefficient theorem, $H_{d}(M;\QQ)=H_d(M)\otimes \QQ$. 
    With this as a base case, the proof follows by induction over the number of elements of the arrangement $\irrone{D}=\irrone{A}\sqcup\irrone{B}$:
    
    Consider an element $S\in \irrone{B}$, so that $B=S\cup B'$ where $B'$ is the union of all $B_i$ in $\irrone{B}\setm\set{S}$. The residue and boundary sequences from \cref{sec:relative-residues} fit into the diagram
\begin{equation*}\xymatrixcolsep{2cm}\xymatrix{
    \ar[d]^{\varpi_S} & \ar[d]
    \\
    H_{d+1-r}(S\setm B',A)  \ar[r]^{\alpha} \ar[d]^{\fibSphere_S} %\ar@{}[dr]|{\boxed{\alpha}}
    & H_{n-d-1}(\overline{S}\setm A,\partial S\cup B')^{\vee} \ar[d]^{\partial_S^{\vee}}
    \\
    H_{d}(M\setm B,A)       \ar[r] \ar[d]
    & H_{n-d}(\overline{M}\setm A,\partial M\cup B)^{\vee} \ar[d]
    \\
    H_{d}(M\setm B',A)      \ar[r]^{\beta} \ar[d]^{\varpi_S}
    & H_{n-d}(\overline{M}\setm A,\partial M\cup B')^{\vee}  \ar[d]
    \\
    &
}\end{equation*}
    built out of the residue sequence in the left column, the dual of the boundary sequence in the right column, and the horizontal maps induced by the intersection pairings. We can add signs such that this diagram commutes: one of the squares is \eqref{eq:intersection-delta-partial-general}, and similarly one checks that the other squares also commute up to sign. We may assume that $\alpha\otimes\QQ$ and $\beta\otimes\QQ$ are isomorphisms, because they arise from intersection pairings of smaller arrangements $\irrone{D}\setm\set{S}$. The induction step then follows from the five-lemma.
    
    Analogously, we can add an element $S\in\irrone{A}$, by considering the diagram built from the boundary sequence (left column) and the dual of the residue sequence (right column).
\end{proof}

\subsection{de Rham cohomology}
Under the identification \eqref{eq:verdier-duality}, the intersection pairing can be viewed as a perfect pairing in singular cohomology,
\begin{equation*}
    H^{d}(M\setm B,A;\QQ) \times H^{n-d}(\overline{M}\setm A,\partial M\cup B;\QQ) \longrightarrow \QQ.
\end{equation*}
Consider now the special case when $M$ is a compact complex manifold without boundary, and $D\subset M$ is a transverse union of smooth complex hypersurfaces. Let $H_{\dR}^{\bullet}$ denote the cohomology of smooth complex differential forms. By the de Rham theorem $H_{\dR}^{\bullet}\cong H^{\bullet}\otimes\CC$, we obtain a perfect pairing
\begin{equation}\label{eq:intersection-dR}
    H_{\dR}^{d}(M\setm B,A) \times H_{\dR}^{n-d}(M\setm A, B) \longrightarrow \CC.
\end{equation}

This incarnation of \cref{thm:verdier-duality} is derived in \cite[\S 3]{BrownDupont:SVdoublecopy} from Verdier duality. Furthermore, by \cite[Proposition~3.10]{BrownDupont:SVdoublecopy} the pairing \eqref{eq:intersection-dR} equals the integral
\begin{equation*}
    ([\nu],[\omega]) \longmapsto \int_M \nu \wedge \omega. %\overline{\omega}.
\end{equation*}
Here the representative forms $\nu$ ($\omega$) have to be chosen such that they have at most simple poles on $B$ ($A$). This is always possible \cite[pp.~121--122]{Leray:CauchyIII}, and ensures that the integral converges. When $M$ and $D$ are algebraic over $\QQ$, then the pairing \eqref{eq:intersection-dR} also respects the $\QQ$-mixed Hodge structures on algebraic de Rham cohomology \cite{BrownDupont:SVdoublecopy}.

\section{Fundamental group and small loops}
\label{sec:fg-codim1}

Consider a closed analytic subvariety $L\subsetneq T$ of a connected complex manifold $T$.
Let $\irrone{L}$ denote the subset of irreducible components of $L$ which have complex codimension one. Their union forms a subset
$
    L'=\bigcup_{\ell\in\irrone{L}} \ell
    \subseteq L
$.
\begin{prop}\label{lem:fg-codim1}
Fix a point $t\in T\setm L$ and take all fundamental groups in the sequel relative to $t$ as the base point.
Then:
    \begin{enumerate}
        \item 
    The inclusion $T\setm L \hookrightarrow T\setm L'$ induces an isomorphism
    \begin{equation*}
        \pi_1(T\setm L) \xrightarrow{\;\cong\;} \pi_1(T\setm L').
    \end{equation*}
    \item The inclusion $T\setm L\hookrightarrow T$ induces a surjection
    \begin{equation*}
        \iota\colon \pi_1(T\setm L)\relbar\joinrel\twoheadrightarrow \pi_1(T).\qquad\quad
    \end{equation*}
    \item The kernel of $\iota$ is generated by simple loops (\cref{def:small-simple}).
    \end{enumerate}
\end{prop}
\begin{cor}
    The fundamental group of $T\setm L$ depends only on the codimension one part $L'$ of $L$.
    When $T$ is simply connected, then $\pi_1(T\setm L)$ is generated by simple loops.
\end{cor}
\Cref{lem:fg-codim1} follows from the existence of stratifications by smooth subvarieties (\cref{sec:stratsandlandau}) and transverse deformations. A smooth map $f\colon X\rightarrow Y$ of manifolds is called \emph{transverse} to a closed submanifold $D\subset Y$, if at every point $x\in X$ with $f(x)\in D$, the tangent space of $Y$ is spanned by the tangent space of $D$ and the image of the differential of $f$:
\begin{equation*}
    T_{f(x)} Y = T_{f(x)}D + \td f_x\left( T_x X\right).
\end{equation*}
Under this condition, $f^{-1}(D)$ is a closed submanifold of $X$, with codimension $\dim Y-\dim D$. In particular, if $\dim D+\dim X<\dim Y$, then transversality implies that the image of $f$ does not intersect $D$, that is $f(X)\subseteq Y\setm D$.
\begin{thm}[{Transversality theorem \cite[\S IV.Corollary (2.5)]{Kosinski:DiffMan}}]
    Given a smooth closed submanifold $D\subset Y$ and a smooth map $f\colon X\rightarrow Y$, there exists a smooth map that is homotopic to $f$ and transverse to $D$.
\end{thm}
\begin{cor}[{\cite[Th\'eor\`eme~2.3]{Godbillon:ElemATn}}]\label{cor:fg-complement}
    For a closed smooth submanifold $D\subset Y$ of real codimension $r\geq 2$, the morphism $\pi_1(Y\setm D)\rightarrow \pi_1(Y)$, induced by the inclusion $Y\setm D\hookrightarrow Y$, is surjective. Whenever $r\geq 3$, it is also injective.
\end{cor}
\begin{proof}
    Every loop $\gamma\colon [0,1]\rightarrow Y$ is homotopic to a loop $\gamma'$ that is transverse to $D$. Because $\dim D+\dim\, [0,1]=\dim Y-r+1$ and $r\geq 2$, $\gamma'$ does not intersect $D$, proving surjectivity.
    For injectivity, suppose $[\gamma]\in\pi_1(Y\setm D)$ is in the kernel, i.e.\ there is a contracting homotopy $F\colon[0,1]^2\rightarrow Y$ of $\gamma$ in $Y$. Choose $F$ transverse to $D$. Then since $\dim ([0,1]^2)=2$ and $r\geq 3$, the homotopy does not intersect $D$. Thus $\gamma$ is contractible already in $Y\setm D$.
\end{proof}
\begin{proof}[{Proof of \cref{lem:fg-codim1}}]
    The subvariety $L\setm L'$ of $Y= T\setm L'$ has complex codimension at least two (definition of $L'$). Consider a stratification $L\setm L'=C_{n-2}\supseteq\cdots\supseteq C_0$ by closed analytic subvarieties $C_k$, so that each difference $D_k=C_k\setm C_{k-1}$ is a smooth closed complex submanifold of $Y\setm C_{k-1}$, with real codimension $2n-2k\geq 4$.
    By \cref{cor:fg-complement}, inclusions induce isomorphisms
    \begin{equation*}\label{eq:pi1-strat-step}\tag{$\ast$}
        \iota_k\colon \pi_1(Y\setm C_k)\longrightarrow \pi_1(Y\setm C_{k-1})
    \end{equation*}
    for each $k\leq n-2$. Their composition $\iota_0\cdots\iota_{n-2}$ yields the isomorphism (1) between the fundamental groups of $Y\setm C_{n-2}=T\setm L$ and $Y=T\setm L'$.

    The claim (2) is proved analogously, by considering a stratification of $L=C_{n-1}\supseteq\cdots\supseteq C_0$ inside $Y=T$. As above, the maps \eqref{eq:pi1-strat-step} are isomorphisms for $k\leq n-2$. At $k=n-1$, we still have surjectivity (real codimension two).
    
    Given the isomorphisms $\pi_1(T\setm C_{n-2})\cong\pi_1(T\setm C_{n-3})\cong\cdots\cong \pi_1(T)$, we set $Y= T\setm C_{n-2}$ and identify the map $\iota$ in claim (3) with the morphism
    \begin{equation*}
        \pi_1(T\setm L)=\pi_1(Y\setm D)\longrightarrow \pi_1(Y)\cong \pi_1(T)
    \end{equation*}
    induced by the inclusion of the complement $T\setm L=Y\setm D$ of the smooth submanifold $D= L\setm C_{n-2}\subset Y$, which has real codimension two. We apply the proof of \cite[Proposition \S V.1.3]{Pham:Singularities}: For $[\gamma]\in\pi_1(Y\setm D,t)$ in the kernel of $\iota$, there is a homotopy $F\colon \square\rightarrow Y$ on the square $\square=[0,1]^2$ with $F(\tau,1)=F(0,\tau)=F(1,\tau)=t$ constant and $\gamma(\tau)=F(\tau,0)$. Choose $F$ transverse to $D$, so $P=\set{p_1,\ldots,p_k}=F^{-1}(D)$ is a submanifold---with dimension zero (a finite set of points). In the group $\pi_1(\square\setm P)$, the boundary of the square is a product of simple loops $\gamma_{i}$ 
    (\cref{fig:punctured-square}), hence, applying $F$,
    \begin{equation*}
        [\gamma] = \eta_{1}\conc\cdots\conc\eta_{k}
    \end{equation*}
    is a product of simple loops $\eta_{i}= [F\circ\gamma_{i}]\in\pi_1(Y\setm D,t)$ around $t_i=F(p_i)\in D$. Consider a holomorphic chart $(z_1,\ldots,z_d)$ of $Y$, centered at $t_i$, with locally $D=\set{z_1=0}$. By transversality of $F$, the map $q\mapsto z_1(F(q))$ is a local diffeomorphism and embeds a sufficiently small disk $\oBallD=\set{q\in\square\colon \norm{q-p}<\epsilon}$. As $\gamma_i$ is conjugate to the circle $\partial\oBallD$, the loop $\eta_i$ is conjugate to the closed curve $F(\partial\oBallD)=\partial(F(\oBallD))$. The latter winds once around $t_i=\set{z=0}= F(\oBallD)\cap L$, and it is homotopic to $\set{\abs{z_1}=\epsilon}\times\set{z_2=\cdots=z_d=0}\subset T\setm L$. Thus indeed, $\eta_i$ (or its inverse, depending on the orientation) represents the class of a simple loop in the sense of \cref{def:small-simple}.
\end{proof}
\begin{figure}
    \centering
    \includegraphics[width=\textwidth]{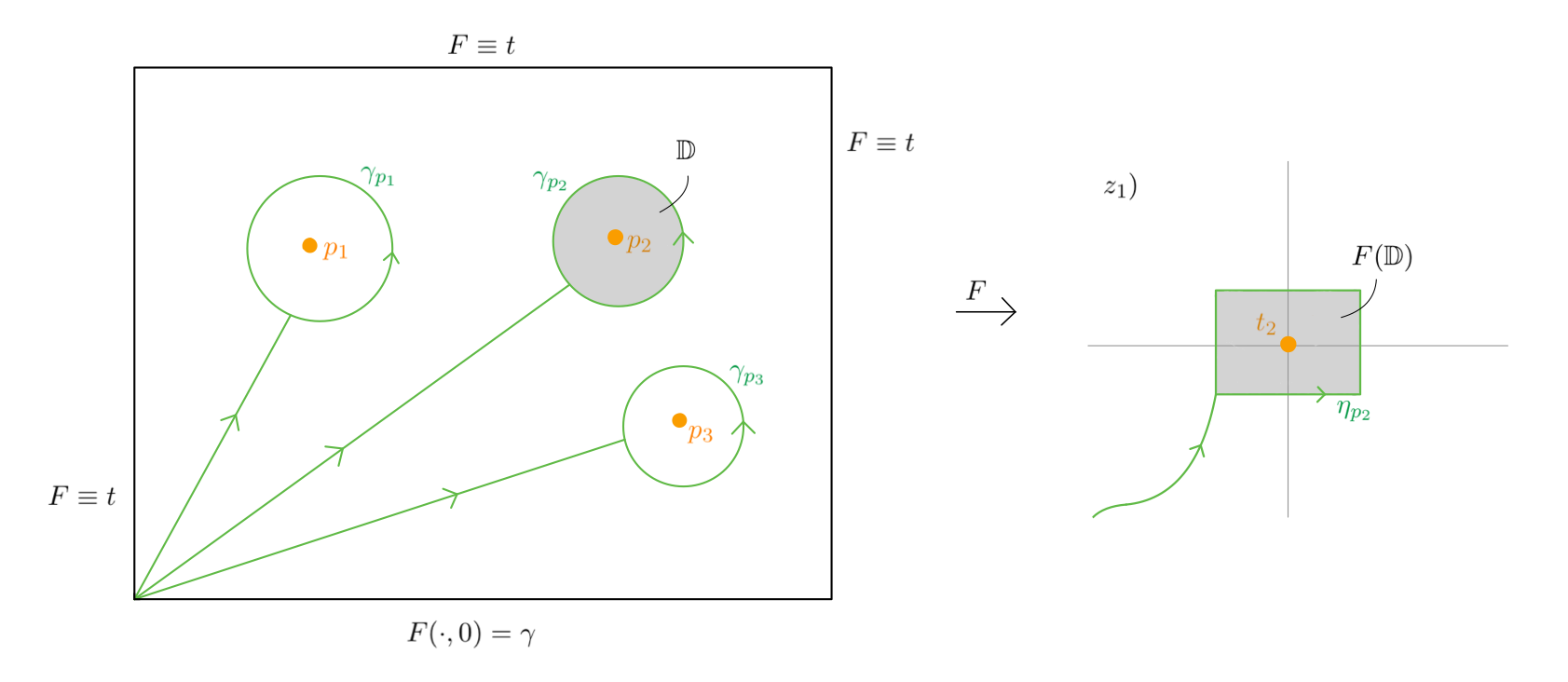}
    \caption{Decomposition of a homotopy into simple loops.}%
    \label{fig:punctured-square}%
\end{figure}
\begin{rem}
Instead of transversality, one can prove \cref{lem:fg-codim1} and \cref{cor:fg-complement} also with the Seifert--van Kampen theorem. Consider the open cover $Y=Y\setm D\cup N$ obtained from a tubular neighbourhood $N\supset D$. By induction over the connected components, we may assume that $D$ is connected. Then we can write the amalgamated product as
\begin{equation}\label{eq:fg-SvK}
    \pi_1(Y) \cong \pi_1(Y\setm D) *_{\pi_1(S)} \pi_1(D),
\end{equation}
exploiting the homotopy equivalences $N \simeq D$ and $N\setm D\simeq S$, where $S\subset N$ denotes a sphere bundle with fibre $\Sphere^{r-1}$. Since the fibre is connected ($r\geq 2$), the long exact homotopy sequence for the bundle $S$ ends in
\begin{equation*}
    \cdots\rightarrow \pi_1(\Sphere^{r-1})\rightarrow \pi_1(S) \xrightarrow{\alpha} \pi_1(D) \rightarrow 0.
\end{equation*}
The surjectivity of $\alpha$ renders the factor $\pi_1(D)$ in \eqref{eq:fg-SvK} superfluous, proving that $\iota\colon\pi_1(Y\setm D)\rightarrow \pi_1(Y)$ surjects. In consequence, the amalgamated product degenerates to a quotient
\begin{equation*}
    \pi_1(Y) \cong \pi_1(Y\setm D)/ \langle G \rangle,
\end{equation*}
where $\ker\iota=\langle G\rangle$ denotes the normal closure of the leftover relations, where $G$ denotes the image in $\pi_1(Y\setm D)$ of the kernel of $\alpha$. For $r\geq 3$, the groups $\pi_1(\Sphere^{r-1})=0$ and thus $\ker\alpha$ are trivial, so that $\iota$ is injective.

Finally, if $r=2$, $\pi_1(\Sphere^1)\cong\ZZ$ shows that $G$ is generated by a single small loop $\gamma$ around $D$. Thus the entire kernel $\langle G\rangle$ of $\iota$ is generated by the conjugates of $\gamma$.
This generation by a single element and its conjugates is explained also in \cite[Lemme~1.2]{Kervaire:NoeudsSup}.
\end{rem}
\begin{rem}
    Any two simple loops $\gamma$, $\gamma'$ around the \emph{same} component $\ell\in\irrone{L}$ are conjugate. Indeed, by definition,
    \begin{equation*}
        \gamma=\eta^{-1} \conc \partial \oBallD\conc\eta
        \quad\text{and}\quad
        \gamma'={\eta'}^{-1} \conc \partial \oBallD'\conc\eta'
    \end{equation*}
    are conjugate to the the boundary circles $\partial\oBallD\cong\Sphere^1$ of small holomorphic discs $\oBallD$, $\oBallD'$ embedded transverse to the smooth part $S=\ell\setm C_{n-2}$. We can shrink them to fit into a tubular neighbourhood $N\supset S$, so that $\partial\oBallD\simeq\pi^{-1}(t_c)$ and $\partial\oBallD'\simeq\pi^{-1}(t_c')$ are homotopic to fibres of the corresponding circle bundle $\pi\colon \partial N\rightarrow S$. Since $S$ is connected, there is a path $\rho\colon[0,1]\rightarrow S$ that connects $t_c$ and $t_c'$. The pullback $\rho^{-1}\pi \cong [0,1]\times\Sphere^1$ is a cylinder and supplies a homotopy $\partial\oBallD \simeq \partial\oBallD'$; as paths with basepoints, these circles are conjugate.
\end{rem}
The closure $S\mapsto\overline{S}$ gives a bijection between the connected components $S\subseteq D=L\setm C_{n-2}$ (strata of complex codimension one), and the Landau singularities $\irrone{L}$ (the irreducible components of $L'$).

For every $\ell\in\irrone{L}$, choose an arbitrary simple loop $\gamma_{\ell}$ around $\ell$. Then by the previous remark, their conjugates provide the set
\begin{equation*}
    H=\set{\eta^{-1}\conc\gamma_{\ell}\conc\eta\colon 
    \ell\in\irrone{L}\ \text{and}\ \eta\in\pi_1(T\setm L)}.
\end{equation*}
of \emph{all} simple loops.
Part (3) of \cref{lem:fg-codim1} states that the kernel of $\iota$ is the subgroup generated by $H$.
\begin{rem}
The subgroup generated by the simple loops $H'=\set{\gamma_{\ell}\colon \ell\in\irrone{L}}$ alone, without including conjugates, can be strictly smaller. For example, \cite[Theorem~2]{Poenaru:genFGcod2} constructs, among several others, a smooth embedding
\begin{equation*}
    L=\ell\sqcup\ell'\quad\hookrightarrow\quad Y=\Sphere^6
\end{equation*}
of two disjoint spheres $\ell\cong\ell'\cong\Sphere^4$ with the following property: For every choice of simple loops $\gamma_{\ell}$ and $\gamma_{\ell'}$, the group generated by $H'=\set{\gamma_{\ell},\gamma_{\ell'}}$ is strictly smaller than the kernel of $\pi_1(Y\setm L)\rightarrow \pi_1(Y)$.

Thus the monodromy representation is not necessarily completely determined by the variation along only one simple loop per Landau component.
\end{rem}

\end{document}